\documentclass[preprint,3p, times]{elsarticle}
\pdfoutput=1

\journal{Journal of Logical and Algebraic Methods in Programming}
\usepackage[T1]{fontenc}

\usepackage{graphicx}
\usepackage{subfigure}
\usepackage[table, dvipsnames]{xcolor}
\usepackage{amsmath,amssymb}
\usepackage{mathtools}
\usepackage{longtable}
\usepackage{booktabs}
\usepackage{multirow,multicol}
\usepackage{array}
\usepackage{fontawesome}

\usepackage{enumitem}

\usepackage{trimclip}

\usepackage{adjustbox}

\usepackage{icgtjlampextras}

\newcommand{\ti}[1]{%
 \ensuremath{\vcenter{\hbox{\includegraphics{diagrams/#1.pdf}}}}%
}

\colorlet{h1color}{blue!70!black} %
\colorlet{h2color}{orange!90!black} %
\colorlet{h3color}{blue!40!white} %
\colorlet{h4color}{green!40!black} %
\usepackage[scope, electronic]{knowledge} %

\knowledgeconfigure{notion}

\knowledgeconfigureenvironment{theorem,proposition,lemma,corollary,definition,proof}{knowledge=attracts} %

\knowledge{anchor point}[anchor points|Anchor points]{notion}
\knowledge{diagnose file}{notion}
\knowledge{knowledge package}{notion,typewriter}
\knowledge{paper mode}{notion}
\knowledge{electronic mode}{notion}
\knowledge{composition mode}{notion}

\knowledge{knowledge}{notion}
\knowledge{knowledges}{synonym}
\knowledge{Knowledges}{synonym}

\knowledge{LaTeX}[latex|LATEX|Latex]{url={https://fr.wikipedia.org/wiki/LaTeX},  text=\LaTeX}

\knowledge{pdflatex}{scope={section:example},url={https://en.wikipedia.org/wiki/PdfTeX}, typewriter}
\begin{scope}\label{package}
   \knowledge{knowledge}{link=knowledge package,link scope}
\end{scope}

\knowledge{notion}
 | Grothendieck fibration
 | Grothendieck fibrations

\begin{scope}\label{gf}

\knowledge{notion}
 | Cartesian morphism
 | Cartesian
 | Cartesianity

\knowledge{notion}
 | Cartesian lifting
 | Cartesian liftings

\end{scope}

\knowledge{notion}
 | Grothendieck opfibration
 | Grothendieck opfibrations

\begin{scope}\label{gopf}

\knowledge{notion}
 | op-Cartesian morphism
 | op-Cartesian
 | op-Cartesianity

\knowledge{notion}
 | op-Cartesian lifting
 | op-Cartesian liftings
 | op-Cartesian lifts

\end{scope}

\knowledge{notion}
 | multi-opfibration
 | multi-opfibrations

\begin{scope}\label{mof}

\knowledge{notion}
 | multi-op-Cartesian liftings
 | multi-op-Cartesian lifting
 | multi-op-Cartesian lifts

\knowledge{notion}
 | Universal property of multi-opfibrations
 | universal property of multi-opfibrations

\knowledge{notion}
 | Essential uniqueness
 | essential uniqueness

\knowledge{notion}
 | strong

\knowledge{notion}
 | lifting property of isomorphisms
 | lifting property of isomorphisms for strong multi-opfibrations
 | isomorphism lifting

\knowledge{notion}
 | Pullback-lifting lemma for strong multi-opfibrations
 | pullback lifting lemma for strong multi-opfibrations
 | pullback lifting

\end{scope}

\knowledge{notion}
 | residual multi-opfibration
 | residual multi-opfibrations

\begin{scope}\label{rmof}

\knowledge{notion}
 | residual multi-op-Cartesian liftings
 | residual multi-op-Cartesian lifting
 | residual multi-op-Cartesian lifts

\knowledge{notion}
 | residue
 | residues

\knowledge{notion}
 | Universal property of residual multi-opfibrations
 | universal property of residual multi-opfibrations
 | universal property

\knowledge{notion}
 | Essential uniqueness
 | essentially unique
 | essential uniqueness

\end{scope}

\knowledge{notion}
 | universal property@rmofrup
 | universal property of residues@rmofrup

\knowledge{notion}
 | residual opfibration

\knowledge{notion}
 | residual op-Cartesian@rof
 | residual op-Cartesian lifts@rof

\knowledge{notion}
 | isoglobular residual opfibration

\knowledge{notion}
 | stable system of monics

\begin{scope}\label{ssm}

\knowledge{notion}
 | stable under pullback
 | stability of $\cM $-morphisms under pullback
 | stability of monomorphisms under pullback
 | $\cM $-morphisms are stable under pullback
 | $\cM $-morphisms stable under pullback

\knowledge{notion}
 | includes all isomorphisms

\knowledge{notion}
 | stable under composition

\knowledge{notion}
 | decomposition property of $\cM $-morphisms

\knowledge{notion}
 | has pullbacks along $\cM $-morphisms
 | pullbacks along $\cM $-morphisms
 | has pullbacks and pushouts along $\cM $-morphisms
 | has pushouts and pullbacks along $\cM $-morphisms
 | pullbacks along $\cM _2$-morphisms
 | has pullbacks along monomorphisms
 | has pullbacks of spans of $\cM $-morphisms

\knowledge{notion}
 | has pushouts along $\cM $-morphisms
 | has pushouts and final pullback complements (FPCs) along $\cM $-morphisms
 | pushout along an $\cM $-morphism
 | has pullbacks, pushouts and FPCs along $\cM $-morphisms
 | has pushouts along monomorphisms
 | Pushouts along monomorphisms
 | has pushouts along regular monomorphisms
 | Pushouts along regular monomorphisms
 | Pushouts along $\cM $-morphisms
 | pushouts along $\cM _1$-morphisms
 | pushouts along $\cM $-morphisms
 | has pushouts and FPCs along $\cM $-morphisms
 | pushouts along regular monomorphisms

\knowledge{notion}
 | has final pullback complements (FPCs) along $\cM $-morphisms
 | has FPCs along $\cM $-morphisms
 | having FPCs along $\cM $-morphisms
 | FPCs along $\cM $-morphisms

\knowledge{notion}
 | $\cM $-morphisms are stable under pushout
 | stability of $\cM $-morphisms under pushout
 | stable under pushout

\knowledge{notion}
 | stability of $\cM $-morphisms under FPCs
 | $\cM $-morphisms are stable under FPCs
 
\knowledge{notion}
 | pushouts along $\cM $-morphisms are stable under $\cM $-pullbacks
 | stability of $\cM $-pushouts under $\cM $-pullbacks

\knowledge{notion}
 | $\cM $-initial
 | $\cM $-initial object

\knowledge{notion}
 | extremal morphisms
 | extremality

\knowledge{notion}
 | has pullbacks of cospans of $\cM $-morphisms

\end{scope}

\knowledge{notion}
 | universal property of pullbacks
 | Universal property of pullbacks (PBs)
 | pullbacks
 | pullback

\knowledge{notion}
 | Universal property of pushouts (POs)
 | universal property of pushouts
 | pushouts
 | pushout

\knowledge{notion}
 | Universal property of final pullback complements (FPCs)
 | universal property of FPCs
 | final pullback complement (FPC)
 | final pullback complement
 | final pullback complements
 | FPCs
 | FPC
 | final pullback complements (FPCs)
 | universal property of FPCs

\knowledge{notion}
 | has pullbacks
 | have pullbacks

\knowledge{notion}
 | has pushouts
 | have pushouts

\knowledge{notion}
 | has final pullback complements (FPCs)

\knowledge{notion}
 | multi-initial pushout complement (mIPC)
 | multi-initial pushout complement
 | multi-initial pushout complements (mIPCs)
 | $\cM $-multi-initial pushout complements ($\cM $-multi-IPCs)
 | $\cM $-multi-IPC
 | multi-initial pushout complements
 | $\cM $-multi-initial pushout complement
 | $\cM $-multi-IPCs
 | mIPC
 | multi-IPC
 | multi-IPCs
 | multi-IPC (mIPC)

\knowledge{notion}
 | has multi-initial pushout complements (mIPCs) along $\cM $-morphisms
 | has $\cM $-multi-IPCs

\begin{scope}\label{mIPC}

\knowledge{notion}
 | universal property
 | universal property of mIPCs
 | universal property of multi-IPCs
 | universal isomorphisms

\end{scope}

\begin{scope}\label{cmtSq}

\knowledge{notion}
 | source

\knowledge{notion}
 | target

\knowledge{notion}
 | domain

\knowledge{notion}
 | codomain

\knowledge{notion}
 | $dom:\textsf {PB}_h(\bfC ,\cM )\rightarrow \bfC $ is a Grothendieck fibration

\knowledge{notion}
 | $dom:\textsf {PB}_h(\bfC ,\cM )\rightarrow \bfC $ is a Grothendieck opfibration

\knowledge{notion}
 | $dom:\textsf {PB}_h(\bfC ,\cM )\rightarrow \bfC $ satisfies a Beck-Chevalley condition (BCC)

\knowledge{notion}
 | (BCC-1)

\knowledge{notion}
 | (BCC-2)
\end{scope}

\begin{scope}\label{domCsqPBT}

\knowledge{notion}
 | stable under pullbacks
 
\knowledge{notion}
 | Pushouts along $\cM $-morphisms are pullbacks

\end{scope}

\knowledge{notion}
 | has pushouts along $\cM $-morphisms@domFpo
 | pushouts along $\cM $-morphisms@domFpo

\knowledge{notion}
 | has FPCs along $\cM $-morphisms@domFfpc
 | FPCs along $\cM $-morphisms@domFfpc

\knowledge{notion}
 | target functor $T:\mathsf {PO}_v(\bfC ,\cM )\rightarrow \bfC \vert _{\cM }$ is a multi-opfibration
 | target functor $T_{\mathsf {PO}}:\mathsf {PO}_v(\bfC ,\cM )\rightarrow \bfC \vert _{\cM }$ is a multi-opfibration

\knowledge{notion}
 | source functor $S:\mathsf {PO}_v(\bfC ,\cM )\rightarrow \bfC \vert _{\cM }$ is a Grothendieck opfibration
 | source functor $S_{\mathsf {PO}}:\mathsf {PO}_v(\bfC ,\cM )\rightarrow \bfC \vert _{\cM }$ is a Grothendieck opfibration

\knowledge{notion}
 | target functor $T: \mathsf {FPC}_v(\bfC ,\cM )\rightarrow \bfC \vert _{\cM }$ is a Grothendieck opfibration
 | target functor $T_{\mathsf {FPC}}: \mathsf {FPC}_v(\bfC ,\cM )\rightarrow \bfC \vert _{\cM }$ is a Grothendieck opfibration

\knowledge{notion}
 | source functor $S:\mathsf {FPC}_v(\bfC ,\cM )\rightarrow \bfC \vert _{\cM }$ is a residual multi-opfibration
 | source functor $S_{\mathsf {FPC}}:\mathsf {FPC}_v(\bfC ,\cM )\rightarrow \bfC \vert _{\cM }$ is a residual multi-opfibration

\knowledge{notion}
 | pushouts along $\cM $-morphisms are stable under pullback@mOPFtrgt

\knowledge{notion}
 | factorization structure for morphisms

\knowledge{notion}
 | $(E,M)$-structured
 | (epi,$\cM $)-structured
 | ($\cE $, $\cM $)-structured
 | ($\cE $,$\cM $)-structured
 | ($\cE _{\text {fin}}$, $\cM _{\text {fin}}$)-structured
 | (epi, mono)-structured

\begin{scope}\label{EMS}

\knowledge{notion}
 | closed under composition with isomorphisms

\knowledge{notion}
 | has $(E,M)$-factorizations of morphisms
 | has (auto-augmented, inert)-factorizations of morphisms
 | (auto-augmented, inert) factorization
 | epi-$\cM $-factorization
 | epi-$\cM $-factorizations
 | $\cE $-$\cM $-factorization
 | $\cE $-$\cM $-factorization
 | epi-mono-factorizations
 | $\cE $-$\cM $-factorizations
 | (auto-augmented, inert)-factorizations
 | epi-mono-factorization
 | epi-regular mono-factorization

\knowledge{notion}
 | has the unique $(E,M)$-diagonalization property
 | has a unique (auto-augmented,inert) diagonalization property
 | unique $(E,M)$-diagonalization property

\knowledge{notion}
 | diagonal

\knowledge{notion}
 | morphisms that are both in $E$ and in $M$ are isomorphisms

\knowledge{notion}
 | closed under composition

\knowledge{notion}
 | $(E,M)$-factorizations are essentially unique
 | epi-$\cM $-factorizations are essentially unique
 | essential uniqueness of epi-$\cM $-factorizations
 | $\cE $-$\cM $-factorizations are essentially unique
 | essential uniqueness

\end{scope}

\knowledge{notion}
 | $\cM $-final pullback complement pushout augmentation (FPA)
 | final pullback complement pushout augmentation (FPA)
 | final pullback complement pushout augmentations (FPAs)
 | final pullback complement pushout augmentation
 | FPAs
 | FPA
 | FPC-pushout-augmentations (FPAs)
 | FPC-pushout-augmentation
 | $\cM $-FPC-pushout-augmentation
 | FPC-pushout-augmentation (FPA)

\knowledge{notion}
 | $\mathsf {FPC}_v(\bfC ,\cM )$ is (auto-augmented,inert)-structured
 | (auto-augmented, inert)-structured

\begin{scope}\label{FPCfact}

\knowledge{notion}
 | auto-augmented FPCs
 | auto-augmented FPC

\knowledge{notion}
 | inert FPCs
 | inert FPC

\end{scope}

\knowledge{notion}
 | double category (DC)
 | double category
 | double categories

\begin{scope}\label{dc}
\knowledge{notion}
 | domain
 | domain pseudofunctor

\knowledge{notion}
 | source
 | source functor

\knowledge{notion}
 | codomain

\knowledge{notion}
 | target
 | target functor

\knowledge{notion}
 | objects

\knowledge{notion}
 | vertical morphisms

\knowledge{notion}
 | horizontal morphisms

\knowledge{notion}
 | squares

\knowledge{notion}
 | horizontal unit
 | $\exists $ horizontal/vertical units
 | horizontal and vertical units

\knowledge{notion}
 | horizontal unit square

\knowledge{notion}
 | vertical composition

\knowledge{notion}
 | horizontal composition
 
\knowledge{notion}
  | globular isomorphisms
 | globular isomorphism

\end{scope}

\knowledge{notion}
 | $S:\mathsf {FPC}_v(\bfC ,\cM )\rightarrow \bfC \vert _{\cM }$ is a residual multi-opfibration

\knowledge{notion}
 | multi-sum
 | $\cM $-multi-sums
 | multi-sums
 | $\cM $-multi-sum

\begin{scope}\label{ms}

\knowledge{notion}
 | (multi-) universal property
 | universal property of multi-sums

\knowledge{notion}
 | has multi-sums

\knowledge{notion}
 | multi-sum extension Lemma
 | multi-sum extension

\end{scope}

\knowledge{notion}
 | compositional rewriting DC (crDC)
 | crDC
 | compositional rewriting double category
 | compositional rewriting double category (crDC)
 | compositional rewriting double categories
 | compositional rewriting double categories (crDCs)
 | crDCs

\begin{scope}\label{crDC}

\knowledge{notion}
 | $\bD _0$ has multi-sums
 | multi-sum structure of $\bD _0$

\knowledge{notion}
 | $\bD _0$ and $\bD _1$ have pullbacks
 | $\bD _0$ has pullbacks
 | $\mathbb {D}_1$ has pullbacks
 | $\bD _0:=\bfC \vert _{\cM }$ has pullbacks

\knowledge{notion}
 | $\hComp :\bD _1\times _{\bD _0}\bD _1\rightarrow \bD _1$ is an isoglobular residual opfibration
 | isoglobular residual opfibration

\knowledge{notion}
 | $S:\bD _1 \rightarrow \bD _0$ is a strong multi-opfibration
 | the source functor $S:\bD _1\rightarrow \bD _0$ is a strong multi-opfibration
 | the source functor is a strong multi-opfibration
 | $S$ is a strong multi-opfibration
 | the source functor $S:\bD _1\rightarrow \bD _0$ is a multi-opfibration
 | $S$ is a multi-opfibration

\knowledge{notion}
 | $T:\bD _1 \rightarrow \bD _0$ is a residual multi-opfibration
 | the target functor $T:\bD _1\rightarrow \bD _0$ is a residual multi-opfibration
 | $T$ is a residual multi-opfibration

\knowledge{notion}
 | horizontal decomposition property
 | horizontal decomposition

\knowledge{notion}
 | complex decomposition property

\knowledge{notion}
 | composite rule
 | composite
 | composites

\knowledge{notion}
 | rule composition
 | composition

\end{scope}

\knowledge{notion}
 | \gcong

\knowledge{notion}
 | Concurrency theorem
 | concurrency theorem
 | concurrency theorems
 | concurrency

\begin{scope}\label{cct}

\knowledge{notion}
 | Synthesis
 | synthesis

\knowledge{notion}
 | Analysis
 | analysis

\end{scope}

\knowledge{notion}
 | Associativity theorem for compositional rewriting double categories
 | associativity theorem
 | associativity theorems

\knowledge{notion}
 | $\cM $-partial map classifier
 | $\mono {\bfC }$-partial map classifier
 | $\cM $-partial map classifiers

\begin{scope}\label{Madh}

\knowledge{notion}
 | $\bfC $ has pushouts and pullbacks along $\cM $-morphisms

\knowledge{notion}
 | has pushouts and pullbacks along $\cM $-morphisms

\knowledge{notion}
 | $\cM $-morphisms are stable under pushout
 | Stability of $\cM $-morphisms under pushout
 | stability of $\cM $-morphisms under pushout
 | under pushouts

\knowledge{notion}
 | $\cM $-van Kampen squares

\knowledge{notion}
 | Pushouts along $\cM $-morphisms are pullbacks
 | pushouts along $\cM $-morphisms are pullbacks

\end{scope}

\knowledge{notion}
 | pushouts along $\cM $-morphisms are pullbacks@lemMIPCexistence

\knowledge{notion}
 | pushouts along $\cM $-morphisms are pullbacks@thmTPOvMOF

\knowledge{notion}
 | category $\mathbf {Graph}$

\knowledge{notion}
 | category of directed simple graphs
 | directed simple graphs
 | $\mathbf {SGraph}$

\knowledge{notion}
 | $\cM $-subobject

\knowledge{notion}
 | finitary
 | finitary objects

\knowledge{notion}
 | finitary category

\knowledge{notion}
 | finitary restriction

\knowledge{notion}
 | category $\mathbf {FinGraph}$

\knowledge{notion}
 | category $\mathbf {FinSGraph}$

\knowledge{notion}
 | extremal monomorphisms

\knowledge{notion}
 | quasi-topos
 | quasi-topoi

\begin{scope}\label{qt}

\knowledge{notion}
 | has pushouts along regular monomorphisms

\knowledge{notion}
 | pushouts along $\cM $-morphisms are pullbacks
 | pushouts along regular monomorphisms are pullbacks

\knowledge{notion}
 | pushouts along regular monomorphisms are stable under pullbacks
 | stability of pushouts along regular monomorphisms under pullbacks

\knowledge{notion}
 | rm-quasi-adhesive
 | rm-quasi-adhesive category

\knowledge{notion}
 | exists a final pullback-complement (FPC)

\knowledge{notion}
 | strict initial object

\knowledge{notion}
 | disjoint coproducts

\knowledge{notion}
 | unions of regular subobjects are effective

\knowledge{notion}
 | solid quasi-topos

\knowledge{notion}
 | weak adhesive HLR category

\end{scope}

\knowledge{notion}
 | (VK-a)

\knowledge{notion}
 | (VK-b)

\knowledge{notion}
 | (X-iii-a)@notationVKa
 | (L-iii-a)@notationVKa
 | (H-iii-a)@notationVKa
 | (V-iii-a)@notationVKa
 | (W-iii-a)@notationVKa
 | (A-iii-a)@notationVKa

\knowledge{notion}
 | (X-iii-b)@notationVKb
 | (V-iii-b)@notationVKb
 | (L-iii-b)@notationVKb
 
\begin{scope}\label{adhVK}

\knowledge{notion}
 | (A-iii)
 
\knowledge{notion}
 | (Q-iii)

\knowledge{notion}
 | (L-iii)
 
\knowledge{notion}
 | (H-iii)

\knowledge{notion}
 | (V-i)

\knowledge{notion}
 | (V-ii)

\knowledge{notion}
 | (V-iii)

\knowledge{notion}
 | (W-iii)

\end{scope}

\knowledge{notion}
 | Pushout-pushout-(de-)composition
 | pushout-pushout composition
 | pushout composition
 | pushout-pushout decomposition
 | PO-PO-

\knowledge{notion}
 | Pullback-pullback-(de-)composition
 | pullback-pullback decomposition
 | pullback-pullback composition

\knowledge{notion}
 | Pushout-pullback-decomposition
 | pushout-pullback decomposition

\knowledge{notion}
 | Pullback-pushout-decomposition
 | pullback-pushout decomposition

\knowledge{notion}
 | Horizontal FPC-FPC-(de-)composition
 | horizontal FPC composition
 | horizontal FPC decomposition
 
\knowledge{notion}
 | Vertical FPC-FPC-(de-)composition
 | vertical FPC composition
 | vertical FPC-FPC decomposition

\knowledge{notion}
 | Vertical FPC-pullback decomposition
 | vertical FPC-pullback decomposition

\knowledge{notion}
 | stable under pullbacks@PO
 | pushouts along $\cM $-morphisms are stable under pullbacks@PO
 | stability under pullbacks@PO

\knowledge{notion}
 | stable under pullbacks@FPC
 | stability of FPCs under pullbacks@FPC

\knowledge{notion}
 | stability of FPCs under pullbacks
 | FPCs are stable under pullbacks

\knowledge{notion}
 | van Kampen (VK) square
 | VK squares
 | van Kampen squares

\knowledge{notion}
 | horizontal weak VK squares

\knowledge{notion}
 | vertical weak VK squares
 | vertical weak van Kampen squares
 | vertical weak VK square

\knowledge{notion}
 | adhesive category
 | adhesive
 | adhesive categories

\knowledge{notion}
 | quasi-adhesive
 | quasi-adhesive categories
 | rm-adhesive category
 | rm-adhesive
 | quasi-adhesive/rm-adhesive
 | rm-adhesive categories

\knowledge{notion}
 | adhesive high-level replacement (HLR) category
 | adhesive HLR
 | adhesive HLR category
 | adhesive HLR categories

\knowledge{notion}
 | adhesivity properties

 \knowledge{notion}
 | horizontal weak adhesive HLR category
 | horizontal
 | hor.\ weak adh.\ HLR
 | horizontal weak adhesive HLR categories

\knowledge{notion}
 | pushouts along $\cM $-morphisms are pullbacks@aps

\knowledge{notion}
 | pushouts of spans of $\cM $-morphisms are pullbacks@hwahlr

\knowledge{notion}
 | vertical weak adhesive HLR category
 | vertical
 | vert.\ weak adh.\ HLR
 | vertical weak adhesive HLR categories
 | vertical weak adhesive HLR
 | $\cM $-adhesive categories
 | $\cM $-adhesive category
 | $\cM $-adhesive 

\knowledge{notion}
 | weak adhesive HLR category
 | weak adhesive HLR categories
 | weak adhesive HLR

\knowledge{notion}
 | regular mono (rm)-quasi-adhesive category
 | rm-quasi-adhesive
 | rm-quasi-adhesive categories

\knowledge{notion}
 | comma category
 | comma categories

\knowledge{notion}
 | directed multigraphs
 | $\mathbf {Graph}$
 | category of directed graphs

\knowledge{notion}
 | $\mathbf {HyperGraph}$
 | directed ordered hypergraphs

\knowledge{notion}
 | $\mathbf {PTNets}$
 | place/transition nets
 
\knowledge{notion}
 | $\mathbf {ElemNets}$
 | elementary Petri nets

\knowledge{notion}
 | undirected multigraphs

\begin{scope}\label{msEM}

\knowledge{notion}
 | Existence:
 | existence

\knowledge{notion}
 | Construction:
 | construction

\knowledge{notion}
 | Refinements:
 | refinements

\end{scope}

\knowledge{notion}
 | quotients
 | quotient

\knowledge{notion}
 | Double-Pushout (DPO) semantics
 | DPO-semantics
 | double-pushout semantics
 | Double-Pushout semantics
 | DPO-
 | Double Pushout (DPO)
 | Double-Pushout

\knowledge{notion}
 | direct derivation@crsDPO
 | direct derivations@crsDPO
 | DPO-type direct derivation@crsDPO
 | DPO-type direct derivations@crsDPO

\knowledge{notion}
 | set of DPO-admissible matches@crsDPO
 | admissible match@crsDPO
 | admissible matches@crsDPO

\knowledge{notion}
 | set of DPO-type admissible matches@crsDPOrComp
 | admissible match@crsDPOrComp

\knowledge{notion}
 | DPO-type rule composition@crsDPOrComp
 | rule composition@crsDPOrComp
 | DPO-type rule compositions@crsDPOrComp

\knowledge{notion}
 | composite rule@crsDPOrComp

\knowledge{notion}
 | suitable for DPO-semantics@crsDPO

\knowledge{notion}
 | Sesqui-Pushout (SqPO) semantics
 | sesqui-pushout (SqPO) semantics
 | sesqui-pushout semantics
 | SqPO-semantics
 | Sesqui-Pushout semantics
 | SqPO rewriting
 | Sesqui-Pushout (SqPO)

 \knowledge{notion}
 | direct derivation@crsSqPO
 | direct derivations@crsSqPO
 | SqPO-type direct derivation@crsSqPO
 | SqPO-type direct derivations@crsSqPO
 | sesqui-pushout direct derivations@crsSqPO

\knowledge{notion}
 | set of SqPO-admissible matches@crsSqPO
 | admissible matches@crsSqPO

\knowledge{notion}
 | suitable for SqPO-semantics@crsSqPO
 | non-linear SqPO-rewriting to be well-posed@crsSqPO

\knowledge{notion}
 | set of SqPO-type admissible matches@crsSqPOrComp
 | admissible SqPO-type matches of rules@crsSqPOrComp
 | admissible match@crsSqPOrComp
 | admissible matches of rules@crsSqPOrComp

\knowledge{notion}
 | SqPO-type rule composition@crsSqPOrComp
 | SqPO-type rule composition@crsSqPO
 | rule composition@crsSqPOrComp
 | rule compositions@crsSqPOrComp

\knowledge{notion}
 | composite rule@crsSqPOrComp
 | composite rules@crsSqPOrComp

\begin{scope}\label{crs}

\knowledge{notion}
 | rule
 | rules

\knowledge{notion}
 | output-linear
 
\knowledge{notion}
 | input-linear
 
\knowledge{notion}
 | linear
 
\knowledge{notion}
 | generic

\knowledge{notion}
 | direct derivation
 | direct derivations

\knowledge{notion}
 | match
 | matches

\knowledge{notion}
 | co-match
 | co-matches

\end{scope}

\begin{scope}\label{crsType}

\knowledge{notion}
 | generic

\knowledge{notion}
 | output-linear
 | output- (or right-) linear

\knowledge{notion}
 | input-linear
 | input- (or left-) linear

\knowledge{notion}
 | linear

\knowledge{notion}
 | semi-linear

\end{scope}

\knowledge{notion}
 | pushouts along $\cM $-morphisms are stable under pullbacks@auxDirectDPO

\knowledge{notion}
 | pushouts along $\cM $-morphisms are pullbacks@FPAconstr

\begin{document}
\begin{frontmatter}

\title{Fundamentals of Compositional Rewriting Theory\tnoteref{mytitlenote}}
\tnotetext[mytitlenote]{This is an invited extended journal version of the ICGT 2021 conference paper entitled ``Concurrency Theorems for Non-linear Rewriting Theories''~\cite{BehrHK21,BehrHK21-ext}.}

\author[IRIF]{Nicolas Behr\corref{cor1}}
\ead{nicolas.behr@irif.fr}
 \cortext[cor1]{Corresponding author.}
 \affiliation[IRIF]{organization={Université Paris Cit\'{e}, CNRS, IRIF},
             addressline={8 Place Aurélie Nemours},
             city={Paris Cedex 13},
             postcode={75205},
             country={France}
             }
\author[ENS]{Russ Harmer}
 \ead{russell.harmer@ens-lyon.fr}
 \affiliation[ENS]{organization={Université de Lyon, ENS de Lyon, UCBL, CNRS, LIP},
             addressline={46 allée d'Italie},
             city={Lyon Cedex 07},
             postcode={69364},
             country={France}}

\author[IRIF]{Jean Krivine}
\ead{jean.krivine@irif.fr}
\begin{abstract}
A foundational theory of compositional categorical rewriting theory is presented, based on a collection of fibration-like properties that collectively induce and intrinsically structure the large collection of lemmata used in the proofs of theorems such as concurrency and associativity. The resulting highly generic proofs of these theorems are given. It is noteworthy that the proof of the concurrency theorem takes only a few lines and, while that of associativity remains somewhat longer, it would be unreadably long if written directly in terms of the basic lemmata. In essence, our framework improves the readability and ease of comprehension of these proofs by exposing latent modularity. A curated list of known instances of our framework is used to conclude the paper with a detailed discussion of the conditions under which the Double Pushout and Sesqui-Pushout semantics of graph transformation are compositional.
\end{abstract}

\begin{keyword}
\texttt{elsarticle.cls}\sep \LaTeX\sep Elsevier \sep template
\MSC[2010] 16B50\sep 60J27\sep 68Q42 (Primary) 60J28\sep 16B50\sep 05E99 (Secondary)
\end{keyword}

\end{frontmatter}

\tableofcontents

\section{Introduction}

The main contribution of the present paper is a novel framework of \kl{compositional rewriting double categories (crDCs)} which,  at its core, is based upon mathematical notions of fibrational structures relevant to categorical rewriting theories. Our motivation for this development has been the technically highly involved nature of the definitions of rule compositions and the resulting \kl{concurrency theorems} for \kl(crsType){generic} \kl{Double-Pushout (DPO) semantics} and \kl{Sesqui-Pushout (SqPO) semantics} as presented in ~\cite{BehrHK21} (with additional details and proofs presented in~\cite{BehrHK21-ext}).

In this extended version, we show that one can modularize the statement of the rewriting semantics, its list of prerequisites and also the statement and proof of the concurrency theorem in a \emph{uniform} fashion by establishing the notion of \kl{crDCs}: once, for a given semantics, the notion of \kl(crs){direct derivations} is specified (and thus, in a certain sense, the very definition of the semantics itself is provided), rather than trying to follow the steps prescribed by the categorical rewriting literature in the tradition of the work of Ehrig et al.~\cite{ehrig:2006fund}, our novel concept of \kl{crDCs} permits to decide whether or not the given semantics is \emph{compositional} (i.e., admits both a \kl{concurrency theorem} and an \kl{associativity theorem}) purely based upon the properties of the \kl(crs){direct derivations} themselves.

It should be noted that establishing that a given semantics indeed yields a \kl{crDC} structure is still a technically involved task (as will become evident when presenting instantiations of our novel framework for some concrete examples of rewriting semantics in Section~\ref{sec:rs}); however, the required reasoning has a rather more mechanical character than that of the proofs of complex theorems in compositional rewriting theory---which are \emph{automatically guaranteed} to hold once a \kl{crDC} structure is verified for a given semantics.

\subsection{Rewriting}\label{sec:intro:bg}

The theory of graph transformation has been under development for about the last fifty years. Over this time, it has gradually evolved from working with specific concrete settings---such as multi or simple graphs, with or without attributes---to being expressed in terms of certain classes of categories---such as \kl{adhesive} \cite{ls2004adhesive,lack2005adhesive,garner2012axioms}, \kl{quasi-adhesive/rm-adhesive} \cite{lack2005adhesive,garner2012axioms,quasi-topos-2007} and \kl{$\cM$-adhesive} \cite{ehrig:2006fund,GABRIEL_2014} categories---that provide sufficient structure to reprove abstractly the key theorems that hold in those concrete settings. This theory has found application in a variety of contexts such as model-driven software engineering \cite{heckel2020graph}, graph databases \cite{braatz2008graph,bonifati2019schema}, graph-based knowledge representation \cite{chein2008graph,harmer2019bio,harmer2020knowledge} and executable representation of complex systems \cite{danos2007rule,faeder2009rule,andersen2016software} as well as more theoretical uses in graph grammars, structural graph theory and string diagrams.

In this theory, a graph transformation \emph{rule} $O\leftharpoonup r-I$ is interpreted in a given category as a span, i.e., as two morphisms $O \leftarrow o_r- K_r -i_r\rightarrow I$ with a common source object $K_r$, called the \emph{context} or \emph{preserved region}, an \emph{input} (or LHS) object $I$ and an \emph{output} (or RHS) object $O$. Note that we depart from the classical representation of rules in this theory in two ways: the application to rule algebra \cite{bp2019-ext,nbSqPO2019,Behr2021} makes it more natural to reverse the orientation of a rule, so as to match with the usual right-to-left ordering of composition of functions, whereupon the traditional terminology of `LHS' and `RHS' becomes rather confusing; as such, we prefer the neutral, self-explanatory terms `input' and `output'. Similarly, we use the terms `input-linear' and `output-linear' instead of the more usual `left-linear' and `right-linear' when speaking of rules where $i_r$ or, respectively, $o_r$ are monomorphisms of some kind.

In typical concrete settings, the two arrows express the correspondence between the entities (nodes or edges, etc.) in $I$ and those in $O$. This naturally suggests a small set of primitive operations---deletion and addition, where $i_r$ and $o_r$ are non-surjective, and cloning and merging, where they are non-injective---that correspond to our intuitive ideas of how graphs can be transformed; a rule then interprets a combination of these primitive operations. In order to formalize the effect of a rule, several distinct, but closely related, semantics have been proposed, the most prominent of which are the Double Pushout (DPO)~\cite{CorradiniMREHL97}, Single Pushout (SPO)~\cite{Loewe_1993} and, more recently, Sesqui-Pushout (SqPO)~\cite{Corradini_2006} semantics.
 
In all of these approaches, a rule is applied to an object $X$ through a so-called \emph{matching} $I\rtail m\rightarrow X$, where $m$ is a monomorphism\footnote{The DPO approach has sometimes been formulated without requiring $m$ to be a monomorphism; we do not consider this variant here.} potentially chosen from a specified restricted class $\mathcal{M}$. In the case of DPO or SqPO semantics, rewriting proceeds in two steps: first, we use $i_r$ and $m$ to construct an intermediate object $K_{r_m}$, a monomorphism $K_r \rtail k_{r_m}\rightarrow K_{r_m}$ and a morphism $K_{r_m}-i_{r_m}\rightarrow X$ (i.e., the square (*) below); then we use $o_r$ and $k_{r_m}$ to construct an object $r_m(X)$ (i.e., the direct derivation of $X$ along $r$ with match $m$), a monomorphism $O\rtail m^{*}\rightarrow r_m(X)$, and a morphism $K_{r_m} - o_{r_m} \rightarrow r_m(X)$ (i.e. the square ($\dagger$) below):
\begin{equation}\label{eq:ddIntro}
\ti{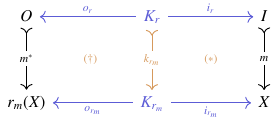}
\end{equation}
In both DPO and SqPO semantics, the second step is determined by taking the pushout (PO) of $k_{r_m}$ and $o_r$ as in square $(\dagger)$; the difference between these semantics arises in the first step: the DPO approach specifies that the square $(*)$ be a PO while the SqPO approach specifies it to be the final pullback complement of $m$ and $i_r$~\cite{dyckhoff1987exponentiable}: given two composable arrows $A - f' \rightarrow B$ and $B - g \rightarrow D$, a \emph{final pullback complement} (FPC; cf.\ diagram below) consists of two composable arrows $A - g' \rightarrow C$ and $C - f \rightarrow D$ such that (i) the resulting square is a pullback (PB); and (ii) for all PB squares such as the outer square in the diagram below, and for all factorizations $A' - a \rightarrow A$ of $\bar{f}'$ through $f'$, there exists a unique arrow $C' - c \rightarrow C$ such that $f \circ c = \bar{f}$ and $c \circ \bar{g}' = g \circ \bar{f}'$.

\begin{equation*}\label{eq:defFPC}
\ti{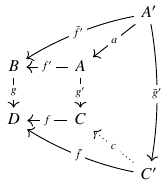}
\end{equation*}\\

The notion of FPC is defined by a universal property\footnote{It can be seen as a categorical generalization of the notion of set difference: the FPC is the largest, i.e., least general, $C$ together with arrows $A-g' \rightarrow C$ and $C-f \rightarrow D$ for which the resulting square is a PB.} so that, under SqPO semantics, given $m$ and $i_r$, $K_{r_m}$ is \emph{essentially} unique (i.e., unique up to unique isomorphism). However, the construction of $K_{r_m}$ under the DPO semantics need not be uniquely determined because, in general, there may be several non-isomorphic objects for which the square $(*)$ is a PO. Nonetheless, in many concrete settings under the assumption that the rule is \kl(crsType){input-linear}, there is in fact \emph{at most one way} for square $(*)$ to be a PO: this follows from the fact that, in the category $\Set$ of sets and functions, there is exactly one $K_{r_m}$ for which square $(*)$ is a PO, provided that $i_r$ is injective. In these settings, it is easy to show that this square---if it exists, as is characterized by the no-dangling edges condition \cite{ehrig:2006fund}---satisfies a universal property that is categorically dual to that of an FPC: it is the smallest, i.e., most general, $K_{r_m}$ together with arrows from and to $K_r$ and $X$ respectively for which the resulting square is a PO. This is a mild generalization of the notion of what is called the \emph{initial pushout complement} in \cite{taentzer1999distributed} or \emph{minimal pushout complement} in \cite{BRAATZ2011246}; despite the risk of slight confusion, we reuse the terminology of \emph{initial pushout complement} (IPC) as this is standard practice in the naming of categorical duals.

In the case that $i_r$ is not a monomorphism, there is no longer any guarantee of uniqueness: even in typical concrete settings  based on sets and functions, there may be several \emph{minimal} but incomparable candidate PO squares. This leads us to consider the less familiar categorical notion of \kl{multi-IPC (mIPC)} which formalizes the notion of a \emph{family} of minimal solutions, i.e., PO squares over $m$ and $i_r$, that are \emph{collectively} universal: any PO square using $i_r$ and factoring through $m$ factors uniquely through a unique family member. This multi-universal property is an instance of the general theory of Diers~\cite{diers1978familles}; in the case at hand, it effectively states that the family of $K_{r_m}$s contains precisely all possible rewrites of $X$ by $i_r$ via $m$ that are compatible with the DPO semantics. We will investigate this construction in more detail in Sections~\ref{sec:targetFunctors} and~\ref{sec:exConstr:mIPCs}.

Finally, in \emph{SPO semantics}~\cite{Loewe_1993}, the rewrite is applied in a single step by taking the PO, in the bi-category of spans, of $m$ and the rule $r$---all subject to the condition that $i_r$ be a monomorphism, i.e., that the rule is \kl(crsType){input-linear}. %
Due to the well-known fact that SPO semantics in the setting where monic matches are used in fact coincides with a special case of \kl{Sesqui-Pushout (SqPO) semantics}~\cite{Corradini_2006}, we will not consider SPO semantics separately in the present paper. Note in particular that the SqPO approach, unlike SPO, does \emph{not} require input-linearity and works with fully general rules.

An interesting consequence of the definition of DPO semantics is that rule applications are always \emph{reversible} because squares $(*)$ and $(\dagger)$ are both POs; in effect, the no-dangling condition prevents the rule from being applied in the case that it would otherwise produce an irreversible transformation---or, alternatively, induce a \emph{side-effect}. This is not necessarily the case for SqPO semantics because there is no a fortiori reason that the square $(*)$, defined by an FPC, be a PO: it is well-known that a rule which deletes a node can always be applied under the SqPO semantics, but induces an irreversible transformation in the case that the deleted node has incident edges in $X$. However, such a rule can only be applied under the DPO semantics if the  no-dangling condition holds, i.e., the targeted node has no incident edges. Equally, there is no a fortiori reason that the square $(\dagger)$ be an FPC: a rule that merges two nodes generally loses the information about their incident edges that would be required to reverse the transformation---this is another kind of side-effect. The special case of \emph{reversible SqPO}, where the FPC is also always a PO and the PO is also always an FPC, was studied in~\cite{reversibleSqPO,harmer2020reversibility}; in practice, at least for linear rules, this amounts to restricting to the DPO semantics since this constraint is then equivalent to the no-dangling condition.

The use of input-non-linear rules under the SqPO semantics allows for the expression of the natural operation of the \emph{cloning} of a node or an edge (when this is meaningful), as explained in~\cite{Corradini_2006,BRAATZ2011246,Corradini_2015}. More recently, such rules have also been used to express operations such as concept refinement in schemata for graph databases \cite{bonifati2019schema} and, more generally, in graph-based knowledge representation \cite{harmer2020knowledge}. In combination with output-non-linear rules, as for (non-linear) DPO rewriting, the SqPO semantics thus allows the expression of all natural \emph{primitive operations on graphs}: \emph{addition} and \emph{deletion} of nodes and edges; and \emph{cloning} and \emph{merging} of nodes and edges. Moreover, the SqPO semantics also allows for all natural side-effects, or failures of reversibility, that arise intuitively in the case of deletion and merging, as mentioned above, and also in the case of cloning; this has notably been exploited in the definition of the semantics of the \texttt{Kappa} language where all rules are linear---so no cloning or merging---but deletion may have side-effects~\cite{Boutillier:2018aa}, and in the \texttt{KAMI} bio-curation framework~\cite{harmer2019bio}, where fully general rules are used with the SqPO semantics.

\subsection{Compositional rewriting}

In the discussion above, we have seen that the definition of a setting for graph transformation requires us to specify a number of things, including: the category in which we work, the classes of rules and matchings under consideration, the semantics we use to apply rules, etc. %
In this paper, we adopt the stance that a choice of parameters in this design space should be compatible with \emph{compositionality} in the sense that a notion of \emph{rule composition} exists, which amounts to requiring that the \kl{concurrency theorem} holds, and that this satisfies an appropriate form of \kl{associativity theorem}. These properties enable the static analysis of collections of rules, such as rule algebras \cite{bp2019-ext,nbSqPO2019,Behr2021} or the causal analysis found in the Kappa language \cite{Boutillier:2018aa,BK2020}, that ultimately depend on the notion of \emph{tracelets}~\cite{behr2019tracelets} that follows from having such an associative rule composition.

\subsubsection{Rule composition}

The composition of two rules, $O_1 \leftharpoonup r_1 - I_1$ and $O_2 \leftharpoonup r_2 - I_2$, is a third rule $r_2 \circ r_1$ whose effect on an object $X$ should be the same as that of applying first $r_1$ to $X$ then $r_2$ to the resulting $r_{m_1}(X)$. However, the effect of applying first $r_1$ then $r_2$ depends critically on the way in which the images of the matchings $m^*_1$ and $m_2$, of $O_1$ and $I_2$ respectively, overlap in $r_{m_1}(X)$: given an initial matching $m_1$ of $I_1$ in $X$, the resulting $m^*_1$ is uniquely determined; however, there may be many possible choices for the matching $m_2$ of $I_2$ in $r_{m_1}(X)$. As such, in general, there is not one single composite rule $r_2 \circ r_1$ but rather one such composite for each possible overlap of $m^*_1$ with some such $m_2$.

In order to express these ideas independently of the particular choice of $X$, we use Diers' notion of multi-sum \cite{diers1978familles} as the means to express the \emph{family} of all possible overlaps of matchings from $O_1$ and $I_2$. This is a generalization of the familiar categorical notion of co-product which replaces the single co-product $O_1 \rightarrow O_1 + I_2 \leftarrow I_2$ with a family of co-spans satisfying a multi-universal property: essentially, any given co-span from $O_1$ and $I_2$ factors through exactly one family member and does so uniquely\footnote{We give a formal definition in Section 2.1 which, although its precise statement differs from this informal account by allowing for essential uniqueness, remains essentially equivalent.}.

This kind of multi-universal construction often arises in concrete settings, based on the category of sets, where we wish to restrict our attention to injective functions. In this case, although the inclusions $O_1 \rightarrow O_1 + I_2 \leftarrow I_2$ are injective, the universal arrow is generally non-injective (unless the images of $O_1$ and $I_2$ are disjoint). The multi-sum construction side-steps this problem by providing all possible overlaps of the images of $O_1$ and $I_2$, including the case where they are disjoint, so that all identifications---that would lead to violations of injectivity---necessary for the usual universal arrow can instead be accounted for by choosing the appropriate family member.

The general notion of rule composition can thus be stated purely at the level of $r_1$ and $r_2$, provided that we are working in a setting where the multi-sum of $O_1$ and $I_2$ is guaranteed to exist, and it provides one composite rule per family member of that multi-sum. The \emph{synthesis} part of the concurrency theorem then states that a sequential application, of $r_1$ then $r_2$, can be simulated in a single step by identifying the relevant multi-sum element and using the appropriate induced composite rule; and the \emph{analysis} part of the theorem states, conversely, that the direct application of such a composite rule can be decomposed back into a sequential application of its constituents with overlap determined by the corresponding multi-sum element.

\subsubsection{Compositionality}

In this paper, we seek to provide foundations for \emph{compositional} rewriting that apply equally to DPO and SqPO semantics. In particular, we provide a single proof of the \kl{concurrency theorem} in Section~\ref{sec:comp:cct}) and a single proof of an appropriate associativity property of the induced notion of rule composition, i.e., an \kl{associativity theorem}, in Section~\ref{sec:comp:at}) that work for linear and non-linear rules under both semantics. This associativity property is important as it guarantees that, for any sequence of rule applications, their overall composite transformation can be computed, by iterating the concurrency theorem, in any order without changing the result (up to isomorphism): this is precisely what we mean by `compositionality'.

In this paper, we consider a rule application as a single unit (analogous to SPO semantics) rather than decomposing, as usual, into two stages that proceed via an intermediate object. This choice eases the path to a characterization of the necessary categorical structure for compositional rewriting in terms of (i) the existence of certain kinds of fibrations; and (ii) a small number of additional axioms specific to rewriting.

From this point of view, the large collection of lemmata used in these proofs---as, for example, collected in the appendix of~\cite{BehrHK21-ext}---fall into two groups: a first group of fundamental results that do not specifically relate to rewriting; and a second group that specializes this theory precisely to the case of rewriting.
The advantage of this new approach is that it enables the use of macros that express the key steps in proofs at a higher level of abstraction than usual and, indeed, the resulting proof of the \kl{concurrency theorem} is very compact. The proof of the \kl{associativity theorem} is significantly longer and technically more involved---although this seems to be intrinsic to its nature---but would have been essentially impossible to express at a lower level of abstraction. Our new approach makes a clear and clean separation of the basic building blocks from their means of combination; we return to this point in the conclusion.

\subsection*{Outline of the paper}

The paper is structured as follows. In Section~\ref{sec:fib}, we present all the preliminary material necessary for the definition of compositional rewriting double categories with a particular emphasis on the required fibrational structures. In Section~\ref{sec:comp}, we define this novel concept and apply it to state and prove the concurrency theorem and the associativity theorem in a universal fashion. In Section~\ref{sec:examplesFib}, we investigate the fibrational structure of various categories of squares (pullbacks, pushouts and final pullback complements). In Section 5, we study some classes of categories that admit constructions that are necessary in order to formulate compositional rewriting theories, and, in Section 6, we focus on the DPO- and SqPO-semantics in order to clarify under what conditions, and for what classes of rules, these semantics are compositional. Finally, we conclude with a detailed comparison of the approach in this paper with that of its conference version \cite{BehrHK21,BehrHK21-ext} as well as a discussion of related and future work.
The reader interested in our fibration-based proofs of concurrency and associativity can therefore read Sections 2 and 3 only;  one more interested in how our new framework can be put to use might prefer to skim those sections and focus principally on Sections 4, 5 and 6.

\section{On multi-sums and fibrational structures}\label{sec:fib}

In this section, we provide some prerequisite material for our compositional rewriting theory framework. We begin with the notion of \emph{multi-sum} which is an instance of the general theory of multi-co-limits developed by Diers \cite{diers1978familles}. We then introduce the mathematical theory for a number of fibrational structures, namely the well-known notions of \kl{Grothendieck fibration} and \kl{Grothendieck opfibration}, but also \kl{multi-opfibrations} and \kl{residual multi-opfibrations} which, to the best of our knowledge, are original results of our work.

\subsection{Multi-sums}\label{sec:ms}

An atypical feature of fibrational structures relevant for compositional rewriting theories is the following type of mathematical property, which may eventually play an important role in the static analysis of rewriting systems.

\begin{definition}\label{def:ms}
Let $\bfC$ be a category. %
A \AP\intro{multi-sum} $\Msum{A}{B}$ of two objects $A$ and $B$ of $\bfC$ is %
a family of cospans $\{A-a_j\rightarrow M_j \leftarrow b_j- B\}_{j\in J}$ such that %
for every cospan $A-a\rightarrow X\leftarrow b- B$, %
there exists a $j\in J$ and morphism $M_j -x\rightarrow X$ %
such that $a=x\circ a_j$ and $b=x\circ b_j$, and %
with the following \AP\intro(ms){(multi-) universal property}: %
for every cospan $A-a'\rightarrow Y\leftarrow b'-B$ and morphism $Y-y\rightarrow X$ such that $a=y\circ a'$ and $b=y\circ b'$, there exists a unique morphism $M_j-m_j\rightarrow Y$ such that $a'=m_j\circ a_j$ and $b'=m_j\circ b_j$: 

\begin{equation}\label{eq:ms:up}
\ti{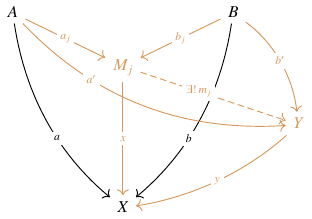}
\end{equation}
We say that $\bfC$ \AP\intro(ms){has multi-sums} if every pair of objects has a multi-sum. 
\end{definition}

While we postpone the presentation of some concrete examples of multi-sum structures to Section~\ref{sec:exConstr:mms}, suffice it here to introduce a technical result that will be necessary in our ensuing constructions:

\begin{lemma}[Multi-sum extension]\label{lem:ms1}
\AP\phantomintro(ms){multi-sum extension Lemma}
Let $\bfC$ be a category that \kl(ms){has multi-sums} and that \kl{has pullbacks}. Then for every commutative diagram such as in~\eqref{eq:lem:multiSumExtensionA} below, where $A\rightarrow M\leftarrow B$ and $C\rightarrow N\leftarrow D$ are multi-sum elements, there exists a universal arrow $M\rightarrow N$ that makes the diagram commute.

\begin{equation}\label{eq:lem:multiSumExtensionA}
\ti{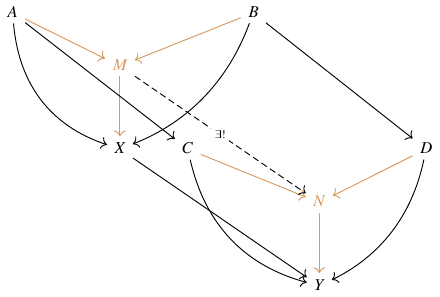}
\end{equation}
\end{lemma}
The proof of this lemma, and indeed of all the results in Section~\ref{sec:fib}, can be found in \ref{app:ps2}.

\subsection{Grothendieck fibrations and opfibrations}\label{sec:gfof}

In the technical constructions developed in this paper, we will require certain generalizations of the notion of \kl{Grothendieck opfibration}. We will therefore employ a notation for fibrations that slightly differs from the standard conventions in category theory (cf.\ e.g.\ \cite{streicher2018fibered,jacobs1999categorical,borceux1994handbook,benabou1985fibered}). Let us therefore briefly recall the definitions of \kl{Grothendieck fibrations} and \kl{Grothendieck opfibrations} for the readers convenience, expressed in our notational conventions:

\begin{definition}\label{def:grothendieckFibration}
A functor $G:\bfE\rightarrow \bfB$ is a \AP\intro{Grothendieck fibration} if the following property holds:
\begin{equation}\label{eq:def:grothendieckFibration}
\begin{aligned}
&\forall\ti{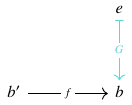}\;:\;
\exists \ti{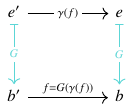}\;:\\
&\qquad \forall\ti{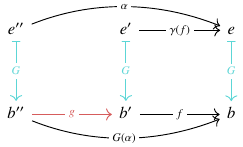}\;:\; \ti{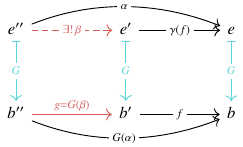}
\end{aligned}
\end{equation}
The second line encodes that $\gamma(f)$ is a \AP\intro(gf){Cartesian morphism}, hence we will refer to it as a \AP\intro(gf){Cartesian lifting} of $f$.
\end{definition}

\begin{definition}\label{def:grothendieckOpfibration}
A functor $G:\bfE\rightarrow \bfB$ is a \AP\intro{Grothendieck opfibration} if the following property holds:
\begin{equation}\label{eq:def:grothendieckOpfibration}
\ti{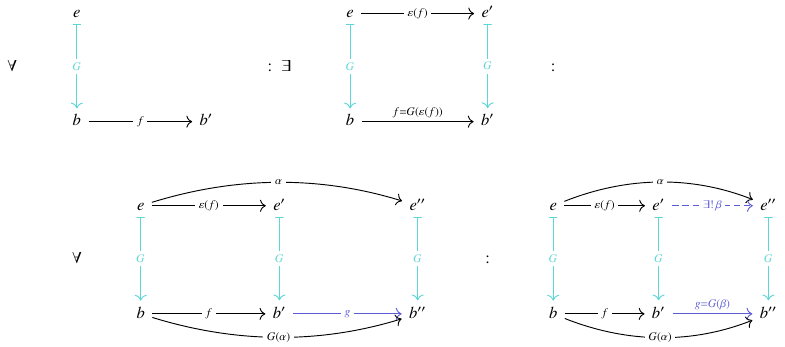}
\end{equation}
The second line encodes that $\varepsilon(f)$ is an \AP\intro(gopf){op-Cartesian morphism}, so we will refer to it as an \AP\intro(gopf){op-Cartesian lifting} of $f$.
\end{definition}

It is well known that the above definitions imply that the (op-)Cartesian lifting of $f$ is essentially unique, i.e., unique up to unique isomorphism; see, for example, Proposition 1.1.4 of \cite{jacobs1999categorical}.
\subsection{Multi-opfibrations}\label{sec:mof}

The definition of a \kl{Grothendieck opfibration} may be generalized in the following form, whereby instead of requiring the existence of \kl(gopf){op-Cartesian lifts} neither existence nor essential uniqueness are required. This particular variant of a fibrational structure postulates instead the existence of a (possibly empty) family of \kl(mof){multi-op-Cartesian lifts}, subject to a somewhat more intricate universal property. As will be demonstrated in Section~\ref{sec:targetFunctors}, this generalized notion is the appropriate fibrational concept capable of formalizing so-called \kl{multi-initial pushout complements}, which in turn play a key role in categorical rewriting semantics. 

\begin{definition}\label{def:multi-opfibration}
A functor $M:\bfE\rightarrow \bfB$ is a \AP\intro{multi-opfibration} if the following property holds:
\begin{equation}\label{def:multi-opfibration}
\begin{aligned}
&\forall \ti{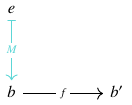}%
\; : \; \exists \left\lbrace 
\ti{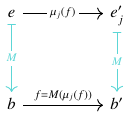}
\right\rbrace_{j\in J_{f;e}}\;:\\
&\qquad \forall \ti{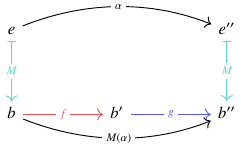}%
\; :\;%
\ti{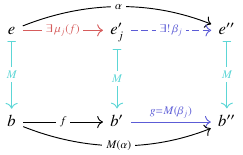}:\\
&\qquad\forall%
\ti{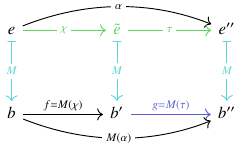}:\; \ti{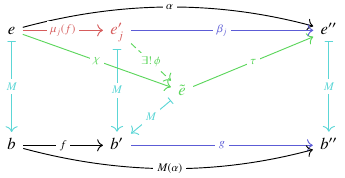}
\end{aligned}
\end{equation}
In words: 
\begin{enumerate}[label=(\roman*)]
\item For every $b-f\rightarrow b'$ in $\bfB$ and $e\in \bfE$ with $M(e)=b$, %
there exists a (possibly empty) family $J_{f;e}$ of \AP\intro(mof){multi-op-Cartesian liftings} $e-\mu_j (f)\rightarrow e_f'$ (with $M(\mu_j(f))=f$). %
\item \AP\intro(mof){Universal property of multi-opfibrations}: Multi-op-Cartesianity of the liftings entails that for all $e-\alpha\rightarrow e''$ in $\bfE$ and $b'-g\rightarrow b''$ in $\bfB$ with $M(\alpha)=g\circ f$, %
there exists a $j\in J_{f;e}$ such that %
there exists a unique $e_j'-\beta_j\rightarrow e''$ with $\alpha =\beta_j\circ \mu_j(f)$ %
and $M(\beta_j)=g$. %
\item \AP\intro(mof){Essential uniqueness}: for a given \kl(mof){multi-op-Cartesian lifting} as in (ii), if there exists some $e-\chi\rightarrow \tilde{e}$ and %
$\tilde{e}-\tau\rightarrow e''$ such that %
$\tau\circ \chi = \alpha$, $M(\chi)=f$ and $M(\tau)=g$, then %
there exists a unique morphism $e'_j-\phi\rightarrow \tilde{e}$ such that %
$\chi=\phi\circ \mu_j(f)$, $\beta_j=\tau\circ \phi$, %
and $M(\phi)=id_{b'}$.
\end{enumerate}
We say that a \kl{multi-opfibration} is \AP\intro(mof){strong} if the morphisms $\phi$ in (iii) above are \emph{isomorphisms}.
\end{definition}

It is useful to note that a \kl{Grothendieck opfibration} is a special case of a \kl{multi-opfibration}, namely when for every $f\in \bfB$, the family of \kl(mof){multi-op-Cartesian lifts} is non-empty, and such that all members of a given family are in the same equivalence class under the universal property (i.e., for all $j,k\in J_{e;f}$, there exists an isomorphism $e'_j-\phi_{jk}\to e'_j$ such that $M(\varphi_{jk})=id_{b'}$ and $\mu_k(f)=\varphi_{jk}\circ \mu_j(f)$). However, a \kl{Grothendieck opfibration} is in general \emph{not} a special case of a \emph{\kl(mof){strong}} \kl{multi-opfibration}. This is relevant since \kl(mof){strong} \kl{multi-opfibrations} enjoy two important technical properties (\kl(mof){isomorphism lifting} and \kl(mof){pullback lifting}, see below) that are crucial for obtaining compositional rewriting theories:

\begin{lemma}\label{lem:mofIsoLifting}
Let $M:\bfE\rightarrow \bfB$ be a \kl(mof){strong} \kl{multi-opfibration}. Then the following \AP\intro(mof){lifting property of isomorphisms} is satisfied:
\begin{equation}\label{eq:lem:mofIsoLifting}
\begin{aligned}
&\forall\;  
\ti{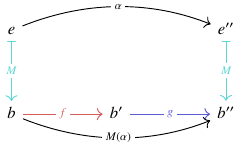}%
\; : \;\forall\; %
\ti{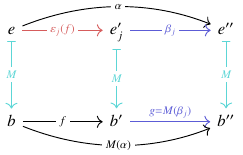}\;:\\
&\qquad  (g\in \iso{\bfB}\;\Rightarrow \; \beta_j\in \iso{\bfE}) \land 
(f\in \iso{\bfB}\;\Rightarrow \; \varepsilon_j(f)\in \iso{\bfE})
\end{aligned}
\end{equation}
\end{lemma}

We conclude the general discussion of multi-opfibrations with the following technical result which will be used in the proof of the \kl{associativity theorem} for compositional rewriting theories in Section~\ref{sec:comp:at}:

\begin{lemma}[Pullback-lifting lemma for strong multi-opfibrations]\label{lem:pbsmopf}
\AP\phantomintro(mof){Pullback-lifting lemma for strong multi-opfibrations}
Let $\bfE$ be a category that \kl{has pullbacks}, and let $M:\bfE\rightarrow \bfB$ be a \kl(mof){strong} \kl{multi-opfibration}. Then the following property holds:
\begin{equation}\label{eq:mofPBsplittingLemma}
\forall\;\ti{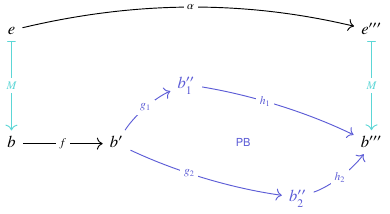}\;:\; 
\ti{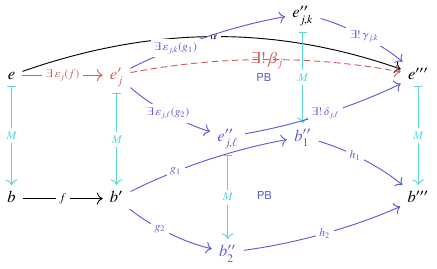}
\end{equation}
More explicitly, for every diagram such as on the left of~\eqref{eq:mofPBsplittingLemma}, whose bottom part contains a pullback square in $\bfB$, the following properties hold:
\begin{enumerate}[label=(\roman*)]
\item There exists an $\bfE$-morphism $e-\varepsilon_j(f)\rightarrow e_j'$ such that there exists a unique $\bfE$-morphism $e_j'-\beta_j\rightarrow e'''$ with $M(\varepsilon_j(f))=f$ and $M(\beta_j)=h_1\circ g_1=h_2\circ g_2$, and such that the diagram commutes.
\item There then exist $\bfE$-morphisms $e_j'-\varepsilon_{j,k}(g_1)\rightarrow e''_{j,k}$ and $e_j'-\varepsilon_{j,\ell}(g_2)\rightarrow e''_{j,\ell}$ such that there exist unique $\bfE$-morphisms $e''_{j,k}-\gamma_{j,k}\rightarrow e'''$ and $e''_{j,\ell}-\delta_{j,\ell}\rightarrow e'''$ such that $M(\varepsilon_{j,k}(g_1))=g_1$, $M(\varepsilon_{j,\ell}(g_2))=g_2$, $M(\gamma_{j,k})=h_1$ and $M(\delta_{j,\ell})=h_2$, and such that the diagram commutes.
\item Moreover, the square in $\bfE$ into $e'''$ is a pullback.
\end{enumerate}
\end{lemma}
\subsection{Residual multi-opfibrations}\label{sec:rmof}

The following concept constitutes yet a further generalization of fibrational concepts -- while \kl{multi-opfibrations} generalize \kl{Grothendieck opfibrations} via replacing (essentially unique) \kl(gopf){op-Cartesian lifts} with \kl(mof){multi-op-Cartesian lifts}, one encounters in compositional rewriting theory situations where moreover morphisms may in general not possess such liftings, but only certain \emph{extensions} of morphisms, referred to as \kl(rmof){residues} in the definition below. The reason for introducing this concept will become evident only when considering the salient examples of fibrational properties of \kl{final pullback complement} squares, and of \kl(crsSqPO){sesqui-pushout direct derivations} in the later parts of this paper.

\begin{definition}\label{def:residualMultiOpfibration}
A functor $R:\bfE\rightarrow \bfB$ is a \AP\intro{residual multi-opfibration} if the following property holds:
\begin{equation}\label{eq:def:residualMultiOpfibration}
\begin{aligned}
&\forall \ti{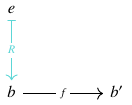}%
\; : \; \exists \left\lbrace 
\ti{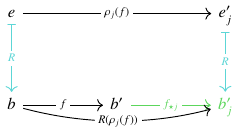}\right\rbrace_{j\in J_{f;e}}\;:\\
&\qquad \forall \ti{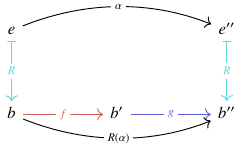}%
\; : \; %
\ti{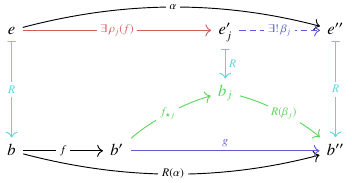}\;:\\
&\qquad \forall \ti{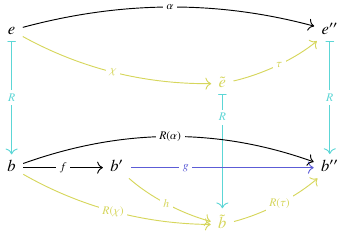}\;:\; \ti{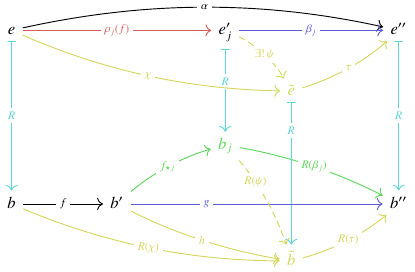}
\end{aligned}
\end{equation}
In words: 
\begin{enumerate}[label=(\roman*)]
\item For every $b-f\rightarrow b'$ in $\bfB$ and $e\in \bfE$ with $R(e)=b$, %
there exists a (possibly empty) family $J_{f;e}$ of \AP\intro(rmof){residual multi-op-Cartesian liftings} $e-\rho_j (f)\rightarrow e_j'$ %
(with $R(\rho_j(f))=f_{\star j}\circ f$, and with $f_{\star j}$ referred to as a %
\AP\intro(rmof){residue} with respect to $(e;f)$).%
\item \AP\intro(rmof){Universal property of residual multi-opfibrations}: Residual multi-op-Cartesianity of the liftings entails that for all $e-\alpha\rightarrow e''$ in $\bfE$ and $b'-g\rightarrow b''$ in $\bfB$ with $R(\alpha)=g\circ f$, %
there exists a $j\in J_{f;e}$ such that %
there exists a unique $e_j'-\beta_j\rightarrow e''$ with %
$\alpha =\beta_j\circ \rho_j(f)$ and %
$g = R(\beta_j)\circ f_{\star j}$. %
\item \AP\intro(rmof){Essential uniqueness}: For all $b-h\rightarrow \tilde{b}$, $e-\chi\to \tilde{e}$ and $\tilde{e}-\tau\to e''$ such that $R(\chi)=h\circ f$ and $g=R(\tau)\circ h$, %
there exists a unique $e'_j-\psi\to \tilde{e}$ such that %
$\chi=\psi\circ \rho_j(f)$, $\beta_j=\tau\circ \psi$ (which then implies moreover that $h=R(\psi)\circ f_{\star_j}$ and $R(\beta_j)=R(\tau)\circ R(\beta_j)$).
\end{enumerate}
\end{definition}

We record the following technical result for \kl{residual multi-opfibrations} which will be crucial later in the paper when it plays a central role in the proof of the \kl{associativity theorem} of Section~\ref{sec:comp:at}:
\begin{lemma}\label{lem:rmopfUP}
Let $R:\bfE\rightarrow \bfB$ be a \kl{residual multi-opfibration}. Then \kl(rmof){residues} have the following \AP\intro(rmofrup){universal property}:
\begin{equation}\label{eq:corUPrmof}
\forall \ti{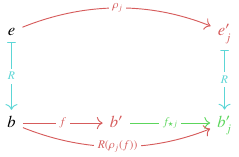}%
\; : \; \exists %
\ti{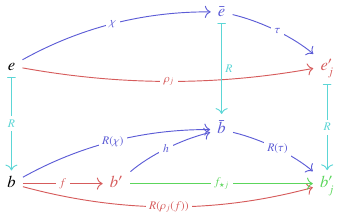}\;\Rightarrow \tau\in \iso{\bfE}\land R(\tau)\in \iso{\bfB}
\end{equation}
In particular, this property entails that if a \kl(rmof){residue} $f_{\star k}$ factorizes a \kl(rmof){residue} $f_{\star j}$ as $f_{\star j}=R(\beta_k)\circ f_{\star k}$ for some $\beta_k\in \bfE$, then the residues $f_{\star j}$ and $f_{\star_k}$ (both of the same morphism $f\in \bfB$) are related by an isomorphism $R(\beta_k)\in \iso{\bfB}$, as are their liftings $\rho_j(f)=\beta_k\circ \rho_k(f)$ via $\beta_k\in \iso{\bfE}$.
\end{lemma}

Finally, we note the following specialization of \kl{residual multi-opfibration}:
\begin{definition}
A functor $R:\bfE\rightarrow \bfB$ is a \AP\intro{residual opfibration} if it is a \kl{residual multi-opfibration} such that for all objects $e$ of $\bfE$ and morphisms $R(e)-f\to b$ of $\bfB$ the family of \kl(rmof){residual multi-op-Cartesian lifts} is non-empty, and such that all lifts in the family are equivalent up to universal isomorphisms (i.e., for all $(\rho_j(f),f_{\star_j})$ and $(\rho_k(f),f_{\star_k})$, there exists a unique isomorphism $\varphi$ in $\bfE$ such that $\rho_k(f)=\varphi\circ \rho_j(f)$ and $f_{\star_k}=R(\varphi)\circ f_{\star_j}$). We will sometimes refer to such lifts as \kl(rof){residual op-Cartesian} for brevity.
\end{definition}

\section{Fundamentals of compositional rewriting theories}\label{sec:comp}

Taking the notion of \kl{double categories} as a convenient ``book-keeping'' device, we will demonstrate in this key section of the present paper that a very general class of compositional rewriting theories---including in particular the ``non-linear'' variants of DPO- and SqPO-semantics~\cite{BehrHK21}---may be elegantly expressed and studied from a fibrational viewpoint. More precisely, based upon and motivated by the fibrational structures presented in Section~\ref{sec:fib}, we introduce the novel notion of \textbf{\kl{compositional rewriting double category (crDC)}} . We then demonstrate that \kl{crDCs} provide a very high-level representation of categorical rewriting theories with compositionality properties in the sense that every \kl{crDC} admits a \kl{concurrency theorem} (Section~\ref{sec:comp:cct}) and an \kl{associativity theorem} (Section~\ref{sec:comp:at}). The crucial point of our novel approach to proving compositional properties via \kl{crDCs} is that the aforementioned \kl{concurrency theorems} and \kl{associativity theorems} may be established in an entirely \emph{universal} form, i.e., entirely independently of the concrete rewriting semantics underlying a given \kl{crDC}.

\subsection{Double categories}\label{sec:crDC}
Throughout this paper, we work exclusively with the ``algebraic'' order in compositions of morphisms and commutative squares (i.e., $g\circ f$ rather than the ``diagram'' order notation $f;g$ common in category theory). For reasons of convenience, we will swap the roles of the classes of morphisms that have a weakly associative composition, usually the \emph{vertical} morphisms~\cite{grandis1999limits}, to be the class of \emph{horizontal} morphisms. (We opted for this particular convention so that it is essentially a 90 degrees clockwise rotation of the standard mathematical one.) Finally, since we will be exclusively interested in \emph{finitary} categories, we will often not mention finitarity explicitly in what follows.

\begin{definition}[Cf.\ e.g.\ \cite{grandis1999limits, kelly1974review, fiore2007pseudo}]\label{def:pcd}
A \AP\intro{double category (DC)} $\bD$ is a weakly internal category in the $2$-category $\mathcal{CAT}$ of all categories~\cite{hansen2019constructing}\footnote{Some authors prefer the term ``pseudo double category'', cf.\ also \href{https://ncatlab.org/nlab/show/double+category}{nLab} article on double categories.}.
\end{definition}

In particular, this entails that a double category consists of a category $\bD_0$ of \AP\intro(dc){objects} and \AP\intro(dc){vertical morphisms}, and a category $\bD_1$ of \AP\intro(dc){horizontal morphisms} and \AP\intro(dc){squares} of $\bD$, equipped with functors $S,T:\bD_1\rightarrow \bD_0$, referred to as \AP\intro(dc){source} and \AP\intro(dc){target} functors, respectively (cf.\ Figure~\ref{fig:dcDef:ST}), and with a functor $U:\bD_0\rightarrow \bD_1$ which maps every object $A$ of $\bD_0$ to a \AP\intro(dc){horizontal unit} $U_A$ (depicted in Figure~\ref{fig:pcdDef:horizontalUnitarity} as identity horizontal morphisms), and every morphism $f$ of $\bD_0$ to a \AP\intro(dc){horizontal unit square} $U_f$ (depicted in Figure~\ref{fig:pcdDef:horizontalUnitarity} as \kl(dc){squares} annotated with the symbol $id_{\dotsc}$ for better readability). %
We denote vertical morphisms by $\rightarrowtail$ and horizontal morphisms by $\leftharpoonup$, respectively. %
We denote by $\vComp$ the \AP\intro(dc){vertical composition} of squares as in Figure~\ref{fig:pdcDef:verticalComposition} (i.e., the associative composition operation of $\bD_1$). %
$\bD$ moreover carries a weakly associative \AP\intro(dc){horizontal composition} of squares (cf.\ Figure~\ref{fig:pdcDef:horizontalComposition}) $\hComp:\bD_1\times_{\bD_0}\bD_1\rightarrow \bD_1$. %
Finally, for technical convenience, we assume without loss of generality\footnote{We follow here the viewpoint of~\cite{FIORE20111174}, whereby utilizing the strictification theorem of pseudo double categories \cite[Thm. 7.5]{grandis1999limits}, this amounts to implicitly utilizing a pseudo-functor into an equivalent double category where unitarity is strict, thus not reducing generality of our constructions.} that both types of compositions are strictly unitary (cf.\  Figures~\ref{fig:pcdDef:verticalUnitarity} and~\ref{fig:pcdDef:horizontalUnitarity}).

\begin{figure}
\centering
	$
	\ti{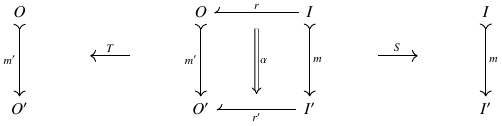}
	$
\caption{\label{fig:dcDef:ST}Convention for source and target functors for double categories.}
\end{figure}

\begin{figure}
\centering
 \subfigure[\label{fig:pdcDef:verticalComposition}Vertical composition $\diamond_v$.]{%
	$
	\ti{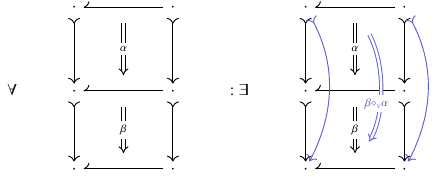}
	$
}%
\\[1.5em]
\subfigure[\label{fig:pdcDef:horizontalComposition}Horizontal composition $\diamond_h$.]{%
	$
	\ti{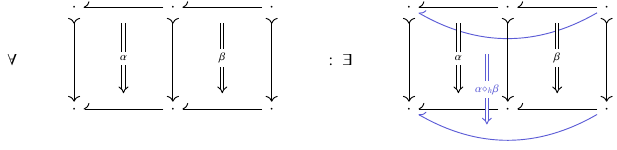}
	$
}%
\\[1.5em]
\subfigure[\label{fig:pcdDef:verticalUnitarity}(Strict) vertical unitarity.]{%
	$
	\ti{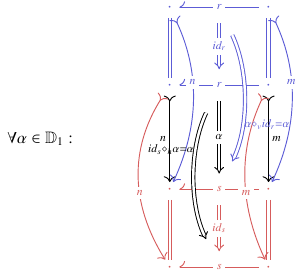}
	$
}%
\\[1.5em]
\subfigure[\label{fig:pcdDef:horizontalUnitarity}(Strict) horizontal unitarity.]{%
	$ 
	\ti{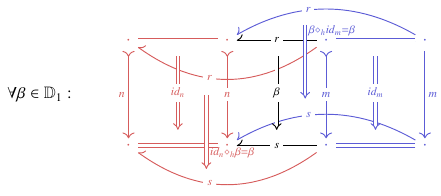}
	$
}%
\caption{On the definition of \kl{double categories}.}
\label{fig:pcdDef}
\end{figure}
\subsection{Compositional rewriting double categories}\label{sec:comp:crDCs}

\begin{definition}\label{def:crc}
A double category (DC) $\bD$ is a \AP\intro{compositional rewriting DC (crDC)} if it has the following properties:
\begin{enumerate}[label=(\roman*)]
\item \AP\phantomintro(crDC){$\bD_0$ has multi-sums}$\bD_0$ \kl(ms){has multi-sums}.
\item \AP\phantomintro(crDC){$\bD_0$ and $\bD_1$ have pullbacks}$\bD_0$ and $\bD_1$ \kl{have pullbacks}. 
\item \AP\phantomintro(crDC){$\hComp:\bD_1\times_{\bD_0}\bD_1\rightarrow \bD_1$ is an isoglobular residual opfibration}%
The \kl(dc){horizontal composition} functor $\hComp:\bD_1\times_{\bD_0}\bD_1\rightarrow \bD_1$ is an \AP\intro{isoglobular residual opfibration}, namely a \kl{residual opfibration} such that\footnote{The definition in fact amounts to a special case of a so-called Street opfibration; this aspect and further variations of fibrational structures in rewriting theory are studied in~\cite{BMZ-2023}.} all \kl(rmof){residues} are \AP\intro(dc){globular isomorphisms} (i.e., isomorphisms $\varphi$ of $\bD_1$ such that $T(\varphi)$ and $S(\varphi)$ are identity morphisms). 
\item \AP\phantomintro(crDC){$S:\bD_1 \rightarrow \bD_0$ is a strong multi-opfibration}The \kl(dc){source functor} $S:\bD_1\rightarrow \bD_0$ is a \kl(mof){strong} \kl{multi-opfibration}.
\item \AP\phantomintro(crDC){$T:\bD_1 \rightarrow \bD_0$ is a residual multi-opfibration} The \kl(dc){target functor} $T:\bD_1\rightarrow \bD_0$ is a \kl{residual multi-opfibration}.
\end{enumerate}
\end{definition}

\begin{remark}\label{rem:hCompDom}
It is worthwhile unpacking the fibrational property of the \kl(dc){horizontal composition} functor in a \kl{crDC} into a more explicit form in view of later applications in the proof of the \kl{associativity theorem} (cf.\ Section~\eqref{sec:comp:at}):
\begin{itemize}
\item Squares in a \kl{crDC} $\bD$ have the following \AP\intro(crDC){horizontal decomposition property}:
\begin{equation}\label{eq:def:crDChorDecomp}
\begin{aligned}
&\forall \ti{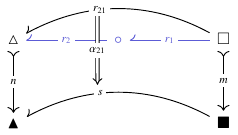}\;:\; 
\exists \; \ti{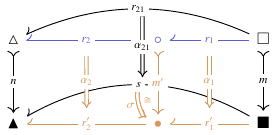}\;:\; \alpha_2\hComp \alpha_1 = \sigma\vComp \alpha_{21}
\end{aligned}
\end{equation}
In particular, utilizing the notation $\AP\intro{\gcong}$ for equality up to \kl(dc){globular isomorphism}, one has $s\kl{\gcong}r_2'\hComp r_1'$.
\item Unpacking the definition of the \kl(rmof){universal property of residual multi-opfibrations} for the case at hand, we find the following \AP\intro(crDC){complex decomposition property} (where in the diagram below $\gamma_2\hComp\gamma_1=\beta_{21}\vComp\alpha_{21}$):
\begin{equation}
\begin{aligned}
&\forall \ti{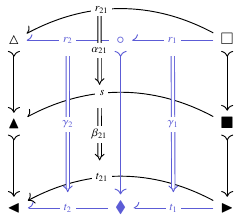}\;:\; 
\exists \; \ti{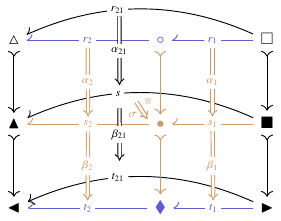}\;:\; \begin{aligned}
\alpha_2\hComp \alpha_1 &= \sigma\vComp \alpha_{21}\\
\land\; \beta_2\hComp \beta_1&= \beta_{21}\vComp \sigma^{-1}
\end{aligned}
\end{aligned}
\end{equation}
\end{itemize}
\end{remark}

\subsection{Concurrency theorem}\label{sec:comp:cct}
Let us finally put the fibrational structures, introduced in Section~\ref{sec:fib}, and the above concept of \kl{compositional rewriting double category} to use by proving a first theorem---the \kl{concurrency theorem}---that plays a key role in the static analysis of rewriting systems.

\begin{theorem}\label{thm:ccThm}
\phantomintro{Concurrency theorem}
Let $\bD$ be a \kl{compositional rewriting double category}. Then the following statements hold:
\begin{equation}\label{eq:concurrencyThm}
\ti{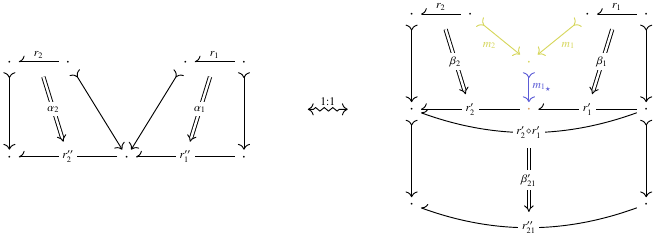}
\end{equation}
\begin{itemize}
\item \textbf{\AP\intro(cct){Synthesis}:} %
For every pair of \kl(dc){squares} $(\alpha_1,\alpha_2)$ of $\bD$ ``adjacent at the foot'' (left diagram above), %
there exist a \kl{multi-sum} element $(m_2,m_1)\in \Msum{S(r_2)}{T(r_1)}$, a \kl(rmof){residue} $m_{1_{\star}}$ for $(r_1;m_1)$, and \kl(dc){squares} $(\beta_1,\beta_2,\beta_{21}')$ (with $T(\beta_1)=m_{1_{\star}}\circ m_1$ and $S(\beta_2)=m_{1_{\star}}\circ m_2$), uniquely determined up to universal isomorphisms, %
such that $r_{21}''\kl{\gcong} r_2''\hComp h_1''$.
\item \textbf{\AP\intro(cct){Analysis}:} For every \kl{multi-sum} element $(m_2,m_1)\in \Msum{S(r_2)}{T(r_1)}$,  \kl(rmof){residue} $m_{1_{\star}}$ for $(r_1;m_1)$, and \kl(dc){squares} $(\beta_1,\beta_2,\beta_{21}')$ (with $T(\beta_1)=m_{1_{\star}}\circ m_1$ and $S(\beta_2)=m_{1_{\star}}\circ m_2$), there exist \kl(dc){squares} $(\alpha_1,\alpha_2)$ of $\bD$, determined uniquely up to universal isomorphisms,  such that %
$r_2''\hComp r_1'' \kl{\gcong} r_{21}''$, where $r_2''=codom(\alpha_2)$, $r_1''=codom(\alpha_1)$, and $r_{21}''=codom(\beta_{21}')$.
\end{itemize}
Consequently, modulo a suitable notion of isomorphisms (induced by essential uniqueness of the respective constructions), the resulting sets of equivalence classes are isomorphic. 
\end{theorem}
\begin{proof}
\textbf{Synthesis part:} Construct the diagram in~\eqref{eq:cctSynthesisPartProof} from the premise as follows:
\begin{itemize}
\item Via the \kl(ms){universal property of multi-sums}, there exists a cospan of $\bD_0$-morphisms $(m_2,m_1)$ into an object $\lozenge$ and a mediating $\bD_0$-morphism $\lozenge \rightarrowtail \cdot$.
\item Since \kl(crDC){the target functor $T:\bD_1\rightarrow \bD_0$ is a residual multi-opfibration}, there exists a residue $m_{1_{\star}}: \lozenge \rightarrowtail \blacklozenge$ with respect to $(r_1;m_1)$ and a $\bD_0$-morphism $\blacklozenge \rightarrowtail \cdot$ such that $T(\beta_1)= m_{1_{\star}}\circ m_1$ and $\alpha_1 = \beta_1' \vComp \beta_1$.
\item Since \kl(crDC){the source functor $S:\bD_1\rightarrow \bD_0$ is a multi-opfibration}, there exist squares $\beta_2$ and $\beta_2'$ such that $S(\beta_1)=m_{1_{\star}}\circ m_2$ and $\alpha_2=\beta_2'\vComp \beta_2$. Thus the claim follows by letting $\beta_{21}':= \beta_2' \hComp \beta_1'$.
\end{itemize}

\begin{equation}\label{eq:cctSynthesisPartProof}
\ti{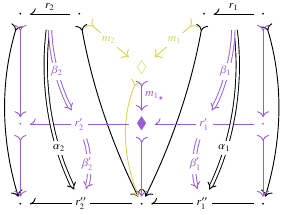}
\end{equation}

\textbf{Analysis part:} Construct the diagram in~\eqref{eq:cctProof:analysisPart} as follows:
\begin{itemize}
\item By the \kl(crDC){horizontal decomposition property} of squares in $\bD$, there exist squares $\beta_2'$ and $\beta_1'$ such that $\sigma\circ \beta_{21}' =  \beta_2' \hComp \beta_1'$ (for some \kl(dc){globular isomorphism} $\sigma$).
\item The claim follows be letting $\alpha_i := \beta_i'\vComp \beta_i$ for $i=1,2$, since $r_2''\hComp r_1''\kl{\gcong} r_{21}''$.
\end{itemize}

\begin{equation}\label{eq:cctProof:analysisPart}
\ti{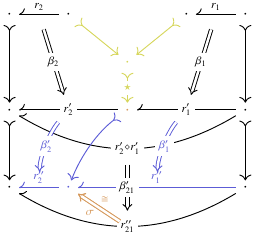}
\end{equation}
\qed \end{proof}

While an interesting mathematical structure on double categories in its own right, the deeper meaning and importance of the \kl{concurrency theorem} in formulating compositional rewriting semantics becomes apparent when interpreting squares in a \kl{compositional rewriting double category} as rewriting steps (traditionally referred to as \emph{direct derivations} in rewriting theory). To this end, consider a two-step rewriting sequence, where the result of the first step is the starting object of the second step, as depicted in the form of the squares $\alpha_2$ and $\alpha_1$ in~\eqref{eq:concurrencyThm}. The \kl{concurrency theorem} then implies that there exists a one-step rewrite, depicted as the square $\beta_{21}$ in~\eqref{eq:concurrencyThm}, from the start object to the end object of the two-step sequence, and uniquely determined up to universal isomorphisms. Moreover, the particular property of the one-step rewrite operation is that it is taken along some \AP\intro(crDC){composite rule} (here: $r_{2}'\hComp r_1'$), which----again up to universal isomorphisms---is uniquely determined from the data of the two-step rewrite sequence. One may thus interpret the top half of the right diagram in~\eqref{eq:concurrencyThm} (i.e., the squares $\beta_2$ and $\beta_1$, the $\bD_1$-object $r_{2}'\hComp r_1'$ as well as the \kl{multi-sum} and \kl(rmof){residue}) as defining a notion of \AP\intro(crDC){rule composition}. Indeed, as we shall illustrate in Section~\ref{sec:nlCRS}, the abstract \kl{crDC}-based formulation of the \kl{concurrency theorem} instantiates precisely to the traditional concepts of concurrency and rule compositions when considering \kl(crsType){linear} \kl{Double-Pushout semantics}~\cite{ehrig:2006fund}, but also provides an abstraction of compositional rewriting for more general \kl(crsType){semi-linear} and \kl(crsType){generic} \kl{Double-Pushout} as well as \kl{Sesqui-Pushout semantics}, as first introduced in~\cite{BehrHK21}.\\

Let us finally note that in comparison to the instantiations of compositional rewriting theories to concrete choices of semantics, our abstract \kl{crDC}-based approach as presented here allows an efficient modularization of the proof of the \kl{concurrency theorem} by clearly separating the concrete definitions of compositional rewriting theories (i.e., proving that a certain semantics and choice of base category gives rise to a \kl{crDC}) from the universal structures provided by \kl{crDCs}.

\subsection{Associativity theorem}\label{sec:comp:at}

Unlike for the case of the \kl{concurrency theorem}, the statement and proof of which were a straightforward and efficient application of the multi-sum and fibrational concepts, our second main theorem has indeed such a complex statement that its proof relies much more non-trivially upon fibrational structures---which serve in a certain sense as a form of ``proof macros''---without which the proof would be extremely long and presumably difficult to follow.

\begin{theorem}\label{thm:assoc}
\AP\phantomintro{Associativity theorem for compositional rewriting double categories}
Let $\bD$ be a \kl{compositional rewriting double category}. Then every diagram as in~\eqref{eq:assocThmA} below (interpreted as encoding a \kl(crDC){composition} of rules $r_2$ and $r_1$, and of the \kl(crDC){composite} with $r_3$),
\begin{equation}\label{eq:assocThmA}
\ti{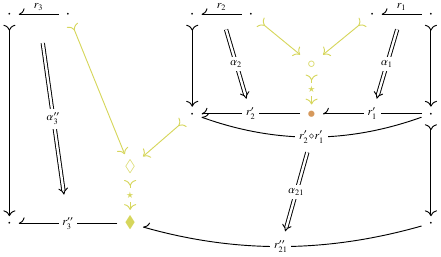}
\end{equation}
determines uniquely up to universal isomorphisms a diagram as in~\eqref{eq:assocThmPartB} below (interpreted as encoding a \kl(crDC){composition} of rules $r_3$ with $r_2$, and of the \kl(crDC){composite} with $r_1$), and vice versa:
\begin{equation}\label{eq:assocThmPartB}
\ti{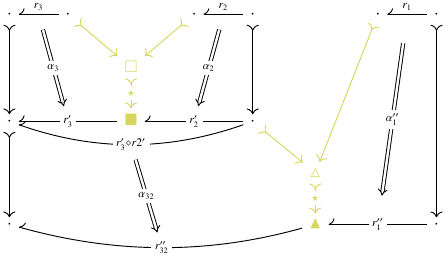}
\end{equation}
Moreover, the equivalence is such that in addition
\begin{equation}
r_3'' \hComp r_{21}'' \cong r_{32}''\hComp r_1''\,.
\end{equation}
Thus for a suitable notion of equivalence up to isomorphisms (induced by the essential uniqueness of the respective constructions), there exists an isomorphism between the sets of equivalence classes of nested \kl(crDC){composites} of the three rules in the two different nesting orders. This amounts to a notion of \emph{associativity} for the \kl(crDC){rule composition} operation.
\end{theorem}
\begin{proof}
For the $\Rightarrow$ direction of the equivalence, construct the following diagram from the premise by applying the \kl(crDC){horizontal decomposition property} to the square $\alpha_{21}$, obtaining a globular isomorphism $\sigma_{21}$ in $\bD_1$ and squares $\alpha_2'$ and $\alpha_1'$ such that $\sigma_{21}\vComp \alpha_{21}=\alpha_2\hComp \alpha_1$ (and with $r_2''\hComp r_1''\kl{\gcong}r_{21}''$):
\begin{equation}
\ti{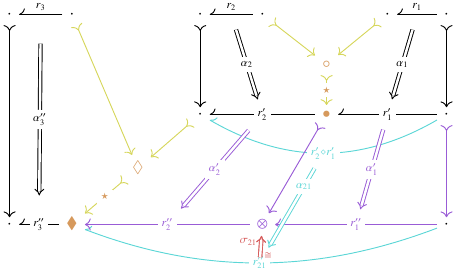}
\end{equation}

Next, we apply a part of the \kl(cct){synthesis} construction of the \kl{concurrency theorem}, in that we synthesize from the squares $\alpha_3''$ and $\alpha_2'\vComp \alpha_2$ a \kl(crDC){composite} of rules $r_3$ and $r_2$ (as encoded via the squares $\beta_3$ and $\beta_2$, with the \kl(crDC){composite rule} itself omitted for clarity) and squares $\beta_3'$ and $\beta_2'$ such that $\alpha_3''=\beta_3'\vComp \beta_3$ and $\alpha_2'\vComp \alpha_2 = \beta_2'\vComp \beta_2$:
\begin{equation}
\ti{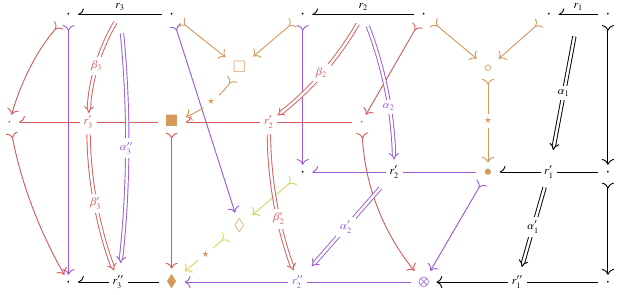}
\end{equation}

By applying the \kl(cct){synthesis} construction of the \kl{concurrency theorem} to the pair of squares $\beta_3'\hComp\beta_2'$ and $\alpha_1'\vComp \alpha_1$, we may obtain the diagram in~\eqref{eq:assocThm:PartA:step2B} below (where $\gamma_{32}'\vComp\gamma_{32}=\beta_3'\hComp \beta_2'$, and $\gamma_1'\vComp\gamma_1= \alpha_1'\vComp\alpha_1$):

\begin{equation}\label{eq:assocThm:PartA:step2B}
\ti{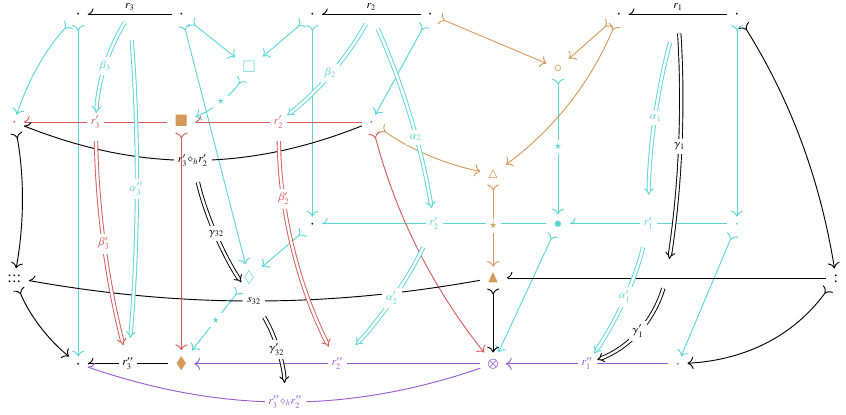}
\end{equation}

We now apply the following sequence of manipulations to obtain the diagram in~\eqref{eq:assocThm:PartA:step3}:
\begin{itemize}
\item Via the \kl(ms){multi-sum extension Lemma}, there exists a $\bD_0$-morphism $\circ \rightarrowtail \triangle$ between the multi-sum objects $\circ$ and $\triangle$ (and analogously between the multi-sum objects $\square$ and $\lozenge$, albeit this is irrelevant for the proof and thus omitted from the diagrams).
\item Via the \kl(crDC){complex decomposition property},  there exist squares $\gamma_3$, $\gamma_3'$, $\gamma_2'$ and $\gamma_2$ and a \kl(dc){globular isomorphism} $\tau_{32}$ such that $\gamma_3\hComp\gamma_2= \tau_{32}\vComp \gamma_{32}$ and $\gamma_3'\hComp \gamma_2'= \gamma_{32}'\vComp \tau_{32}^{-1}$.
\end{itemize}

\begin{equation}\label{eq:assocThm:PartA:step3}
\ti{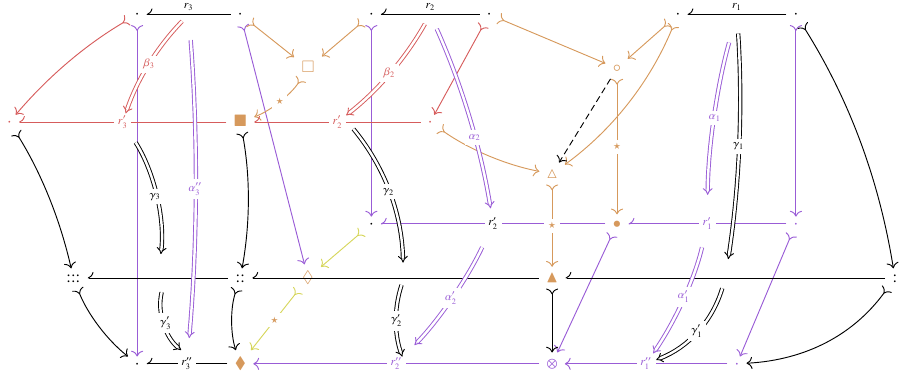}
\end{equation}
For the final step of the $\Rightarrow$ part of the proof, we construct the diagram in~\eqref{eq:assocThm:PartA:step4} below via the following steps:
\begin{itemize}
\item Take a pullback (admissible since \kl(crDC){$\bD_0$ has pullbacks}) in order to obtain the object marked $;$ on the back rightmost part of the diagram in~\eqref{eq:assocThm:PartA:step4}, yielding a number of morphisms as indicated (all of which are in $\bD_0$, again since \kl(crDC){$\bD_0$ has pullbacks}). 
\item Since \kl(crDC){the source functor is a strong multi-opfibration}, by applying the \kl(mof){pullback lifting lemma for strong multi-opfibrations} we obtain squares $\delta_1$, $\delta_1'$ and $\varepsilon_1$ such that $\alpha_1=\varepsilon_1 \vComp \delta_1$ and $\gamma_1=\delta_1'\vComp \delta_1$. The lemma also implies that since the square from the object marked $;$ was by construction a pullback, so is the square from $\triangledown$, which by the \kl{universal property of pullbacks} yields the existence of a morphism into $\triangledown$ (marked $+$), which is a $\bD_0$-morphism since \kl(crDC){$\bD_0$ has pullbacks}.
\item Applying the \kl(mof){pullback lifting lemma for strong multi-opfibrations} once again, we may obtain the configuration in the middle of the diagram in~\eqref{eq:assocThm:PartA:step4}, i.e., squares $\delta_2$, $\delta_2'$ and $\varepsilon_2$ such that $\alpha_2=\varepsilon_2\vComp \delta_2$ and $\gamma_2\vComp\beta_2 =\delta_2'\vComp\delta_2$. The lemma also entails that since the commutative square from $\triangledown$ is a pullback, the square from the object marked $;;$ is a pullback, too, and there exists the $\bD_0$-morphism $codom(r_2)\rightarrowtail ;;$.
\item By the \kl(rmofrup){universal property of residues}, since $\alpha_1=\varepsilon_1\vComp\delta_1$, and since the residue $\circ\rightarrowtail \bullet$ marked $\star$ (which forms the second factor of $T(\alpha_1)$) factors through $\triangledown \rightarrowtail \bullet$ (i.e., through $T(\varepsilon_1)$), we find that the square $\varepsilon_1$ is an isomorphism in $\bD_1$. By the \kl(mof){lifting property of isomorphisms for strong multi-opfibrations}, the square $\varepsilon_2$ is then an isomorphism in $\bD_1$, too.
\item The latter point entails that we may form the cospan $\cdot -S(\gamma_3\vComp\beta_3)\rightarrow :: \leftarrow T(\delta_2'\vComp \varepsilon_2^{-1})-\cdot$ of morphisms in $\bD_0$; hence by the \kl(ms){universal property of multi-sums}, there exists a $\bD_0$-morphism $\lozenge \rightarrowtail ::$.
\item The existence of the morphism $\lozenge \rightarrowtail ::$ together with $\alpha_2'\hComp \alpha_1' = (\gamma_2'\hComp \gamma_1') \vComp (\delta_2'\hComp\delta_1')\vComp 
(\varepsilon_2^{-1}\hComp\varepsilon_1^{-1})$ implies via the \kl(rmofrup){universal property of residues} that $\gamma_2'\hComp \gamma_1'$ is an isomorphism in $\bD_1$. 
\item Since $\gamma_2'\hComp \gamma_1'$ is an isomorphism in $\bD_1$, $S(\gamma_2'\hComp \gamma_1')=S(\gamma_1')$ is an isomorphism in $\bD_0$; thus applying the \kl(mof){lifting property of isomorphisms for strong multi-opfibrations} repeatedly, we find that the squares $\gamma_1'$, $\gamma_2'$ and $\gamma_3'$ are all isomorphisms in $\bD_1$, which concludes the proof of the $\Rightarrow$ part of the theorem.
\end{itemize}

\begin{equation}\label{eq:assocThm:PartA:step4}
\ti{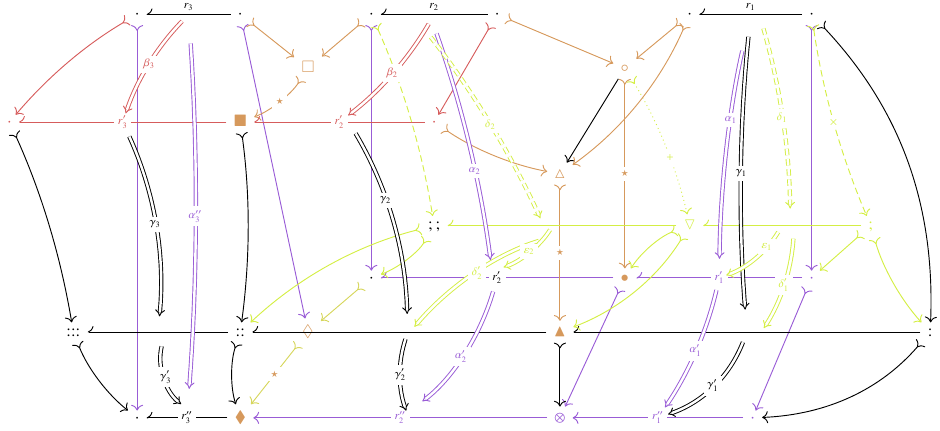}
\end{equation}

For the $\Leftarrow$ part of the claim, via the \kl(crDC){horizontal decomposition property} (here for the square $\alpha_{32}$), we obtain a globular isomorphism $\sigma_{32}$ and squares $\alpha_3'$ and $\alpha_2'$ such that $\sigma_{32}\vComp \alpha_{32}=\alpha_3'\hComp \alpha_2'$ (and with $r_{32}'\kl{\gcong}r_3'\hComp r_2'$):

\begin{equation}\label{eq:assocThm:proof:PartB:step1}
\ti{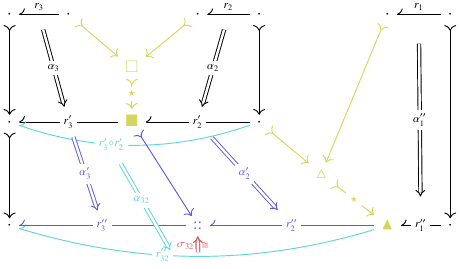}
\end{equation}

To proceed, we may now apply the \kl(cct){synthesis} part of the \kl{concurrency theorem} to the sequence formed by the composite square $\alpha_2'\vComp \alpha_2$ and the square $\alpha_1''$, again not explicitly carrying out the horizontal composition of squares in the last step of the construction. We thus arrive at a diagram as in~\eqref{eq:assocThm:proof:PartB:step2} below, with the squares $\beta_1$, $\beta_1'$, $\beta_2$ and $\beta_2'$ arising from the aforementioned construction (where the existence of the $\bD_0$-morphism $\circ\rightarrowtail\triangle$ follows from the \kl(ms){multi-sum extension Lemma}):

\begin{equation}\label{eq:assocThm:proof:PartB:step2}
\ti{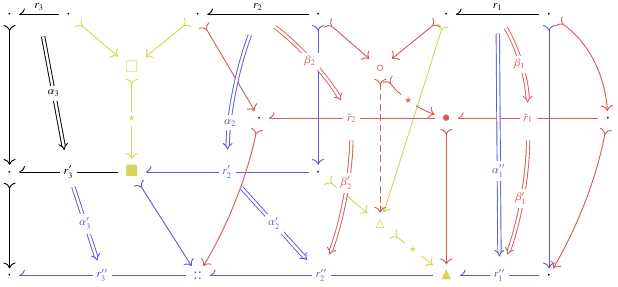}
\end{equation}

By applying the \kl(cct){synthesis} construction of the \kl{concurrency theorem} to the pair of squares $\alpha_3'\vComp \alpha_3$ and $\beta_2'\hComp\beta_1'$, we may obtain the diagram in~\eqref{eq:assocThm:PartB:step2B} below (where $\gamma_3'\vComp\gamma_3= \alpha_3'\vComp\alpha_3$, and $\gamma_{21}'\vComp\gamma_{21}=\beta_2'\hComp \beta_1'$):

\begin{equation}\label{eq:assocThm:PartB:step2B}
\ti{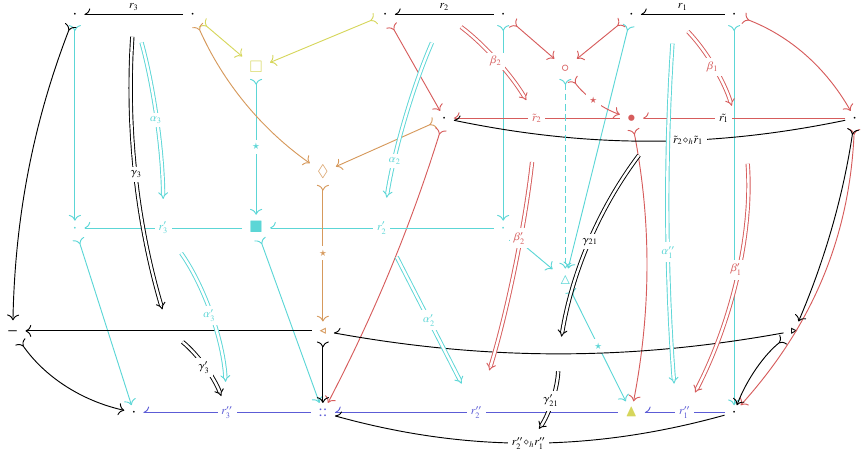}
\end{equation}

Next, we apply the following sequence of manipulations to obtain the diagram in~\eqref{eq:assocThm:proof:PartB:step3}:
\begin{itemize}
\item Via the \kl(ms){multi-sum extension Lemma}, there exists a $\bD_0$-morphism $\square  \rightarrowtail \lozenge$ between the multi-sum objects $\square$ and $\lozenge$.
\item  Via the \kl(crDC){complex decomposition property},  there exist squares $\gamma_2$, $\gamma_2'$, $\gamma_1'$ and $\gamma_1$ and a \kl(dc){globular isomorphism} $\tau_{21}$ such that $\gamma_2\hComp\gamma_1= \tau_{21}\vComp \gamma_{21}$ and $\gamma_2'\hComp \gamma_1'= \gamma_{21}'\vComp \tau_{21}^{-1}$.
\end{itemize}

\begin{equation}\label{eq:assocThm:proof:PartB:step3}
\ti{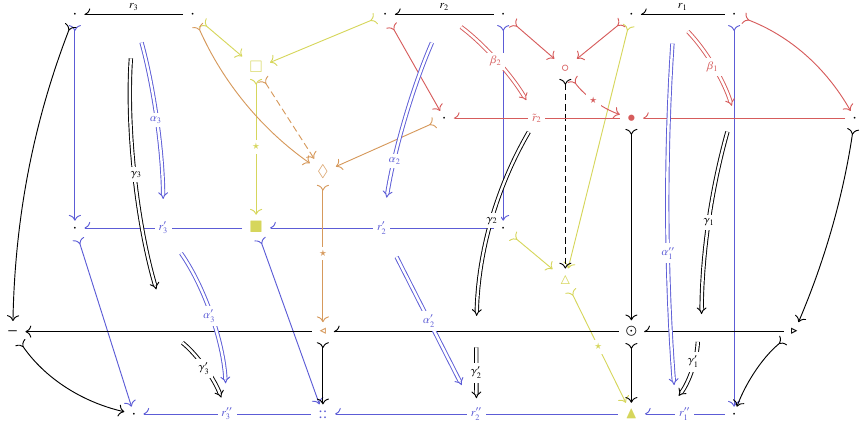}
\end{equation}

In order to complete the proof, we proceed as follows (yielding the diagram in~\eqref{eq:assocThm:proof:PartB:step4}):

\begin{itemize}
\item Take a pullback to obtain the span $\cdot \leftarrowtail \oplus \rightarrowtail \odot$. Since \kl(crDC){$\bD_0$ has pullbacks}, the span consists of two $\bD_0$-morphisms. By the \kl{universal property of pullbacks}, there exists a morphism $\cdot \rightarrowtail \oplus$ (marked as a dashed arrow in~\eqref{eq:assocThm:proof:PartB:step4}), which again since \kl(crDC){$\bD_0$ has pullbacks} is also a $\bD_0$-morphism.
\item Applying the \kl(mof){pullback lifting lemma for strong multi-opfibrations} to the pullback square over the object $\oplus$, we obtain squares $\delta_2$, $\delta_2'$ and $\varepsilon_2$ such that $\alpha_2 = \varepsilon_2\vComp\delta_2$ and $\gamma_2\vComp\beta_2 = \delta_2'\vComp\delta_2$. The lemma also entails that the square under the object marked $\otimes$ is a pullback. By the \kl{universal property of pullbacks}, this entails the existence of a morphism $\square \rightarrow \otimes$, which is a $\bD_0$-morphism since \kl(crDC){$\bD_0$ has pullbacks}.
\item By the \kl(rmofrup){universal property of residues}, since %
$\alpha_2=\varepsilon_2\vComp\delta_2$, and since %
the residue $\square\rightarrowtail \blacksquare$ marked $\star$ (which forms the second factor of of $T(\alpha_2)$) factors through $\otimes \rightarrowtail \blacksquare$ (i.e., through $T(\varepsilon_2)$), we find that the square $\varepsilon_2$ is an isomorphism in $\bD_1$. Thus in particular the morphism $\oplus \rightarrowtail \cdot$ marked $\dag$ in~\eqref{eq:assocThm:proof:PartB:step4}, i.e., $S(\varepsilon_2)$, is an isomorphism.
\item The latter fact entails that by applying the \kl(ms){universal property of multi-sums} to the cospan $dom(r_2') \rightarrowtail \oplus\rightarrowtail \odot \leftarrowtail codom(r_1)$, there exists a $\bD_0$-morphism $\triangle \rightarrowtail \odot$ (marked $\dag\dag$ in~\eqref{eq:assocThm:proof:PartB:step4}). 
\item By the \kl(rmofrup){universal property of residues}, since %
$\alpha_1''=\gamma_1'\vComp(\gamma_1\vComp\beta_1)$, and since %
the residue $\triangle\rightarrowtail \blacktriangle$ marked $\star$ (which forms the second factor of of $T(\alpha_1)$) factors through $\odot \rightarrowtail \blacktriangle$ (i.e., through $T(\gamma_1')$), we find that the square $\gamma_1'$ is an isomorphism in $\bD_1$, and thus in particular $T(\gamma_1')$ is an isomorphism in $\bD_0$; thus by repeated application of the \kl(mof){lifting property of isomorphisms for strong multi-opfibrations}, the squares $\gamma_2'$ and $\gamma_3'$ are found to be isomorphisms in $\bD_1$. This concludes the proof of the $\Leftarrow$ part of the theorem.
\end{itemize}

\begin{equation}\label{eq:assocThm:proof:PartB:step4}
\ti{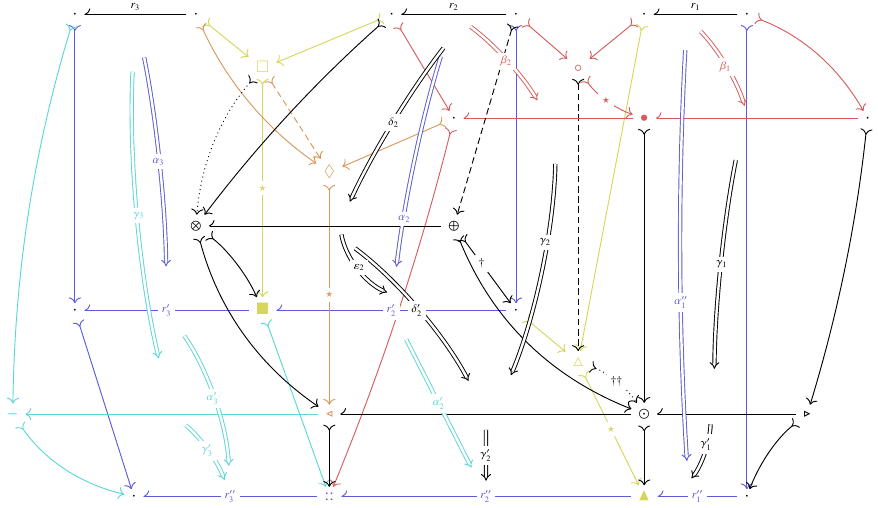}
\end{equation}

\qed \end{proof}

This proof of the \kl{associativity theorem} in a crDC provides a strong indication that modularizing the categorical structures in this form renders vastly complex mathematical developments feasible and, at the same time, provides some deep structural insights. Most importantly, our characterization of a given categorical rewriting semantics to qualify as being \emph{compositional} is based exclusively on verifying properties of just the \kl(dc){squares} of a \kl{double category} (which model \emph{direct derivations}) and on the existence of \kl{multi-sums}, i.e., only on the very definition of the rewriting semantics being formalized as a \kl{compositional rewriting double category}.

\section{Examples of fibrational structures relevant for rewriting theory}\label{sec:examplesFib}

In this section, we demonstrate that a number of constructions of commutative squares that form the building blocks of standard categorical rewriting semantics in fact carry fibrational structures; this will eventually allow us to instantiate our general compositional rewriting theory to these standard semantics. After a quick review of the notion of stable system of monics, we define various categories of pullback, pushout or final pullback complement squares where composition is defined by either horizontal or vertical pasting. In the remainder of the section, we then analyze the fibrational structures on the four natural boundary functors (domain, codomain, target and source) from these categories of squares.

\subsection{Categories of squares}\label{sec:cos}

A system of monics $\cM$ in a category $\bfC$ is a collection of monomorphisms that \AP\intro(ssm){includes all isomorphisms} and is \AP\intro(ssm){stable under composition}. Throughout the remainder of this paper, we will reserve the notation $\rightarrowtail$ for monos in $\cM$, and $\hookrightarrow$ for generic monomorphisms.
We say that $\bfC$ \AP\intro(ssm){has pullbacks along $\cM$-morphisms} if pullbacks of cospans of the form $A\rightarrow B\leftarrowtail B'$ always exist in $\bfC$.

For a category $\bfC$, $\cM$ is a \AP\emph{\intro{stable system of monics}} \cite{COCKETT2002223} if $\bfC$ \kl(ssm){has pullbacks along $\cM$-morphisms} and $\cM$ is \AP\intro(ssm){stable under pullback}: if $A \leftarrow m' - A' \rightarrow B'$ is a pullback of $A \rightarrow B \leftarrow m \ltail B'$ then $m'\in \cM$. The morphisms in $m \in \cM$ satisfy the following \AP\intro(ssm){decomposition property of $\cM$-morphisms} \cite{COCKETT2002223}: if $m = m' \circ f$ where $m'$ is a mono then $f \in \cM$.\\

For later convenience, we introduce the following auxiliary definitions, which permit us to succinctly express whether or not a given category admits pullbacks, pushouts or final pullback complements for generic input data, or only when the morphisms on the input are of a certain nature relative to a stable system of monics:

\begin{definition}
Let $\bfC$ be a category. 
\begin{enumerate}[label=(\roman*)]
\item $\bfC$ \AP\intro{has pullbacks} if $\bfC$ admits pullbacks of all cospans. 
\item $\bfC$ \AP\intro{has pushouts} if $\bfC$ admits pushouts of all spans.
\item $\bfC$ \AP\intro{has final pullback complements (FPCs)} if $\bfC$ admits \kl{FPCs} along arbitrary sequences of composable morphisms $A\rightarrow B\rightarrow B'$.
\end{enumerate}
If $\bfC$ has a \kl{stable system of monics} $\cM$, we define also the following variants and additional concepts:
\begin{enumerate}[label=(\roman*')]
\item $\bfC$ \AP\intro(ssm){has pushouts along $\cM$-morphisms} if pushouts of spans of the form $A\leftarrow B\rightarrowtail B'$ exist in $\bfC$.
\item $\bfC$ \AP\intro(ssm){has final pullback complements (FPCs) along $\cM$-morphisms} iff \kl{FPCs} of sequences of composable morphisms of the form $A\rightarrow B\rightarrowtail B'$ exist in $\bfC$.
\item \AP\intro(ssm){$\cM$-morphisms are stable under pushout} in $\bfC$ if whenever $A'\rightarrow B'\leftarrow \beta-B$ is a pushout of a span of the form $A'\leftarrow \alpha \ltail A\rightarrow B$, then $\beta\in \cM$.
\item \AP\intro(ssm){pushouts along $\cM$-morphisms are stable under $\cM$-pullbacks}\footnote{Throughout this paper, in order to avoid confusion, we follow the convention that ``stable under \emph{pullback}'' exclusively refers to the stability of morphisms when considering individual pullback squares (as in the definition of \kl(ssm){$\cM$-morphisms stable under pullback}), while ``stable under \emph{pullbacks}''  always refers to stability properties that involve commutative cubes with vertical squares being pullbacks (as in the definition of \kl(ssm){stability of $\cM $-pushouts under $\cM $-pullbacks}).} in $\bfC$ if for all diagrams of the form below,
\begin{equation}\label{eq:def:MstableMpushouts}
\ti{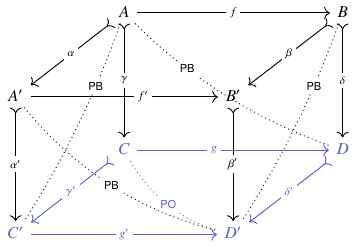}
\end{equation}
where the bottom square is a \kl(ssm){pushout along an $\cM$-morphism}, and the vertical squares are \kl(ssm){pullbacks along $\cM$-morphisms}, then the top square is a pushout.
\end{enumerate}
\end{definition}

\begin{definition}\label{def:cstCats}
Let $\bfC$ be a category with a \kl{stable system of monics} $\cM$%
. Let $\mathsf{T}$ be a type of commutative squares, for which we consider $\mathsf{PB}$ (\kl{pullbacks}), $\mathsf{PO}$ (\kl{pushouts}), or $\mathsf{FPC}$ (\kl{final pullback complements}). Then we define the following categories:
\begin{enumerate}[label=(\roman*)]
\item $\mathsf{T}_h(\bfC,\cM)$ has as objects the morphisms of $\cM$, and as morphisms commutative squares of type $\mathsf{T}$ along arbitrary morphisms of $\bfC$, and a morphism composition induced by \emph{horizontal pasting} of squares of type $\mathsf{T}$.
\item $\mathsf{T}_v(\bfC,\cM)$ has as objects the morphisms of $\bfC$, and as morphisms commutative squares of type $\mathsf{T}$ along $\cM$-morphisms, and a morphism composition induced by \emph{vertical pasting} of squares of type $\mathsf{T}$.
\end{enumerate}
In Figure~\ref{fig:boundaryFunctors}, we depict a square of type $\mathsf{T}$ (with $\cM$-morphism drawn vertically) and the action of four  ``boundary functors'' that naturally arise from the above definitions:
\begin{enumerate}[label=(\alph*)]
\item The \AP\intro(cmtSq){domain} functor $dom: \mathsf{T}_h(\bfC,\cM)\rightarrow \bfC$ and the \AP\intro(cmtSq){codomain} functor $codom: \mathsf{T}_h(\bfC,\cM)\rightarrow \bfC$.
\item The \AP\intro(cmtSq){source} functor $S: \mathsf{T}_v(\bfC,\cM)\rightarrow \bfC\vert_{\cM}$ and the \AP\intro(cmtSq){target} functor $codom: \mathsf{T}_v(\bfC,\cM)\rightarrow \bfC\vert_{\cM}$, where $\bfC\vert_{\cM}$ has the same objects as $\bfC$, and as morphisms those of $\cM$.
\end{enumerate}
\end{definition}

\begin{figure}
\centering
$\ti{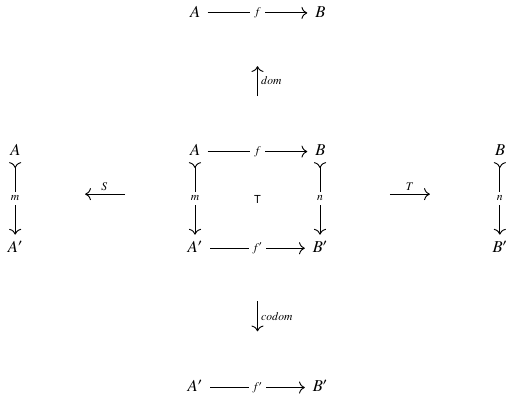}$
\caption{Boundary functors.\label{fig:boundaryFunctors}}
\end{figure}
The following result, whose proof can be found in \ref{app:ps4}, establishes that these categories are indeed well-defined.
\begin{lemma}\label{lem:PBPOFPCcats}
The categories $\mathsf{T}_h(\bfC,\cM)$ and $\mathsf{T}_v(\bfC,\cM)$ for $\mathsf{T}\in \{\mathsf{PB}, \mathsf{PO} , \mathsf{FPC}\}$ as introduced in Definition~\ref{def:cstCats} are well-defined, i.e., their composition operations are well-typed, associative and unital.
\end{lemma}
\subsection{Fibrational properties of the domain and codomain functors}\label{sec:domFunctors}

We now begin to investigate a number of interesting fibrational structures carried by the boundary functors of the various categories of squares, considering first the case of the domain functor.

\begin{theorem}\label{thm:domFunctorFP}
Let $\bfC$ be a category with a \kl{stable system of monics} $\cM$, and with the following additional properties:
\begin{enumerate}
\item $\bfC$ \kl{has pullbacks}.
\item $\bfC$ \kl(ssm){has pushouts and final pullback complements (FPCs) along $\cM$-morphisms}. 
\item Pushouts along $\cM$-morphisms are \AP\phantomintro(domCsqPBT){stable under pullbacks}\kl(PO){stable under pullbacks}.
\item \emph{\AP\intro(domCsqPBT){Pushouts along $\cM$-morphisms are pullbacks}}.
\end{enumerate}
Then the domain functor $dom:\textsf{PB}_h(\bfC,\cM)\rightarrow \bfC$ from the category of pullback squares along $\cM$-morphisms and under horizontal composition to the underlying category $\bfC$ satisfies the following properties:
\begin{enumerate}[label=(\roman*)]
\item \AP\phantomintro(cmtSq){$dom:\textsf{PB}_h(\bfC,\cM)\rightarrow \bfC$ is a Grothendieck fibration}  $dom:\textsf{PB}_h(\bfC,\cM)\rightarrow \bfC$ is a \kl{Grothendieck fibration}, with the \kl(gf){Cartesian liftings} given by FPCs.
\item \AP\phantomintro(cmtSq){$dom:\textsf{PB}_h(\bfC,\cM)\rightarrow \bfC$ is a Grothendieck opfibration}  $dom:\textsf{PB}_h(\bfC,\cM)\rightarrow \bfC$ is a \kl{Grothendieck opfibration}, with the \kl(gopf){op-Cartesian liftings} given by pushouts.
\item \AP\intro(cmtSq){$dom:\textsf{PB}_h(\bfC,\cM)\rightarrow \bfC$ satisfies a Beck-Chevalley condition (BCC)}: adopting the notation $m-(f,f')\rightarrow n$ for morphisms in $\mathsf{PB}_h(\bfC,\cM)$ (cf.\ Figure~\ref{fig:boundaryFunctors}), consider a commutative square in $\mathsf{PB}_h(\bfC,\cM)$ that is mapped by $dom$ into a pullback square in $\bfC$:
\begin{equation}\label{eq:BCCpremise}
\ti{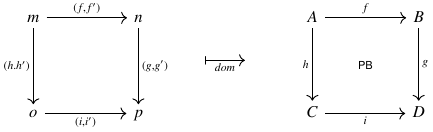}
\end{equation}
Then the following two equivalent conditions hold:
\begin{itemize}
\item \AP\intro(cmtSq){(BCC-1)}: $(f,f')$ is \kl(gopf){op-Cartesian} if $(i,i')$ is \kl(gopf){op-Cartesian} and $(g,g')$ and $(h,h')$ are \kl(gf){Cartesian}.
\item \AP\intro(cmtSq){(BCC-2)}: $(g,g')$ is \kl(gf){Cartesian} if $(h,h')$ is \kl(gf){Cartesian} and $(f,f')$ and $(i,i')$ are \kl(gopf){op-Cartesian}.
\end{itemize}
\end{enumerate}
\end{theorem}
\begin{proof} 
As the first two parts of the proof will demonstrate, the $dom$ functor is in a certain sense a prototypical example of a Grothendieck bifibration, in that the fibration and opfibration structures arise directly from universal properties of FPCs and pushouts, respectively.

\emph{Ad (i) --- $dom$ is a \kl{Grothendieck fibration}:} this statement follows by specializing the defining equation~\eqref{eq:def:grothendieckFibration} to the case of the $dom$ functor. The existence of \kl(gf){Cartesian liftings} is guaranteed since the category $\bfC$ by assumption \kl(ssm){has FPCs along $\cM$-morphisms}, while the requisite universal property that qualifies the liftings as being Cartesian (i.e., the second line of~\eqref{eq:def:grothendieckFibration}) is satisfied via the \kl{universal property of FPCs}.

\emph{Ad (ii) --- $dom$ is a \kl{Grothendieck opfibration}:}  specializing the defining equation~\eqref{eq:def:grothendieckOpfibration} to the case of the $dom$ functor, we find that the \kl(gopf){op-Cartesian liftings} exist in the form of pushouts (which are guaranteed to exist since $\bfC$ by assumption \kl(ssm){has pushouts along $\cM$-morphisms}), while the universal property which qualifies these liftings as op-Cartesian (i.e., the second line of~\eqref{eq:def:grothendieckOpfibration}) is satisfied via \kl{pullback-pushout decomposition}.

\emph{Ad (iii) --- Beck-Chevalley condition (BCC):} The proof can be found in \ref{app:ps4}.
\qed \end{proof}

As the above results indicate, the domain functor $dom:\mathsf{PB}_h(\bfC,\cM)\rightarrow \bfC$ is (for suitable categories $\bfC$) a Grothendieck bifibration, i.e., simultaneously a Grothendieck fibration and opfibration. An interesting variant of this type of result---which moreover has important computational meaning in its own right---arises when considering the domain functors from the categories $\mathsf{PO}_h$ and $\mathsf{FPC}_h$ instead, which permits to state fibrational properties under considerably weaker assumptions on the underlying categories $\bfC$:

\begin{theorem}\label{cor:domPOFPC}
Let $\bfC$ be a category with a \kl{stable system of monics} $\cM$.
\begin{enumerate}[label=(\roman*)]
\item If $\bfC$ \kl(ssm){has pushouts along $\cM$-morphisms}, the functor $dom: \mathsf{PO}_h(\bfC,\cM)\rightarrow \bfC$ is a \kl{Grothendieck opfibration}.
\item If $\bfC$ \kl(ssm){has FPCs along $\cM$-morphisms}, the functor $dom: \mathsf{FPC}_h(\bfC,\cM)\rightarrow \bfC$ is a \kl{Grothendieck fibration}.
\end{enumerate}
\end{theorem}
\begin{proof}
It is straightforward to demonstrate that, for case $(i)$, \kl(ssm){pushouts along $\cM$-morphisms} provide the \kl(gopf){op-Cartesian liftings} (as was also the case for $dom:\mathsf{PB}_h(\bfC,\cM)\rightarrow \bfC$), while the \kl(gopf){op-Cartesianity} properties of the liftings are realized in the form of \kl{pushout-pushout decomposition}. For case $(ii)$, \kl(ssm){FPCs along $\cM$-morphisms} provide the \kl(gf){Cartesian liftings}, while the \kl(gf){Cartesianity} properties of liftings are realized in the form of \kl{horizontal FPC decomposition}.  
\qed \end{proof}

Let us briefly compare the results of Theorem~\ref{thm:domFunctorFP} and Theorem~\ref{cor:domPOFPC}. The \kl(gopf){op-Cartesianity} of \kl(gopf){op-Cartesian liftings} for the functor $dom:\mathsf{PB}_h(\bfC,\cM)\rightarrow \bfC$ relies on \kl{pullback-pushout decomposition} while, for $dom:\mathsf{PO}_h(\bfC,\cM)\rightarrow \bfC$, \kl{pushout-pushout decomposition}, valid in any category, is all that is required. On the other hand, the \kl(gf){Cartesianity} of \kl(gf){Cartesian liftings} for the functor $dom:\mathsf{PB}_h(\bfC,\cM)\rightarrow \bfC$ relies on the \kl{universal property of FPCs} while it is a consequence of \kl{horizontal FPC decomposition} for $dom:\mathsf{FPC}_h(\bfC,\cM)\rightarrow \bfC$. Since the requisite properties of FPCs hold in any category that admits FPCs, it appears interesting to note that the strong requirements necessary for $dom:\mathsf{PB}_h(\bfC,\cM)\rightarrow \bfC$ to carry bifibrational structures appear to be caused mainly by the \kl{Grothendieck opfibration} part of the structure. 

\begin{remark}
In contrast to the domain functors discussed in the previous section, only the codomain functor $codom:\mathsf{PB}_h(\bfC,\cM)\rightarrow \bfC$ appears to admit some fibrational structure (see below), while $codom:\mathsf{PO}_h(\bfC,\cM)\rightarrow \bfC$ and $codom:\mathsf{FPC}_h(\bfC,\cM)\rightarrow \bfC$ fail to do so. Since none of these three codomain functors play a role in our constructions, this causes no technical problems, but we found it interesting to mention the following result here for symmetry nonetheless.

\begin{theorem}
Let $\bfC$ be a category with a \kl{stable system of monics}. Then $codom:\mathsf{PB}_h(\bfC,\cM)\rightarrow \bfC$ is a \kl{Grothendieck fibration}.
\end{theorem} 
\begin{proof}
\kl(gf){Cartesian liftings} are provided by taking pullbacks, while the \kl(gf){Cartesianity} of the liftings amounts to \kl{pullback-pullback decomposition}.
\qed \end{proof}
\end{remark}

\subsection{Fibrational properties of the target functors}\label{sec:targetFunctors}

As we show in this section, the target functors will have rather different fibrational structures: 
\begin{itemize}
\item $T:\mathsf{PB}_v(\bfC,\cM)\rightarrow \bfC\vert_{\cM}$ carries no fibrational structures.
\item $T:\mathsf{FPC}_v(\bfC,\cM)\rightarrow \bfC\vert_{\cM}$ carries a \kl{Grothendieck opfibration} structure.
\item $T:\mathsf{PO}_v(\bfC,\cM)\rightarrow \bfC\vert_{\cM}$ carries a \kl(mof){strong} \kl{multi-opfibration} structure.
\end{itemize}

We begin with the following theorem, that deals with the case of $T:\mathsf{FPC}_v(\bfC,\cM)\rightarrow \bfC\vert_{\cM}$, whose full proof can be found in~\ref{app:ps4}.

\begin{theorem}\label{thm:trgtFPC-FP}
Let $\bfC$ be a category with a \kl{stable system of monics}  $\cM$ and that \kl(ssm){has FPCs along $\cM$-morphisms}. \AP\phantomintro{target functor $T: \mathsf{FPC}_v(\bfC,\cM)\rightarrow \bfC\vert_{\cM}$ is a Grothendieck opfibration}Then the target functor $T: \mathsf{FPC}_v(\bfC,\cM)\rightarrow \bfC\vert_{\cM}$ is a \kl{Grothendieck opfibration}.
\end{theorem}
It is interesting to note that the proof strategy for \kl(gopf){op-Cartesianity} (cf.\ \ref{app:ps4}) would fail if we were to work in the category $\mathsf{PB}_v(\bfC,\cM)$ rather than in $\mathsf{FPC}_v(\bfC,\cM)$, since the existence of the isomorphism $A'-\eta\rightarrow P$ and the uniqueness of $A'-\alpha'\rightarrow A''$ relied upon the \kl{universal property of FPCs} (i.e., both of the FPC in the front and in the back of the diagram). Indeed, if we were to consider the analogue of the diagrams in~\eqref{eq:proof:trgtFPCgopfClaimExplicit} in $\mathsf{PB}_v(\bfC,\cM)$, i.e., where the front vertical square would be merely a pullback, taking a pullback as indicated would only yield that the squares under and over $P-p'\rightarrow B'$ are pullbacks. By the \kl{universal property of FPCs} (of the back vertical FPC square, i.e., the one of the lifting), we could only conclude that there exists a unique mediating arrow $P-\eta'\rightarrow A'$, but this arrow will in general not be an isomorphism, hence we can indeed not prove \kl(gopf){op-Cartesianity} of the liftings in $\mathsf{PB}_v(\bfC,\cM)$.

Let us now turn our attention to the remaining variant of the target functor, i.e., $T:\mathsf{PO}_v(\bfC,\cM)\rightarrow \bfC\vert_{\cM}$. This yields a first example of a \kl{multi-opfibration}. In order to formulate this result, we require the following multi-universal notion:

\begin{definition}\label{def:mipc}
Let $\bfC$ be a category with a \kl{stable system of monics} $\cM$. %
For all composable sequences of morphisms of the form $A-f\rightarrow B\rtail\beta\rightarrow B'$ (i.e., with $\beta\in \cM$), we define the following class:
\begin{equation}\label{eq:def:mIPC}
\mIPC{f}{\beta} := \{(A\rtail\alpha\rightarrow A', A'-f'\rightarrow B')\in \mor{\bfC} \times \mor{\bfC} \mid \alpha\in \cM \land (f',\beta) =\pO{\alpha,f}\}\,.
\end{equation}
More explicitly, $\mIPC{f}{\beta}$ consists of all composable sequences of morphisms $A\rtail\alpha\rightarrow A'-f'\rightarrow B'$ such that there exists a pushout square in $\bfC$ whose boundary is given by $(\alpha,f')$ and $(f,\beta)$. We refer to $\mIPC{f}{\beta}$ as the ($\cM$-) \AP\intro{multi-initial pushout complement (mIPC)} of $(f,\beta)$ if the class satisfies the following \AP\intro(mIPC){universal property}:
\begin{equation}\label{def:mIPC:UP}
\begin{aligned}\forall \ti{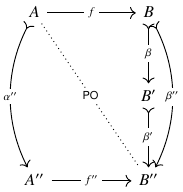}: &\exists (\alpha,f')\in \mIPC{f}{\beta}:  \exists!\, \alpha'\in \cM\;:
\ti{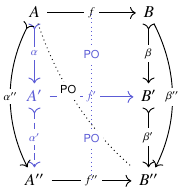}:\\
&\quad \forall \ti{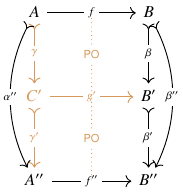}: \exists!\, A'-\varphi\rightarrow C'\in \iso{\bfC}: \gamma=\varphi\circ \alpha \land \alpha'= \gamma'\circ \varphi
\end{aligned}
\end{equation}
We say that $\bfC$ has \AP\intro{has multi-initial pushout complements (mIPCs) along $\cM$-morphisms} if $\bfC$ has an mIPC for every composable sequence of morphisms of the form $A-f\rightarrow B\rtail\beta\rightarrow B'$.
\end{definition}

\begin{remark} It is worthwhile pointing out that just as for ordinary (``non-multi-'') pushout complements, a \kl{multi-IPC} for a given composable sequence of morphisms may be an \emph{empty} set. For example, in $\mathbf{Graph}$, the category of \kl{directed multigraphs} , the \kl{multi-IPC} of the sequence $\varnothing \rightarrowtail \bullet \rightarrow \bullet\!\!\!\! \to\!\!\!\!-\!\circ$ is empty, a well-known phenomenon in the graph rewriting literature, interpreted as the impossibility to apply a vertex-deletion operation in DPO-semantics to a vertex with incident edges (i.e., since deletion of a vertex with incident edges would leave \emph{``dangling'' edges}).
\end{remark}

The following lemma establishes sufficient conditions to guarantee that $\bfC$ has mIPCs along $\cM$-morphisms; the proof is given in~\ref{app:ps4}.

\begin{lemma}\label{prop:mIPCexistence}
Let $\bfC$ be a category with a \kl{stable system of monics} $\cM$. If \kl(ssm){pushouts along $\cM$-morphisms are stable under $\cM$-pullbacks}, %
and if \AP\intro(lemMIPCexistence){pushouts along $\cM$-morphisms are pullbacks}, %
then $\bfC$ \kl{has multi-initial pushout complements (mIPCs) along $\cM$-morphisms}.
\end{lemma}

After this somewhat lengthy excursion, a direct comparison of the notion of \kl(mof){strong} \kl{multi-opfibration} (Definition~\ref{def:multi-opfibration}) and of \kl{multi-initial pushout complement} yields the following important result:

\begin{theorem}\label{thm:TPOvMOF}
Let $\bfC$ be a category with a \kl{stable system of monics} $\cM$. If \kl(ssm){pushouts along $\cM$-morphisms are stable under $\cM$-pullbacks} in $\bfC$, %
and if \AP\intro(thmTPOvMOF){pushouts along $\cM$-morphisms are pullbacks}, %
\AP\phantomintro{target functor $T:\mathsf{PO}_v(\bfC,\cM)\rightarrow \bfC\vert_{\cM}$ is a multi-opfibration}then the target functor $T:\mathsf{PO}_v(\bfC,\cM)\rightarrow \bfC\vert_{\cM}$ is a \kl(mof){strong} \kl{multi-opfibration}.
\end{theorem}
\begin{proof}
The \kl(mof){multi-op-Cartesian liftings} are provided by \kl{multi-initial pushout complements (mIPCs)}, whose existence and uniqueness up to isomorphism are guaranteed under the stated assumptions according to Lemma~\ref{prop:mIPCexistence}.
\qed \end{proof}

\subsection{Fibrational properties of the source functors}\label{sec:sourceFunctors}

Finally, let us investigate the fibrational structures of the source functors. This gives rise to the following results:
\begin{itemize}
\item $S:\mathsf{PB}_v(\bfC,\cM)\rightarrow \bfC\vert_{\cM}$ carries no fibrational structures.
\item $S:\mathsf{PO}_v(\bfC,\cM)\rightarrow \bfC\vert_{\cM}$ carries a \kl{Grothendieck opfibration} structure.
\item $S:\mathsf{FPC}_v(\bfC,\cM)\rightarrow \bfC\vert_{\cM}$ carries a \kl{residual multi-opfibration} structure.
\end{itemize}

\begin{theorem}\label{thm:sourcePOgopf}
Let $\bfC$ be a category with a \kl{stable system of monics} $\cM$, that \kl(ssm){has pushouts along $\cM$-morphisms}, and such that \kl(ssm){$\cM$-morphisms are stable under pushout}. \AP\phantomintro{source functor $S:\mathsf{PO}_v(\bfC,\cM)\rightarrow \bfC\vert_{\cM}$ is a Grothendieck opfibration}Then the source functor $S:\mathsf{PO}_v(\bfC,\cM)\rightarrow \bfC\vert_{\cM}$ is a \kl{Grothendieck opfibration}, with the \kl(gopf){op-Cartesian liftings} provided by pushouts.
\end{theorem}
\begin{proof}
It suffices to instantiate the definition of \kl{Grothendieck opfibration} to the case at hand:
\begin{equation}\label{eq:proofSpoGopfA}
\begin{aligned}
&\forall \ti{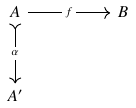}\;:\;
\exists \ti{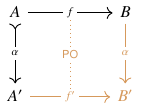}\;: \forall \ti{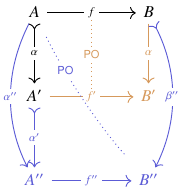}\;:\; \ti{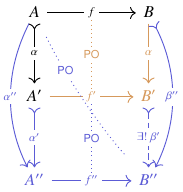}
\end{aligned}
\end{equation}
Here, the existence of \kl(gopf){op-Cartesian liftings} is provided by the assumption that $\bfC$ \kl(ssm){has pushouts along $\cM$-morphisms}, while the \kl(gopf){op-Cartesianity} of the liftings follows from the \kl{universal property of pushouts} (yielding the existence of a unique morphism $B'-\beta'\rightarrow B''$), \kl{pushout-pushout decomposition} (which ensures the bottom square in the rightmost diagram in~\eqref{eq:proofSpoGopfA} is a pushout), and finally from the assumption that \kl(ssm){$\cM$-morphisms are stable under pushout} (ensuring that $\beta'\in \cM$, so that the pushout square over it indeed qualifies as a morphism in $\mathsf{PO}_v(\bfC,\cM)$).
\qed \end{proof}

It is worthwhile considering whether the above proof strategy for the \kl{Grothendieck opfibration} structure of $S:\mathsf{PO}_v(\bfC,\cM)\rightarrow \bfC\vert_{\cM}$ could be adapted to the case of the source functor $S:\mathsf{PB}_v(\bfC,\cM)\rightarrow \bfC\vert_{\cM}$. However, even under the additional assumption that pushouts along $\cM$-morphisms are pullbacks, we could \emph{not} prove that $\beta'\in \cM$ for the analogue of the last diagram in~\eqref{eq:proofSpoGopfA} where the outer square is merely a pullback (this was true for~\eqref{eq:proofSpoGopfA}, because here we could rely upon the assumed \kl(ssm){stability of $\cM$-morphisms under pushout}). Nevertheless, it is interesting to observe that $S:\mathsf{PB}_v(\bfC\vert_{\cM},\cM)\rightarrow \bfC\vert_{\cM}$ (i.e., restricting to pullback squares where all morphisms are in $\cM$) \emph{does} have the structure of a \kl{Grothendieck opfibration}, with the \kl(gopf){op-Cartesian liftings} given by pushouts, and \kl(gopf){op-Cartesianity} ensured if the \kl{pullback-pushout decomposition} lemma holds (which requires certain additional assumptions on $\bfC$).

Finally, let us consider the case of the source functor $S:\mathsf{FPC}_v(\bfC,\cM)\rightarrow \bfC\vert_{\cM}$ from the category of FPCs along $\cM$-morphisms with vertical pasting. This case requires the introduction of a novel universal construction, that of $\cM$-final pullback complement pushout augmentation. Before giving the definition, we first quote some prerequisite standard concepts from category theory pertaining to factorization structures on morphisms, which will be used in this paper for instance in the form of \kl(EMS){epi-$\cM$-factorizations}, but also to demonstrate a certain factorization structure on FPC squares (seen as morphisms in $\mathsf{FPC}_v(\bfC,\cM)$).

\begin{definition}[\cite{adamek2006}, Def. 14.1]\label{def:EMstruct}
For a category $\bfC$, let $E$ and $M$ be classes of morphisms. By convention, in commutative diagrams, let morphisms in $E$ be depicted as $\twoheadrightarrow$, and morphisms in $M$ by $\rightarrowtail$. Then $(E,M)$ is called a \AP\intro{factorization structure for morphisms} in $\bfC$, and $\bfC$ is called \AP\intro{$(E,M)$-structured} iff
\begin{enumerate}[label=(\roman*)]
\item both $E$ and $M$ are \AP\intro(EMS){closed under composition with isomorphisms},
\item $\bfC$ \AP\intro(EMS){has $(E,M)$-factorizations of morphisms} (i.e., for every morphism $f$ in $\bfC$, there exist $m\in M$ and $e\in E$ such that $f=m\circ e$),
\item $\bfC$ \AP\intro(EMS){has the unique $(E,M)$-diagonalization property}:
\begin{equation}\label{eq:def:EMdiagProp}
\forall\ti{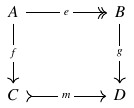}\;:\;\ti{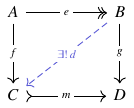}
\end{equation}
In words: for all commutative squares as in~\eqref{eq:def:EMdiagProp} above, where $e\in E$ and $m\in M$, there exists a unique morphism $d$ (referred to as the \AP\intro(EMS){diagonal}) such that $f=d\circ e$ and $g=m\circ d$.
\end{enumerate}
\end{definition}

\begin{proposition}[\cite{adamek2006}, Prop. 14.4 and 14.6]
Let $\bfC$ be a category that is \kl{$(E,M)$-structured}. Then the following properties hold:
\begin{enumerate}[label=(\roman*)]
\item \AP\phantomintro(EMS){morphisms that are both in $E$ and in $M$ are isomorphisms}$E\cap M =\iso{\bfC}$.
\item the classes $E$ and $M$ are both \AP\intro(EMS){closed under composition}.
\item \AP\intro(EMS){$(E,M)$-factorizations are essentially unique}:
\begin{enumerate}[label=(\alph*)]
\item If $e_1\circ m_1 =e_2\circ m_2$ (for $m_1,m_2\in M$ and $e_1,e_2\in E$), then there exists an isomorphism $h$ such that $e_2=h\circ e_1$ and $m_1=m_2\circ h$.
\item If $f=e\circ m$ (for $e\in E$ and $m\in M$), and if  $e'=h\circ e$ and $m'=m\circ h$ for an isomorphism $h$, then $(e',m')$ is also an $(E,M)$-factorization of $f$.
\end{enumerate}
\end{enumerate}
\end{proposition}

\begin{definition}\label{def:FPA}
Let $\bfC$ be a category with a \kl{stable system of monics} $\cM$, and such that\footnote{Here and in the following, it is important to emphasize that we do not require any particular properties of the class of morphisms $\cE$ other than that it is a class such that $\bfC$ is \kl{($\cE$, $\cM$)-structured}. While in many applications of interest $\cE$ will coincide with the class of epimorphisms of $\bfC$ or a subclass thereof, in some cases $\cE$ will not even be a class of epimorphisms~\cite{GABRIEL_2014}.} $\bfC$ is \kl{($\cE$,$\cM$)-structured}. Given a pushout square along an $\cM$-morphism as in back of the diagram below,
\begin{equation}\label{eq:defFPA}
\ti{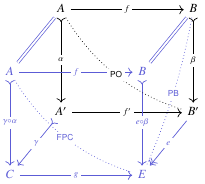}
\end{equation}
we define an \AP\intro{$\cM$-final pullback complement pushout augmentation (FPA)} as a triple of morphisms $(\gamma, g, e)$ such that
\begin{enumerate}[label=(\roman*)]
\item $\gamma\in \cM$, $e\in \cE$, $e\circ \beta \in \cM$, and $g\circ \gamma = e\circ f'$,
\item $(\gamma\circ \alpha, g)$ is an FPC of $(f, e\circ \beta)$, and
\item $(id_B, \beta)$ is a pullback of $(e\circ \beta, e)$.
\end{enumerate}
We denote the class of all FPAs of a given pushout square $(f,\beta,\alpha,f')$ by $\FPA{f,\beta,\alpha,f'}$.
\end{definition}

\begin{remark}
It appears worthwhile to note that for a diagram as in~\eqref{eq:defFPA}, since the left vertical square is a pullback (given that $\gamma$ is in $\cM$ and thus a monomorphism), and since the front vertical square is an FPC and thus a pullback, by \kl{pullback-pullback composition} the composite of the right and back vertical squares is a pullback; therefore, we find that the morphism $\gamma$ must coincide with the unique morphism from $A'$ to $C$ that exists by the \kl{universal property of FPCs}. The subtlety in the definition of \kl{FPAs} then lies in the nature of $\gamma$ as being in $\cM$: supposing for a moment that $\gamma$ is a generic morphism, if the pushout square in the back is known to be a pullback (which will be the case in all categories of interest), then the composite of the right and back vertical is again a pullback, thus $\gamma$ again has to coincide with the unique morphism that exists by the universal property of FPCs; however, we would not be able to conclude from this set of assumptions that $\gamma$ must be in $\cM$, hence this is indeed found to be a non-trivial part of our set of assumptions.
\end{remark}

Throughout this paper, we will exclusively be interested in situations where the definition above may be slightly simplified: 
\begin{lemma}\label{lem:FPAaux}
Let $\bfC$ be a category with a \kl{stable system of monics} $\cM$, and such that $\bfC$ is \kl{($\cE$,$\cM$)-structured}. If \kl(ssm){pushouts along $\cM$-morphisms are stable under $\cM$-pullbacks} in $\bfC$, condition (iii) in Definition~\ref{def:FPA} is automatically satisfied, i.e., $(id_B, \beta)$ is a pullback of $(e\circ \beta, e)$ in the notations of~\eqref{eq:defFPA}.  
\end{lemma}
\begin{proof}
The proof follows by invoking \kl{pullback-pushout decomposition} (which holds due to the assumption that \kl(ssm){pushouts along $\cM$-morphisms are stable under $\cM$-pullbacks} in $\bfC$) to the commutative diagram in~\eqref{eq:defFPA}: since the front and left vertical squares compose into a pullback, the vertical morphisms are in $\cM$, and since the back vertical square is a \kl(ssm){pushout along an $\cM$-morphism}, the right vertical square is a pullback.
\qed \end{proof}

The concept of \kl{FPC-pushout-augmentations (FPAs)} introduced above gives rise to an interesting factorization system on FPCs as the following theorem explains; its proof is in~\ref{app:ps4}.

\begin{theorem}\label{thm:FPCfact}
Let $\bfC$ be a category with a \kl{stable system of monics} $\cM$, that is \kl{($\cE$,$\cM$)-structured}, that \kl(ssm){has pushouts and FPCs along $\cM$-morphisms}, such that \kl(ssm){$\cM$-morphisms are stable under pushout}, and such that \kl(ssm){pushouts along $\cM$-morphisms are stable under $\cM$-pullbacks}.  
\AP\phantomintro{$\mathsf{FPC}_v(\bfC,\cM)$ is (auto-augmented,inert)-structured}Then the category $\mathsf{FPC}_v(\bfC, \cM)$ is \kl{(auto-augmented, inert)-structured}. Here, the class of \AP\intro(FPCfact){auto-augmented FPCs}  is defined as
\begin{equation}\label{eq:defAAfpc}
\ti{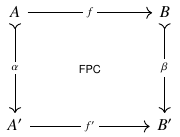} \in \mor{\mathsf{FPC}_v(\bfC,\cM)}\vert_{auto-augmented}
\; :\Leftrightarrow\; \exists\;
\ti{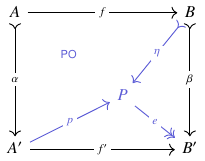}
\end{equation}
In words: an FPC square along an $\cM$-morphism (seen as a morphism in $\mathsf{FPC}_v(\bfC,\cM)$ is auto-augmented iff when taking a pushout of the span within the FPC, the mediating morphism into the cospan object of the FPC is a morphism in $\cE$.\footnote{Note that since we admit arbitrary morphisms of $\bfC$ for the horizontal morphisms, the mediating morphism would in general be a morphism with a non-trivial \kl(EMS){$\cE$-$\cM$-factorization}, hence for this morphism to be an $\ cE$-morphism is indeed a non-trivial requirement.}
Moreover, the class of \AP\intro(FPCfact){inert FPCs} is defined as 
\begin{equation}\label{eq:defIfpc}
\mor{\mathsf{FPC}_v(\bfC,\cM)}\vert_{inert} :=
\left.\left\lbrace\ti{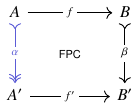}\;\right\vert \alpha \in \cE\cap \cM = \iso{\bfC}\right\rbrace
\end{equation}

\end{theorem}

We refer the interested readers to Example~\ref{ex:FPA} of Section~\ref{sec:exConstr:FPAs} for an illustration of the practical meaning of the \kl(EMS){(auto-augmented, inert) factorization} of FPCs, for the case of directed simple graphs, where it will be demonstrated that, in a certain sense, the factorization provides a static analysis of the classes of cloning with implicit deletion that can be modeled by FPCs.

Besides its quintessential role in factorizations of FPCs, the concept of \kl{FPC-pushout-augmentations (FPAs)} is also crucial for the fibrational structure of the source functor $S:\mathsf{FPC}_v(\bfC,\cM)\rightarrow \bfC\vert_{\cM}$ as the following theorem shows. Its proof can be found in~\ref{app:ps4}.

\begin{theorem}\label{thm:SourceFPCvRMOF}
Let $\bfC$ be a category with a \kl{stable system of monics} $\cM$, that is \kl{($\cE$, $\cM$)-structured}, that \kl(ssm){has pullbacks, pushouts and FPCs along $\cM$-morphisms}, such that \kl(ssm){$\cM$-morphisms are stable under pushout}, and such that \kl(ssm){pushouts along $\cM$-morphisms are stable under $\cM$-pullbacks}. %
\AP\phantomintro{source functor $S:\mathsf{FPC}_v(\bfC,\cM)\rightarrow \bfC\vert_{\cM}$ is a residual multi-opfibration}Then $S:\mathsf{FPC}_v(\bfC,\cM)\rightarrow \bfC\vert_{\cM}$ is a \kl{residual multi-opfibration}.
\end{theorem}

Let us finally note here that, in all of our applications, we will only consider base categories which are \kl{finitary}, so that one may indeed provide algorithms for the various universal constructions that yield finite sets of solutions up to isomorphisms.

\section{Examples of categories suitable for defining categorical constructions with fibrational properties}\label{sec:ccrts}

This section is structured into two parts: in Sections~\ref{sec:caps} and~\ref{sec:qt}, we present two classes of categories that may ultimately serve as a basis for defining compositional rewriting theories (cf.\ Section~\ref{sec:rs}), i.e., categories with \kl{adhesivity properties} and \kl{quasi-topoi}, respectively; and in Section~\ref{sec:exConstr}, we demonstrate that these classes of categories admit certain key categorical constructions of with fibrational properties.

\subsection{Categories with adhesivity properties}\label{sec:caps}

Starting in the early 2000s, the seminal work of Lack and Sobocinski~\cite{ls2004adhesive,lack2005adhesive,quasi-topos-2007} introducing \kl{adhesive} and \kl{quasi-adhesive} categories, which was later generalized by Ehrig et al.\cite{ehrig:2006fund,ehrig2004adhesive,ehrig2010categorical} to \kl{adhesive HLR} and \kl{weak adhesive HLR} categories and their variants, constituted a significant breakthrough in formalizing and standardizing the theory of Double-Pushout (DPO) rewriting. In this section, we will quote the salient definitions as well as key results from this research, with the purpose of providing a curated list of categories of practical interest that carry one of the variants of adhesivity properties mentioned above. We refer the interested readers to~\cite{ehrig:2006fund,ehrig2010categorical} (cf.\ also~\cite{bp2019-ext}) for further background materials.

In order to formulate the various notions of adhesivity, we require the following definitions:
\begin{definition}[Notions of van Kampen (VK) squares]\label{def:VKS}
Let $\bfC$ be a category. Then a pushout square is a \AP\intro{van Kampen (VK) square} iff for any commutative diagram as in~\eqref{eq:def:VKS} below, where the bottom square highlighted in blue is the aforementioned pushout square, and where the back squares are pullbacks, the following conditions hold:
\begin{itemize}
\item \AP\intro{(VK-a)} If the front and the right squares are pullbacks, then the top square is a pushout.
\item \AP\intro{(VK-b)} If the top square is a pushout, then the front and the right squares are pullbacks.
\end{itemize}
\begin{equation}\label{eq:def:VKS}
\ti{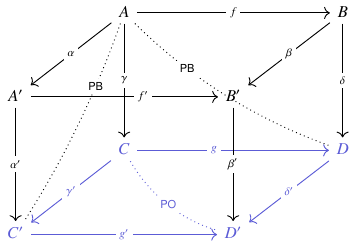}
\end{equation}
If $\bfC$ has a \kl{stable system of monics} $\cM$, we define the following weakenings of the notion of VK squares:
\begin{itemize}
\item \AP\intro{horizontal weak VK squares} are defined as pushouts whose morphism are all in $\cM$, and that are \kl{VK squares}.
\item \AP\intro{vertical weak VK squares} are defined as pushouts which satisfy the defining properties of \kl{VK squares} modulo the constraint that all vertical morphisms in~\eqref{eq:def:VKS} are in $\cM$ (i.e., when $\alpha',\beta',\gamma,\delta\in \cM$).
\end{itemize}
\end{definition}

\begin{definition}\label{def:adhesivityProperties}
We define the following variants of \AP\intro{adhesivity properties} for categories:
\begin{enumerate}
\item Let $\bfA$ be a category. Then $\bfA$ is said to be an \AP\intro{adhesive category}~\cite{ls2004adhesive} iff it has the following properties:
\begin{enumerate}[label=(A-\roman*)]
\item $\bfA$ \kl{has pullbacks}.
\item $\bfA$ \kl(ssm){has pushouts along monomorphisms}.
\item \AP\phantomintro(adhVK){(A-iii)}\kl(ssm){Pushouts along monomorphisms} in $\bfA$ are \kl{van Kampen squares}.
\end{enumerate}
\item Let $\bfQ$ be a category, and let $\regmono{\bfQ}$ denote the class of regular monomorphisms of $\bfQ$. Then $\bfQ$ is said to be \AP\intro{quasi-adhesive}~\cite{lack2005adhesive} (sometimes also referred to as \kl{rm-adhesive}~\cite{garner2012axioms}) iff $\bfQ$ satisfies the following properties:
\begin{enumerate}[label=(Q-\roman*)]
\item $\bfA$ \kl{has pullbacks}.
\item $\bfA$ \kl(ssm){has pushouts along regular monomorphisms}.
\item \AP\phantomintro(adhVK){(Q-iii)}\kl(ssm){Pushouts along regular monomorphisms} in $\bfA$ are \kl{van Kampen squares}.
\end{enumerate}
\item Let $\bfL$ be a category that admits a \kl{stable system of monics} $\cM$. Then $\bfL$ is said to be an \AP\intro{adhesive high-level replacement (HLR) category}~\cite{ehrig:2006fund} iff $\bfL$ satisfies the following properties:
\begin{enumerate}[label=(L-\roman*)]
\item $\bfL$ \kl(ssm){has pullbacks along $\cM$-morphisms}.
\item $\bfL$ \kl(ssm){has pushouts along $\cM$-morphisms}, and \kl(ssm){$\cM$-morphisms are stable under pushout}.
\item \AP\phantomintro(adhVK){(L-iii)}\kl(ssm){Pushouts along $\cM$-morphisms} in $\bfL$ are \kl{van Kampen squares}.
\end{enumerate}
\item Let $\bfH$ be a category that admits a \kl{stable system of monics} $\cM$. Then $\bfH$ is said to be a \AP\intro{horizontal weak adhesive HLR category}~\cite{ehrig2010categorical} iff $\bfH$ satisfies the following properties:
\begin{enumerate}[label=(H-\roman*)]
\item $\bfH$ \kl(ssm){has pullbacks along $\cM$-morphisms}.
\item $\bfH$ \kl(ssm){has pushouts along $\cM$-morphisms}, and \kl(ssm){$\cM$-morphisms are stable under pushout}.
\item \AP\phantomintro(adhVK){(H-iii)}$\bfH$ has \kl{horizontal weak VK squares}.
\end{enumerate}
\item Let $\bfV$ be a category that admits a \kl{stable system of monics} $\cM$. Then $\bfV$ is said to be a \AP\intro{vertical weak adhesive HLR category}~\cite{ehrig2010categorical} (often alternatively referred to as an \kl{$\cM$-adhesive category}) iff $\bfV$ satisfies the following properties:
\begin{enumerate}[label=(V-\roman*)]
\item \AP\phantomintro(adhVK){(V-i)} $\bfV$ \kl(ssm){has pullbacks along $\cM$-morphisms}.
\item \AP\phantomintro(adhVK){(V-ii)} $\bfV$ \kl(ssm){has pushouts along $\cM$-morphisms}, and \kl(ssm){$\cM$-morphisms are stable under pushout}.
\item \AP\phantomintro(adhVK){(V-iii)}\kl(ssm){Pushouts along $\cM$-morphisms} in $\bfV$ are \kl{vertical weak van Kampen squares}.
\end{enumerate}
\item Let $\bfW$ be a category that admits a \kl{stable system of monics} $\cM$. Then $\bfW$ is said to be a \AP\intro{weak adhesive HLR category}~\cite{ehrig2010categorical} iff $\bfW$ has the properties of both a \kl{horizontal} and a \kl{vertical weak adhesive HLR category}.
\end{enumerate}
Finally, since in many of the proofs that rely upon the above \kl{adhesivity properties} one in fact needs different sub-statements of the axioms (X-iii) (i.e., of the VK-type axioms), we will use the notation \AP\intro(notationVKa){(X-iii-a)} for the part of axiom (X-iii) referring to stability under pullbacks, and to \AP\intro(notationVKb){(X-iii-b)} for the variant of the statement of axiom \kl{(VK-b)} in the definition of \kl{van Kampen squares}.
\end{definition}

\begin{remark}
The above list of definitions of categories with \kl{adhesivity properties} might appear to have a certain ``asymmetry'' in that for \kl{adhesive} and for \kl{quasi-adhesive categories}, stability of the relevant class of monics under pushout is not explicitly stated. However, one may prove~\cite[Prop.~6.4]{lack2005adhesive} that this stability in fact follows from the other axioms for these kinds of adhesivity. 
\end{remark}

The motivation of the seemingly peculiarly long list of \kl{adhesivity properties} for categories is indeed given by the intricate nature of requirements on categories to admit various notions of rewriting semantics (cf.\ Section~\ref{sec:rs}). It should also be noted that an oddity in this line of research is~\cite{ehrig2010categorical}  that to date only a single example of a category is known that is a \kl{vertical}, but not a \kl{horizontal weak adhesive HLR category} (namely the category $\mathbf{lSet}$ of list sets as introduced by Heindel in~\cite{10.1007/978-3-642-15928-2_17}), while all other known examples of categories with weak forms of \kl{adhesivity properties} are indeed \kl{weak adhesive HLR categories}. This is illustrated in Table~\ref{tab:adh}, which is an adaptation of a similar table presented in~\cite{bp2019-ext}, and which lists both examples of categories with \kl{adhesivity properties} as well as examples of \kl{quasi-topoi} (cf.\ Section~\ref{sec:qt}).

Before presenting some examples of categories with \kl{adhesivity properties} in further detail, it is worthwhile stating the following sufficient condition for when vertical implies horizontal weak HLR adhesivity, which requires the following well-known result:
\begin{theorem}[\cite{ehrig:2006fund}, Thm.~4.26(1)]\label{cor:adhPOPB}
Let $\bfC$ be a category with one of the variants of \kl{adhesivity properties} for some \kl{stable system of monics} $\cM$ (which for the case of $\bfC$ being an \kl{adhesive category} is $\cM=\mono{\bfC}$). Then \AP\intro(aps){pushouts along $\cM$-morphisms are pullbacks} in all cases but one, i.e., when $\bfC$ is a \kl{horizontal weak adhesive HLR category}, in which case \AP\intro(hwahlr){pushouts of spans of $\cM$-morphisms are pullbacks}.
\end{theorem}

\begin{lemma}\label{lem:vwadhTohwadh}
Let $\bfV$ be a \kl{vertical weak adhesive HLR category} with respect to a \kl{stable system of monics} $\cM$. Then a sufficient condition for $\bfV$ to also carry the structure of a \kl{horizontal weak adhesive HLR category} (and thus overall of a \kl{weak adhesive HLR category}) is that \kl(ssm){pushouts along $\cM$-morphisms} are \kl(PO){stable under pullbacks}.
\end{lemma}
\begin{proof}
Since by assumption \kl(ssm){pushouts along $\cM$-morphisms} are \kl(PO){stable under pullbacks}, it remains to prove that in the diagram below, where the top and bottom squares are pushouts of spans of $\cM$-morphisms, and where the left and back squares are pullbacks, the front and right squares are pullbacks:
\begin{equation}
	\ti{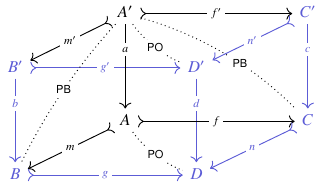}
\end{equation}
Since yet again by assumption \kl(ssm){pushouts along $\cM$-morphisms} are \kl(PO){stable under pullbacks}, and since in a \kl{vertical weak adhesive HLR category}, \kl(aps){pushouts along $\cM$-morphisms are pullbacks}, the claim follows by invoking \kl{pullback-pushout decomposition}.
\qed \end{proof}

The paradigmatic example of an \kl{adhesive category} is the following one:
\begin{definition}\label{def:Graph}
The category $\mathbf{Graph}$ of \AP\intro{directed multigraphs} is defined as the presheaf category $\mathbf{Graph}:=(\bG^{op}\rightarrow \mathbf{Set})$, where $\bG:=(\cdot \rightrightarrows \star)$ is a category with two objects and two morphisms~\cite{ls2004adhesive}. Objects $G=(V_G,E_G,s_G,t_G)$ of $\mathbf{Graph}$ are given by a set of vertices $V_G$, a set of directed edges $E_G$ and the source and target functions $s_G,t_G:E_G\rightarrow V_G$. Morphisms of $\mathbf{Graph}$ between $G,H\in \obj{\mathbf{Graph}}$ are of the form $\varphi=(\varphi_V,\varphi_E)$, with $\varphi_V:V_G\rightarrow V_H$ and $\varphi_E:E_G\rightarrow E_H$ such that $\varphi_V\circ s_G=s_H\circ \varphi_E$ and $\varphi_V\circ t_G=t_H\circ \varphi_E$.
\end{definition}
\begin{theorem}
The category \kl{$\mathbf{Graph}$} is an \kl{adhesive category} and (by definition) a \emph{presheaf topos}~\cite{ls2004adhesive} (and thus in particular a \kl{quasi-topos}), with strict-initial object $\mIO=(\emptyset, \emptyset,\emptyset\to\emptyset,\emptyset\to\emptyset)$ the empty graph, and with the following additional properties:
\begin{itemize}
\item Morphisms are in the classes $\mono{\mathbf{Graph}}$/$\epi{\mathbf{Graph}}$/$\iso{\mathbf{Graph}}$ if they are component-wise injective/surjective/bijective functions, respectively. All monos in $\mathbf{Graph}$ are regular, and $\mathbf{Graph}$ therefore possesses an epi-mono-factorization.
\item For each $G\in \obj{\mathbf{Graph}}$~\cite[Sec.~2.1]{Corradini_2015}, $\eta_G:G\rightarrow T(G)$ is defined as the embedding of $G$ into $T(G)$, where $T(G)$ is defined as the graph with vertex set $V_G':=V_G\uplus \{\star\}$ and edge set $E_G\uplus E_G'$. Here, $E_G'$ contains one directed edge $e_{n,p}: v_n\rightarrow v_p$ for each pair of vertices $(v_n,v_p)\in V_G'\times V_G'$.
\end{itemize}
\end{theorem}

Many of the examples listed in Table~\ref{tab:adh} are obtained via the following construction:
\begin{definition}[\cite{ehrig:2006fund}, Def. A.41]
Let $F:\bfA\rightarrow \bfC$ and $G:\bfB\rightarrow \bfC$ be two functors, and let $\cI$ be an index set. Then the \AP\intro{comma category} $\mathbf{ComCat}(F,G;\cI)$ is defined as a category whose objects are of the form
\begin{equation}
\obj{\mathbf{ComCat}(F,G;\cI)}:=\{(A,B,op=\{op_i\}_{i\in \cI})\mid \forall i\in \cI: op_i\in \mor{F(A),G(B)}\}\,,
\end{equation}
and whose morphisms $f:(A,B;op)\rightarrow (A',B';op')$ consist of pairs of morphisms $f_{\bfA}:A\rightarrow A'$ and $f_{\bfB}:B\rightarrow B'$ such that $G(f_{\bfB})\circ op_i=op_i'\circ F(f_{\bfA})$ for all $i\in \cI$.
\end{definition}

The main interest in this definition of \kl{comma categories} is that they enjoy a number of important properties that render them extremely useful for determining whether various datatypes of practical importance have \kl{adhesivity properties}:
\begin{theorem}\label{thm:comCats}
Let $\mathbf{ComCat}(F,G;\cI)$ be a \kl{comma category}, for $F:\bfA\rightarrow \bfC$ and $G:\bfB\rightarrow \bfC$ two functors, and where $\cI$ is an index set.
\begin{enumerate}[label=(\roman*)]
\item Morphisms $f=(f_{\bfA},f_{\bfB})$ in $\mathbf{ComCat}(F,G;\cI)$ are mono-/epi-/isomorphisms iff they are component-wise mono-/epi-/isomorphisms, respectively~\cite[Fact A.43]{ehrig:2006fund}.
\item If $\bfA$ and $\bfB$ \kl{have pushouts} and $F$ preserves pushouts, then $\mathbf{ComCat}(F,G;\cI)$ \kl{has pushouts}, and these are constructed component-wise~\cite[Fact A.43]{ehrig:2006fund}.
\item If $\bfA$ and $\bfB$ \kl{have pullbacks} and $G$ preserves pullbacks, then $\mathbf{ComCat}(F,G;\cI)$ \kl{has pullbacks}, and these are constructed component-wise~\cite[Fact A.43]{ehrig:2006fund}.
\item If $(\bfA,\cM_1)$ and $(\bfB,\cM_2)$ are \kl{adhesive HLR categories}, $F:\bfA\rightarrow \bfC$ preserves \kl(ssm){pushouts along $\cM_1$-morphisms} and $G:\bfB\rightarrow \bfC$ preserves pullbacks, then  $\mathbf{ComCat}(F,G;\cI)$ is an \kl{adhesive HLR category} with respect to the \kl{stable system of monics} $\cM=(\cM_1\times\cM_2)\cap \mor{\mathbf{ComCat}(F,G;\cI)}$~\cite[Thm.~4.15.4]{ehrig:2006fund}.
\item If $(\bfA,\cM_1)$ and $(\bfB,\cM_2)$ are \kl{adhesive HLR categories}, $F:\bfA\rightarrow \bfC$ preserves \kl(ssm){pushouts along $\cM_1$-morphisms} and $G:\bfB\rightarrow \bfC$ preserves \kl(ssm){pullbacks along $\cM_2$-morphisms}, then  $\mathbf{ComCat}(F,G;\cI)$ is a \kl{weak adhesive HLR category} with respect to the \kl{stable system of monics} $\cM=(\cM_1\times\cM_2)\cap \mor{\mathbf{ComCat}(F,G;\cI)}$~\cite[Thm.~4.15.4]{ehrig:2006fund}.
\end{enumerate}
\end{theorem}

Generalizing from directed graphs to hypergraphs, it is interesting to note that the various notions of hypergraphs yield different notions of \kl{adhesivity properties}. We present here one of the standard constructions in the literature:
\begin{definition}
Let $\bf{HyperGraph}$ denote the category of  \kl{directed ordered hypergraphs}~\cite[Fact~4.17]{ehrig:2006fund}, defined as the \kl{comma category} $(Id_{\mathbf{Set}}, \square^{*}; \{1,2\})$. Here, $Id_{\mathbf{Set}}$ denotes the identity functor on the category $\mathbf{Set}$, while $\square^{*}:\mathbf{Set}\rightarrow \mathbf{Set}$ denotes the free monoid functor (which assigns to each set $A$ the free monoid $A^{*}$ on $A$, and to each set morphism $f:A\rightarrow B$ the free monoid morphism $f^{*}:A^{*}\rightarrow B^{*}$). More explicitly, an object of $\mathbf{HyperGraph}$ is a tuple $H=(V_H,E_H,s_H,t_H)$, where $V_H$ is the set of vertices, $E_H$ the set of hyperedges, and where $s_H,t_H:E\rightarrow V^{*}$ are the source and target functions (assigning to each edge an ordered list of source and target vertices). A morphism $\varphi:H\rightarrow H'$ in $\mathbf{HyperGraph}$ is given by a pair of morphisms $\varphi_V:V_H\rightarrow V_{H'}$ and $\varphi_E:E_H\rightarrow E_{H'}$ such that the diagram in~\eqref{def:HyperGraphMorph} below commutes.
\begin{equation}\label{def:HyperGraphMorph}
\ti{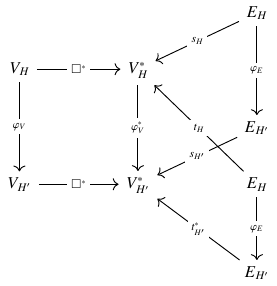}
\end{equation}
\end{definition}

\begin{proposition}[\cite{ehrig:2006fund}, Fact~4.17]
The category $\mathsf{HyperGraph}$ is an \kl{adhesive HLR category} with respect to the \kl{stable system of monics} $\cM_{\mathbf{HyperGraph}}$ given by morphisms $\varphi=(\varphi_V,\varphi_E)$ where $\varphi_V$ and $\varphi_E$ are both monomorphisms.
\end{proposition}

We conclude our brief presentation of examples by mentioning a number of slightly more sophisticated cases. Many interesting examples of categories with \kl{adhesivity properties} may be obtained by using the construction of presheaves (cf.\ e.g.\ \cite[Sec.~5]{Grochau_Azzi_2019} for a review within the context of categorical rewriting theory). Remarkable examples include the category of asynchronous graphs as introduced in~\cite{10.1145/3373718.3394762}, which permit to model certain structures in game semantics, and various notions of attributed and symbolic graphs as discussed in~\cite{Grochau_Azzi_2019}. Many other examples concern comma category constructions, with a number of illustrative examples provided in Table~\ref{tab:ccsExamples}. More intricate examples still have been developed in the context of so-called hierarchical graphs, which are obtained via comma-categorical constructions along various notions of super-power functors, and whose \kl{adhesivity properties} have been studied in~\cite{10.1007/978-3-319-61470-0_2, padberg2017towards} (see also~\cite{cgm:fossacs22}).

\begin{table}[h]
\centering
\begin{tabular}{cl}
$\ti{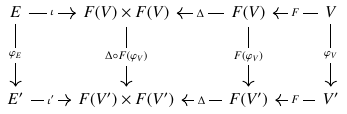}$ 
& 
\begin{tabular}{cL{7.5cm}}
$F$ & Description\\
\toprule
$Id_{\mathbf{Set}}$ & \kl{directed multigraphs}\\
\midrule
$\square^{*}$ & directed ``ordered'' hypergraphs with multiple incidences (\AP\intro{$\mathbf{HyperGraph}$}~\cite[Fact~4.17]{ehrig:2006fund} aka $\mathbf{PNet}$~\cite[Ex.~7]{quasi-topos-2007})\\
\midrule
$\cM$ & directed ``unordered'' hypergraphs with multiple incidences (= \AP\intro{$\mathbf{PTNets}$} of~\cite[Fact~4.21]{ehrig:2006fund})\\
\midrule
$\cP$ & directed ``unordered'' hypergraphs with simple incidences (= \AP\intro{$\mathbf{ElemNets}$} of~\cite[Fact~4.20]{ehrig:2006fund})\\[1em]
\end{tabular}
\\
\midrule 
$\ti{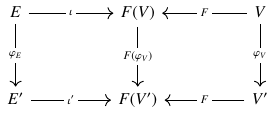}$ 
& 
\begin{tabular}{cL{7.5cm}}
$F$ & Description\\
\toprule
$\cP^{(1,2)}$ & \AP\intro{undirected multigraphs}~\cite{bp2019-ext}\\
\midrule
$\square^{*}$ & undirected ``ordered'' hypergraphs with multiple incidences (i.e.\ lists)\\
\midrule
$\cM$ & undirected ``unordered'' hypergraphs with multiple incidences\\
\midrule
$\cP$ & undirected ``unordered'' hypergraphs with simple incidences\\[1em]
\end{tabular}
\\
\end{tabular}
\caption{\label{tab:ccsExamples}Collection of examples for categories with \kl{adhesivity properties} based upon two ``schemas'' of \kl{comma category} constructions. Here, we employ the notations $\square^{*}$ for the free monoid functor, $\cM$ (also denoted $\oplus^{*}$ in~\cite{ehrig:2006fund}) for the free commutative monoid functor, $\cP$ for the covariant powerset functor, and $\cP^{(1,2)}$ for the restricted version thereof (cf.\ e.g.\ \cite{padberg2017towards}).}
\end{table}

\begin{table}[htpb]
\centering
\vspace{2em}
{\setlength{\extrarowheight}{5pt}
  \begin{tabular}{C{5.5cm}ccccccccc}
Category\newline(underlying data type)\newline&
\rot{\kl{quasi-topos}} &
\rot{\kl{adhesive}} &
\rot{\kl{quasi-adhesive}} &
\rot{\kl{adhesive HLR}} &
\rot{\kl{hor.\ weak adh.\ HLR}} &
\rot{\kl{vert.\ weak adh.\ HLR}} &
references
    \\[-0.5em] \toprule
$\mathbf{Set}$\newline (sets)  &%
\YES	 	& %
\YES 	& %
\YES	 	& %
\YES 	& %
\YES		& %
\YES		& %
\cite{lack2005adhesive}	\\%
$\mathbf{Graph}$\newline(\kl{directed multigraphs})	&%
\YES		& %
\YES		& %
\YES		& %
\YES		& %
\YES		& %
\YES		& %
\cite{lack2005adhesive}		\\ %
$\mathbf{HyperGraph}$\newline(\kl{directed ordered hypergraphs}) 	& %
\YES	& %
\YES	& %
\YES	& %
	\YES & %
	\YES & %
	\YES & %
 \cite[Ex.~7]{quasi-topos-2007}	\\ %
 $\mathbf{Sig}$\newline(algebraic signatures) 	& %
\YES	& %
\YES	& %
\YES	& %
\YES	& %
\YES & %
\YES & %
 \cite[Ex.~6]{quasi-topos-2007}	\\ %
$\hat{\mathbf{S}}$\newline(presheaves on category $\mathbf{S}$)	& %
\YES		& %
\YES		& %
\YES		& %
\YES		& %
\YES		& %
\YES		& %
\cite{Grochau_Azzi_2019,Lack2006}	\\ %
$\hat{\mathbf{T}_{\Sigma}}$\newline(term graphs over a signature $\Sigma$)	& %
\MAYBE		& %
		& %
\YES		& %
\YES		& %
\YES		& %
\YES		& %
\cite{CORRADINI200543}	\\ %
$\mathbf{TripleGraph}$\newline(functor category $[\mathbf{S}_3,\mathbf{Graph}]$) 	& %
\MAYBE	& %
	& %
	& %
	\YES & %
	\YES & %
	\YES & %
 \cite[Fact~4.18]{ehrig:2006fund}\\ %
$\mathbf{AGraph}_{\Sigma}$\newline(attributed graphs over signature $\Sigma$)	& %
\MAYBE	& %
	& %
	& %
\YES	& %
\YES 	& %
\YES 	& %
\cite[Thm.~11.11]{ehrig:2006fund}, \cite{GABRIEL_2014,Grochau_Azzi_2019}	\\ %
$\mathbf{SymbGraph}_{D}$\newline(symbolic graphs over $\Sigma$-algebra $D$)	& %
\MAYBE	& %
	& %
	& %
\YES 	& %
\YES 	& %
\YES 	& %
\cite{Grochau_Azzi_2019}, \cite[Thm.~2]{orejas2010symbolic}	\\ %
$\mathbf{uGraph}$\newline(\kl{undirected multigraphs}) & %
\MAYBE	& %
		& %
		& %
		& %
\YES	& %
\YES	& %
\cite{bp2019-ext}	\\ %
$\mathbf{ElemNets}$\newline(\kl{elementary Petri nets})	& %
\MAYBE	& %
	& %
	& %
{\color{orange}(!)}	& %
\YES	& %
\YES	& %
\cite{GABRIEL_2014}	\\ %
$\mathbf{PTnets}$\newline(\kl{place/transition nets})	& %
\MAYBE	& %
	& %
	& %
	& %
\YES	& %
\YES	& %
\cite[Fact~4.21]{ehrig:2006fund}, \cite{GABRIEL_2014}	\\ %
$\mathbf{Spec}$\newline(algebraic specifications)	& %
\YES	& %
	& %
	& %
	& %
\YES	& %
\YES	& %
\cite[Fact 4.24]{ehrig:2006fund}, \cite[Ex.~6]{quasi-topos-2007}	\\ %
$\mathbf{SGraph}$\newline(\kl{directed simple graphs}) & %
\YES	& %
		& %
		& %
		& %
{\color{green}\YES}	& %
\YES	& %
\cite[Prop.~17]{quasi-topos-2007}, {\color{green}Corollary~\ref{cor:main}(q-v)}	\\ %
$\mathbf{Set}_F$\newline(coalgebras for $F:\mathbf{Set}\rightarrow \mathbf{Set}$)	& %
$(*)$ 	& %
	& %
	& %
	& %
	& %
$(\dag)$	& %
\cite{quasi-topos-2007}, \cite[Thm.~1]{padberg2017towards}	\\ %
$\mathbf{lSets}$\newline(list sets)	& %
\MAYBE	& %
	& %
	& %
	& %
	& %
\YES	& %
 \cite{10.1007/978-3-642-15928-2_17}	\\ %
\bottomrule
\end{tabular}
}
\caption{\label{tab:adh}Examples of categories exhibiting various forms of \kl{adhesivity properties}. The symbol \MAYBE~indicates when a certain property is (to the best of our knowledge) not known to hold. Note that for the HLR variants of adhesivity properties, the information not contained in the table is the precise nature (cf.\ references provided) of the \kl{stable system of monics} $\cM$ for which the adhesivity properties hold. Moreover, the precise conditions $(*)$ and $(\dag)$ under which the category $\mathbf{Set}_F$ of $F$-coalgebras has quasi-topos or adhesivity properties are provided in~\cite{quasi-topos-2007} and~\cite[Thm.~1]{padberg2017towards}, respectively.}
\vspace{2em}
\end{table}

\subsection{Quasi-topoi}\label{sec:qt}

Quasi-topoi have been considered in the context of rewriting theories as a natural generalization of adhesive categories in~\cite{lack2005adhesive}. While several adhesive categories of interest to rewriting are topoi, including in particular the category $\mathbf{Graph}$ of directed multigraphs (cf.\ Definition~\ref{def:Graph}), it is not difficult to find examples of categories equally relevant to rewriting theory that fail to be topoi. A notable such example is the category $\mathbf{SGraph}$ of directed simple graphs (cf.\ Definition~\ref{def:SGraph}). 

Let us first recall a number of results from the work of Cockett and Lack~\cite{COCKETT2002223,COCKETT200361} on restriction categories. We will only need a very small fragment of their theory, namely the definition and existence guarantees for $\cM$-partial map classifiers, so we will follow mostly \cite{Corradini_2015}. We will in particular not be concerned with the notion of $\cM$-partial maps itself.

\begin{definition}[\cite{Corradini_2015}, Sec.~2.1; compare~\cite{COCKETT200361}, Sec.~2.1]
For a \kl{stable system of monics} $\cM$ in a category $\bfC$, an \AP\intro{$\cM$-partial map classifier} $(T,\eta)$ is a functor $T:\bfC\rightarrow \bfC$ and a natural transformation $\eta:ID_{\bfC}\xrightarrow{.} T$ such that 
\begin{enumerate}
\item for all $X\in \obj{\bfC}$, $\eta_X:X\rightarrow T(X)$ is in $\cM$.
\item for each span $(A\xleftarrow{m}X\xrightarrow{f}B)$ with $m\in \cM$, there exists a unique morphism $A\xrightarrow{\varphi(m,f)}T(B)$ such that $(m,f)$ is a pullback of $(\varphi(m,f),\eta_B)$.
\end{enumerate}
\end{definition}

\begin{proposition}[\cite{Corradini_2015},  Prop.~6]
For every $\cM$-partial map classifier $(T,\eta)$, $T$ preserves pullbacks, and $\eta$ is Cartesian, i.e., for each $X\xrightarrow{f}Y$, $(\eta_x,f)$ is a pullback of $(T(f),\eta_Y)$.
\end{proposition}

\begin{definition}[\cite{quasi-topos-2007}, Def.~9]
A category $\bfC$ is a \AP\emph{\intro{quasi-topos}} iff
\begin{enumerate}
\item it has finite limits and colimits.
\item it is locally Cartesian closed.
\item it has a regular-subobject-classifier.
\end{enumerate}
\end{definition}

Based upon a variety of different results from the rich literature on quasi-topoi, we will now exhibit that quasi-topoi indeed possess all technical properties required in order for \kl(crsSqPO){non-linear SqPO-rewriting to be well-posed}:
\begin{corollary}\label{cor:main}
Every \emph{quasi-topos} $\bfC$ enjoys the following properties:
\begin{enumerate}[label=(q-\roman*)]
\item It has (by definition) a \kl{stable system of monics} $\cM=\regmono{\bfC}$ (the class of regular monos), which coincides with the class of \AP\emph{\intro{extremal monomorphisms}}~\cite[Cor.~28.6]{adamek2006}, i.e., if $m=f\circ e$ for $m\in \regmono{\bfC}$ and $e\in \epi{\bfC}$, then $e\in \iso{\bfC}$.
\item It has (by definition) an \kl{$\cM$-partial map classifier} $(T,\eta)$.
\item It is \AP\intro(qt){rm-quasi-adhesive},  i.e., %
it \AP\phantomintro(qt){has pushouts along regular monomorphisms}\kl(ssm){has pushouts along regular monomorphisms}, %
these are \AP\phantomintro(qt){pushouts along regular monomorphisms are stable under pullbacks}\kl(PO){stable under pullbacks}, and %
\AP\intro(qt){pushouts along regular monomorphisms are pullbacks}~\cite{garner2012axioms}.
\item It is a \kl{vertical weak adhesive HLR category} (sometimes referred to as $\cM$-adhesive category)~\cite[Lem.~13]{10.1007/978-3-642-15928-2_17}.
\item The latter entails according to Lemma~\ref{lem:vwadhTohwadh} that every quasi-topos is in fact \AP\phantomintro(qt){weak adhesive HLR category} a \kl{weak adhesive HLR category}.
\item For all pairs of composable morphisms $A\xrightarrow{f}B$ and $B\xrightarrow{m}C$ with $m\in \cM$, there \AP\intro(qt){exists a final pullback-complement (FPC)} $A\xrightarrow{n}F\xrightarrow{g}C$, and with $n\in \cM$ (\cite[Thm.~1]{Corradini_2015}; cf.\ Theorem~\ref{thm:FPC}).
\item It possesses an \kl(EMS){epi-$\cM$-factorization}~\cite[Prob.~28.10]{adamek2006}: each morphism $A\xrightarrow{f}B$ factors as $f=m\circ e$, with morphisms $A\xrightarrow{e}\bar{B}$ in $\epi{\bfC}$ and $\bar{B}\xrightarrow{m}A$ in $\cM$ (uniquely up to isomorphism in $\bar{B}$).
\item It possesses a \AP\intro(qt){strict initial object} $\mIO\in \obj{\bfC}$~\cite[A1.4]{johnstone2002sketches}, i.e., for every object $X\in \obj{\bfC}$, there exists a morphism $i_X:\mIO\rightarrow X$, and if there exists a morphism $X\rightarrow \mIO$, then $X\cong \mIO$. 
\end{enumerate}
If in addition the strict initial object $\mIO$ is \AP\intro(ssm){$\cM$-initial}, i.e., if for all objects $X\in \obj{\bfC}$ the unique morphism $i_X:\mIO\rightarrow X$ is in $\cM$, then $\bfC$ has \AP\intro(qt){disjoint coproducts}, i.e., for all $X,Y\in \obj{\bfC}$, the pushout of the $\cM$-span $X\leftarrowtail\mIO\rightarrowtail Y$ is $X\rightarrowtail X+Y\leftarrowtail Y$ (cf.\ \cite[Thm.~3.2]{MONRO1986141}, which also states that this condition is equivalent to requiring $\bfC$ to be a \AP\intro(qt){solid quasi-topos}), and the coproduct injections are $\cM$-morphisms as well. Finally, if pushouts along regular monos of $\bfC$ are van Kampen, $\bfC$ is a \kl{rm-adhesive category}~\cite[Def.~1.1]{garner2012axioms}.
\end{corollary}
\begin{remark}\label{rem:qt}
An interesting (and, as it turns out, highly relevant) curiosity of the above list of properties enjoyed by every \kl{quasi-topos} is that, while \kl{quasi-topoi} in general \emph{fail} to be \kl{adhesive HLR categories} (cf.\ e.g.\ \cite{quasi-topos-2007,GABRIEL_2014} for the famous and paradigmatic example of $\mathbf{SGraph}$, the category of \kl{directed simple graphs}), they \emph{do} satisfy axiom \kl(notationVKa){(L-iii-a)}, i.e., \kl(ssm){pushouts along regular monomorphisms} are \kl(PO){stable under pullbacks}. Therefore, as we will demonstrate in Section~\ref{sec:rs}, (\kl{finitary}) \kl{quasi-topoi} are a suitable type of category for all variants of \kl{Sesqui-Pushout (SqPO) semantics}, while they in general do \emph{not} have sufficient properties to support \kl(crsType){generic} \kl{Double-Pushout (DPO) semantics}  (cf.\ Table~\ref{tab:main}).
\end{remark}

The prototypical example of quasi-topoi in rewriting is the following notion of directed simple graphs:

\begin{definition}\label{def:SGraph}
Let $\mathbf{SGraph}$, the \AP\intro{category of directed simple graphs}\footnote{Some authors prefer to not consider directly the category $\mathbf{BRel}$, but rather define $\mathbf{SGraph}$ as some category equivalent to $\mathbf{BRel}$, where simple graphs are of the form $\langle V,E\rangle$ with $E\subseteq V\times V$. This is evidently equivalent to directly considering $\mathbf{BRel}$, whence we chose to not make this distinction in this paper.}, be defined as the category of binary relations $\mathbf{BRel}\cong\mathbf{Set}\git\Delta$~\cite{quasi-topos-2007}. Here, $\Delta:\mathbf{Set}\rightarrow\mathbf{Set}$ is the pullback-preserving diagonal functor defined via $\Delta X:= X\times X$, and $\mathbf{Set}\git\Delta$ denotes the full subcategory of the slice category $\mathbf{Set}/\Delta$ defined via restriction to objects $m:X\rightarrow \Delta X$ that are monomorphisms. More explicitly, an object of $\mathbf{Set}\git\Delta$ is given by $S=(V,E, \iota)$, where $V$ is a set of vertices, $E$ is a set of directed edges, and where $\iota:E\rightarrow V\times V$ is an injective function. A morphism $f=(f_V,f_E)$ between objects $S$ and $S'$ is a pair of functions $f_V:V\rightarrow V'$ and $f_E:E\rightarrow E'$ such that $\iota'\circ f_E=(f_V\times f_V)\circ\iota$ (see~\eqref{eq:epiRMfactorizationSGraph}).
\end{definition}

The category $\mathbf{SGraph}$ satisfies the following well-known properties:
\begin{theorem}\label{thm:SGraphProperties}
The category $\mathbf{SGraph}$ is \emph{not} adhesive, but it is a quasi-topos~\cite{quasi-topos-2007}, and with the following additional properties:
\begin{enumerate}[label=(S-\roman*)]
\item In $\mathbf{SGraph}$~\cite{quasi-topos-2007} (compare~\cite[Prop.~9]{BRAATZ2011246}), morphisms $f=(f_V,f_E)$ are monic (epic) if $f_V$ is monic (epic), while isomorphisms satisfy that both $f_V$ and $f_E$ are bijective. \emph{Regular} monomorphisms in $\mathbf{SGraph}$ are those for which $(\iota, f_E)$ is a pullback of $(\Delta(f_V),\iota')$~\cite[Lem.~14(ii)]{quasi-topos-2007}, i.e., a monomorphism is regular iff it is \emph{edge-reflecting}. As is the case for any \kl{quasi-topos}, $\mathbf{SGraph}$ possesses an epi-regular mono-factorization.
\item The regular mono-partial map classifier $(T,\eta)$ of $\mathbf{SGraph}$ is defined as follows~\cite[Ex.~28.2(3)]{adamek2006}: for every object $S=(V,E,\iota)\in \obj{\mathbf{SGraph}}$, 
\begin{equation}
T(S):= (V_{\star}=V\uplus \{\star\}, E_{\star}=E\uplus (V\times\{\star\})\uplus (\{\star\}\times V)\uplus \{(\star,\star)\}, \iota_{\star})\,,
\end{equation}
where $\iota_{\star}$ is the evident inclusion map, and moreover $\eta_{S}:S\rightarrowtail T(S)$ is the (by definition edge-reflecting) inclusion of $S$ into $T(S)$.
\item $\mathbf{SGraph}$ possesses a regular mono-initial object $\mIO=(\emptyset, \emptyset,\emptyset\to\emptyset)$.
\end{enumerate}
\end{theorem}
\begin{proof}
While most of these results are standard, we briefly demonstrate that the epi-regular mono-factorization of $\mathbf{SGraph}$~\cite{quasi-topos-2007} is  ``inherited'' from the epi-mono-factorization of the adhesive category $\mathbf{Set}$. To this end, given an arbitrary morphism $f=(f_V,f_E)$ in $\mathbf{SGraph}$ as on the left of~\eqref{eq:epiRMfactorizationSGraph}, the epi-mono-factorization $f_V=m_V\circ e_V$ lifts via application of the diagonal functor $\Delta$ to a decomposition of the morphism $f_V\times f_V$. Pulling back $(\Delta(m_v),\iota')$ results in a span $(\tilde{\iota},f_E'')$ and (by the \kl{universal property of pullbacks}) an induced morphism $f_E'$ that makes the diagram commute. By stability of monomorphisms under pullbacks, $\tilde{\iota}$ is a monomorphism, thus the square marked $(*)$ precisely constitutes the data of a regular monomorphism in $\mathbf{SGraph}$, while the square marked $(\dag)$ is an epimorphism in $\mathbf{SGraph}$ (since $e_V\in \epi{\mathbf{Set}}$).
\qed \end{proof}

\begin{equation}\label{eq:epiRMfactorizationSGraph}
\ti{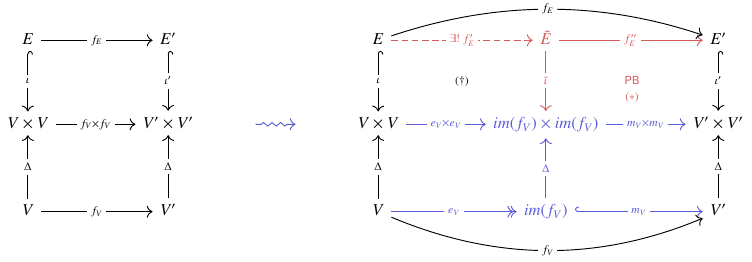}
\end{equation}

A number of additional examples of quasi-topoi relevant for rewriting applications is referenced in Table~\ref{tab:adh}. Most of these examples arise via so-called \emph{Artin gluing} as demonstrated in~\cite{quasi-topos-2007}.

\subsection{Explicit constructions of universal categorical operations (multi-sums, FPCs, multi-IPCs and FPAs)}\label{sec:exConstr}

The focus of this section is to provide some results on constructions (as opposed to merely existence) of some of the key concepts necessary for practically working with compositional rewriting theories.

\subsubsection{$\cM$-multi-sums}\label{sec:exConstr:mms}

An important technical ingredient for our constructions is the notion of \kl{$\cM$-multi-sum} (see Definition~\ref{def:ms}), a special case of the general theory of multi-(co-)limits due to Diers~\cite{diers1978familles}. In order to provide a more constructive version of this definition, it is in practice often useful to consider a certain notion of finiteness of objects in the underlying categories in the sense of~\cite{GABRIEL_2014}:

\begin{definition}[\cite{GABRIEL_2014}, Def.~2.8 \& Def.~4.1]
Let $\bfC$ be a category with a \kl{stable system of monics} $\cM$.
\begin{enumerate}[label=(F-\roman*)]
\item An object $A$ of $\bfC$ is said to be \AP\intro{finitary} (or finitely $\cM$-well-powered) if it has only finitely many $\cM$-subobjects up to isomorphism. Here, an \AP\intro{$\cM$-subobject} of $A$ is an $\cM$-morphism $X\rtail x \rightarrow A$, and an $\cM$-subobject $Y\rtail y \rightarrow A$ is defined to be isomorphic to $x$ if there exists an isomorphism $X-\phi\rightarrow Y$ such that $y=\phi\circ x$. $\bfC$ is a \AP\intro{finitary category} (w.r.t.\ $\cM$) if every object of $\bfC$ is \kl{finitary}.
\item The \AP\intro{finitary restriction} of $\bfC$, denoted $\bfC_{\text{fin}}$, is defined by the restriction to \kl{finitary objects} and morphisms thereof.
\end{enumerate}
\end{definition}

The importance of the work presented in~\cite{GABRIEL_2014} for constructing compositional rewriting theories is in particular that it provides an elegant method to demonstrate that finitary restrictions of suitable base categories preserve the requisite \kl{adhesivity properties}, and in addition yield categories that are guaranteed to possess a certain form of factorization system: 
\begin{theorem}[\cite{GABRIEL_2014}, Thm.~4.6 \& Fact~3.4]\label{thm:finitarityCats}
Let $\bfC$ be a \kl{vertical weak adhesive HLR category} with respect to a \kl{stable system of monics} $\cM$. Denote by $C_{\text{fin}}$ the \kl{finitary restriction} of $\bfC$.
\begin{enumerate}[label=(\roman*)]
\item $C_{\text{fin}}$ has a \kl{stable system of monics} $\cM_{\text{fin}}= \cM\cap C_{\text{fin}}$.
\item $\bfC_{\text{fin}}$ is a \kl{vertical weak adhesive HLR category} with respect to $\cM_{\text{fin}}$.
\item $C_{\text{fin}}$ is \kl{($\cE_{\text{fin}}$, $\cM_{\text{fin}}$)-structured}, where $\cE_{\text{fin}}$ denotes the class of \AP\intro(ssm){extremal morphisms} (w.r.t.\ $\cM_{\text{fin}}$) defined as\footnote{It is instructive to compare this definition to the case of a category with \kl(EMS){epi-mono-factorizations}: here, since $m\circ f$ being an epimorphism implies that $m$ is an epimorphism, then if $m$ is also a monomorphism, this indeed implies that it is an isomorphism.  However, it is important to note that as highlighted in~\cite{GABRIEL_2014}, there exist finitary categories where $\cE_{\text{fin}}$ is \emph{not} a class of epimorphisms.}
\begin{equation}
\cE_{\text{fin}} := \{ e \in \bfC_{\text{fin}} \mid \forall m,f \in \bfC_{\text{fin}}: e=m\circ f: m\in \cM_{fin} \Rightarrow m\in \iso{\bfC_{\text{fin}}}\}\,.
\end{equation}
\end{enumerate}
\end{theorem}

With these preparations, we may verify that as presented below, under suitable conditions on the underlying category, \kl{$\cM$-multi-sums} (as e.g.\ also considered in~\cite{lmcs:1573}) coincide with the concept referred to as \emph{($\cE'$,$\cM$)-pair factorizations} in the graph rewriting literature~\cite{ehrig:2006fund}: 

\begin{lemma}[\cite{ehrig:2006fund}; \cite{GABRIEL_2014}, Fact~A.3.7]\label{lem:constrMS}
Let $\bfC$ be a \kl{finitary} \kl{vertical weak adhesive HLR category} with respect to a \kl{stable system of monics} $\cM$, and denote by $\cE$ the class of \kl(ssm){extremal morphisms} with respect to $\cM$.
\begin{enumerate}[label=(\roman*)]
\item \AP\intro(msEM){Existence:} If $\bfC$ has binary coproducts, then every cospan of $\cM$-morphisms $A\rtail a\rightarrow Z\leftarrow b \ltail B$ factors essentially uniquely through a cospan of $\cM$-morphisms $A\rtail y_A\rightarrow Y\leftarrow y_B \ltail B$ and an $\cM$-morphism $Y\rtail m\rightarrow Z$, where $m$ is obtained via the \kl(EMS){$\cE$-$\cM$-factorization} $A+B-e\twoheadrightarrow Y \rtail m\rightarrow Z$ of the induced morphism $A+B-[a,b]\rightarrow Z$, and where $y_A=e\circ in_A$ and $y_B=e\circ in_B$.
\item \AP\intro(msEM){Construction:} if $\bfC$ in addition has an \kl(ssm){$\cM$-initial object} $\mIO$, then $\Msum{A}{B}$ consists of cospans of $\cM$-morphisms obtained as pushouts $A\rtail p_A \rightarrow P\leftarrow p_B\ltail B$ of $\cM$-spans $A\leftarrow x_A \ltail X\rtail x_B\rightarrow B$ (i.e., ``$\cM$-partial overlaps'') extended by $\cE$-morphisms $P-q\twoheadrightarrow Q$ such that $q_A=q\circ p_A$ and $q_B=q\circ p_B$ are in $\cM$.
\item \AP\intro(msEM){Refinements:} if $\bfC$ in addition \kl{has pullbacks}, and if \kl(ssm){pushouts along $\cM$-morphisms} in $\bfC$ are \kl(PO){stable under pullbacks}, then the extension morphisms $P-q\twoheadrightarrow Q$ are morphisms in $\cE\cap \mono{\bfC}$ (so-called ``refinements'').
\end{enumerate}
\end{lemma}
\begin{corollary}
Every quasi-topos with $\cM$-initial object $\mIO$ has \kl{$\cM$-multi-sums} and \kl(msEM){refinements} according to Lemma~\ref{lem:constrMS}. 
\end{corollary}

\begin{figure}[ht]
\centering
    \includegraphics[scale=0.5]{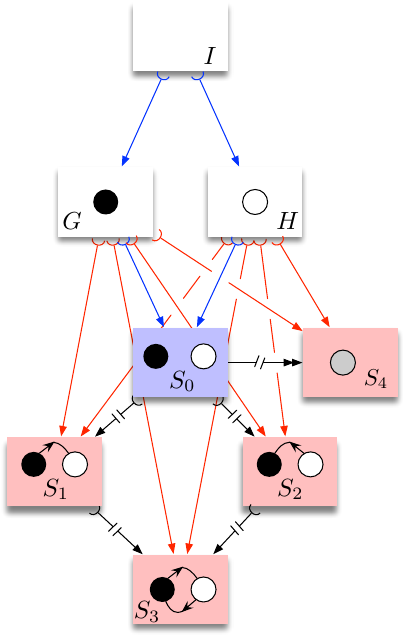}
\caption{\label{fig:msExample}Example for an \kl{$\cM$-multi-sum} $\Msum{\bullet}{\bullet}$ of one-vertex graphs in \kl{$\mathbf{SGraph}$} (i.e., for $\cM=\regmono{\mathbf{SGraph}}$). Note in particular the monic-epis that extend the two-vertex graph $S_0$ into the graphs $S_1$, $S_2$ and $S_3$, all of which have the same vertices as $S_0$ (recalling that a morphism in $\mathbf{SGraph}$ is monic/epic if it is so on vertices), yet additional edges, so that in particular none of the morphisms $S_0\rightarrow S_j$ for $j=1,2,3$ is edge-reflecting.}
\end{figure}
Since in an adhesive category all monomorphisms are regular~\cite{ls2004adhesive}, in this case the multi-sum construction simplifies to the statement that every monic cospan can be uniquely factorized as a cospan obtained as the pushout of a monic span composed with a monomorphism. It is however worthwhile emphasizing that for generic quasi-topoi $\bfC$ one may have $\cM\neq \mono{\bfC}$, as is the case in particular for the quasi-topos $\mathbf{SGraph}$ of simple graphs. We illustrate this phenomenon in Figure~\ref{fig:msExample}, via presenting the multi-sum construction for $A=B=\bullet$.  

\subsubsection{FPCs along $\cM$-morphisms}\label{sec:mFPCs}

Let us now turn to the question of the existence of FPCs. To this end, it will prove useful to recall from~\cite{Corradini_2015} the following constructive result:

\begin{theorem}[\cite{Corradini_2015}, Thm.~1]\label{thm:FPC}
For a category $\bfC$ with \kl{$\cM$-partial map classifier} $(T,\eta)$, the \kl{final pullback complement (FPC)} of a composable sequence of arrows $A\xrightarrow{f}B$ and $B\xrightarrow{m}C$ with $m\in \cM$ is guaranteed to exist, and is constructed via the following algorithm:
\begin{enumerate}
\item Let $\bar{m}:=\varphi(m,id_B)$ (i.e., the morphism that exists by the universal property of $(T,\eta)$, cf.\ square $(1)$ below).
\item Construct $T(A)\xleftarrow{\bar{n}}F\xrightarrow{g}C$ as the pullback of $T(A)\xrightarrow{T(f)}T(B)\xleftarrow{\bar{m}}C$ (cf.\ square $(2)$ below); by the \kl{universal property of pullbacks}, this in addition entails the existence of a morphism $A\xrightarrow{n}F$.
\end{enumerate}
Then $(n,g)$ is the FPC of $(f,m)$, and $n$ is in $\cM$.
\end{theorem}

\begin{equation}\label{eq:M-FPC}
\ti{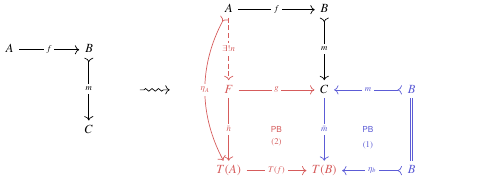}
\end{equation}

An interesting (and to the best of our knowledge open) question that arises in considering (input-) linear Sesqui-Pushout (SqPO) rewriting semantics concerns the nature of the morphism $g$ in the above definition of FPCs under the constraint that $f$ is an $\cM$-morphism, which we refer to as \AP\intro(ssm){stability of $\cM$-morphisms under FPCs}. We may provide a first partial answer for the case where $\cM=\mono{\bfC}$:
\begin{lemma}\label{lem:FPCauxMonoStab}
Let $\bfC$ be a category with a \kl{stable system of monics} $\cM=\mono{\bfC}$ that comprises all monomorphisms of $\bfC$, that \kl(ssm){has pullbacks along monomorphisms}, and that has a \kl{$\mono{\bfC}$-partial map classifier}. Then an FPC $A-n\rightarrow F-g\rightarrow C$ of a composable sequence of monomorphisms $A\rtail f\rightarrow B\rtail m\rightarrow C$ satisfies the property that both $n$ and $g$ are monomorphisms.
\end{lemma}
\begin{proof}
The property that $n\in \mono{\bfC}$ follows from the \kl(ssm){stability of monomorphisms under pullback}, hence it remains to show that $g$ is a monomorphism. This is by definition equivalent to showing that for arbitrary morphisms $H-h_1\rightarrow F$ and $H-h_2\rightarrow F$ such that $g\circ h_1=g\circ h_2$, it follows that $h_1=h_2$. To this end, let $d_H=g\circ h_1=g\circ h_2$, and obtain the span $\langle p,p_B\rangle$ via taking a pullback of the cospan $\rangle d_H, m\langle$:
\begin{equation}
\ti{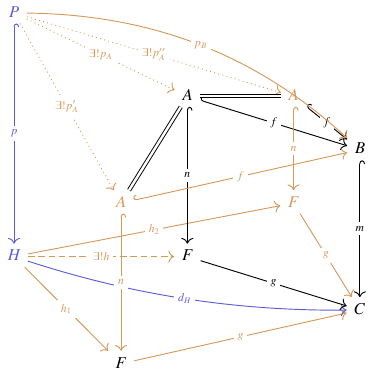}
\end{equation}
\begin{itemize}
\item By the \kl{universal property of pullbacks}, there exist unique morphisms $p_A'$ and $p_A''$ that make the front and back parts of the diagram commute, and since $f$ is a monomorphism and thus the span $\langle id_A,id_A\rangle$ is a pullback of the cospan $\rangle f,f\langle$, there also exists a unique morphism $p_A$ that makes the inner part of the diagram commute. The latter entails in particular that $p_A=p_A'=p_A''$.
\item Since by assumption $(n, g)$ is an FPC of $(f,m)$, by the \kl{universal property of FPCs} there exists a unique morphism $H-h\rightarrow F$ that makes the inner diagram commute, and such that in particular $h\circ g=h_1\circ g=h_2\circ g$. Thus uniqueness entails that $h=h_1=h_2$, which proves the claim.
\end{itemize}
\end{proof}

The above result ensures in particular that \kl{adhesive categories} with mono-partial map classifiers yield suitable ``linear'' SqPO rewriting semantics (as already noted in~\cite{Corradini_2006}). For more general scenarios, we quote here the following result from the literature on restriction categories:
\begin{proposition}[\cite{COCKETT200361}, Prop.~4.16]
Let $\bfC$ be a category with a \kl{stable system of monics} $\cM$ and an \kl{$\cM$-partial map classifier} $(T,\eta)$. Then $T(m)\in \cM$ for all $m\in \cM$ iff $T(\eta_C)\in \cM$ for all objects $C$ in $\bfC$.
\end{proposition}
We suspect that the above result should provide a method for determining whether or not, e.g., comma categories constructed from categories with partial map classifiers that preserve the respective stable systems of monics will have the sought-after stability properties of monics under FPCs, as is the case for the example of \kl{directed simple graphs}, yet we leave further investigations to future work.

\subsubsection{Multi-initial pushout complements}\label{sec:exConstr:mIPCs}

While the standard literature on graph rewriting (cf.\ \cite{ehrig:2006fund}) provides some examples of explicit constructions of pushout complements, we require in general a construction of \kl{multi-initial pushout complements}, as in Definition~\ref{def:mipc}, for the ``non-linear'' variants of rewriting semantics. In this subsection, we demonstrate sufficient conditions under which  a construction for such mIPCs can be made explicit.

\begin{proposition}\label{prop:constrMIPC}
Let $\bfC$ be a category with a \kl{stable system of monics} $\cM$ that \kl(ssm){has pushouts along $\cM$-morphisms} and that possesses an \kl{$\cM$-partial map classifier} $(T,\eta)$.
Assume further that \kl(ssm){pushouts along $\cM$-morphisms} are pullbacks in $\bfC$. Then for any composable sequence of morphisms $A-f\rightarrow B\rtail m \rightarrow C$ with $m$ in $\cM$, the following construction provides $\mIPC{f}{m}$:
\begin{enumerate}
\item Take the \kl{FPC} $A-n\rightarrow F-g\rightarrow C$ of $(f,m)$ (cf.\ diagram $(i)$ in~\eqref{eq:mIPCconstruction} below).
\item For every factorization $A\rtail n'\rightarrow F' -n''\rightarrow F$ of $A\rtail n\rightarrow F$, where $n'$ is an $\cM$-morphism, take a pushout $F'-f'\rightarrow C'\leftarrow m' \ltail B$, which by the \kl{universal property of pushouts} entails that there exists a universal morphism $C'-m''\rightarrow C$. If $m''$ is an isomorphism, then $A\rtail n' \rightarrow F' - m''\circ f'\rightarrow C$ is an element of $\mIPC{f}{m}$.
\end{enumerate}
\end{proposition} 
\begin{equation}\label{eq:mIPCconstruction}
\ti{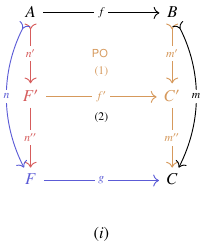}\qquad\qquad 
\ti{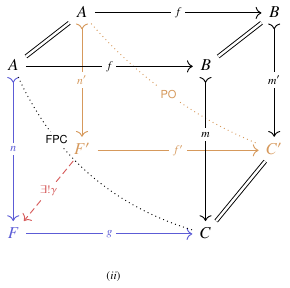}
\end{equation}
\begin{proof} Since $\bfC$ is assumed to possess an \kl{$\cM$-partial map classifier}, by Theorem~\ref{thm:FPC} this entails that $\bfC$ \kl(ssm){has FPCs along $\cM$-morphisms}. %
For any element of $\mIPC{f}{m}$, i.e., for a pushout square as the vertical back square in diagram $(ii)$ of~\eqref{eq:mIPCconstruction}, if $(n, g)$ is an FPC of $(f,m)$, then by the \kl{universal property of FPCs} there exists a unique morphism $F'-\gamma\rightarrow F$ that makes the diagram commute. Consequently, up to essential uniqueness, every pushout square that is an element of $\mIPC{f}{\beta}$ fits a diagram of shape $(i)$ in~\eqref{eq:mIPCconstruction}.
\qed\end{proof}

The above construction is typically assumed to be applied in situations where the underlying category is \kl{finitary} with respect to a \kl{stable system of monics} $\cM$. According to~\cite[Prop.~14.9(2)]{adamek2006}, a sufficient condition for a composition of morphisms $f\circ g$ that is an $\cM$-morphism to also satisfy that $g\in \cM$ is that $f$ is a monomorphism (not necessarily in $\cM$). In concrete examples (see below), it seems that all factorizations used in constructions of mIPCs in the finitary setting are of the aforementioned form, yet it remains unclear to us whether this is in fact the only possible situation, hence we defer a full investigation of this point to future work.

An example of an \kl{$\cM$-multi-IPC} construction both in \kl{$\mathbf{SGraph}$} and in \kl{$\mathbf{Graph}$} is given in the diagram below. Note that in \kl{$\mathbf{Graph}$}, the \kl{$\cM$-multi-IPC} does not contain the FPC contribution (since in \kl{$\mathbf{Graph}$} the pushout of the relevant span would yield to a graph with a multi-edge).
\begin{equation}\label{eq:FPCexamples}
\vcenter{\hbox{\includegraphics[scale=0.5]{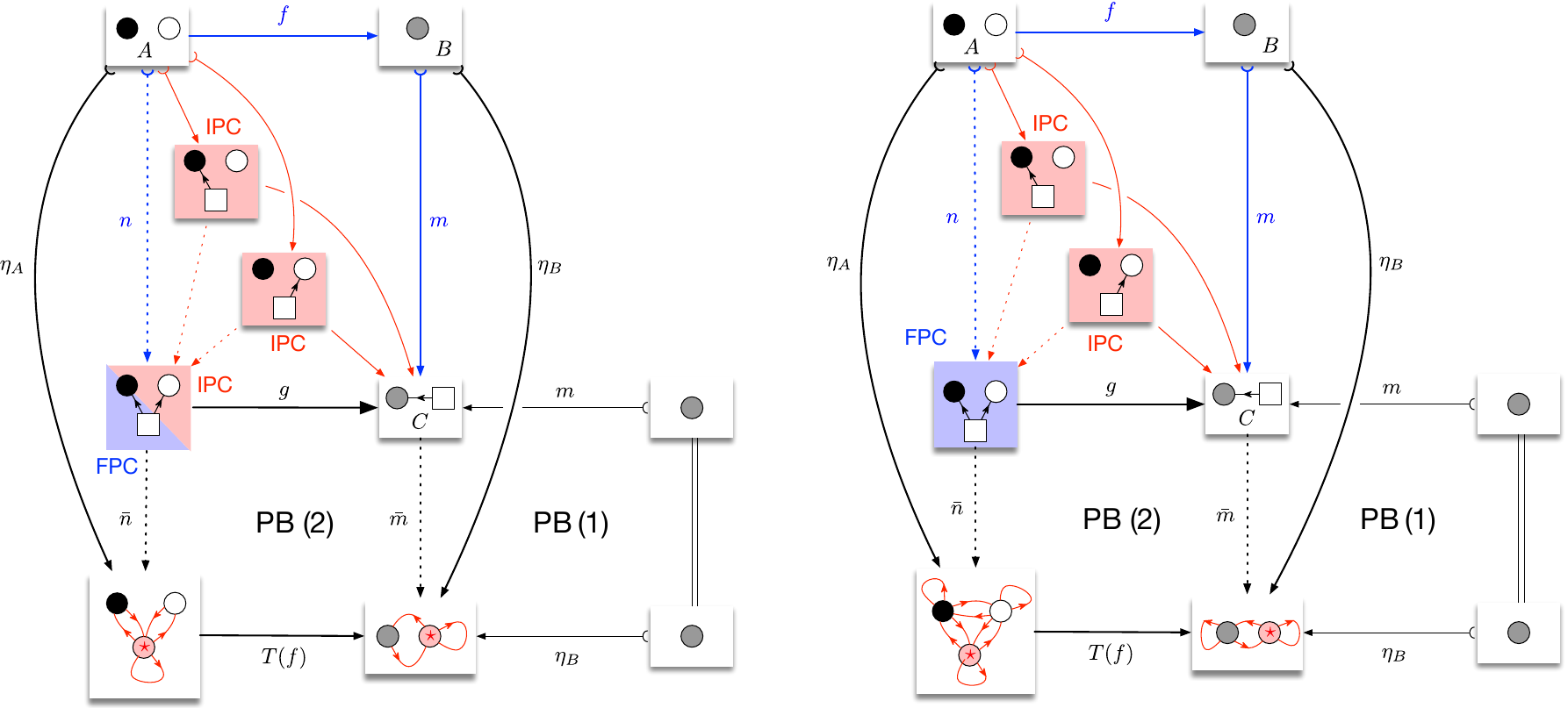}}}
\end{equation}

We conclude this subsection by mentioning the following useful result, which clarifies that under suitable assumptions on the underlying category, mIPCs along sequences of composable $\cM$-morphisms are essentially unique when they exist (in which case it is customary to speak of (``non-multi-'') \emph{pushout complements}, compare e.g.\ \cite[Fact~4.27.4]{ehrig:2006fund})]

\begin{proposition}\label{prop:POFPCs}
Let $\bfC$ be a category with a \kl{stable system of monics} $\cM$ and with an \kl{$\cM$-partial map classifier} $(T,\eta)$, such that \kl(ssm){pushouts along $\cM$-morphisms are stable under $\cM$-pullbacks}, such that \kl(ssm){pushouts along $\cM$-morphisms} are pullbacks, and such that \kl(ssm){$\cM$-morphisms are stable under FPCs}. Then if the \kl{$\cM$-multi-initial pushout complement} $\mIPC{f}{\beta}$ for a composable sequence of $\cM$-morphisms $A\rtail f\rightarrow B\rtail \beta \rightarrow B'$ (i.e., where both $f$ and $\beta$ are in $\cM$) is non-empty, then $\mIPC{f}{\beta}$ is essentially unique, and moreover every element of $\mIPC{f}{\beta}$ yields a square that is both a pushout and an FPC. 
\end{proposition}
\begin{proof}
It suffices to note that assumptions on $\bfC$ ensure that \kl{pullback-pushout decomposition} is applicable, and that according to Proposition~\ref{prop:constrMIPC}, every element of an mIPC fits into a diagram of the shape $(i)$ in~\eqref{eq:mIPCconstruction}, where $f$, $f'$ and $c$ are in $\cM$, and where $p$ is an isomorphism: since the outer square in~\eqref{eq:mIPCconstruction}(i) is an FPC and thus a pullback, the top square is a pushout, and the horizontal morphisms are in $\cM$, the bottom square is a pullback. Since isomorphisms are stable under pullback, $\gamma$ is an isomorphism, hence every pushout that is an element of $\mIPC{f}{\beta}$ is isomorphic (via the respective morphism $\gamma$) to the FPC square, demonstrating that if $\mIPC{f}{\beta}$ is non-empty, then every element yields both a pushout and an FPC.
\qed\end{proof}

\subsubsection{Final pullback complement pushout augmentations (FPAs)}\label{sec:exConstr:FPAs}

Let us now consider the final construction of interest, that of \kl{FPAs}, as given in Definition~\ref{def:FPA}. %
This naturally suggests the following explicit construction:
\begin{lemma}\label{lem:FPAconstr}
Let $\bfC$ be a category with a \kl{stable system of monics} $\cM$ that is a \kl{vertical weak adhesive HLR category} and \kl{finitary} with respect to $\cM$ (which entails by Theorem~\ref{thm:finitarityCats}(iii) that $\bfC$ is \kl{($\cE$,$\cM$)-structured}, for $\cE$ the class of \kl(ssm){extremal morphisms} w.r.t.\ $\cM$), such that \kl(ssm){pushouts along $\cM$-morphisms are stable under $\cM$-pullbacks} in $\bfC$, and such that \AP\intro(FPAconstr){pushouts along $\cM$-morphisms are pullbacks}. Then given a \kl(ssm){pushout along an $\cM$-morphism} such as the square in the diagram below, one may construct the \kl{FPAs} of this pushout as follows:
\begin{enumerate}
\item For every $\cE$-morphism $P-e\twoheadrightarrow E$ such that $e\circ \beta$ is in $\cM$, and such that $(\beta,id_B)$ is a pullback of $(e,e\circ \beta)$, take an FPC $(\varphi,g)$ of $(f,e\circ \beta)$. 
\item Since by assumption \kl(FPAconstr){pushouts along $\cM$-morphisms are pullbacks}, by \kl{pullback-pullback composition} $(\alpha,f)$ is a pullback of $(e\circ f',e\circ \beta)$, hence by the \kl{universal property of FPCs}, there exists a unique morphism $A'-\gamma\rightarrow F$ into the FPC object $F$. If $\gamma$ is in $\cM$, then $(\gamma, g,e)$ is in the \kl{FPA} of the pushout square.
\end{enumerate}
\begin{equation}\label{eq:FPAconstr}
\ti{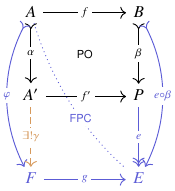}
\end{equation}
\end{lemma}

In the above definition, while we assumed that $\bfC$ is \kl{finitary} with respect to the \kl{stable system of monics} $\cM$, i.e., $\cM$-well-powered, in order for construction to be practicable, we strictly speaking also have to assume that $\bfC$ is \emph{$\cE$-well-copowered} (``co-finitary''). The latter entails that every object $X$ of $\bfC$ has only finitely many \kl{quotients}, where a \AP\intro{quotient} of $X$ is an isomorphism class of $\cE$-morphisms $X-e\twoheadrightarrow E$, with $X-e'\twoheadrightarrow E$ isomorphic to $e$ if there exists an isomorphism $E-\varepsilon \rightarrow E'$ such that $e'=\varepsilon \circ e$. One might wonder whether under the assumptions of the above definition, the assumptions on $\bfC$ (i.e., $\bfC$ being \kl{finitary}, \kl(ssm){having FPCs along $\cM$-morphisms}, being a \kl{vertical weak adhesive HLR category} and thus \kl{($\cE$,$\cM$)-structured}) might be sufficient in order for $\bfC$ to be $\cE$-well-copowered, albeit this does not seem to be a standard result in category theory  to the best of our knowledge (perhaps apart from the special case of Grothendieck topoi~\cite{johnstone2002sketches}, which are \kl{adhesive categories}~\cite{Lack2006} and thus in particular also \kl{vertical weak adhesive HLR categories}). 

\begin{example}\label{ex:FPA}
In order to illustrate the notion of \kl{FPAs} in their application to \kl(EMS){(auto-augmented, inert)-factorizations} of \kl{final pullback complements}, consider the example of an FPC in the category $\mathbf{Graph}$ of \kl{directed multigraphs} as shown in the diagram below left:
\begin{equation}
\vcenter{\hbox{\includegraphics[scale=0.18]{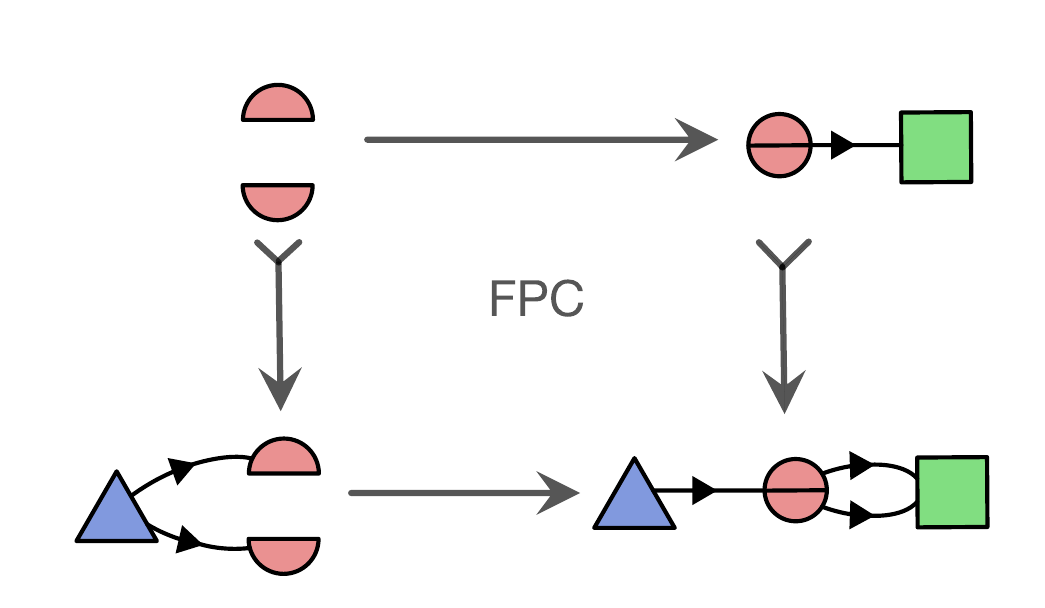}}} \xrightarrow{\text{take $\mathsf{PO}$}}
\vcenter{\hbox{\includegraphics[scale=0.18]{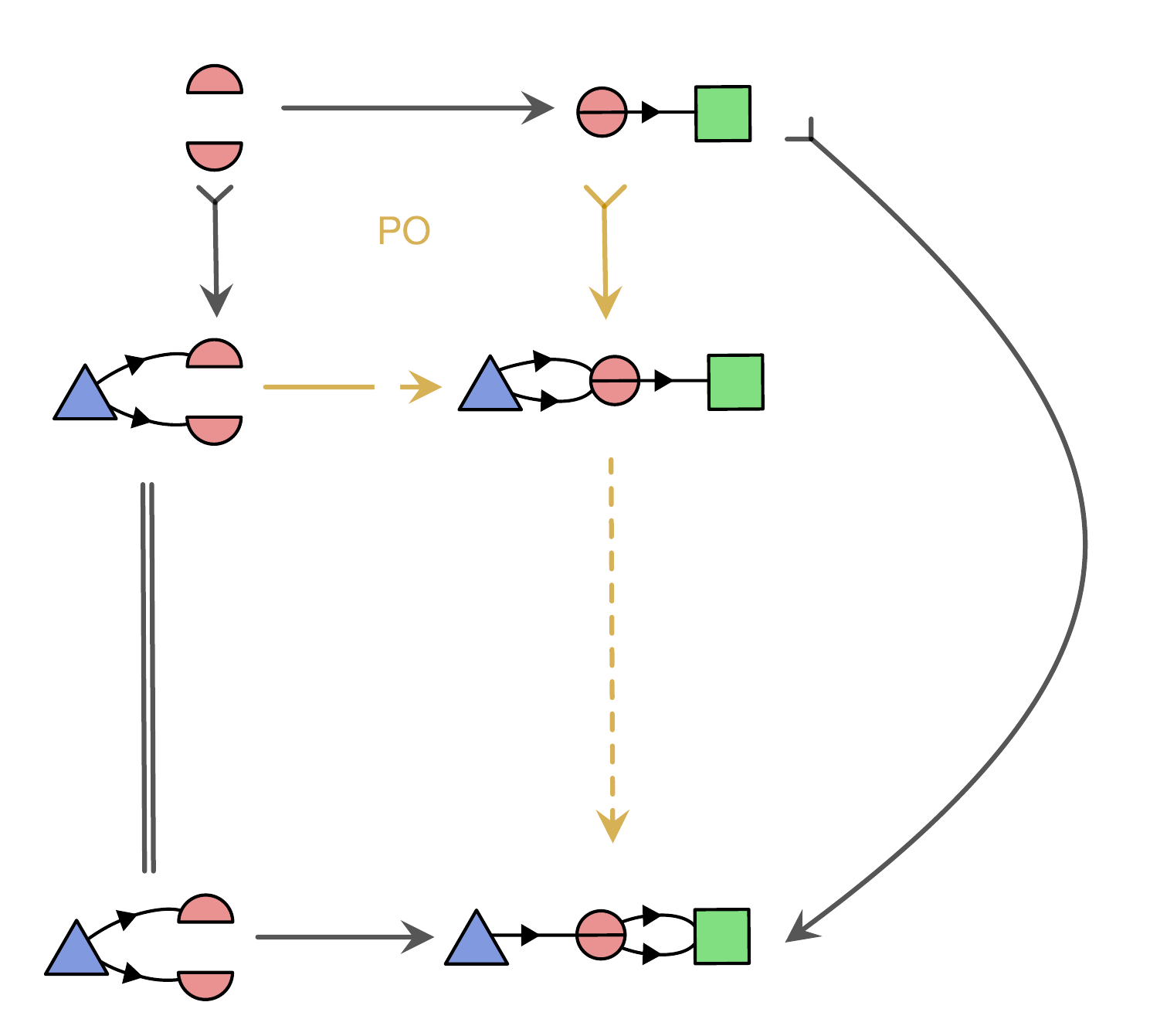}}} \xrightarrow[\text{take $\mathsf{PB}$}]{\text{epi-mono-fact.}}
\vcenter{\hbox{\includegraphics[scale=0.18]{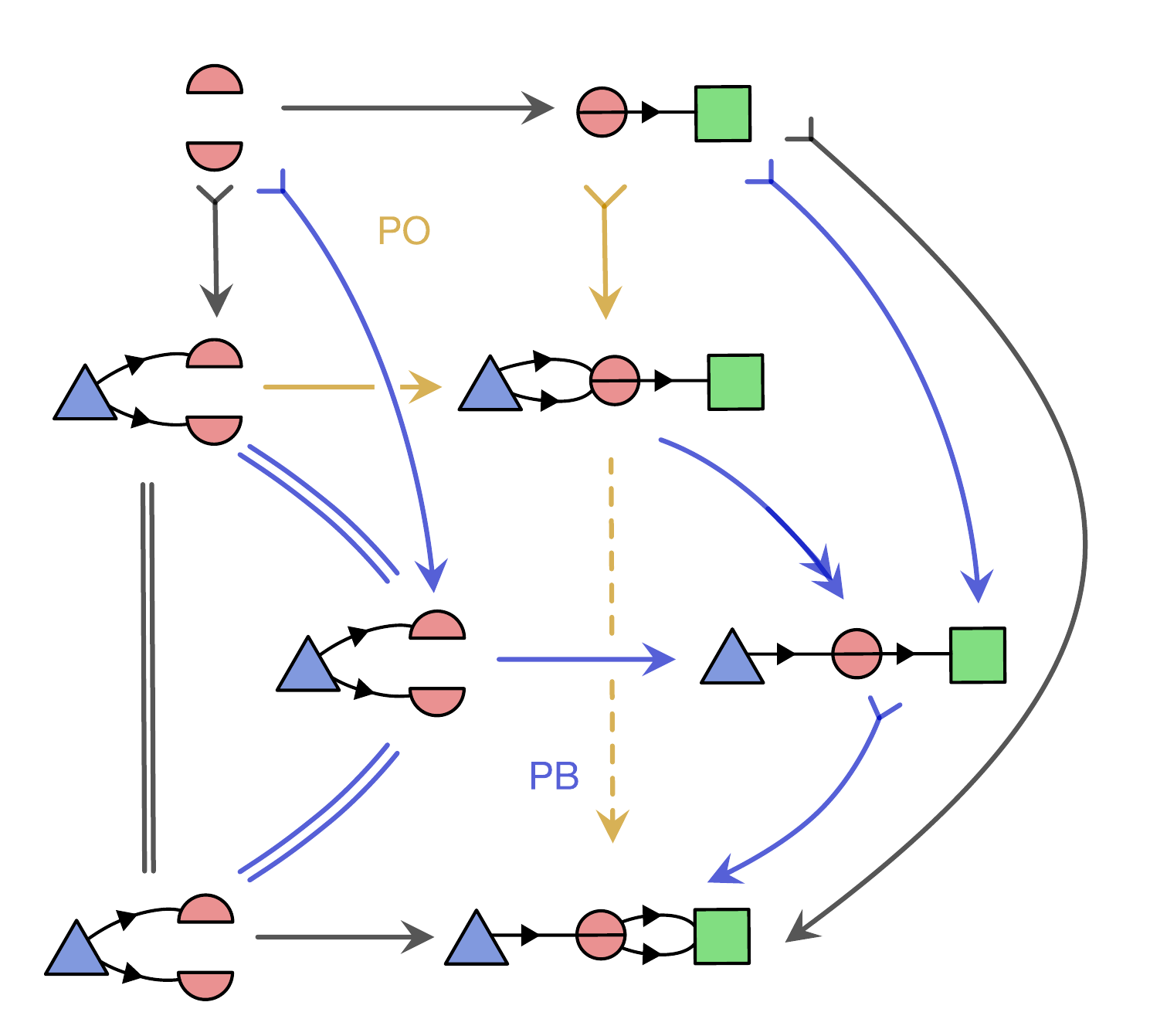}}} 
\end{equation}
In order to exhibit the FPC factorization, first take a pushout of the span of the FPC, which by the \kl{universal property of pushouts} yields the dashed arrow. Since $\mathbf{Graph}$ is \kl{(epi, mono)-structured}, we may take the \kl(EMS){epi-mono-factorization} of the dashed arrow, followed by taking a pullback along the monomorphism produced by the factorization. By \kl{vertical FPC-pullback decomposition}, the bottom square as well as the vertical composite of the middle and the top square are FPCs. One may verify that the aforementioned composite square is indeed an \kl(FPCfact){auto-augmented FPC}, while the bottom square is an \kl(FPCfact){inert FPC}, in the sense of Theorem~\ref{thm:FPCfact}. If we interpret the original FPC as a direct derivation under \kl{sesqui-pushout (SqPO) semantics} (i.e., along a rule of the form $r=(K\leftarrow id_K-K-i\rightarrow I)$), the \kl(FPCfact){auto-augmented FPC} encodes the minimal context for a SqPO-type direct derivation along the given rule to produce the graph at the bottom left of the original FPC square. On the other hand, the \kl(FPCfact){inert FPC} encodes a context extension that does not modify the outcome of the SqPO-type direct derivation, since it only adds edges to the context that are in effect implicitly deleted during the direct derivation.
\end{example}

\section{Examples of compositional rewriting semantics}\label{sec:rs}

In this section, we investigate a number of variants of \kl{DPO-semantics} and \kl{SqPO-semantics}, and illustrations thereof, for categorical rewriting, parametrized by the choice of the class of rules considered. The results of this section are summarized in Table~\ref{tab:main}, which makes explicit the conditions under which DPO and SqPO rewriting are compositional.

\begin{table}[h]
\centering
\renewcommand{\arraystretch}{1.2}
\begin{tabular}{C{2.5cm}|*{6}{C{1.5cm}|}C{1.5cm}}
 & \multicolumn{3}{c|}{\kl{Double-Pushout semantics}} & \multicolumn{4}{c}{\kl{Sesqui-Pushout semantics}} \\
Property & \kl(crsType){linear} & \kl(crsType){semi-linear} & \kl(crsType){generic} & 
\kl(crsType){linear} & \kl(crsType){output-linear} & 
\kl(crsType){input-linear} & \kl(crsType){generic}\\
\hline
\kl(crDC){$\bD_0$ has multi-sums} &
\multicolumn{7}{c}{$\YES$ (Lemma~\ref{lem:constrMS})}\\
\hline
\kl(dc){$\exists$ horizontal/vertical units} & \multicolumn{7}{c}{$\YES$ (Corollary~\ref{cor:CRShvUnits})}\\
\hline
\kl(dc){vertical composition} & \multicolumn{3}{c|}{$\YES$ (by \kl{pushout-pushout composition})} & \multicolumn{4}{c}{$\YES$ (by \kl{PO-PO-} and \kl{vertical FPC composition})}\\
\hline
\multirow{2}{2.5cm}{\centering\kl(dc){horizontal composition} (Proposition~\ref{prop:hCompCRDC})}
& \multicolumn{2}{C{3cm}|}{$\bfC$ \kl(ssm){has pullbacks along $\cM$-morphisms}} 
& $\bfC$ \kl{has pullbacks}
& \multicolumn{4}{c}{$\bfC$ \kl{has pullbacks}} \\
 & $\land$ \kl(notationVKa){(V-iii-a)} & 
 $\land$ \kl(notationVKa){(W-iii-a)} & 
 $\land$ \kl(notationVKa){(L-iii-a)} & 
 $\land$ \kl(notationVKa){(V-iii-a)} & 
 $\land$ \kl(notationVKa){(H-iii-a)} & 
 $\land$ \kl(notationVKa){(V-iii-a)} & 
 $\land$ \kl(notationVKa){(L-iii-a)}\\
\hline
\multirow{2}{2.5cm}{\centering\kl(crDC){horizontal decomposition} (Proposition~\ref{prop:hDeCompCRDC})} &
\multicolumn{3}{c|}{$\bfC$ \kl(ssm){has pushouts along $\cM$-morphisms}} & 
\multicolumn{4}{C{6cm}}{$\bfC$ \kl{has pullbacks} $\land$ \kl(ssm){has pushouts and FPCs along $\cM$-morphisms}}\\
&
\multicolumn{2}{C{3cm}|}{$\land$ \kl(ssm){has pullbacks along $\cM$-morphisms} $\land$ \kl(notationVKb){(V-iii-b)}} & 
$\land$ \kl{has pullbacks} $\land$ \kl(notationVKb){(L-iii-b)} &
$\land$ \kl(notationVKa){(V-iii-a)} $\lor$ \kl(notationVKa){(H-iii-a)} & 
$\land$ \kl(notationVKa){(H-iii-a)} & 
$\land$ \kl(notationVKa){(V-iii-a)} & 
$\land$ \kl(notationVKa){(L-iii-a)}\\
\hline
\kl(crDC){$\mathbb{D}_1$ has pullbacks} (Proposition~\ref{prop:D1PBcrDC})& 
\multicolumn{7}{c}{\kl(notationVKa){(V-iii-a)}} \\
\hline
\kl(crDC){$S$ is a multi-opfibration} (Theorem~\ref{thm:SmofTrmof}) &
\multicolumn{3}{C{4.5cm}|}{$\bfC$ is a \kl{vertical weak adhesive HLR category}} 
& \multicolumn{4}{C{6cm}}{$\bfC$ is \kl{vertical weak adhesive HLR} and \kl(ssm){has FPCs along $\cM$-morphisms}}\\
\hline
\kl(crDC){$T$ is a residual multi-opfibration} (Theorem~\ref{thm:SmofTrmof}) &
\multicolumn{3}{C{4.5cm}|}{$\bfC$ is a \kl{vertical weak adhesive HLR category}} 
& \multicolumn{4}{C{6cm}}{$\bfC$ is \kl{vertical weak adhesive HLR} and \kl(ssm){has FPCs along $\cM$-morphisms}}\\
\end{tabular}
\renewcommand{\arraystretch}{1}
\caption{\label{tab:main}Requirements on the underlying category for giving rise to compositional rewriting semantics of the various kinds. For all cases, we minimally assume that $\bfC$ has a \kl{stable system of monics}, %
with respect to which $\bfC$ is \kl{finitary}, %
with respect to which the variants of \kl{adhesivity properties} are required to hold, %
and such that \kl(crDC){$\bD_0:=\bfC\vert_{\cM}$ has pullbacks}. %
The latter is equivalent to requiring that $\bfC$ \AP\intro(ssm){has pullbacks of cospans of $\cM$-morphisms}, which is true for all of the listed \kl{adhesivity properties}. %
We moreover use the abbreviation \AP\intro(adhVK){(W-iii)} to denote \kl(adhVK){(V-iii)} $\land$ \kl(adhVK){(H-iii)}.}
\end{table}

\subsection{DPO and SqPO semantics}

In much of the traditional work on graph- and categorical rewriting theories~\cite{ehrig:2006fund}, while it was appreciated early in its development that SqPO-rewriting permits the \emph{cloning of subgraphs}~\cite{Corradini_2006}, and that both SqPO- and DPO-semantics permit the \emph{fusion of subgraphs} (i.e., via \kl(crsType){input-linear}, but output-non-linear rules), the non-uniqueness of pushout complements along non-monic morphisms for the DPO- and the lack of a \kl{concurrency theorem} in the SqPO-case in general has prohibited a detailed development of non-linear rewriting theories to date. Interestingly, the SqPO-type concurrency theorem for \kl(crsType){linear} rules as developed in~\cite{nbSqPO2019} exhibits the same obstacle for the generalization to non-linear rewriting as the DPO-type concurrency theorem, i.e., the non-uniqueness of certain pushout complements. 

Our proof for non-linear rules identifies in addition a new and highly non-trivial  ``back-propagation effect'', which will be highlighted in Section~\ref{sec:nlCRS} (cf.\ also Example~\ref{ex:SqPOexplicit} for an in-detail heuristic discussion of this effect). It may be worthwhile emphasizing that there exists previous work that aimed at circumventing some of the technical obstacles of non-linear rewriting either via specializing the semantics e.g.\ from double pushout to a version based upon so-called \emph{minimal} pushout complements~\cite{BRAATZ2011246}, or from sesqui-pushout to \emph{reversible} SqPO-semantics~\cite{reversibleSqPO,harmer2020reversibility} or other variants such as AGREE-rewriting~\cite{Corradini_2015}. In contrast, we will in the following introduce the ``true'' extensions of both SqPO- and DPO-rewriting to the non-linear setting, with our constructions based upon \kl{multi-sums}, \kl{multi-IPCs} and \kl{FPAs}.

We focus here on the following eight variants of categorical rewriting semantics:
\begin{definition}\label{def:crs}
Let $\bfC$ be a category with a \kl{stable system of monics} $\cM$. 
\begin{enumerate}[label=(\roman*)]
\item A \AP\intro(crs){rule}, denoted $O\leftharpoonup r-I$, is a span $r= (O\leftarrow o_r-K_r-i_r\rightarrow I)$ in $\bfC$. We refer to a rule as 
\begin{itemize}
\item \AP\intro(crs){output-linear} if $o_r$ is in $\cM$, 
\item \AP\intro(crs){input-linear} if $i_r$ is in $\cM$, and
\item  \AP\intro(crs){linear} if both $o_r$ and $i_r$ are in $\cM$.
\end{itemize}
We will also refer to arbitrary spans as \AP\intro(crs){generic} rules.
\item In \AP\intro{Double-Pushout (DPO) semantics}, a \AP\intro(crsDPO){direct derivation} is defined as a commutative diagram as in~\eqref{eq:defCRS} below, where the vertical morphisms are in $\cM$, and where the square marked $(\dag_{\alpha})$ is a pushout, while the square marked $(*_{\alpha})$ is an element of an \kl{$\cM$-multi-IPC} (and thus in particular also a pushout). A category $\bfC$ is thus \AP\intro(crsDPO){suitable for DPO-semantics} if it \kl{has multi-initial pushout complements (mIPCs) along $\cM $-morphisms}, if it \kl(ssm){has pushouts along $\cM$-morphisms}, and if \kl(ssm){$\cM$-morphisms are stable under pushout}.
\item In \AP\intro{Sesqui-Pushout (SqPO) semantics}, a \AP\intro(crsSqPO){direct derivation} is defined as a commutative diagram as in~\eqref{eq:defCRS} below, where the vertical morphisms are in $\cM$, and where the square marked $(\dag_{\alpha})$ is a pushout, while the square marked $(*_{\alpha})$ is a \kl{final pullback complement (FPC)}. A category $\bfC$ is thus \AP\intro(crsSqPO){suitable for SqPO-semantics} if it \kl(ssm){has FPCs along $\cM$-morphisms}, if it \kl(ssm){has pushouts along $\cM$-morphisms}, and if \kl(ssm){$\cM$-morphisms are stable under pushout}.
\end{enumerate}
\begin{equation}\label{eq:defCRS}\AP\phantomintro(crs){direct derivation}
\ti{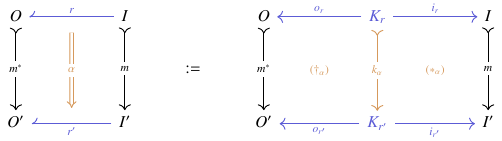}
\end{equation}
It is conventional to refer to the $\cM$-morphisms $m$ and $m^{*}$ as \AP\intro(crs){match} and \AP\intro(crs){co-match}, respectively. Finally, we will refer to either of the two semantics as \AP\intro(crsType){generic} if no special restrictions are imposed upon the underlying \kl(crs){rules}, and as \AP\intro(crsType){linear}, \AP\intro(crsType){output-linear} or \AP\intro(crsType){input-linear} if rules are restricted to being  \kl(crsType){linear}, \kl(crsType){output-linear} or \kl(crsType){input-linear}, respectively.  We will sometimes also use the term \AP\intro(crsType){semi-linear} as an abbreviation for ``\kl(crsType){output-linear} or \kl(crsType){input-linear}''.
\end{definition}
As discussed in further detail in Section~\ref{sec:intro:bg}, each of these eight types of semantics permits a different set of features, e.g., for the rewriting of \kl{directed multigraphs}, where the type of linearity of the rules entails whether or not fusing or cloning of subgraphs are possible, and where the choice of SqPO- versus DPO-semantics yields a difference also in whether or not edges may be implicitly deleted (in addition to modifying the precise type of cloning semantics for the non-input-linear variants of the semantics). %
It should also be noted that evidently there are many more kinds of categorical rewriting semantics available in the literature, including cases where the \kl(crs){matches} and \kl(crs){co-matches} are not required to be $\cM$-morphisms, yet we focus here on the aforementioned eight variants for concreteness as a sufficiently diverse set of test cases for our new theoretical framework for categorical rewriting theories.

We will now follow the proof-strategy set forth via the formalism introduced in Section~\ref{sec:comp} in order to determine efficiently sets of sufficient conditions under which the eight different semantics of Definition~\ref{def:crs} give rise to \kl{compositional rewriting double categories (crDCs)}, and thus to compositional rewriting theories. In this way, we are able to demonstrate the high level of modularization afforded by our novel approach, and at the same time highlight some of the similarities and crucial mathematical differences between the various rewriting semantics.

\subsection{Double-categorical structures}\label{sec:crs:dc}

For all eight semantics of Definition~\ref{def:crs}, we will let $\bD_0$ be defined as $\bfC\vert_{\cM}$, i.e., the restriction of $\bfC$ along $\cM$ (with objects the objects of $\bfC$, and morphisms the morphisms of $\cM$). We let $\bD_1$ be defined as \kl(crs){rules} for the \kl(dc){horizontal morphisms} (i.e., the objects of $\bD_1$), and via \kl(crs){direct derivations} (i.e., diagrams of the form in~\eqref{eq:defCRS}) for the \kl(dc){squares} (i.e., the morphisms of $\bD_1$). This  identifies the crDCs we will construct as double categories obtained via restriction of the double category $\mathbb{S}pan(\bfC)$ of spans (cf.\ e.g.\ \cite[Ex.12.3.16]{Johnson2021}), with $\bD_0=\bfC$, with spans of $\bfC$ as \kl(dc){horizontal morphisms} (i.e., as objects of $\bD_1$), and with commutative diagrams of the form below (without any restrictions on the squares other than commutativity) for the \kl(dc){squares} of $\mathbb{S}pan(\bfC)$.
\begin{equation}
\ti{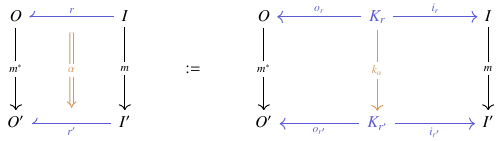}
\end{equation}
Rather than having to work through a lengthy set of coherence conditions for our \kl{crDCs} to indeed qualify as double categories, the fact that they are all obtained as restrictions of $\mathbb{S}pan(\bfC)$ simplifies this task down to verifying the following properties, which ensure that the restrictions are compatible with the existence of horizontal and vertical units, and with horizontal and vertical compositions:

\begin{corollary}\label{cor:CRShvUnits}
For all eight semantics of Definition~\ref{def:crs}, the resulting definitions of $\bD_0$ and $\bD_1$ have \kl(dc){horizontal and vertical units} in the following form:
\begin{equation}\label{eq:crDCunitsCRS}
\ti{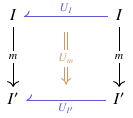} := 
\ti{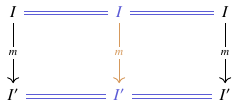} \;,\qquad 
\ti{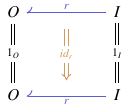} := 
\ti{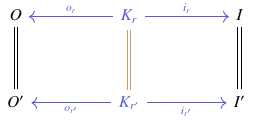}
\end{equation}
\end{corollary}
\begin{proof}
The only non-trivial statement to prove is that the diagrams in~\eqref{eq:crDCunitsCRS} qualify as direct derivations of the respective types in a given semantics according to Definition~\ref{def:crs}. But this follows immediately from the results of Lemma~\ref{lem:PBPOFPCcats}, whereby all commutative squares of the types occurring in the direct derivations depicted in~\eqref{eq:crDCunitsCRS} are simultaneously pushouts and final pullback complements. Moreover, since by assumption $\cM$ is a \kl{stable system of monics}, it contains in particular all isomorphisms, which completes the proof that the direct derivations in~~\eqref{eq:crDCunitsCRS} are well-formed.
\end{proof}

As mentioned in Table~\ref{tab:main}, \kl(dc){vertical composition} is guaranteed to be well-posed in all cases because of \kl{pushout composition} and \kl{vertical FPC composition}. In contrast, it is considerably more intricate to prove that \kl(dc){horizontal composition} is well-posed, which is the first instance where \kl{adhesivity properties} are required in different forms depending on the precise nature of the chosen rewriting semantics:
\begin{proposition}\label{prop:hCompCRDC}
Under the additional assumptions on $\bfC$ presented in Table~\ref{tab:main}, each of the rewriting semantics of Definition~\ref{def:crs} yields a well-posed \kl(dc){horizontal composition} for direct derivations.
\end{proposition}
\begin{proof}
Since the \kl{crDCs} for the different semantics are obtained via suitable restrictions of the double category $\mathbb{S}pan(\bfC)$, the operation $\hComp$ is for all situations induced via span composition of the horizontal morphisms:
\begin{subequations}
\begin{align}
\ti{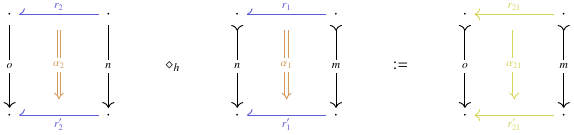}\label{eq:def:hCompCRSa}\\
\ti{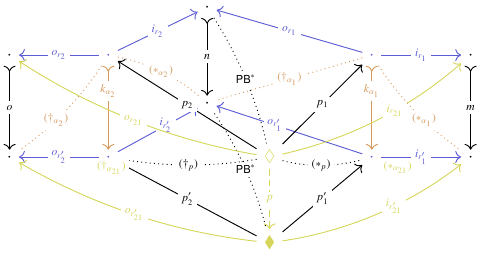}\label{eq:def:hCompCRSb}
\end{align}
\end{subequations}
In the commutative diagram in~\eqref{eq:def:hCompCRSb}, the notation $\mathsf{PB}^{*}$ indicates that the definition of $\hComp$ involves \emph{choices} of pullbacks for each cospan, so that $\hComp$  will in particular only be a pesudo-functor.\\

With regards to prerequisites on the underlying category, the definition of $\hComp$ via taking pullbacks requires that the underlying category \kl{has pullbacks} for the case of \kl(crsType){generic} rules, while for all variants of linearity it suffices that $\bfC$ \kl(ssm){has pullbacks along $\cM$-morphisms}.  Moreover, since by assumption $\cM$ is a \kl{stable system of monics}, \kl(ssm){$\cM$-morphisms are stable under pullback}, hence the types of the composite spans are indeed compatible with the types specified in \kl(crsType){generic}, \kl(crsType){output-linear}, \kl(crsType){input-linear} or \kl(crsType){linear} rewriting semantics, respectively.\\

Next, the \kl{universal property of pullbacks} entails the existence of a unique morphism $\lozenge-p\rightarrow \blacklozenge$ that makes the diagram in~\eqref{eq:def:hCompCRSb} commute. By \kl{pullback-pullback decomposition}, the front left and right vertical squares marked $(\dag_p)$ and $(*_p)$ in~\eqref{eq:def:hCompCRSb} are pullbacks.\\

It remains to demonstrate that the squares marked $(\dag_p)$ and $(*_p)$ in~\eqref{eq:def:hCompCRSb} are not only pullbacks, but indeed of the correct type (i.e., pushouts or FPCs, respectively) as required for the chosen rewriting semantics.
\begin{enumerate}[label=(\roman*)]
\item For the square marked $(\dag_p)$, since in all eight types of semantics according to Definition~\ref{def:crs} the square marked $(\dag_{\alpha_1})$ is a pushout, we require the appropriate notion of stability of this type of pushout under pullbacks (compare Table~\ref{tab:main}). More precisely, the distinction depends on the character of the horizontal morphisms in the pushout square $(\dag_{\alpha_1})$, and in the nature of the morphisms in the pullback squares over $(\dag_{\alpha_1})$ (i.e., $i_{r_2}$, $i_{r_2'}$, $p_1$ and $p_1'$), which depending on the rewriting semantics are either generic morphisms or $\cM$-morphisms:
\begin{itemize}
\item For \kl(crsType){generic} semantics, $(\dag_p)$ is a \kl(ssm){pushout along an $\cM$-morphism}, and the morphisms in the pullback squares over $(\dag_{\alpha_1})$ are generic morphisms, hence we require $\bfC$ to satisfy that \kl(ssm){pushouts along $\cM$-morphisms} are \kl(PO){stable under pullbacks} (i.e., axiom \kl(notationVKa){(L-iii-a)} of the definition of \kl{adhesive HLR categories}).
\item For \kl(crsType){output-linear} semantics, $(\dag_p)$ is a pushout of a span of $\cM$-morphisms, and the morphisms in the pullback squares over $(\dag_{\alpha_1})$ are generic morphisms, hence we require $\bfC$ to satisfy axiom \kl(notationVKa){(H-iii-a)} of the definition of \kl{horizontal weak adhesive HLR categories}.
\item For \kl(crsType){input-linear} semantics, $(\dag_p)$ is a \kl(ssm){pushout along an $\cM$-morphism}, and the morphisms in the pullback squares over $(\dag_{\alpha_1})$ are $\cM$-morphisms, hence we require $\bfC$ to satisfy axiom \kl(notationVKa){(V-iii-a)} of the definition of \kl{vertical weak adhesive HLR categories}.
\item For \kl(crsType){linear} semantics, $(\dag_p)$ is a pushout of a span of $\cM$-morphisms, and the morphisms in the pullback squares over $(\dag_{\alpha_1})$ are $\cM$-morphisms, hence we require $\bfC$ to satisfy axiom \kl(notationVKa){(V-iii-a)} of the definition of \kl{vertical weak adhesive HLR categories}.\footnote{Coincidentally, it would also be sufficient for $\bfC$ to satisfy axiom \kl(notationVKa){(H-iii-a)} of the definition of \kl{horizontal weak adhesive HLR categories}; however, it will become evident in the following that axiom \kl(notationVKa){(V-iii-a)} is in fact required for other properties of \kl{crDCs} to be satisfied (cf.\ Table~\ref{tab:main}).} 
\end{itemize}
\item For the square marked $(*_{\alpha_2})$ in~\eqref{eq:def:hCompCRSb}, since this square is a \kl{final pullback complement (FPC)} for \kl{sesqui-pushout semantics}, and under the condition that $\bfC$ \kl{has pullbacks} (which is also one of the necessary assumptions for the construction of \kl{FPCs} via an \kl{$\cM$-partial map classifier}, cf.\ Theorem~\ref{thm:FPC}), we obtain that the square marked $(*_p)$ is an FPC by \kl{stability of FPCs under pullbacks}. For \kl{double-pushout semantics}, we may repeat the analysis of the previous step (i) to demonstrate that since the square marked $(*_{\alpha_2})$ is a pushout for this semantics, under suitable conditions on $\bfC$ the square marked $(*_p)$ is a pushout as well. In particular, we find for the \kl(crsType){output-linear} and \kl(crsType){input-linear} variants of \kl{DPO-semantics} that $\bfC$ has to satisfy axiom \kl(notationVKa){(W-iii-a)} of the definition of \kl{weak adhesive HLR categories}, i.e., both of axioms \kl(notationVKa){(V-iii-a)} and \kl(notationVKa){(H-iii-a)}.
\end{enumerate}
Finally, by \kl{pushout composition} and \kl{horizontal FPC composition}, respectively, one may demonstrate that the horizontal composite of $(\dag_{\alpha_2})$ and $(\dag_p)$ is a pushout, while $(*_p)$ and $(*_{\alpha_1})$ compose into a pushout for \kl{DPO-semantics}, and into an FPC for \kl{SqPO-semantics}, which concludes the proof.
\qed\end{proof}

\begin{remark}
In earlier work on the \kl(crsType){linear} variant of \kl{sesqui-pushout semantics}~\cite{nbSqPO2019,behrRaSiR}, instead of requiring that $\bfC$ \kl{has pullbacks}, an alternative argument was utilized in order to prove that square $(*_p)$ in~\eqref{eq:def:hCompCRSb} is an \kl{FPC}: after completing step (i) in order to prove that $(\dag_p)$ is a pushout as above, and utilizing that pushouts of $\cM$-spans are also FPCs (compare Proposition~\ref{prop:POFPCs}), the pushout squares $(\dag_{\alpha_1})$ and $(\dag_p)$ are FPCs. Thus by \kl{horizontal FPC composition} the composite of squares $(*_{\alpha_2})$ and  $(\dag_p)$ is an FPC, hence applying \kl{horizontal FPC decomposition}, one may demonstrate that $(*_p)$ is an FPC. However, there exists to the best of our knowledge no example of a category that \kl(ssm){has FPCs along $\cM$-morphisms} where FPCs are not constructed via an \kl{$\cM$-partial map classifier} as in Theorem~\ref{thm:FPC}, and since the latter theorem requires that the category \kl{has pullbacks}, it appears to be more efficient to apply \kl{stability of FPCs under pullbacks} in order to prove that $(*_p)$ is an FPC. 
\end{remark}

\subsection{Properties specific to compositional rewriting double categories}

Having established the conditions on the underlying category under which direct derivations of one of the eight semantics of Definition~\ref{def:crs} give rise to a double category, it remains to determine whether additional conditions are required such that these double categories indeed qualify as \kl{compositional rewriting double categories}. The results of this part of the derivation are summarized in Table~\ref{tab:main}.

\begin{proposition}\label{prop:hDeCompCRDC}
Let $\bfC$ be a category suitable for one of the rewriting semantics of Definition~\ref{def:crs}, and such that $\bfC$ also satisfies the relevant additional assumptions stated in Table~\ref{tab:main}. Then the \kl(dc){horizontal composition} functor $\hComp$ of the  \kl{crDC} for the given choice of $\bfC$ and rewriting semantics is a \kl(crDC){isoglobular residual opfibration}.
\end{proposition}
\begin{proof}
Let us first verify that, under the given assumptions, $\hComp$ possesses \kl(rof){residual op-Cartesian lifts} whose \kl(rmof){residues} are \kl(dc){globular isomorphisms}. %
Consider thus a diagram of the form below, where $\mathsf{PB}^{*}$ denotes a pullback as chosen in the definition of $\hComp$,  $(\dag_{\alpha_{21}})$ is a pushout, and $(*_{\alpha_{21}})$ is a pushout for \kl{DPO-semantics} and an \kl{FPC} for \kl{SqPO-semantics}:
\begin{equation}\label{eq:proof:hDeCompCRDCa}
\ti{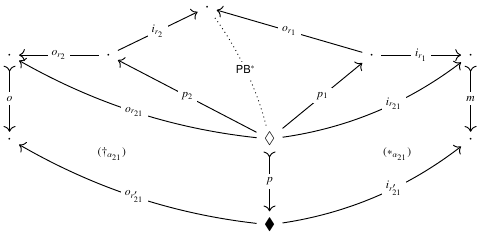}
\end{equation}
We have to prove that for each of the semantics of Definition~\ref{def:crs}, one may obtain essentially uniquely a horizontal composition of direct derivations. To this end, consider first the case of \kl{DPO-semantics}, for which we transform the diagram of~\eqref{eq:proof:hDeCompCRDCa} into the diagram below:
\begin{equation}\label{eq:proof:hDeCompCRDCb}
\ti{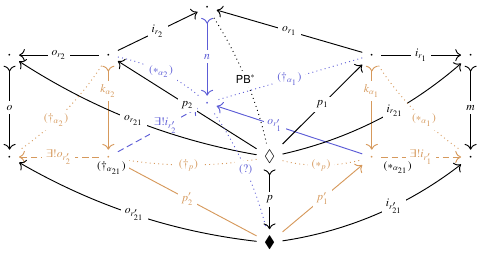}
\end{equation}
\begin{itemize}
\item Take pushouts to obtain the squares marked $(\dag_p)$ and $(*_p)$, which by the \kl{universal property of pushouts} entails that there exist unique morphisms $o_{r_2'}$ and $i_{r_1'}$. Moreover, by \kl{pushout-pushout decomposition}, the squares marked $(\dag_{\alpha_2})$ and $(*_{\alpha_1})$ are pushouts.
\item Take another pushout to obtain the square\footnote{Evidently, we could have equivalently obtained the square marked $(*_{\alpha_2})$ first by taking a pushout.} marked $(\dag_{\alpha_1})$, which by the \kl{universal property of pushouts} entails that there exists a unique morphism $i_{r_2'}$. Moreover, by \kl{pushout-pushout decomposition}, the square marked $(*_{\alpha_2})$ is a pushout.
\item It then remains to invoke the version of the van Kampen square property applicable to the given variant of \kl{DPO-semantics} (i.e., axiom  \kl(notationVKb){(L-iii-b)}  for the \kl(crsType){generic} and axiom \kl(notationVKb){(V-iii-b)} for the other variants, cf.\ Table~\ref{tab:main}) in order to demonstrate that the bottom square marked $(?)$ is indeed a pullback.
\end{itemize}

Finally, the pullback square marked $(?)$ will in general \emph{not} coincide with the pullback of the cospan chosen as part of the definition of $\hComp$; therefore, it remains to form the diagram below (where $\mathsf{PB}^{*}$ marks the chosen pullback):
\begin{equation}\label{eq:proof:hCompDPOchosenPB}
\ti{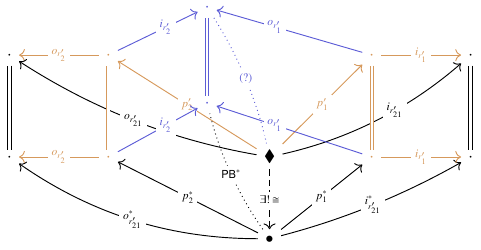}
\end{equation}

\begin{itemize}
\item By the \kl{universal property of pullbacks}, there exists a unique mediating isomorphism between the two pullback squares denoted $(?)$ and $\mathsf{PB}^{*}$.
\item Since each of the vertical squares with isomorphisms for vertical morphisms in~\eqref{eq:proof:hCompDPOchosenPB} is a pushout square (and also an FPC square, cf.\ proof of Lemma~\ref{lem:PBPOFPCcats}), we conclude that the morphism from the span $(o_{r_{21}'},i_{r_{21}'})$ to $(o_{r_{21}'}^{*},i_{r_{21}'}^{*})$ (i.e., the frontmost curved vertical squares) is indeed a \kl(dc){globular isomorphism} in the \kl{crDC} for the chosen DPO-type semantics.
\end{itemize}

For the case of \kl{SqPO-semantics}, we transform the diagram in~\eqref{eq:proof:hDeCompCRDCa} as follows (where once again $\mathsf{PB}^{*}$ marks a pullback as chosen in the definition of $\hComp$):
\begin{equation}\label{eq:proof:hDeCompCRDCc}
\ti{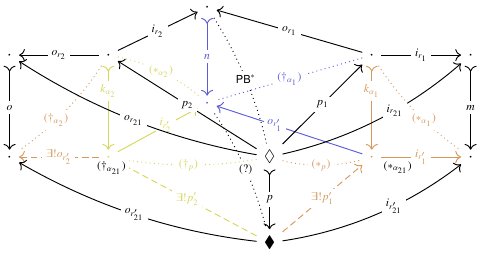}
\end{equation}
\begin{itemize}
\item Take an \kl{FPC} to obtain the square marked $(*_{\alpha_1})$, which by the \kl{universal property of FPCs} entails the existence of a unique morphism $p_1'$, and thus by \kl{horizontal FPC decomposition} that the square marked $(*_p)$ is an FPC.
\item Take a pushout to obtain the square marked $(\dag_{\alpha_1})$, and an FPC to obtain the square marked $(*_{\alpha_2})$, which by the \kl{universal property of FPCs} yields also a unique morphism $p_2'$, and thus by \kl{pullback-pullback decomposition}, the square marked $(\dag_p)$ is a pullback.
\item For the case of \kl(crsType){generic} \kl{SqPO-semantics}, by invoking the Beck-Chevalley-Condition \kl(cmtSq){(BCC-1)} of Theorem~\ref{thm:domFunctorFP}, which allows us to conclude that the square marked $(\dag_p)$ is an FPC, and the bottom square marked $(?)$ is a pullback. It is then related to the chosen pullback $\mathsf{PB}^{*}$ according to the definition of $\hComp$ by a universal isomorphism (i.e., by a span isomorphism).
\item For the other types of \kl{SqPO-semantics}, we may develop more general variants of the Beck-Chevalley-Condition \kl(cmtSq){(BCC-1)} by suitably adapting the proof strategy of Theorem~\ref{thm:domFunctorFP}. To this end, consider the diagrammatic statement presented in~\eqref{eq:proof:BCC-1variant} below (which is a 3D-rotated and relabeled version of the statement in~\eqref{eq:proof:BCC-1} in order to facilitate the comparison to the diagram in~\eqref{eq:proof:hDeCompCRDCc}), where $\mathsf{PB}^{*}$ marks a chosen pullback according to the definition of $\hComp$. In all three cases, the proof strategy consists in (i) taking a chosen pullback (marked $\mathsf{PB}^{*}$) to obtain the second diagram in~\eqref{eq:proof:BCC-1variant} (where by the \kl{universal property of pullbacks} entails that there exist unique arrows $\lozenge -q\rightarrow \bullet$ and $\blacklozenge-q'\rightarrow \bullet$); (ii) using \kl{pullback-pullback decomposition} to prove that all squares of the interior commutative cube are pullbacks; (iii) invoking a suitable variant of stability of pushouts under pullbacks to show that the front left inner vertical square is a pushout; and finally (iv) to apply \kl(FPC){stability of FPCs under pullbacks} in order to demonstrate that the front right inner vertical square is an FPC, such that by the \kl{universal property of FPCs} the morphism $\lozenge -q\rightarrow \bullet$ is an isomorphism. It thus remains to clarify the variant of stability property of pushouts necessary for each kind of semantics:
\begin{itemize}
\item For \kl(crsType){output-linear} \kl{SqPO-semantics}, all morphisms of the back right vertical square are guaranteed to be $\cM$-morphisms, hence the claim follows if $\bfC$ satisfies axiom \kl(notationVKa){(H-iii-a)}.
\item For \kl(crsType){input-linear} \kl{SqPO-semantics}, the morphisms $i_{r_2}$, $i_{r_2'}$, $p_1$, and (by \kl(ssm){stability of $\cM$-morphisms under pullback}) $p_1''$ are guaranteed to be $\cM$-morphisms, hence the claim follows if $\bfC$ satisfies axiom \kl(notationVKa){(V-iii-a)}.
\item Since \kl(crsType){linear} \kl{SqPO-semantics} is a special case both of \kl(crsType){output-linear} and \kl(crsType){input-linear} \kl{SqPO-semantics}, the claim follows if $\bfC$ satisfies either \kl(notationVKa){(H-iii-a)} or \kl(notationVKa){(V-iii-a)}.
\end{itemize}
\end{itemize}
\begin{equation}\label{eq:proof:BCC-1variant}
\ti{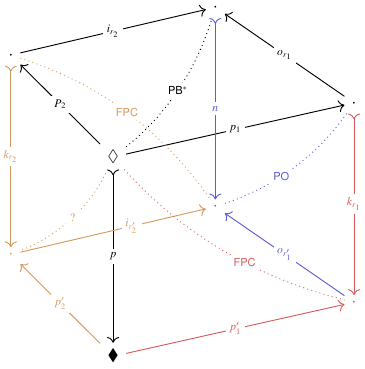} \xrightarrow{\text{take $\mathsf{PB}^{*}$}}
\ti{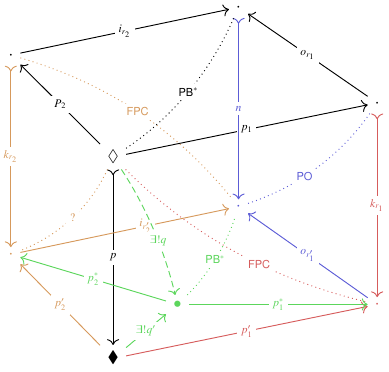}
\end{equation}
Finally, an argument analogous to the one demonstrated in~\eqref{eq:proof:hCompDPOchosenPB} then reveals that the isomorphism $q'$ in~\eqref{eq:proof:BCC-1variant} gives rise to a \kl(dc){globular isomorphism} in the chosen \kl{crDC} of SqPO-type.

We have thus proved for \kl{DPO-semantics} and for \kl{SqPO-semantics} that the functors $\hComp$ of the corresponding \kl{crDCs} possess \kl(rof){residual op-Cartesian lifts}, whose residues are indeed \kl(dc){globular isomorphisms}. Concretely, we found that from each diagram as in~\eqref{eq:proof:hDeCompCRDCa}, one may obtain a diagram of the following shape:
\begin{equation}\label{eq:proof:hCompIGRaux}
\ti{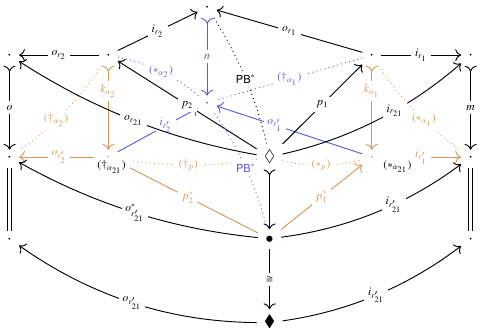}
\end{equation}
Moreover, it follows from the universal properties of pushouts, pullbacks and FPCs that the two constructions are essentially unique.\\

It remains to prove the \kl(crDC){complex decomposition property}, which for \kl{crDCs} of either DPO- or SqPO-semantics takes as its premise a diagram of the following shape:
\begin{equation}\label{eq:proof:cmplxDecompA}
\ti{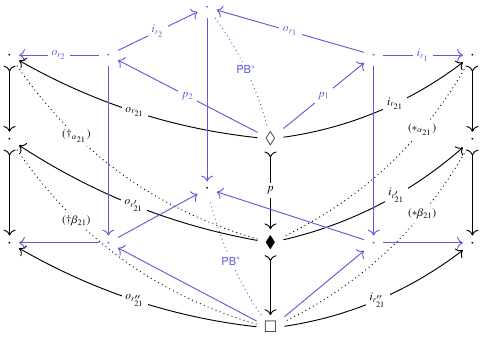}
\end{equation}
The claim then follows by constructing the following diagram, where the top half is constructed analogously to the DPO- or SqPO-variants of the \kl(crDC){horizontal decomposition property}:
\begin{equation}\label{eq:proof:cmplxDecompPartB}
\ti{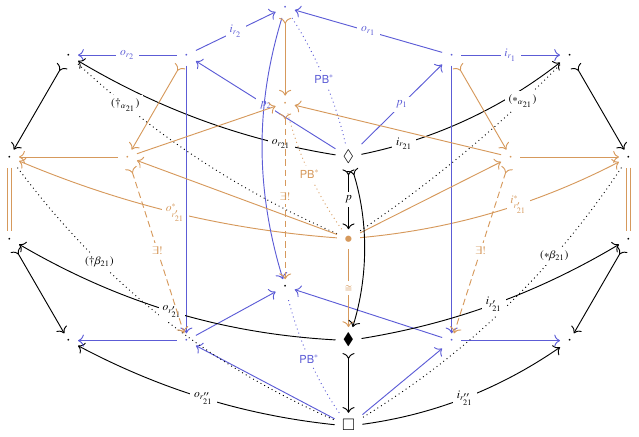}
\end{equation}
The existence of the dashed vertical arrows in the bottom half of the diagram in~\eqref{eq:proof:cmplxDecompPartB} may be derived from the respective universal properties of the pushout and FPC squares present. Finally, the various splitting lemmata for pushouts and FPCs then permit to demonstrate that the back squares of the middle and bottom half of the diagram (i.e., the vertical squares adjacent to the dashed arrows) indeed constitutes a horizontal composition of two DPO- or SqPO-type direct derivations, which concludes the proof.
\qed\end{proof}

\begin{proposition}\label{prop:D1PBcrDC}
Let $\bfC$ be a category with a \kl{stable system of monics} $\cM$, and such that $\bfC$ is suitable for the chosen rewriting semantics according to Definition~\ref{def:crs}. Let~$\bD_1$ denote the category of \kl(crs){rules} as objects and direct derivations of the chosen semantics as morphisms. Then if $\bfC$ is a \kl{vertical weak adhesive HLR category}, \kl(crDC){$\mathbb{D}_1$ has pullbacks}. 
\end{proposition}
\begin{proof}
Consider a cospan in $\bD_1$, which amounts to a diagram of the form below:
\begin{equation}\label{eq:proof:D1PBcrDCa}
\ti{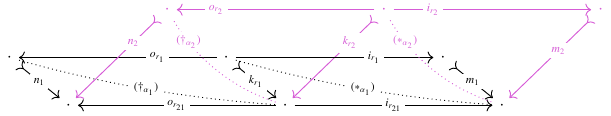}
\end{equation}
By assumption on the underlying category, $\bfC$ \kl(ssm){has pullbacks along $\cM$-morphisms}, which permits us to construct the diagram below from the one in~\eqref{eq:proof:D1PBcrDCa} via taking three pullbacks:
\begin{equation}\label{eq:proof:D1PBcrDCb}
\ti{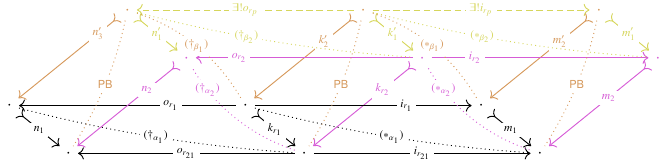}
\end{equation}
By the \kl{universal property of pullbacks}, there exist unique morphisms $o_{r_p}$ and $i_{r_p}$ that make the diagram commute, and thus by \kl{pullback-pullback decomposition}, we find that the squares marked $(\dag_{\beta_1})$, $(\dag_{\beta_2})$, $(*_{\beta_1})$ and $(*_{\beta_2})$ are pullbacks.
\begin{itemize}
\item Since axiom \kl(notationVKa){(V-iii-a)} of the definition of \kl{vertical weak adhesive HLR categories} holds in $\bfC$ (i.e., if \kl(ssm){pushouts along $\cM$-morphisms are stable under $\cM$-pullbacks}), the squares marked $(\dag_{\beta_1})$ and $(\dag_{\beta_2})$ are pushouts.
\item For \kl{DPO-semantics}, since \kl(notationVKa){(V-iii-a)} holds in $\bfC$, the squares marked $(*_{\beta_1})$ and $(*_{\beta_2})$ are pushouts.
\item For \kl{SqPO-semantics}, by \kl{stability of FPCs under pullbacks}, $(*_{\beta_1})$ and $(*_{\beta_2})$ are FPCs.
\end{itemize}
It then remains to demonstrate that the construction provided indeed yields a pullback in $\bD_1$. To this end, consider a diagram as below, where the upper blue squares encode a span in $\bD_1$ that together with the cospan in $\bD_1$ that was already depicted in~\eqref{eq:proof:D1PBcrDCa} yields a commutative square in $\bD_1$:
\begin{equation}\label{eq:proof:D1PBcrDCc}
\ti{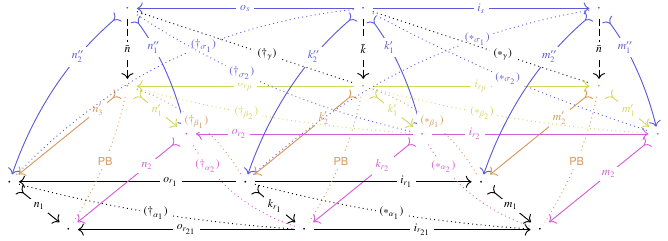}
\end{equation}
We have to prove that there exists a unique mediating morphism in $\bD_1$ (i.e., the dashed $\bD_0$-morphisms that make the diagram commute, and such that the squares marked $(\dag_{\gamma})$ and $(*_{\gamma})$ are of the correct kinds for the given semantics:
\begin{itemize}
\item By the \kl{universal property of pullbacks}, there exist uniquely the morphisms $\bar{n}$, $\bar{k}$ and $\bar{m}$ marked with dashed arrows in~\eqref{eq:proof:D1PBcrDCc}, which make the diagram commute. By the \kl(ssm){decomposition property of $\cM$-morphisms}, these morphisms are moreover in $\cM$.
\item By \kl{pushout-pullback decomposition}, the square marked $(\dag_{\gamma})$ is a pushout.
\item For the case of \kl{DPO-semantics}, yet again by \kl{pushout-pullback decomposition}, the square marked $(*_{\gamma})$ is a pushout.
\item For the case of \kl{SqPO-semantics}, by \kl{vertical FPC-pullback decomposition}, the square marked $(*_{\gamma})$ is an FPC.
\end{itemize}
In summary, we have thus demonstrated the unique existence of a $\bD_1$-morphism consisting of the squares marked $(\dag_{\gamma})$ and $(*_{\gamma})$ that make the diagram in $\bD_1$ commute, which concludes the proof.
\qed\end{proof}

Finally, taking full advantage of the results presented in Section~\ref{sec:fib}, we will investigate the existence of the requisite fibrational structures for the source and target functors on the \kl{double categories} for all of the categorical rewriting semantics of Definition~\ref{def:crs}. Let us first recall the properties that have to be satisfied by a category $\bfC$ to be suitable to carry \kl{DPO-semantics} or \kl{SqPO-semantics}:
\begin{itemize}
\item For \kl{DPO-semantics}, it is required that $\bfC$ \kl{has $\cM$-multi-IPCs}, that it \kl(ssm){has pushouts along $\cM$-morphisms}, and that \kl(ssm){$\cM$-morphisms are stable under pushout} (i.e., the latter two points amount to axiom \kl(adhVK){(V-ii)}).
\item For \kl{SqPO-semantics}, it is required that $\bfC$ \kl(ssm){has FPCs along $\cM$-morphisms}, that it \kl(ssm){has pushouts along $\cM$-morphisms}, and that \kl(ssm){$\cM$-morphisms are stable under pushout} (i.e., the latter two points amount to axiom \kl(adhVK){(V-ii)}).
\end{itemize}

For \kl{DPO-semantics}, recall from Lemma~\ref{prop:mIPCexistence} that a sufficient condition to ensure that $\bfC$ \kl{has $\cM$-multi-IPCs} is that $\bfC$ \kl(ssm){has pullbacks along $\cM$-morphisms} (i.e., axiom \kl(adhVK){(V-i)}), that \kl(ssm){pushouts along $\cM$-morphisms are stable under $\cM$-pullbacks} (i.e., axiom \kl(notationVKa){(V-iii-a)}), and that \kl(aps){pushouts along $\cM$-morphisms are pullbacks} (cf.\ Theorem~\ref{cor:adhPOPB}); hence, in summary, it is sufficient to require that $\bfC$ is a \kl{vertical weak adhesive HLR category}. For SqPO-semantics, in addition to asking that $\bfC$ be a vertical weak adhesive HLR category, we must further ask that it has FPCs along $\cM$-morphisms.

We can now state the theorem:

\begin{theorem}\label{thm:SmofTrmof}
Let $\bfC$ be a category that is \kl{finitary} and a \kl{vertical weak adhesive HLR category} with respect to a \kl{stable system of monics} $\cM$. For the case of \kl{SqPO-semantics}, we assume further that $\bfC$ \kl(ssm){has FPCs along $\cM$-morphisms}. Let $\bD$ denote the \kl{double category} based upon $\bfC$ and direct derivations of the respective kind as introduced in Section~\ref{sec:crs:dc}. Then the following fibrational properties hold:
\begin{enumerate}[label=(\roman*)]
\item The functor $S:\bD_1\rightarrow \bD_0$ is a \kl{multi-opfibration}.
\item The functor $T:\bD_1\rightarrow \bD_0$ is a \kl{residual multi-opfibration}.
\end{enumerate}
\end{theorem}
\begin{proof}
In the case of DPO-semantics, the category $\bfC$ therefore supports the following fibrational structures:
\begin{itemize}
\item By Theorem~\ref{thm:TPOvMOF}, since $\bfC$ \kl(ssm){has pullbacks along $\cM$-morphisms}, and since \kl(ssm){pushouts along $\cM$-morphisms are stable under $\cM$-pullbacks} in $\bfC$, the \kl{target functor $T_{\mathsf{PO}}:\mathsf{PO}_v(\bfC,\cM)\rightarrow \bfC\vert_{\cM}$ is a multi-opfibration}.
\item By Theorem~\ref{thm:sourcePOgopf}, since $\bfC$ \kl(ssm){has pushouts along $\cM$-morphisms}, and since \kl(ssm){$\cM$-morphisms are stable under pushout}, the \kl{source functor $S_{\mathsf{PO}}:\mathsf{PO}_v(\bfC,\cM)\rightarrow \bfC\vert_{\cM}$ is a 
Grothendieck opfibration}.
\end{itemize}

As for \kl{SqPO-semantics}, according to Theorem~\ref{thm:finitarityCats}(iii) that since $\bfC$ is \kl{finitary} and a \kl{vertical weak adhesive HLR category} with respect to the \kl{stable system of monics} $\cM$, $\bfC$ is \kl{($\cE$, $\cM$)-structured}, for $\cE$ the class of \kl(ssm){extremal morphisms} (w.r.t.\ $\cM$). We thus find the following results from Section~\ref{sec:fib}:
\begin{itemize}
\item By Theorem~\ref{thm:trgtFPC-FP}, since $\bfC$ \kl(ssm){has pullbacks along $\cM$-morphisms} and it \kl(ssm){has FPCs along $\cM$-morphisms}, the \kl{target functor $T_{\mathsf{FPC}}: \mathsf{FPC}_v(\bfC,\cM)\rightarrow \bfC\vert_{\cM}$ is a Grothendieck opfibration}.
\item By Theorem~\ref{thm:SourceFPCvRMOF}, since $\bfC$ is \kl{($\cE$, $\cM$)-structured}, \kl(ssm){has pullbacks, pushouts and FPCs along $\cM$-morphisms}, such that \kl(ssm){$\cM$-morphisms are stable under pushout}, and such that \kl(ssm){pushouts along $\cM$-morphisms are stable under $\cM$-pullbacks}, the \kl{source functor $S_{\mathsf{FPC}}:\mathsf{FPC}_v(\bfC,\cM)\rightarrow \bfC\vert_{\cM}$ is a residual multi-opfibration}.
\end{itemize}
With these preparations, it then remains to prove that indeed the functors $S,T:\bD_1\rightarrow \bD_0$ from $\bD_1$ (i.e., the category with rules as objects, and direct derivations as morphisms) to $\bD_0$ (i.e., the category $\bfC\vert_{\cM}$) are a \kl{residual multi-opfibration} in the case of $S$, and a \kl{multi-opfibration} in the case of $T$, respectively.

As for the functor $S:\bD_1\rightarrow \bD_0$, the existence of \kl(mof){multi-op-Cartesian liftings} is induced from the property that the functor $T_{\mathsf{PO}}:\mathsf{PO}_v(\bfC,\cM)\rightarrow \bfC\vert_{\cM}$ is a \kl{multi-opfibration}, and that both $S_{\mathsf{PO}}:\mathsf{PO}_v(\bfC,\cM)\rightarrow \bfC\vert_{\cM}$ and $T_{\mathsf{FPC}}:\mathsf{FPC}_v(\bfC,\cM)\rightarrow \bfC\vert_{\cM}$ are \kl{Grothendieck opfibrations}:
\begin{itemize}
\item Consider the following diagram:
\begin{equation}
\ti{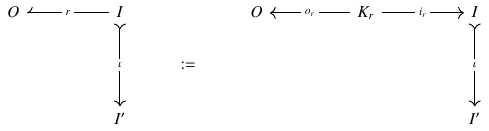}
\end{equation}
Since $T_{\mathsf{PO}}:\mathsf{PO}_v(\bfC,\cM)\rightarrow \bfC\vert_{\cM}$ is a \kl{multi-opfibration}, while $T_{\mathsf{FPC}}:\mathsf{FPC}_v(\bfC,\cM)\rightarrow \bfC\vert_{\cM}$ is a \kl{Grothendieck opfibration} (i.e., a special case of a \kl{multi-opfibration} where each family is a singleton), this entails the existence of a family of \kl(mof){multi-op-Cartesian liftings} (the $\kappa_{r'_j}$s, in blue in the diagram below), each of whose elements via the \kl{Grothendieck opfibration} property of $S_{\mathsf{PO}}:\mathsf{PO}_v(\bfC,\cM)\rightarrow \bfC\vert_{\cM}$ (the $\omega_j$s and $o_{r'_j}$s, cf.\ orange part of the diagram below) lifts into a \kl(crsDPO){DPO-type direct derivation} or an \kl(crsSqPO){SqPO-type direct derivation}, respectively, i.e., to an element of $\bD_1$:
\begin{equation}
\ti{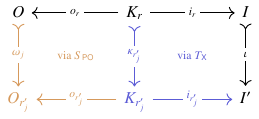}
\end{equation}
Moreover, consider a diagram as the one marked $(i)$ below:
\begin{equation}
\ti{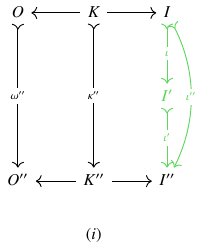} \xrightarrow{\text{via }T_{\mathsf{X}}}
\ti{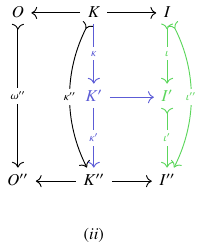}\xrightarrow{\text{via }S_{\mathsf{PO}}}
\ti{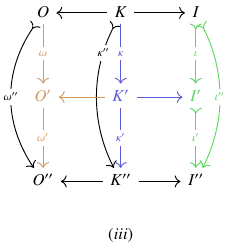}
\end{equation}
\begin{itemize}
\item Diagram $(ii)$ is obtained via invoking the fact that $T_{\mathsf{PO}}$ is a \kl{multi-opfibration} in the \kl{DPO-semantics} case, or via the fact that $T_{\mathsf{FPC}}$ is a \kl{Grothendieck opfibration} in the \kl{SqPO-semantics} case.
\item Diagram $(iii)$ is obtained via using that $S_{\mathsf{PO}}$ is a \kl{Grothendieck opfibration}.
\end{itemize}
\item Since evidently the above constructions are essentially unique, we have thus proved that $S:\bD_1\rightarrow \bD_0$ ``inherits'' a \kl{multi-opfibration} structure from the properties of $S_{\mathsf{PO}}$ and $T_{\mathsf{PO}}$ or $T_{\mathsf{FPC}}$, respectively.
\end{itemize}

Next, for the case of the target functor $T:\bD_1\rightarrow \bD_0$ \kl{DPO-semantics}, due to the symmetry in the definition of \kl{DPO-semantics}, the derivation that $T$ is a \kl{multi-opfibration} follows the same line of arguments as the one for $S:\bD_1\rightarrow \bD_0$ for this semantics. Moreover, since a \kl{multi-opfibration} is a special case of a \kl{residual multi-opfibration}, namely the case when each \kl(rmof){residue} is an identity morphism, this demonstrates that indeed $T:\bD_1\rightarrow \bD_0$ also carries the structure of a \kl{residual multi-opfibration}.

Finally, the proof for the fibrational property of $T:\bD_1\rightarrow \bD_0$ for the case of \kl{SqPO-semantics} is considerably more involved. 
\begin{itemize}
\item Consider a diagram as below:
\begin{equation}\label{eq:proofTrmofA}
\ti{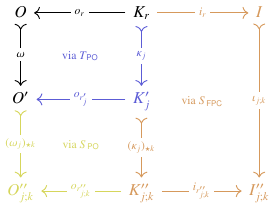}
\end{equation}
\begin{itemize}
\item By the \kl{multi-opfibration} property of $T_{\mathsf{PO}}$, %
there exists a family of \kl(mof){multi-op-Cartesian liftings} (blue part of the diagram in~\eqref{eq:proofTrmofA}). %
\item By the \kl{residual multi-opfibration} property of $S_{\mathsf{FPC}}$, %
for each element of the aforementioned \kl(mof){multi-op-Cartesian lifting} (indexed by $j$), %
there exists a family of \kl(rmof){residual multi-op-Cartesian liftings} %
(orange part of the diagram in~\eqref{eq:proofTrmofA}), %
where each such lifting (indexed by $k$) consists of the data of an \kl{FPA} %
(i.e., of a \kl(rmof){residue} $K'_{j} \rtail(\kappa_j)_{\star k}\rightarrow K''_{j;k}$ and a pair %
of morphisms $K''_{j;k}-i_{r''_{j;k}}\rightarrow I''_{j;k}$ and %
$I\rtail \iota_{j;k}\rightarrow I''_{j;k}$ such that %
the right commutative square of the above diagram is an FPC). %
Finally, in order to obtain an \kl(crsSqPO){SqPO-type direct derivation}, we use the \kl{Grothendieck opfibration} property of $S_{\mathsf{PO}}$ to obtain the yellow parts of the diagram in~\eqref{eq:proofTrmofA} (which in effect amounts to taking a pushout to obtain the cospan $O'\rtail (\omega'_j)_{\star k}\rightarrow O''_{j;k}\leftarrow o_{r''_{j;k}}-K''_{j;k}$). By \kl{pushout-pushout composition}, the composite of the top left and bottom left commutative squares in~\eqref{eq:proofTrmofA} yields a pushout, and hence the overall diagram indeed encodes an \kl(crsSqPO){SqPO-type direct derivation}.
\end{itemize}
\item In order to prove that $T:\bD_1\rightarrow \bD_0$ indeed satisfies the \kl(rmof){universal property} of \kl{residual multi-opfibrations}, consider diagram $(i)$ below left:
\begin{equation}\label{eq:proofSrmofB}
\ti{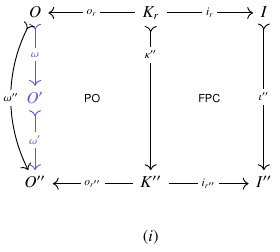}
\qquad\qquad
\ti{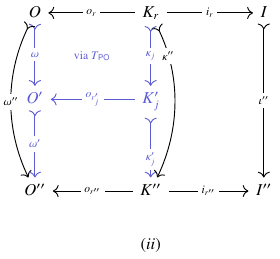}
\end{equation}
\begin{itemize}
\item Invoking the \kl{multi-opfibration} property of $T_{\mathsf{PO}}$ yields a family of \kl(mof){multi-op-Cartesian liftings}, i.e., the blue parts of diagram $(ii)$ in~\eqref{eq:proofSrmofB} (indexed by $j$).
\item For each element of the aforementioned lifting, which in particular includes a sequence of $\cM$-morphisms $K\rtail \kappa_j\rightarrow K'_j\rtail \kappa'_{j}\rightarrow K''$, invoke the \kl{residual multi-opfibration} property of $S_{\mathsf{FPC}}$ in order to obtain a family of \kl(rmof){residual multi-op-Cartesian liftings}, i.e., the orange parts of the diagram $(iii)$ in~\eqref{eq:proofSrmofC} below (indexed by $k$), where each element of the family consists of an \kl{FPA}, and with an induced bottom right square in~\eqref{eq:proofSrmofC} that is an FPC.
\begin{equation}\label{eq:proofSrmofC}
\ti{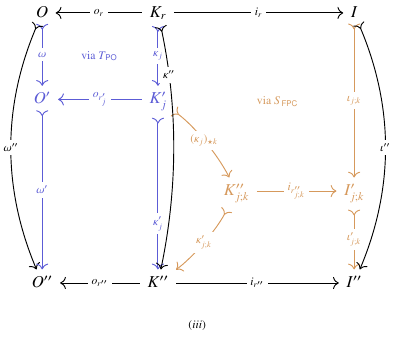}
\end{equation}
\item Finally, invoking the \kl{Grothendieck opfibration} property of $S_{\mathsf{PO}}$ allows to effectively split the bottom left part of the diagram in~\eqref{eq:proofSrmofC} into two pushout squares, i.e., the parts of diagram $(iv)$ in~\eqref{eq:proofSrmofD} below colored in yellow.
\begin{equation}\label{eq:proofSrmofD}
\ti{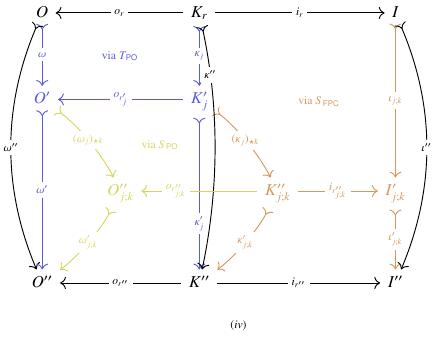}
\end{equation}
\end{itemize}
\item Essential uniqueness of the fibrational properties of $S:\bD_1\rightarrow \bD_0$ are inherited from the essential uniqueness of the functors $S_{\mathsf{PO}}$, $S_{\mathsf{FPC}}$ and $T_{\mathsf{PO}}$.
\end{itemize}
In summary, we have thus succeeded in demonstrating that $S:\bD_1\rightarrow \bD_0$ for \kl{SqPO-semantics} carries a \kl{residual multi-opfibration} structure, which concludes the proof.
\qed\end{proof}

We conclude this discussion of theoretical results with the following observations:

\begin{remark}\label{rem:crs}
A folklore result of categorical rewriting theory, and especially in the tradition of Ehrig et al.~\cite{ehrig2010categorical} has been that the notion of \kl{vertical weak adhesive HLR categories} is a reasonably general characterization of categories with sufficient properties to support some form of compositional semantics. While previous works did not consider the validity of an \kl{associativity theorem} as a prerequisite for a rewriting theory to be compositional, the main criterion was indeed the existence of a \kl{concurrency theorem} for the given theory. As our analysis demonstrates, \kl{vertical weak adhesive HLR categories} are \emph{almost} the main type of categories to support compositional rewriting, were it not for the additional properties required as presented in Table~\ref{tab:main} for the various generalizations of \kl(crsType){linear} semantics (which in effect was the only kind of semantics fully analyzed in the traditional literature~\cite{ehrig:2006fund}). Indeed, the discriminating factors in this regard are the \kl(dc){horizontal composition} (Proposition~\ref{prop:hCompCRDC}) and the \kl(crDC){horizontal decomposition} (Proposition~\ref{prop:hDeCompCRDC}) properties required for a given semantics to yield a \kl{compositional rewriting double category}, which for instance disqualifies the category $\mathbf{SGraph}$ of \kl{directed simple graphs} to support compositional \kl(crsType){generic} \kl{DPO-semantics} (i.e., due to failure of axiom \kl(adhVK){(L-iii)} in $\mathbf{SGraph}$; cf.\ also the discussion in Example~\ref{ex:SGraphCRSgenericDPOfailure}). On the other hand, referring to Table~\ref{tab:adh} for a list of practically relevant examples of categories with \kl{adhesivity properties}, in many cases properties beyond weak adhesivity such as the existence of all pullbacks are indeed verified, which raises the interesting theoretical question of whether it might be possible to find a more general classification of categories that takes the additional properties presented in Table~\ref{tab:main} as its basis, and that would permit an easier access to determining the kind of semantics a given category supports. Moreover, since many examples provided in Table~\ref{tab:adh} are indeed obtained as some form of \kl{comma category} construction based upon \kl{adhesive categories} such as $\mathbf{Set}$, one might envision an extension of Theorem~\ref{thm:comCats} that would permit to also determine whether a given comma category possesses additional structures such as an \kl{$\cM$-partial map classifier}, existence of pullbacks or an \kl(ssm){$\cM$-initial object}. We leave these open questions to future work.
\end{remark}

\subsection{Illustration: compositional non-linear double- and sesqui-pushout rewriting}\label{sec:nlCRS}

In this final part of the paper, we will present in some further detail the quintessential examples of compositional rewriting theories in the sense of our novel framework, i.e., the ``non-linear'' variants of double- and sesqui-pushout rewriting over suitable categories. By suitably restricting the formulae provided in the following to the relevant notion of linearity, one may moreover obtain explicit formulae also for the remaining six types of semantics according to Definition~\ref{def:crs}.\\

The aim of the ensuing results consists in providing explicit formulae for both the notion of \kl(crs){direct derivations} and of \emph{rule compositions}, i.e., in a formulation perhaps somewhat more familiar to experts in graph rewriting theory. This involves in particular extracting the important notion of \emph{rule compositions} from the \kl{concurrency theorems}:

\begin{lemma}\label{lem:genericDPOexplicit}
Let $\bfC$ be a category that \kl{has pullbacks}, that is \kl{finitary} and that is an \kl{adhesive HLR category} with respect to a \kl{stable system of monics} $\cM$. Consider \kl(crsType){generic} \kl{Double-Pushout (DPO) semantics} over $\bfC$, where \kl(crsDPO){direct derivations} are defined more explicitly as follows (compare Definition~\ref{def:crs}):
\begin{itemize}
\item The \AP\intro(crsDPO){set of DPO-admissible matches} of \kl(crs){rule} $r=(O\xleftarrow{o_r} K_r \xrightarrow{i_r} I)\in span(\bfC)$ into object $X\in \obj{\bfC}$ is defined as
\begin{equation}
\MatchGT{DPO}{r}{X}:= \{ (m,k_{\alpha},i_{\alpha})\mid m\in \cM \,,\; (k_{\alpha},i_{\alpha})\in \mIPC{i_r}{m}\}\diagup_{\sim}\,,
\end{equation}
where equivalence $\sim$ is defined as equivalence up to \kl(mIPC){universal isomorphisms} of \kl{$\cM$-multi-IPCs}.
\item A \kl(crsDPO){DPO-type direct derivation} of $X\in \obj{\bfC}$ with rule $r$ along $m\in \MatchGT{DPO}{r}{X}$ is defined as a diagram in~\eqref{eq:defDPONLRSdd} below, where $(1)$ is the \kl{$\cM$-multi-IPC} element chosen as part of the data of the \kl(crsDPO){admissible match}, while $(2)$ is formed as a pushout.
\begin{equation}\label{eq:defDPONLRSdd}
\ti{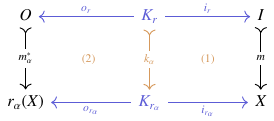}
\end{equation}
\end{itemize}

Moreover, the \kl(cct){synthesis} part of the \kl{concurrency theorem} (cf.\ in particular~\eqref{eq:concurrencyThm}) yields the following notions: 
\begin{itemize}
\item Given $r_2,r_1\in span(\bfC)$, let the \AP\intro(crsDPOrComp){set of DPO-type admissible matches} of rule $r_2$ into $r_1$ (also referred to as the \emph{dependency relation}~\cite{ehrig:2006fund}) be defined as follows:
\begin{equation}
\RMatchGT{DPO}{r_2}{r_1}:= \{
(j_2, j_1, \kappa_2, i_{\overline{r}_2}, \kappa_1,o_{\overline{r}_1})
\mid  (j_2,j_1)\in \Msum{I_2}{O_1} \,,\; (\kappa_2, i_{\overline{r}_2})\in \mIPC{i_{r_2}}{j_2}\,,\; %
(\kappa_1,o_{\overline{r}_1})\in \mIPC{o_{r_1}}{j_1}
\}\diagup_{\sim}
\end{equation}
Here, the equivalence $\sim$ by which we quotient is defined via the compatible universal isomorphisms of \kl{$\cM$-multi-sums} and \kl{$\cM$-multi-IPCs} (i.e., ``compatible'' relative to the diagram in~\eqref{eq:defNLRDPOrc} below).
\begin{equation}\label{eq:defNLRDPOrc}
\ti{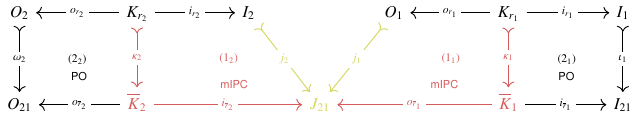}
\end{equation}
\item A \AP\intro(crsDPOrComp){DPO-type rule composition} of two general rules $r_1,r_2\in span(\bfC)$ along an \kl(crsDPO){admissible match} $\mu\in\RMatchGT{DPO}{r_2}{r_1}$  is defined via a diagram as in~\eqref{eq:defNLRDPOrc} above, where $(1_2)$ and $(1_1)$ are the \kl{$\cM$-multi-IPC} elements chosen as part of the data of the match, while $(2_2)$ and $(2_1)$ are pushouts. %
We then define the \AP\intro(crsDPOrComp){composite rule} via span composition:
\begin{equation}\label{eq:def:DPOrComp}
\comp{r_2}{\mu}{r_1}:= (O_{21}\leftarrow\overline{K}_{2} \rightarrow J_{21})\circ (J_{21}\leftarrow \overline{K}_1\rightarrow I_{21})
\end{equation}
\end{itemize}
With these definitions, one recovers a variant of the \kl{concurrency theorem} whereby the statement of~\eqref{eq:concurrencyThm} is expressed as follows: 
\begin{itemize}
\item \textbf{Synthesis:} given an object $X_0\in \obj{\bfC}$, for every pair $((r_2,\nu_2),(r_1,\nu_1))$ of \kl(crs){rules} and \kl(crsDPO){admissible matches}, where $\nu_1\in \MatchGT{DPO}{r_1}{X_0}$ and $\nu_2\in \MatchGT{DPO}{r_2}{X_1}$ with $X_1:=r_{1_{\nu_1}}(X_0)$, there exists an \kl(crsDPOrComp){admissible match} $\mu\in \RMatchGT{DPO}{r_2}{r_1}$ of rule $r_2$ into rule $r_1$ and an admissible match $\nu_{21}\in  \MatchGT{DPO}{r_{2_{\mu}1}}{X_0}$ of the \kl(crsDPOrComp){composite rule} $r_{2_\mu 1}$ defined as in~\eqref{eq:def:DPOrComp} such that $(r_{2_\mu 1})_{\nu_{21}}(X_0)\kl{\gcong} r_{2_{\nu_2}}(r_{1_{\nu_1}}(X_0))$
\item \textbf{Analysis:} for every \kl(crsDPOrComp){admissible match} $\mu\in \RMatchGT{DPO}{r_2}{r_1}$ of rule $r_2$ into rule $r_1$ and for every admissible match $\nu_{21}\in  \MatchGT{DPO}{r_{2_{\mu}1}}{X_0}$ of the \kl(crsDPOrComp){composite rule} $r_{2_\mu 1}$ into the object $X_0$, there exists %
a pair $((r_2,\nu_2),(r_1,\nu_1))$ of \kl(crsDPO){admissible matches} $(\nu_1,\nu_2)$, where $\nu_1\in \MatchGT{DPO}{r_1}{X_0}$ and $\nu_2\in \MatchGT{DPO}{r_2}{r_{1_{\nu_1}}(X_0)}$, such that $r_{2_{\nu_2}}(r_{1_{\nu_1}}(X_0))\kl{\gcong} (r_{2_\mu 1})_{\nu_{21}}(X_0)$.
\end{itemize}
\end{lemma}
\begin{proof}
According to Lemma~\ref{prop:mIPCexistence}, $\bfC$ has \kl{$\cM$-multi-IPCs}, hence the notion of \kl(crsDPO){direct derivations} is well-posed. %
The rest of the proof then follows by instantiating~\eqref{eq:concurrencyThm} for the case of \kl(crsType){generic} \kl{DPO-semantics}. In particular, the explicit formula for the \kl(crsDPOrComp){rule composition} is obtained by taking advantage of the results of Theorem~\ref{thm:SmofTrmof}, i.e., noting that the \kl(dc){source functor} $S:\bD_1\rightarrow \bD_0$ of the \kl{compositional rewriting double category} is a \kl{multi-opfibration}, and that the \kl(dc){target functor} $T:\bD_1\rightarrow \bD_0$ is a \kl{multi-opfibration}, and hence a special type of \kl{residual multi-opfibration}. The latter statement is illustrated in the diagram below, which makes explicit the ``identity-\kl(rmof){residue}'' $\overline{K}_1-id_{\overline{K}_1}\rightarrow \overline{K}_1$:
\begin{equation}\label{eq:proof:explDPOrComp}
\ti{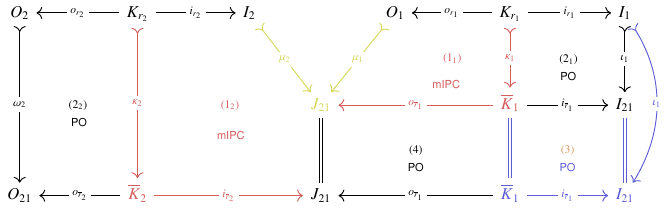}
\end{equation}
\qed\end{proof}

\begin{example}\label{ex:SGraphCRSgenericDPOfailure}
One of the most striking outcomes of the analysis presented in this paper is the \textbf{failure} of the category $\mathbf{SGraph}$ of \kl{directed simple graphs} to support \kl(crsType){generic} \kl{DPO-semantics}. It is worthwhile emphasizing that even though in many cases the full-fledged generality of this type of semantics (which, as discussed also in Section~\ref{sec:intro:bg}, supports cloning and fusing subobjects) might not be needed, the case of $\mathbf{SGraph}$ is indeed much more fundamental, and, in a certain sense, was one of the main motivations for the developments presented in this paper. To wit, it is well-known (cf.\ e.g.\ \cite{quasi-topos-2007}) that $\mathbf{SGraph}$ is a \kl{vertical weak adhesive HLR category} only with respect to the \kl{stable system of monics} $\cM_{\mathbf{SGraph}}=\regmono{\mathbf{SGraph}}$ of \emph{regular monomorphisms}, which are the \emph{edge-reflecting} monomorphisms (cf.\ Theorem~\ref{thm:SGraphProperties}(S-ii)), i.e., in particular \emph{not} with respect to the class of all monomorphisms in $\mathbf{SGraph}$. Moreover, it is well-known that $\mathbf{SGraph}$ is \emph{not} an \kl{adhesive HLR category} w.r.t.\ the class $\regmono{\mathbf{SGraph}}$~\cite{ehrig:2006fund,quasi-topos-2007,GABRIEL_2014}, and hence indeed strictly does \textbf{not} support \kl(crsType){generic} \kl{DPO-semantics} (cf.\ Table~\ref{tab:main}). This failure is particularly startling since it is perfectly well possible to define \kl(crsDPO){DPO-type direct derivations}  (with $\cM=\regmono{\mathbf{SGraph}}$), as $\mathbf{SGraph}$ \emph{does} possess the properties required according to Definition~\ref{def:crs}. In fact, the examples of \kl{multi-initial pushout complements} depicted in the left diagram of~\eqref{eq:FPCexamples} may be interpreted as \kl(crsDPO){DPO-type direct derivations} along a rule with an identity output morphism, demonstrating that \kl(crsType){generic} \kl{DPO-semantics} is at least in principle definable in $\mathbf{SGraph}$.\\\
\end{example}

Finally, let us turn towards \kl{SqPO-semantics}, presented in the remainder of this section in explicit detail for its \kl(crsType){generic} variant\footnote{Note that in the original conference version~\cite{BehrHK21} of this paper, we had provided a variant of this definition for the case of the underlying category being a \kl{quasi-topos}, yet the results of the present paper permit to formulate this definition for a more general class of categories.} (from which explicit definitions for the other variants of \kl{SqPO-semantics} may be obtained by restricting the horizontal morphisms to be in $\cM$ as appropriate for the given semantics).

\begin{lemma}\label{lem:NLRSsqpo}
Let $\bfC$ be a category that \kl{has pullbacks}, that is \kl{finitary} and that is an \kl{adhesive HLR category} with respect to a \kl{stable system of monics} $\cM$. Assume further that $\bfC$ \kl(ssm){has FPCs along $\cM$-morphisms}. Consider \kl(crsType){generic} \kl{Sesqui-Pushout (SqPO) semantics} over $\bfC$, where \kl(crsSqPO){direct derivations} are defined more explicitly as follows (compare Definition~\ref{def:crs}):
\begin{itemize}
\item The \AP\intro(crsSqPO){set of SqPO-admissible matches} of a rule \emph{rule} $r=(O\leftarrow K\rightarrow I)\in span(\bfC)$ into an object $X\in \obj{\bfC}$ is defined as
\begin{equation}
\MatchGT{SqPO}{r}{X}:= \{ I-m\rightarrow X\mid m\in \cM\}\,.
\end{equation}
\item A \kl(crsSqPO){SqPO-type direct derivation}~\cite{Corradini_2006} of $X\in \obj{\bfC}$ with rule $r$ along $m\in \MatchGT{SqPO}{r}{X}$ is defined as a diagram as in~\eqref{eq:defNLRSdd} below, where $(1)$ is formed as an FPC, while $(2)$ is formed as a pushout.
\begin{equation}\label{eq:defNLRSdd}
\ti{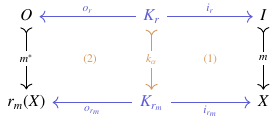}
\end{equation}
\end{itemize}
Moreover, the \kl(cct){synthesis} part of the \kl{concurrency theorem} (cf.\ in particular~\eqref{eq:concurrencyThm}) yields the following notions: 
\begin{itemize}
\item Given $r_2,r_1\in span(\bfC)$, the \AP\intro(crsSqPOrComp){set of SqPO-type admissible matches} of rule  $r_2$ into $r_1$  is defined as\footnote{In the conference version of this paper, we had opted for a slightly different variant of the definition of \kl{FPAs} than in the current paper, i.e., where the pushout square was not explicitly mentioned; however, due to the nature of the equivalence relation $\sim$ in~\eqref{eq:defSqPOrCompMatches}, we in fact arrive at an equivalent notion of \kl(crsSqPOrComp){admissible SqPO-type matches of rules}.}
\begin{equation}\label{eq:defSqPOrCompMatches}
\begin{aligned}
\RMatchGT{SqPO}{r_2}{r_1}:= \{
(j_2,j_1,
o_{r_1'}, \kappa_1, 
(\kappa_1)_{\star k}, i_{r_{1_k}}, e_{1_k}
)&\mid (j_2,j_1)\in \Msum{I_2}{O_1}\,,\;
 (o_{r_1'}, \kappa_1)\in \mIPC{o_{r_1}}{j_1}\,,\\
&\quad ((\kappa_1)_{\star k}, i_{r_{1_k}}, e_{1_k})\in \FPA{i_{r_1},{\color{blue}\iota_1},\kappa_1,{\color{blue}i_{r_1'}}}\,,\\
&\qquad ({\color{blue}i_{r_1'}},{\color{blue}\iota_1}) \in \pO{\kappa_1,i_{r_1}}
\}\diagup_{\sim}\,,
\end{aligned}
\end{equation}
where the notation $({\color{blue}i_{r_1'}},{\color{blue}\iota_1}) \in \pO{\kappa_1,i_{r_1}}$ entails that the cospan $({\color{blue}i_{r_1'}},{\color{blue}\iota_1})$ is a pushout of the span $(\kappa_1,i_{r_1})$, and %
where equivalence is defined up to the compatible universal isomorphisms of \kl{$\cM$-multi-sums}, \kl{$\cM$-multi-IPCs} and \kl{FPAs} (i.e., ``compatibility'' relative to the diagram in~\eqref{eq:defNLRSrc} below).
\begin{equation}\label{eq:defNLRSrc}
\ti{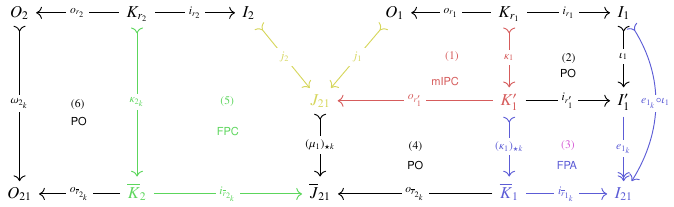}
\end{equation}
\item An \AP\intro(crsSqPOrComp){SqPO-type rule composition} of two general rules $r_1,r_2\in span(\bfC)$ along an admissible match $\mu\in \RMatchGT{SqPO}{r_2}{r_1}$ is defined via a diagram as in~\eqref{eq:defNLRSrc} above, which may be constructed step-wise (going clockwise) by letting square $(1)$ be an \kl{mIPC}, square $(2)$ a pushout, square $(3)$ an \kl{FPA}, square $(4)$ a pushout, square $(5)$ an \kl{FPC}, and finally square $(6)$ a pushout. We then define the \AP\intro(crsSqPOrComp){composite rule} via span composition:
\begin{equation}\label{eq:def:SqPOrCompRule}
\sqComp{r_2}{\mu}{r_1}:= (O_{21}\leftarrow\overline{K}_{2} \rightarrow \overline{J}_{21})\circ (\overline{J}_{21}\leftarrow \overline{K}_1\rightarrow I_{21})
\end{equation}
\end{itemize}
With these definitions, one recovers a variant of the \kl{concurrency theorem} whereby the statement of~\eqref{eq:concurrencyThm} is expressed as follows: 
\begin{itemize}
\item \textbf{Synthesis:} given an object $X_0\in \obj{\bfC}$, for every pair $((r_2,m_2),(r_1,m_1))$ of \kl(crs){rules} and \kl(crsSqPO){admissible matches}, where $m_1\in \MatchGT{SqPO}{r_1}{X_0}$ and $m_2\in \MatchGT{SqPO}{r_2}{X_1}$ with $X_1:=r_{1_{m_1}}(X_0)$, there exists an \kl(crsSqPOrComp){admissible match} $\mu\in \RMatchGT{SqPO}{r_2}{r_1}$ of rule $r_2$ into rule $r_1$ and an admissible match $m_{21}\in  \MatchGT{SqPO}{r_{2_{\mu}1}}{X_0}$ of the \kl(crsSqPOrComp){composite rule} $r_{2_{\mu}1}$ defined as in~\eqref{eq:def:SqPOrCompRule} such that $(r_{2_\mu1})_{m_{21}}(X_0)\kl{\gcong} r_{2_{m_2}}(r_{1_{m_1}}(X_0))$.
\item \textbf{Analysis:} for every \kl(crsSqPOrComp){admissible match} $\mu\in \RMatchGT{SqPO}{r_2}{r_1}$ of rule $r_2$ into rule $r_1$ and for every admissible match $m_{21}\in  \MatchGT{SqPO}{r_{2_{\mu}1}}{X_0}$ of the \kl(crsSqPOrComp){composite rule} $r_{2_{\mu}1}$ into an object $X_0$, there exists a pair $(m_1,m_2)$ of \kl(crsSqPO){admissible matches}, where $m_1\in \MatchGT{SqPO}{r_1}{X_0}$ and $m_2\in \MatchGT{SqPO}{r_2}{r_{1_{m_1}}(X_0)}$, such that $(r_{2_\mu1})_{m_{21}}(X_0)\kl{\gcong} r_{2_{m_2}}(r_{1_{m_1}}(X_0))$.
\end{itemize}
\end{lemma}
\begin{proof}
By assumption, $\bfC$ \kl(ssm){has FPCs along $\cM$-morphisms}, hence the definition of \kl(crsSqPO){SqPO-type direct derivations} is well-posed. Moreover, according to Lemma~\ref{lem:FPAconstr}, the assumptions on $\bfC$ suffice to demonstrate that $\bfC$ has \kl{FPAs}, hence the notions of \kl(crsSqPOrComp){admissible matches of rules} and of \kl(crsSqPOrComp){composite rules} are also well-posed. %
The rest of the proof then follows by instantiating~\eqref{eq:concurrencyThm} for the case of \kl(crsType){generic} \kl{SqPO-semantics}. In particular, the explicit formula for the \kl(crsSqPOrComp){rule composition} is obtained by taking advantage of the results of Theorem~\ref{thm:SmofTrmof}, i.e., noting that the \kl(dc){source functor} $S:\bD_1\rightarrow \bD_0$ of the \kl{compositional rewriting double category} is a \kl{multi-opfibration}, and that the \kl(dc){target functor} $T:\bD_1\rightarrow \bD_0$ is a \kl{residual multi-opfibration}.
\qed\end{proof}

\begin{example}\label{ex:SqPOexplicit}
Let us illustrate the notion of \kl(crsSqPO){SqPO-type rule composition}, as given in Lemma~\ref{lem:NLRSsqpo}, with the following example in the setting of \kl{directed multigraphs}. %
\begin{equation}\label{eq:SqPOexample}
\vcenter{\hbox{\includegraphics[width=0.9\textwidth]{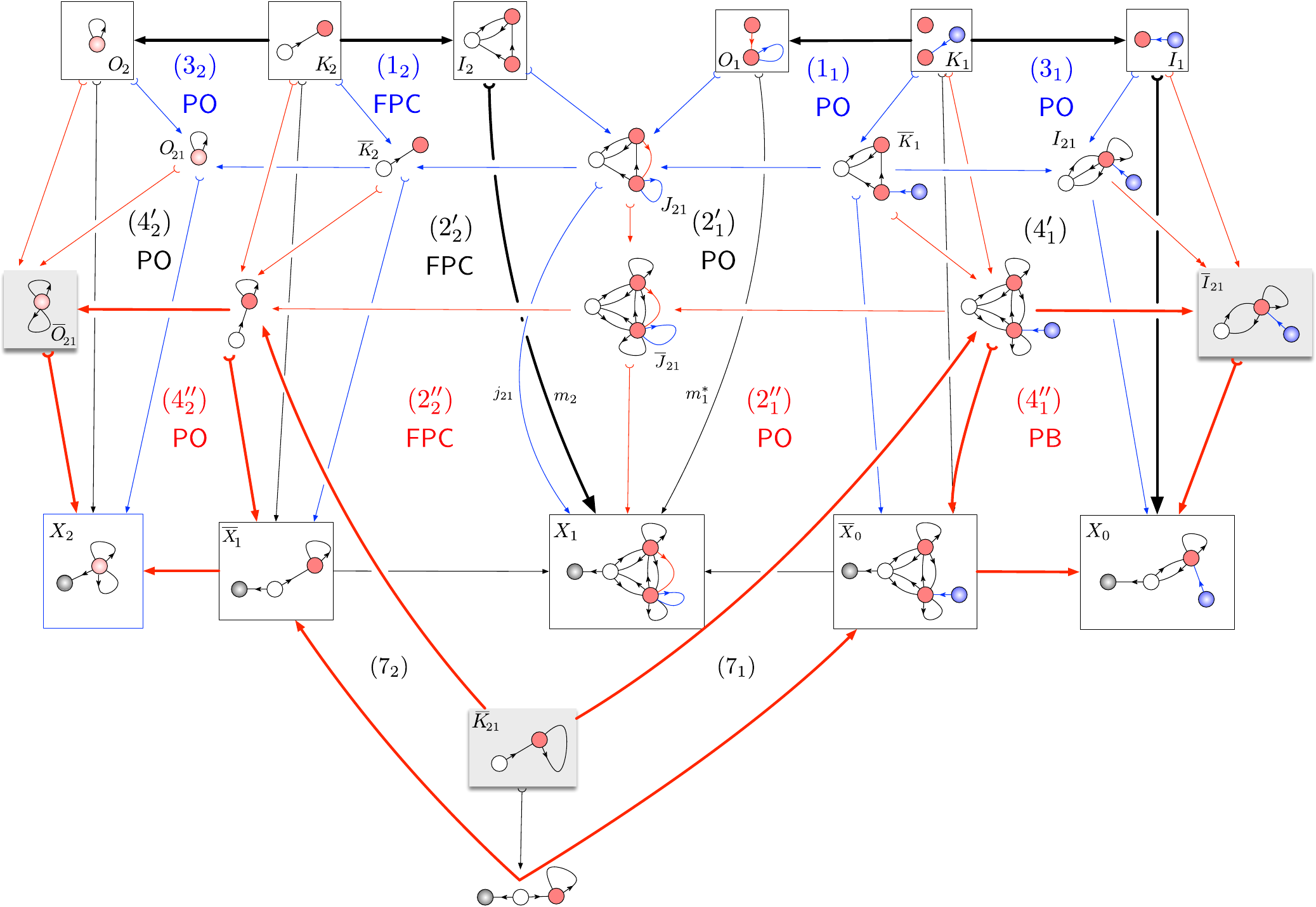}}}
\end{equation}
In order to provide the interested readers with some further intuitions for the relatively complex structure of \kl(crsSqPOrComp){rule compositions} in \kl(crsType){generic} \kl{SqPO-semantics}, we will present in the following a heuristic explanation for the precise shape of the diagram in~\eqref{eq:SqPOexample} (which is constructed via invoking Lemma~\ref{lem:NLRSsqpo}).

In this example, we have two rules. The first clones one node\footnote{Note that the structure of the homomorphisms may be inferred from the node positions, with the exception of the vertex clonings that are explicitly mentioned in the text.}, but not its incident edge, adds a new edge between that original node and its clone and then merges that original node with the other (blue) node of the input graph. The second rule deletes one node and then merges the two remaining nodes. The given applications to the graphs $X_0$ and $X_1$ illustrate some of the idiosyncrasies of SqPO-rewriting: 
\begin{itemize}
\item Since the node of $X_0$ that is being cloned possesses a self-loop, the result of cloning is two nodes, each with a self-loop, with one edge going each way between them. 
\item In the application of the second rule to $X_1$, we see the side-effect whereby all edges incident to the deleted node are themselves deleted (as also occurs in SPO-, but not in DPO-rewriting). 
\end{itemize}
The overall effect of the two rewrites can be seen in $X_2$; as usual, this depends on the overlap between the images of $O_1$ and $I_2$ in $X_1$. %
This overlap is precisely the \kl{multi-sum} element $J_{21}$. %
Since our example is set in an \kl{adhesive category}, this can be most easily computed by taking the \kl{pullback} of $m_1^*$ and $m_2$ and then the \kl{pushout} of the resulting span. %
The \kl{pushout} that defines the rewrite from $\overline{X}_0$ to $X_1$ can now be factorized by computing the \kl{pullback} of the arrow $j_{21}$ from $J_{21}$ to $X_1$ and the arrow from $\overline{X}_0$ to $X_1$; %
this determines $\overline{K}_1$ and its universal arrow from $K_1$ with the consequence that $(1_1)$ and $(2_1)$, the vertical pasting of $(2'_1)+(2''_1)$, are both \kl{pushouts}. %
Let us note that $\overline{K}_1$ is the appropriate member of the \kl{multi-IPC}, as determined by the particular structure of $\overline{X}_0$.

The \kl{pushout} $(3_1)$ induces a universal arrow from $I_{21}$ to $X_0$; %
but an immediate inspection reveals that this homomorphism is not a monomorphism (nor an epimorphism in this case). %
As such, we cannot hope to use $I_{21}$ as the input side of the composite rule. %
Furthermore, we find that the square $(4_1)$, the vertical pasting of $(4'_1)+(4''_1)$, is neither a \kl{pullback} nor a \kl{pushout}. %
However, the \kl{FPA} $\overline{I}_{21}$ resolves these problems by enabling a factorization of this square, %
based upon a factorization of the morphism $I_{21}\rightarrow X_0$ into an epimorphism $I_{21}\xrightarrow{e_{21}} \overline{I}_{21}$ and a monomorphism $\overline{I}_{21}\xrightarrow{m_{21}}X_0$. Note also that $(4''_1)$ and $(3_1)+(4'_1)$ are \kl{pullbacks} and indeed \kl{FPCs}. %
This factorization, as determined by $e_{21}$, can now be \emph{back-propagated} to factorize $(2_1)$ into \kl{pushouts} $(2'_1)$ and $(2''_1)$, %
which gives rise to an augmented version  $\overline{J}_{21}$ of the \kl{multi-sum} object $J_{21}$. %
Note moreover that the effect of back-propagation concerns also the contribution of the second rule in the composition: %
the final output graph contains an extra self-loop (compared to the graph $O_{21}$ defined by the \kl{pushout} $(3_2)$), which is induced by the extra self-loop of $\overline{J}_{21}$ that appears due to back-propagation.

We may then compute the composite rule via taking a \kl{pullback} to obtain $\overline{K}_{21}$, yielding in summary the rule $\overline{O}_{21}\leftarrow\overline{K}_{21}\rightarrow\overline{I}_{21}$. %
Performing the remaining steps of the \kl(cct){synthesis} part of the \kl{concurrency theorem} then amounts to constructing the commutative cube in the middle of the diagram, yielding the \kl{FPC} $(7_1)$ and the \kl{pushout} $(7_2)$, and thus finally the one-step \kl(crsSqPO){SqPO-type direct derivation} from $X_0$ to $X_2$ along the \kl(crsSqPOrComp){composite rule} $\overline{O}_{21}\leftarrow\overline{K}_{21}\rightarrow\overline{I}_{21}$.

Let us finally note, as a general remark, that if the first rule in an SqPO-type rule composition is \kl(crsType){output-linear} then the \kl{mIPC} is uniquely determined; and if it is \kl(crsType){input-linear} then the \kl{pushout} $(3_1)$ is also an \kl{FPC} (cf.\ Lemma~\ref{lem:FPCauxMonoStab}) and $(4_1)$ is a \kl{pullback}, by Lemma 2(h) of~\cite{nbSqPO2019}. In this case, the \kl{FPA} is trivial, and consequently so is the back-propagation process. Our rule composition can thus be seen as a conservative extension of that defined for \kl(crsType){linear} rules in~\cite{nbSqPO2019}. 
\end{example}

\section{Conclusions}\label{sec:conclusions}

\subsection{Concluding remarks}

We have presented a novel formalism for graph transformation and exploited this to provide generic results---the \kl{concurrency} and \kl{associativity theorems}---that characterize the categorical rewriting theories as compositional. %
These results have been proved once and for all and in a \emph{universal}, i.e., semantics-independent, fashion based upon our novel notion of \kl{compositional rewriting double categories (crDCs)}.
We have then investigated the conditions under which a variety of variants of both \kl{DPO-} and \kl{SqPO-semantics} yield \kl{crDCs}, and thus compositional categorical rewriting theories. %
In the case of \kl{SqPO-semantics}, we have in particular established that rewriting is compositional for fully general rules in the concrete setting of simple graphs. %
In the case of \kl{DPO-semantics}, our results also establish general conditions under which rewriting is compositional; it is worthwhile emphasizing that these conditions exclude the case of non-linear rules on \kl{directed simple graphs}. This failure (due concretely to \kl{$\mathbf{SGraph}$} not satisfying axiom \kl(notationVKb){(L-iii-b)}; cf.\ Table~\ref{tab:main}) appears to be indicative of the utility of our fine-grained analysis, in that the precise role played by each of the assumptions found to be sufficient for a given compositional rewriting semantics are highlighted, such that perhaps in future work alternative semantics or categorical constructions might be constructible in order to overcome these limitations. 

Our results are based on a new formalism that expresses the required categorical structure in terms of certain fibrational structures together with a small number of axioms specific to the case of rewriting. %
One particularly significant aspect of our approach is that it provides an intrinsic, and cognitively convenient, structuring of the very large number of lemmata used in graph transformation whose statements and proofs are scattered across the literature. %
We hope that this higher-level structuring will aid the process of formalizing this area of mathematics via proof assistants such as \texttt{Coq}~\cite{bertot2013interactive}, \texttt{Isabelle/HOL}~\cite{nipkow2002isabelle}, or \texttt{Lean}~\cite{conf/cade/MouraKADR15}. 
In this paper, the approach already enables more compact proofs and, indeed, we feel that the \kl{associativity theorem} would simply not be possible to express without these means. %
We are convinced that additional benefit may be extracted from this general setting.

\subsection{Relation to the extended abstract}

In the extended abstract of this paper~\cite{BehrHK21,BehrHK21-ext}, we showed that the important use case of transformation of \kl{directed simple graphs} under \kl{SqPO-semantics} requires the use of a restricted notion of matching---to \emph{edge-reflecting} injective homomorphisms, abstractly characterized by the so-called \emph{regular monos}---in order to prove the \kl{concurrency theorem} that provides a proper notion of rule composition. This led us to investigate more generally the categorical structure required to support rule composition, under the SqPO semantics, for fully general non-linear rules, and we established that \kl{quasi-topoi}~\cite{garner2012axioms,quasi-topos-2007,COCKETT2002223,adamek2006,COCKETT200361} naturally possess all the necessary structure.
In particular, our theorems do therefore apply to the category of \kl{directed simple graphs}---which forms a \kl{quasi-topos}~\cite{quasi-topos-2007,adamek2006} but fails to satisfy the axioms of \kl{adhesive} or \kl{quasi-adhesive} categories \cite{quasi-topos-2007}---for which no satisfactory prior \kl{concurrency theorem} has been proved\footnote{The setting of $\cM$-linear rules, where $\cM$ is the regular monos, prevents rules from deleting or adding edges and so has limited use in practice.}.
These results significantly generalized previous work on concurrency theorems for \kl(crsType){linear} \kl{SqPO rewriting} over \kl{adhesive categories}~\cite{nbSqPO2019}, and for \kl(crsType){linear} \kl{SqPO rewriting} with conditions in \kl{$\cM$-adhesive categories}~\cite{BK2020,behrRaSiR}.

This proof of the \kl{concurrency theorem} under the \kl{SqPO-semantics} relied on the existence of certain structures in quasi-topoi that, to the best of our knowledge, had not been previously noted in the literature: restricted notions of \kl{multi-sum} and \kl{multi-IPC} (mIPC) plus that of \kl{FPC-pushout-augmentation (FPA)}. The \kl{multi-sum} construction provides a generalization of the property of effective unions (in adhesive categories) that guarantees that all necessary monos are regular. The notions of \kl{mIPC} and \kl{FPA} handle the ``backward non-determinism'' introduced by non-linear rules: given a rule and a matching from its output graph, we cannot---unlike with linear or reversible non-linear rules---uniquely determine a matching from the input graph of the rule.

In the case of \kl{DPO-semantics}, we established (again in~\cite{BehrHK21,BehrHK21-ext}) a generalization of the \kl{concurrency theorem} to \kl(crsType){generic} \kl{DPO-semantics}, and presented \kl{rm-adhesive categories} as possessing sufficient properties in order to support this semantics. However, unlike in the case of SqPO-rewriting, this does not capture the case of generic DPO-rewriting in the category of simple graphs as this setting is well-known to satisfy only the weaker axioms of rm-quasi-adhesive categories. In the present paper, we have demonstrated that, slightly more generally, \kl{adhesive HLR categories} can also support \kl(crsType){generic} \kl{DPO-semantics}, thereby extending the range of applications of this type of semantics, although this still does not cover the important case of simple graphs.

In comparison to the much more technically involved ``direct'' proofs found in~\cite{BehrHK21,BehrHK21-ext}, the high-level abstraction offered by the novel fibrational approach to compositional rewriting theories permits the modularization of the proofs of the concurrency and associativity theorems in a very efficient manner. In particular, this relies upon making a clear separation of the concrete definitions of compositional rewriting theories, i.e., proving that a certain semantics and choice of base category gives rise to a \kl{compositional rewriting double category (crDC)}, from the universal structures offered by a crDC.

\subsection{Related work}

Conditions under which \kl{final pullback complements} (FPCs) are guaranteed to exist have been studied in~\cite{dyckhoff1987exponentiable} and also in~\cite{Corradini_2015} which provides a direct construction assuming the existence of appropriate \kl{$\cM$-partial map classifiers}~\cite{10.1007/978-3-642-15928-2_17,COCKETT200361}. We make additional use of these \kl{$\cM$-partial map classifiers} in order to construct \kl{multi-initial pushout complements}. This construction is a mild, but necessary for our purposes, generalization of the notion of minimal pushout complement defined in~\cite{BRAATZ2011246} that requires the universal property with respect to a larger class of encompassing pushouts---precisely analogous to the definition of \kl{FPC}---that additionally allows us to specify a family of solutions that collectively satisfy this universal property so as to handle the backward non-determinism of non-linear rules.

We also exploit the \kl(EMS){$\cE$-$\cM$-factorization} which every \kl{finitary} \kl{$\cM$-adhesive category} possesses~\cite{GABRIEL_2014}, where $\cE$ is the class of \kl(ssm){extremal morphisms} with respect to $\cM$ (cf.\ Theorem~\ref{thm:finitarityCats}). Noting that this factorization coincides with the well-known \kl(EMS){epi-regular mono-factorization}~\cite{adamek2006} in \kl{quasi-topoi}, we extended our approach from~\cite{BehrHK21,BehrHK21-ext} for constructing \kl{$\cM$-multi-sums} to the more general setting of \kl{finitary} \kl{$\cM$-adhesive categories} (cf.\ Lemma~\ref{lem:constrMS}). This makes it clear that the notion of \kl{$\cM$-multi-sum} is in fact nothing other than the well-known notion of \emph{($\cE'$, $\cM$)-pair factorizations} from the traditional literature on \kl{DPO-semantics} for \kl{$\cM$-adhesive categories}~\cite{ehrig:2006fund,GABRIEL_2014}. Moreover, we also find that the notion of \kl(EMS){$\cE$-$\cM$-factorization} permits a more general construction of \kl{FPC-pushout-augmentations (FPAs)} than the one we originally presented in~\cite{BehrHK21,BehrHK21-ext}, thereby opening up \kl(crsType){generic} \kl{SqPO-semantics} to categories other than \kl{quasi-topoi} (cf.\ Section~\ref{sec:exConstr:FPAs}). %

Overall, and as discussed throughout the paper, but in particular in full detail in Sections~\ref{sec:ccrts} and 6, our approach to determining suitable classes of categories supporting the various rewriting semantics relies heavily on the categorical rewriting and category theory literature, such as the traditional framework of Ehrig et al.\ that is based upon the notion of \kl{$\cM$-adhesive categories}~\cite{ehrig:2006fund}, but also a large number of works on a variety of other categories with \kl{adhesivity properties}~\cite{ls2004adhesive,lack2005adhesive,garner2012axioms,quasi-topos-2007,Corradini_2006,COCKETT2002223,ehrig2004adhesive,ehrig2010categorical,CORRADINI200543,COCKETT200361}.

While traditional results mostly concerned \kl(crsType){linear} or \kl(crsType){semi-linear} \kl{DPO-semantics}, for which \kl{$\cM$-adhesive categories} were found to pose a very general class of suitable categories~\cite{ehrig:2006fund,ehrig2010categorical}, and while every \kl{quasi-topos} is indeed also an \kl{$\cM$-adhesive category}~\cite[Lem.~13]{10.1007/978-3-642-15928-2_17} (for $\cM$ the class of regular monos), \kl{$\cM$-adhesive categories} or \kl{quasi-topoi} are in general \emph{not} \kl{adhesive HLR categories}, and may thus fail to support \kl(crsType){generic} \kl{DPO-semantics}. In particular, as mentioned above, the category \kl{$\mathbf{SGraph}$} of \kl{directed simple graphs} does \emph{not} support \kl(crsType){generic} \kl{DPO-semantics}. This is highly significant in view of practical applications: since restricting to edge-reflecting monomorphisms for defining the \kl(crs){rules} would mean that no edges could be created or deleted, in many practical applications, even those with no need for cloning or merging, it would be necessary to work with \emph{generic} monomorphisms, which consequently entails that one no longer has even a \kl{concurrency theorem} available for analyzing the resulting rewriting systems.

A notable application example in this regard is the \texttt{MØD} framework~\cite{andersen2016software}, which aims to implement a rewriting system capable of modeling organo-chemical reaction systems. As explained in further detail in~\cite{BK2020}, since this framework uses typed simple graphs with additional complex type- and degree-constraints for its base category, the operations of rule compositions in \texttt{MØD}~\cite{andersen2016software} would have to be considered strictly \textbf{not} mathematically consistent, since as demonstrated above $\mathbf{SGraph}$ does \textbf{not} support even a \kl{concurrency theorem} in \kl(crsType){generic} \kl{DPO-semantics} (and thus the same holds for typed and constrained variants of \kl{directed simple graphs}). This would be a technically severe failure, since ultimately organo-chemical rewriting is intended to faithfully model the \emph{continuous-time Markov chains (CTMCs)} that encode organo-chemical reaction systems; failure of being a compositional rewriting theory more explicitly would entail the absence of a suitable \emph{rule algebra} and \emph{stochastic mechanics} formalism, rendering a mathematically consistent formulation of the CTMCs impossible.

In summary, the work presented in the present paper represents a clear warning regarding the imprecise interpretation of rewriting semantics, and, more positively, provides a first step towards streamlining the framework of compositional rewriting theory such that verifying the consistency of applications of rewriting theory in a more transparent way.

\subsection{Future work}

The immediately preceding discussion about the failure of generic DPO-rewriting in certain concrete settings leads us naturally to the notion of $\cal{M},\cal{N}$-adhesive categories~\cite{10.1007/978-3-642-33654-6_15,habel2012m}: according to recent results of \cite[Thm.~3.1]{cgm:fossacs22}, $\mathbf{SGraph}$ (referred to as ``$\mathbf{DGraph}$'' in loc. cit.) is in fact \emph{($\regmono{\mathbf{SGraph}}$, $\mono{\mathbf{SGraph}}$)-adhesive} and \emph{($\mono{\mathbf{SGraph}}$, $\regmono{\mathbf{SGraph}}$)-adhesive}, i.e., carries two types of \emph{($\cM$,$\cN$)-adhesivity} structures. Based upon this type of property, as demonstrated in~\cite{10.1007/978-3-642-33654-6_15,habel2012m} one may modify the definition of \kl(crsType){generic} \kl{DPO-semantics} to a variant where the vertical morphisms in the definition of \kl(crsDPO){direct derivations} (cf.\ \eqref{eq:defCRSdpoMN}) are required to be in the class $\cN$, while the horizontal morphisms (i.e., those from which the rules are constructed) to be in the class $\cM\subseteq \mono{\bfC}$. For the concrete case of the category $\mathbf{SGraph}$, one possible choice of this modified type of DPO-semantics would be to let $\cN=\regmono{\mathbf{SGraph}}$ (depicted in the diagram below with $\rightarrowtail$ arrows) and $\cM=\mono{\mathbf{SGraph}}$ (depicted below\footnote{Coincidentally, this diagram is close in structure to the semantics of direct derivations in the aforementioned \texttt{MØD} framework, in that rules therein are in particular identities on the vertices of the graphs involved (albeit $\mathbf{SGraph}$ itself does of course not take into account the type- and degree-constraints relevant to organic chemistry).
} with $\hookrightarrow$ arrows):
\begin{equation}\label{eq:defCRSdpoMN}
\ti{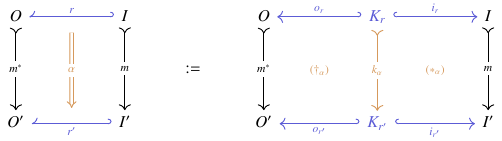}
\end{equation}

According to~\cite[Thm.~1.5]{habel2012m}, pushout complements in an ($\cM$,$\cN$)-adhesive category are essentially unique if they exist, so the aforementioned notion of modified DPO-semantics lies in a certain sense ``in between'' \kl(crsType){linear} and \kl(crsType){generic} \kl{DPO-semantics}. More importantly, it was demonstrated in~\cite{habel2012m} that this modified DPO-semantics over ($\cM$,$\cN$)-adhesive categories admits a concurrency theorem (and most of the other properties of adhesive HLR categories in a suitably modified form), hence we believe it would be highly interesting to submit this notion of rewriting to our novel analysis method, i.e., to determine if (or under which additional conditions) the modified DPO-rewriting semantics yields a \kl{compositional rewriting double category}.

Clearly, it would also be interesting to study some other graph transformation semantics, such as PBPO(+)~\cite{CORRADINI2019213,10.1007/978-3-030-78946-6_4} or AGREE~\cite{Corradini_2015}, from the new viewpoint of our fibrational approach. It would also be fruitful to investigate how the proofs of additional key theorems, such as local confluence, might carry over to this framework as this would increase our confidence in its general applicability.
In general, it will be highly desirable to develop a curated collection of mathematical techniques (ideally formalized in proof assistants) that will permit to efficiently construct and analyze categories used in rewriting theories with suitable \kl{adhesivity properties}, including the aforementioned generalized notion thereof, as well as \kl{quasi-topos} structures, such that for practical applications of compositional categorical rewriting theories the entry barrier posed by the considerably technically intricate theoretical framework may be considerably lowered. As already presented in Sections~\ref{sec:caps} and~\ref{sec:qt}, at present there already exists a certain amount of mathematical methodology in this regard, most notably \kl{comma category} constructions of categories with \kl{adhesivity properties}, and a variant thereof, so-called \emph{Artin gluing}~\cite{quasi-topos-2007}, for constructing \kl{quasi-topoi} (cf.\ also~\cite{cgm:fossacs22} for recent advances in constructing ($\cM$,$\cN$)-adhesive categories). On the other hand, as summarized in Table~\ref{tab:main}, we are left to wonder whether it is indeed \kl{adhesivity properties} that most generally characterize categories suitable for \kl{SqPO-semantics}, since this semantics does \emph{not} require the fully-fledged variant of the relevant van Kampen square axioms (i.e., only axioms \kl(notationVKa){(X-iii-a)}, but not axioms \kl(notationVKb){(X-iii-b)}), while on the other hand in order to support compositional SqPO semantics, it is required that the underlying category \kl{has pullbacks}, and that it \kl(ssm){has FPCs along $\cM$-morphisms}. It would thus be highly desirable to find a more fine-grained and better adapted classification scheme for categories supporting compositional rewriting semantics of various kinds, for which in the present paper we have provided a first stepping stone.

\subsection*{Acknowledgements}

The authors would like to thank Richard Garner for insightful comments and suggestions related to the original conference version of this paper~\cite{BehrHK21}. In particular, he provided us with a sketch of the results presented in Theorem~\ref{thm:domFunctorFP} of Section~\ref{sec:examplesFib}; the developments for the other boundary functors were inspired by this particular example. %
We are also grateful to Paul-André Melliès and Noam Zeilberger for very fruitful related discussions which motivated us to investigate the fibrational properties of source and target functors and also the notion of double categories. %
Finally, we would like to thank the anonymous referees for their detailed remarks and questions which have greatly improved the present manuscript.

%
%

%

\clearpage
\appendix

\section{Collection of definitions and auxiliary properties}\label{sec:aux}

\subsection{Universal properties}\label{app:UP}

\begin{lemma}\label{lem:UP}
Let $\bfC$ be a category.
\begin{equation*}
\ti{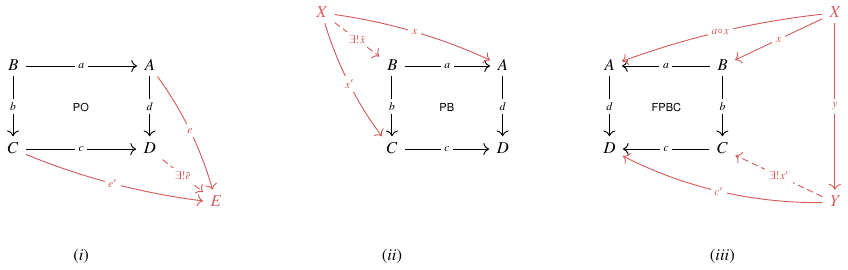}
\end{equation*}
Then the following properties hold:
\begin{enumerate}
\item \AP\emph{\intro{Universal property of pushouts (POs)}}: Given a commutative diagram as in $(i)$, there exists a unique morphism $D-\bar{e}\to E$ such that $\bar{e}\circ d=e$ and $\bar{e}\circ c=e'$.
\item \AP\emph{\intro{Universal property of pullbacks (PBs)}}: Given a commutative diagram as in $(ii)$, there exists a unique morphism $X-\bar{x}\to A$ such that $a\circ\bar{x}=x$ and $b\circ\bar{x}=x'$.
\item \AP\emph{\intro{Universal property of final pullback complements (FPCs)}}: Given a commutative diagram as in $(iii)$ where $(a\circ x.y)$ is a PB of $(d,c')$, there exists a unique morphism $Y-x'\to C$ such that $c\circ x'=c'$, $x'\circ y=b\circ x$, and which satisfies that $(x,y)$ is the PB of $(b,x')$. 
\end{enumerate}
\end{lemma}

\subsection{Stability properties}\label{sec:sp}

\begin{definition}\label{def:sPOFPC}
Let $\bfC$ be a category.
\begin{equation}\label{eq:stability}
\vcenter{\hbox{\ti{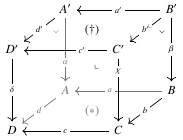}
}}
\end{equation}
\begin{itemize}
\item A pushout $(*)$ in $\bfC$ is said to be \AP\intro(PO){stable under pullbacks} iff for every commutative cube over the pushout $(*)$ such as in the diagram above where all vertical squares are pullbacks, the top square $(\dag)$ is a pushout.
\item A final pullback complement (FPC) $(*)$ in $\bfC$ is said to be \AP\intro(FPC){stable under pullbacks} iff for every commutative cube over the FPC $(*)$ such as in the diagram above where all vertical squares are pullbacks, the top square $(\dag)$ is an FPC.
\end{itemize}
\end{definition}

\begin{lemma}\label{lem:poSP}
Two important examples of categories for which suitable stability properties for pushouts hold are given as follows:
\begin{enumerate}
\item In every \kl{adhesive category} $\bfC$, pushouts along monomorphisms (i.e., pushouts such as $(*)$ in~\eqref{eq:stability} with $a\in \mono{\bfC}$ or $b\in\mono{\bfC}$) are stable under pullbacks~\cite{ls2004adhesive}. This property is indeed axiom \kl(notationVKa){(A-iii-a)} of the van Kampen property of adhesive categories~\cite{garner2012axioms}.
\item In a \AP\intro{regular mono (rm)-quasi-adhesive category}~\cite[Def.~1.1 and Cor.~4.7]{garner2012axioms}, all pushouts along regular monomorphisms exist, these pushouts are also pullbacks, and in particular \kl(ssm){pushouts along regular monomorphisms} are \kl(PO){stable under pullbacks}. A useful characterization of \kl{rm-quasi-adhesive categories} is the following: a small category $\bfC$ which \kl{has pullbacks} and which \kl(ssm){has pushouts along regular monomorphisms} is rm-quasi-adhesive iff it has a full embedding into a \kl{quasi-topos} (preserving the aforementioned two properties). 
\end{enumerate}
\end{lemma}

\begin{lemma}[\cite{reversibleSqPO}, Lem.~1]\label{lem:FPCs}
\AP Let $\bfC$ be a category that \kl{has pullbacks}. Then \intro{FPCs are stable under pullbacks}.
\end{lemma}

\begin{proposition}\label{prop:qtt}
In a \kl{quasi-topos} $\bfC$,  \AP\intro(qt){unions of regular subobjects are effective}~\cite[Prop.~10]{quasi-topos-2007}, i.e., the union of two subobjects is computed as the pushout of their intersection, and moreover the following property holds: in a commutative diagram such as below, where $(c,a)$ is the pullback of $(h,p)$, $(d,b)$ the pushout of $(c,a)$, where all morphisms (except possibly $x$) are monomorphisms, and where either $p\in \regmono{\bfC}$ or $h\in \regmono{\bfC}$, then the induced morphism $x:D\rightarrow E$ is a \emph{monomorphism}~\cite[Prop.~2.4]{garner2012axioms}:
\begin{equation}
\ti{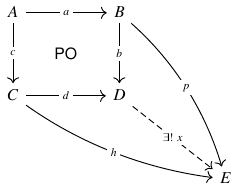}
\end{equation}
\end{proposition}

\subsection{Single-square lemmata specific to $\cM$-adhesive categories}\label{app:SSL}

\begin{lemma}\label{lem:dsl}
Let $\bfC$ be an \kl{$\cM$-adhesive category}.
\begin{equation}
\ti{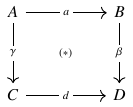}
\end{equation}
\begin{enumerate}
\item \AP\intro(Madh){Pushouts along $\cM$-morphisms are pullbacks}: if $(*)$ is a pushout and $\gamma \in \cM$, then $(*)$ is also a pullback.
\item \AP\intro(Madh){Stability of $\cM$-morphisms under pushout}: if $(*)$ is a pushout and $\gamma \in \cM$, then $\beta \in \cM$.
\end{enumerate}
\end{lemma}

\subsection{Double-square lemmata}\label{app:DSL}

\begin{lemma}\label{lem:dsl}
Let $\bfC$ be a category. If in the following statements a class of monics $\cM$ is mentioned, these statements require $\bfC$ to possess a \kl{stable system of monics} $\cM$.

\begin{equation}\label{eq:lemDSL}
\ti{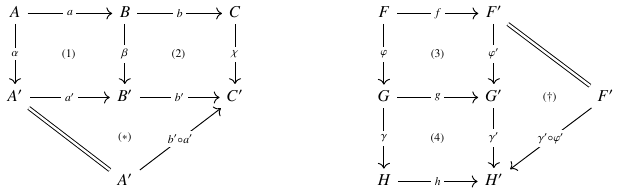}
\end{equation}
Given commutative diagrams as above, the following statements hold:
\begin{enumerate}
\item \AP \emph{\intro{Pushout-pushout-(de-)composition}}: if $(1)$ is a pushout, $(1)+(2)$ is a pushout iff $(2)$ is a pushout.
\item \AP\emph{\intro{Pullback-pullback-(de-)composition}}: if $(2)$ is a pullback, $(1)+(2)$ is a pullback iff $(1)$ is a pullback.
\item \AP \emph{\intro{Pushout-pullback-decomposition}}~\cite[Lem.~4]{reversibleSqPO}: if $(1)+(2)$ is a \textbf{stable pushout}\footnote{Here, ``stable'' refers to stability under pullbacks.} and $(1)$, $(2)$, $(*)$ are pullbacks, then $(1)$ and $(2)$ are both pushouts. (Note: If $a'$ and $b'$ are monomorphisms, the condition on $(*)$ is always satisfied.) Alternatively~\cite[Thm.~4.26.2]{ehrig:2006fund}, if $\bfC$ is a \kl{vertical weak adhesive HLR category} with respect to a \kl{stable system of monics} $\cM$, then if $a'$ and $b'$ are in $\cM$, $(1)+(2)$ is a pushout, and $(2)$ is a pullback, then $(1)$ and $(2)$ are both pushouts. If the category is in fact a \kl{weak adhesive HLR category} (i.e., if it also has \kl{horizontal weak VK squares}), then the decomposition also holds if instead of $a'$ the morphism $\alpha$ is in $\cM$.
\item \AP\emph{\intro{Pullback-pushout-decomposition}}~\cite[Lem.~B.2]{GOLAS2014}: if $\bfC$ is a \kl{vertical weak adhesive HLR category} with respect to a \kl{stable system of monics} $\cM$, $\chi$ is in $\cM$, $(1)+(2)$ is a pullback and $(1)$ is \textbf{stable pushout}, then $(1)$ and $(2)$ are both pullbacks.
\item \AP \emph{\intro{Horizontal FPC-FPC-(de-)composition}}~\cite[Lem.~2 \& 3]{Corradini_2006}, \cite[Prop.~36]{Loewe_2015}: if $(2)$ is an FPC (i.e., $(\beta,b')$ is an FPC of $(b,\chi)$), $(1)+(2)$ is an FPC iff $(1)$ is an FPC.
\item \AP \intro{Vertical FPC-FPC-(de-)composition}~\cite[Prop.~36]{Loewe_2015}: if $(3)$ is an FPC (i.e., $(\varphi,g)$ is an FPC of $(f,\varphi')$), 
\begin{enumerate}
\item if $(4)$ is an FPC (i.e., $(\gamma,h)$ is an FPC of $(g, \gamma')$), then $(3)+(4)$ is an FPC (i.e., $(\gamma\circ \varphi,h)$ is an FPC of $(f,\gamma'\circ\varphi')$)
\item if $(3)+(4)$ is an FPC (i.e., $(\gamma\circ \varphi,h)$ is an FPC of $(f,\gamma'\circ\varphi')$) and if $(4)$ is a pullback, then $(4)$  is an FPC (i.e., $(\varphi,g)$ is an FPC of $(f,\varphi')$).
\end{enumerate}
\item \AP \emph{\intro{Vertical FPC-pullback decomposition}}~\cite[Lem.~3]{reversibleSqPO}: if $(3)+(4)$ is an FPC (i.e., $(\gamma\circ \varphi,h)$ is an FPC of $(f,\gamma'\circ\varphi')$), both $(4)$ and $(\dag)$ are pullbacks, and if the diagram commutes, then $(3)$  is an FPC (i.e., $(\varphi,g)$ is an FPC of $(f,\varphi')$) and $(4)$ is an FPC (i.e., $(\gamma,h)$ is an FPC of $(g, \gamma')$). (Note: If $\gamma'$ and $\varphi'$ are monomorphisms, the condition on $(\dag)$ is always satisfied.)
\end{enumerate}
\end{lemma}

\section{Proofs not included in the main text}

\subsection{Proofs of Section 2}\label{app:ps2}

\noindent\textbf{Lemma~\ref{lem:ms1} (Multi-sum extension).}
Let $\bfC$ be a category that \kl(ms){has multi-sums} and that \kl{has pullbacks}. Then for every commutative diagram such as in~\eqref{eq:lem:multiSumExtension} below, where $A\rightarrow M\leftarrow B$ and $C\rightarrow N\leftarrow D$ are multi-sum elements, there exists a universal arrow $M\rightarrow N$ that makes the diagram commute.

\begin{equation}\label{eq:lem:multiSumExtension}
\ti{lemmaMultiSumExtension}
\end{equation}
\begin{proof}
Construct the diagram in~\eqref{eq:mselProof} below as follows:
\begin{itemize}
\item Since $\bfC$ has pullbacks, take a pullback in order to obtain the span $X\leftarrow P\rightarrow N$. Then by the \kl{universal property of pullbacks}, there exist morphisms $A\rightarrow P$ and $B\rightarrow P$.
\item Since $\bfC$ \kl(ms){has multi-sums}, %
and since the cospan $A\rightarrow P\leftarrow B$ and %
the morphism $P\rightarrow X$ provide a factorization of %
the cospan $A\rightarrow X \leftarrow B$, %
by the \kl(ms){universal property of multi-sums} there exists a unique morphism $M\rightarrow P$ that makes the diagram commute.
\item The morphism $M\rightarrow N$ claimed to exist is then obtained as the composition of the morphisms $M\rightarrow P$ and $P\rightarrow N$. Moreover, if $X\leftarrow P'\rightarrow N$ is another pullback of $X\leftarrow Y\rightarrow N$, by the \kl{universal property of pullbacks} there exists a unique isomorphism $P\rightarrow P'$; therefore, the composites of $M\rightarrow P$ and $P\rightarrow N$, and of $M\rightarrow P'$ and $P'\rightarrow N$, respectively, yield the same morphism $M\rightarrow N$, hence demonstrating unique existence, which concludes the proof.
\end{itemize}

\begin{equation}\label{eq:mselProof}
\ti{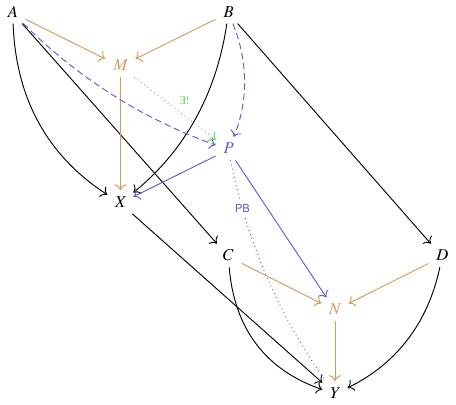}
\end{equation}

\qed \end{proof}

\noindent\textbf{Lemma~\ref{lem:mofIsoLifting}.}
Let $M:\bfE\rightarrow \bfB$ be a \kl(mof){strong} \kl{multi-opfibration}. Then the following \kl(mof){lifting property of isomorphisms} is satisfied:
\begin{equation}\label{eq:lem:mofIsoLifting}
\begin{aligned}
&\forall\;  
\ti{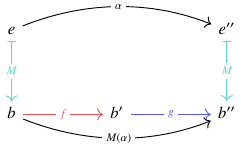}%
\; : \;\forall\; %
\ti{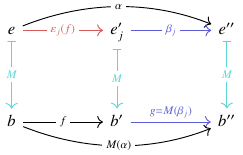}\;:\\
&\qquad  (g\in \iso{\bfB}\;\Rightarrow \; \beta_j\in \iso{\bfE}) \land 
(f\in \iso{\bfB}\;\Rightarrow \; \varepsilon_j(f)\in \iso{\bfE})
\end{aligned}
\end{equation}
\begin{proof}
Let us first consider the claim $g\in \iso{\bfB}\;\Rightarrow \; \beta_j\in \iso{\bfE}$:
\begin{itemize}
\item Since $M$ is a functor, we have that $M(e'')=b''$ implies $M(id_{e''})=id_{b''}$. 
\item Since $b'-g\rightarrow b''$ is by assumption an isomorphism in $\bfB$, there exists a morphism $b''-g^{-1}\rightarrow b'$ such that $id_{b''}=g\circ g^{-1}$. 
\item By the defining properties of \kl{multi-opfibrations}, this in turn entails that there exist morphisms $e''-\varepsilon_k(g^{-1})\rightarrow e'_k$ and  $e'_k-\gamma_k\rightarrow e''$ in $\bfE$ such that $\gamma_k\circ \varepsilon_k(g^{-1})= id_{e''}$, $M(\varepsilon_k(g^{-1}))=g^{-1}$ and $M(\gamma_k)=g$, as depicted in diagram~\eqref{eq:mofIsoLiftingProofA} below.
\end{itemize}
\begin{equation}\label{eq:mofIsoLiftingProofA}
\ti{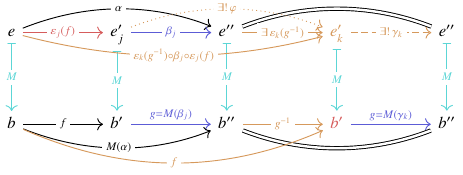}
\end{equation}
\begin{itemize}
\item Since $M$ is a functor, $M(\varepsilon_k(g^{-1})\circ \beta_j\circ \varepsilon_j(f)) = f$. Thus the diagram in~\eqref{eq:mofIsoLiftingProofA} encodes two different liftings of $g\circ f$ (one via $\beta_j\circ \varepsilon_j(f)$, and the other one via $\gamma_k\circ(\varepsilon_k(g^{-1})\circ \beta_j\circ \varepsilon_j(f))$), hence by \kl(mof){essential uniqueness} of \kl(mof){strong} \kl{multi-opfibrations}, there exists an isomorphism $e'_j-\varphi\rightarrow e'_k$ that makes the diagram commute, so that in particular $\varepsilon_k(g^{-1})\circ \beta_j$ is an isomorphism.
\item By standard category theory, (i) a morphism is an isomorphism iff it is both a section and a retraction~\cite[Prop.~7.36]{adamek2006}, (ii) if the composite $y\circ x$ of two morphisms $y$ and $x$ is an isomorphism, $x$ is a section and $y$ is a retraction~\cite[Prop.~7.21 \& 7.27]{adamek2006},  and (iii) the composite of two isomorphisms is an isomorphism~\cite[Prop.~3.14]{adamek2006}.
\begin{itemize}
\item Since $\varepsilon_k(g^{-1})\circ \beta_j$ is an isomorphism, $\varepsilon_k(g^{-1})$ is a retraction (and $\beta_j$ a section).
\item Since $\gamma_k\circ \varepsilon_k(g^{-1})$ is an identity morphism and thus an isomorphism, $\varepsilon_k(g^{-1})$ is a section (and $\gamma_k$ a retraction).
\end{itemize}
Since $\varepsilon_k(g^{-1})$ is thus both a section and a retraction, it is an isomorphism, hence $\beta_j= (\varepsilon_k(g^{-1}))^{-1}\circ\varphi$ is the composite of two isomorphisms and therefore an isomorphism, which proves the claim.
\end{itemize}

The proof of the claim $f\in \iso{\bfB}\;\Rightarrow \; \varepsilon_j(f)\in \iso{\bfE}$ is completely analogous, with the salient steps summarized in diagram~\eqref{eq:mofIsoLiftingPartB} below:
\begin{equation}\label{eq:mofIsoLiftingPartB}
\ti{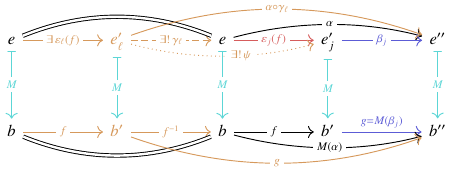}
\end{equation}
\begin{itemize}
\item $f$ being an isomorphism entails that $f^{-1}\circ f=id_b$, hence since $M$ is a functor and $M(e)=b$, $M(id_e)=id_b = f^{-1}\circ f$.
\item By the \kl(mof){universal property of multi-opfibrations}, there exists an $\bfE$-morphism $e-\varepsilon_{\ell}(f)\rightarrow e_{\ell}'$ such that there exists a unique $\bfE$-morphism $e_{\ell}'-\gamma_{\ell}\rightarrow e$ satisfying $M(\varepsilon_{\ell}(f))=f$, $M(\gamma_{\ell})=f^{-1}$, and $\gamma_{\ell}\circ \varepsilon_{\ell}(f)=id_e$. 
\item By the \kl(mof){essential uniqueness} property of \kl(mof){strong} \kl{multi-opfibrations}, there exists a unique isomorphism $e'_{\ell}-\psi\rightarrow e'_j$ that makes the diagram commute, and so that in particular $\varepsilon_j(f)\circ\gamma_{\ell}=\psi$.
\item Since $\varepsilon_j(f)\circ\gamma_{\ell}=\psi$ is an isomorphism, $\gamma_k$ is a section; since $\gamma_{\ell}\circ \varepsilon_{\ell}(f)=id_e$ is an identity morphism and thus an isomorphism, $\gamma_{\ell}$ is a retraction; since thus $\gamma_{\ell}$ is both a retraction and a section, it is an isomorphism.
\item Finally, since $\varepsilon_j(f)=\psi\circ \gamma_{\ell}^{-1}$, $\varepsilon_j(f)$ is an isomorphism, which concludes the proof.
\end{itemize}
\qed \end{proof}

\noindent\textbf{Lemma~\ref{lem:pbsmopf} (Pullback-lifting lemma for strong multi-opfibrations).}
Let $\bfE$ be a category that \kl{has pullbacks}, and let $M:\bfE\rightarrow \bfB$ be a \kl(mof){strong} \kl{multi-opfibration}. Then the following property holds:
\begin{equation}\label{eq:mofPBsplittingLemmaApp}
\forall\;\ti{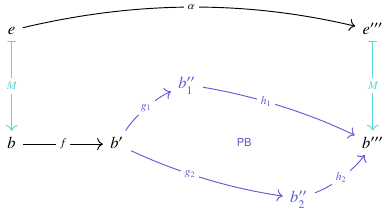}\;:\; 
\ti{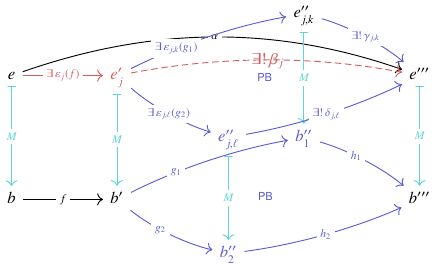}
\end{equation}
More explicitly, for every diagram such as on the left of~\eqref{eq:mofPBsplittingLemmaApp}, whose bottom part contains a pullback square in $\bfB$, the following properties hold true:
\begin{enumerate}[label=(\roman*)]
\item There exists an $\bfE$-morphism $e-\varepsilon_j(f)\rightarrow e_j'$ such that there exists a unique $\bfE$-morphism $e_j'-\beta_j\rightarrow e'''$ with $M(\varepsilon_j(f))=f$ and $M(\beta_j)=h_1\circ g_1=h_2\circ g_2$, and such that the diagram commutes.
\item There then exist $\bfE$-morphisms $e_j'-\varepsilon_{j,k}(g_1)\rightarrow e''_{j,k}$ and $e_j'-\varepsilon_{j,\ell}(g_2)\rightarrow e''_{j,\ell}$ such that there exist unique $\bfE$-morphisms $e''_{j,k}-\gamma_{j,k}\rightarrow e'''$ and $e''_{j,\ell}-\delta_{j,\ell}\rightarrow e'''$ such that $M(\varepsilon_{j,k}(g_1))=g_1$, $\varepsilon_{j,\ell}(g_2))=g_2$, $M(\gamma_{j,k})=h_1$ and $M(\delta_{j,\ell})=h_2$, and such that the diagram commutes.
\item Moreover, the square in $\bfE$ into $e'''$ is a pullback.
\end{enumerate}
\begin{proof}
Claims (i) and (ii) follow directly from repeated applications of the \kl(mof){universal property of multi-opfibrations}. It thus remains to prove claim (iii), i.e., that the square in $\bfE$ on the top right of the diagram in~\eqref{eq:mofPBsplittingLemmaApp} is indeed a pullback. To this end, we construct the auxiliary diagram below by taking a pullback:
\begin{equation}\label{eq:proofMOFpbSplittingLemma}
\ti{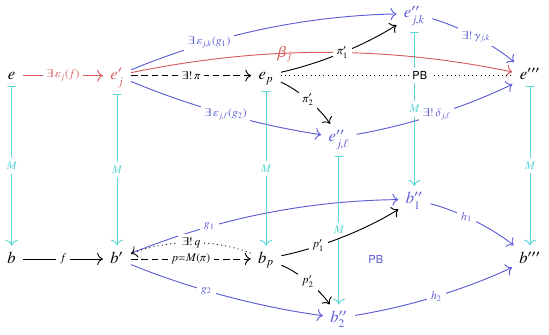}
\end{equation}
\begin{itemize}
\item By the \kl{universal property of pullbacks}, there exists an $\bfE$-morphism $e_j'-\pi\rightarrow e_p$ (where $e_p$ denotes the pullback object) that makes the diagram commute.
\item Since $M$ is a functor, we also obtain $\bfB$-morphisms $b'-p\rightarrow b_p$ (where $b_p=M(e_p)$), $p_1'=M(\pi_1')$ and $p_2'=M(\pi_2')$ that make the diagram commute.
\item By the \kl{universal property of pullbacks}, there exists a unique morphism $b_p-q\rightarrow b'$ that makes the diagram commute; since $(p\circ q)\circ p = p\circ ( q\circ p) = p$ and $p$ is unique, $(p\circ q)= id_{b'}$ and $q\circ p = id_b$, i.e., $p$ is both a section and a retraction, hence an isomorphism (and thus also $q=p^{-1}$).
\item Finally, by applying Lemma~\ref{lem:mofIsoLifting} for  $f=M(\pi)$ and $g=id_{b_p}$, we may demonstrate that $\pi$ is an isomorphism, hence indeed the claim that the square in $\bfE$ marked in blue is a pullback.
\end{itemize}
\qed \end{proof}

\noindent\textbf{Lemma~\ref{lem:rmopfUP}.}
Let $R:\bfE\rightarrow \bfB$ be a \kl{residual multi-opfibration}. Then \kl(rmof){residues} have the following \kl(rmofrup){universal property}:
\begin{equation}\label{eq:corUPrmofapp}
\forall \ti{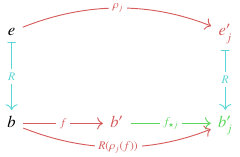}%
\; : \; \exists %
\ti{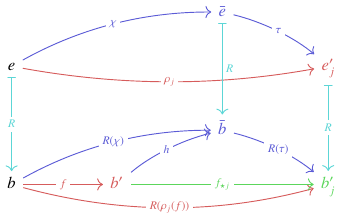}\;\Rightarrow \tau\in \iso{\bfE}\land R(\tau)\in \iso{\bfB}
\end{equation}
In particular, this property entails that if a residue $f_{\star k}$ factorizes a residue $f_{\star j}$ as $f_{\star j}=R(\beta_k)\circ f_{\star k}$ for some $\beta_k\in \bfE$, then the residues $f_{\star j}$ and $f_{\star_k}$ (both of the same morphism $f\in \bfB$) are related by an isomorphism $R(\beta_k)\in \iso{\bfB}$, as are their liftings $\rho_j(f)=\beta_k\circ \rho_k(f)$ via $\beta_k\in \iso{\bfE}$.
\begin{proof}
It suffices to restate the second diagram in~\eqref{eq:corUPrmofapp} in the following equivalent form:
\begin{equation}
\ti{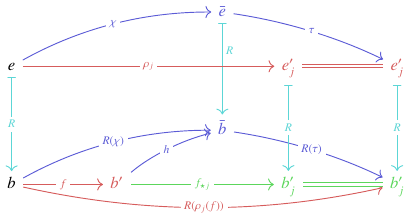}
\end{equation}
The claim then follows by \kl(rmof){essential uniqueness} of \kl{residual multi-opfibrations}.
\qed \end{proof}

\subsection{Proofs of Section 4}\label{app:ps4}

\noindent\textbf{Lemma~\ref{lem:PBPOFPCcats}.}
The categories $\mathsf{T}_h(\bfC,\cM)$ and $\mathsf{T}_v(\bfC,\cM)$ for $\mathsf{T}\in \{\mathsf{PB}, \mathsf{PO} , \mathsf{FPC}\}$ as introduced in Definition~\ref{def:cstCats} are well-defined, i.e., their composition operations are well-typed, associative and unital.
\begin{proof}
Well-definedness of the horizontal and vertical composition operations is a standard result for pullback and pushout squares, while for final pullback complements (FPCs) this is a slight generalization of Lemma~\ref{lem:dsl}(5. \& 6.):
\begin{itemize}
\item \emph{horizontal FPC composition:} in sub-diagram $(a)$ of~\eqref{eq:FPCcompositions} below, given two horizontally composed FPC squares, an outer square which is a pullback, and morphisms $X\rightarrow A$ and $X\rightarrow C$ such that the diagram commutes, we have to prove that there exists a unique morphism $Y\rightarrow A'$ such that the diagram commutes and such that the square over $Y\rightarrow A'$ is a pullback.
\begin{enumerate}
\item Since the square $XCC'Y$ is a pullback and the triangle $XBC$ commutes, by the \kl{universal property of FPCs} there exists a unique morphism $Y-b\rightarrow B'$ that makes the diagram commute, and such that the square $XBB'Y$ is a pullback.
\item Since the square $XBB'Y$ is a pullback and the triangle $XAB$ commutes, by the \kl{universal property of FPCs} there exists a unique morphism $Y-a\rightarrow A'$ that makes the diagram commute, and such that the square $XAA'Y$ is a pullback.
\end{enumerate}
\item \emph{vertical FPC composition:} in sub-diagram $(b)$ of~\eqref{eq:FPCcompositions} below, given two vertically composed FPC squares, an outer square which is a pullback, and morphisms $X\rightarrow A$ and $X\rightarrow B$ such that the diagram commutes, we have to prove that there exists a unique morphism $Y\rightarrow A''$ such that the diagram commutes and such that the square over $Y\rightarrow A''$ is a pullback.
\begin{enumerate}
\item Take a pullback of $Y-b\rightarrow B''\leftarrowtail B'$, obtaining the span $Y\leftarrowtail P\rightarrow B'$ (where $Y\leftarrowtail P$ is in $\cM$ by \kl(ssm){stability of $\cM$-morphisms under pullback}). Since the square $XBB'Y$ commutes, by the \kl{universal property of pullbacks}, there exists a unique morphism $X-x\rightarrow P$ that makes the diagram commute. 
\item By \kl{pullback-pullback decomposition}, the square $XBB'P$ over $P\rightarrow B'$ is a pullback. Since moreover the triangle $XAB$ commutes, by the \kl{universal property of FPCs}, there exists a unique morphism $P-p\rightarrow A'$ that makes the diagram commute, and such that the square $XAA'P$ over $P-p\rightarrow A'$ is a pullback.
\item Invoking the \kl{universal property of FPCs} yet again, since the triangle $PB'A'$ commutes and the square $PB'B''Y$ is a pullback, there exists a unique morphism $Y\rightarrow A''$ that makes the diagram commute, and such that the square $PA'A''Y$ over this morphism is a pullback. By \kl{pullback-pullback decomposition} (or, equivalently, by vertical pasting of pullback squares), the square $XAA''Y$ is a pullback, which concludes the proof.
\end{enumerate}
\end{itemize}
\begin{equation}\label{eq:FPCcompositions}
\ti{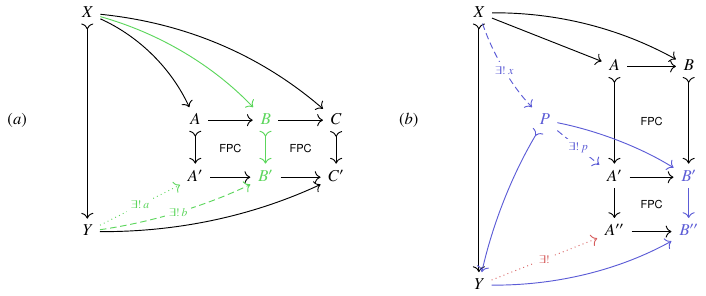}
\end{equation}
Associativity of the horizontal and vertical pasting operations is manifest from the definition. Thus it remains to verify that these compositions are unital. To this end, note first that the units of the horizontal and vertical compositions are squares of the form $id_m$ and $id_f$, respectively, as below:
\begin{equation}\label{eq:sqUnits}
\ti{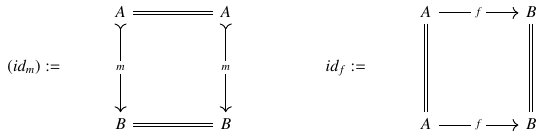}
\end{equation}
The only non-trivial part to prove is that squares of shapes $id_m$ and $id_f$ are simultaneously pullbacks, pushouts and FPCs. The first two properties are a standard exercise to prove, yet the proof of the FPC property deserves a brief clarification:
\begin{equation}\label{eq:FPCidProof}
\ti{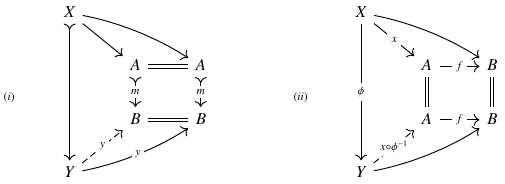}
\end{equation}
As depicted in~\eqref{eq:FPCidProof}(i), horizontal unitality trivially follows from \kl{pullback-pullback decomposition}. However, in order to prove that vertical unitality holds, in a situation as depicted in~\eqref{eq:FPCidProof}(ii), the additional observation necessary is that isomorphisms are stable under pullbacks, from which then together with \kl{pullback-pullback decomposition} the claim follows.
\qed \end{proof}

\noindent\textbf{Theorem~\ref{thm:domFunctorFP}.}
Let $\bfC$ be a category with a \kl{stable system of monics} $\cM$, and with the following additional properties:
\begin{enumerate}
\item $\bfC$ \kl{has pullbacks}.
\item $\bfC$ \kl(ssm){has pushouts and final pullback complements (FPCs) along $\cM$-morphisms}. 
\item Pushouts along $\cM$-morphisms are \kl(domCsqPBT){stable under pullbacks}.
\item \kl(domCsqPBT){Pushouts along $\cM$-morphisms are pullbacks}.
\end{enumerate}
Then the domain functor $dom:\textsf{PB}_h(\bfC,\cM)\rightarrow \bfC$ from the category of pullback squares along $\cM$-morphisms and under horizontal composition to the underlying category $\bfC$ satisfies the following properties:

\begin{enumerate}[label=(\roman*)]
\item \kl(cmtSq){$dom:\textsf{PB}_h(\bfC,\cM)\rightarrow \bfC$ is a Grothendieck fibration}  $dom:\textsf{PB}_h(\bfC,\cM)\rightarrow \bfC$ is a \kl{Grothendieck fibration}, with the \kl(gf){Cartesian liftings} given by FPCs.
\item \kl(cmtSq){$dom:\textsf{PB}_h(\bfC,\cM)\rightarrow \bfC$ is a Grothendieck opfibration}  $dom:\textsf{PB}_h(\bfC,\cM)\rightarrow \bfC$ is a \kl{Grothendieck opfibration}, with the \kl(gopf){op-Cartesian liftings} given by pushouts.
\item \kl(cmtSq){$dom:\textsf{PB}_h(\bfC,\cM)\rightarrow \bfC$ satisfies a Beck-Chevalley condition (BCC)}: adopting the notation $m-(f,f')\rightarrow n$ for morphisms in $\mathsf{PB}_h(\bfC,\cM)$ (cf.\ Figure~\ref{fig:boundaryFunctors}), consider a commutative square in $\mathsf{PB}_h(\bfC,\cM)$ that is mapped by $dom$ into a pullback square in $\bfC$:
\begin{equation}\label{eq:BCCpremise}
\ti{BCCpremise}
\end{equation}
Then the following two equivalent conditions hold:
\begin{itemize}
\item \kl(cmtSq){(BCC-1)}: $(f,f')$ is \kl(gopf){op-Cartesian} if $(i,i')$ is \kl(gopf){op-Cartesian} and $(g,g')$ and $(h,h')$ are \kl(gf){Cartesian}.
\item \kl(cmtSq){(BCC-2)}: $(g,g')$ is \kl(gf){Cartesian} if $(h,h')$ is \kl(gf){Cartesian} and $(f,f')$ and $(i,i')$ are \kl(gopf){op-Cartesian}.
\end{itemize}
\end{enumerate}
\begin{proof} 
\emph{Ad (iii) --- Beck-Chevalley condition (BCC):}

\begin{itemize}
\item \kl(cmtSq){(BCC-1)}: the premise of this condition is explicitly depicted in~\eqref{eq:proof:BCC-1}, i.e., the top and back squares are pullbacks, the left and right squares are FPCs, and the front square is a pushout. In order to demonstrate that this entails that the back square is a pushout, we take a pullback of the cospan $C'-i'\rightarrow D'\leftarrow g'-B'$, obtaining a span $C'\leftarrow P\rightarrow D'$.
\begin{itemize}
\item By the \kl{universal property of pullbacks}, there exist unique morphisms $A\rightarrow P$ and $A'\rightarrow P$ as indicated with dashed arrows.
\item By \kl{pullback-pullback decomposition}, the squares over $C'\leftarrow P$ and over $P\rightarrow B'$ are both pullbacks.
\item Since by assumption pushouts are \kl(domCsqPBT){stable under pullbacks}, the square over $P\rightarrow B'$ is a pushout.
\item Since according to Lemma~\ref{lem:FPCs} in a category such as $\bfC$ which by assumption \kl{has pullbacks}, FPCs are \kl(FPC){stable under pullbacks}, the square over $C'\leftarrow P$ is an FPC.
\item By the \kl{universal property of FPCs}, the arrow $A'\rightarrow P$ is an isomorphism, hence the back square is a pushout, which proves \kl(cmtSq){(BCC-1)}.
\end{itemize}
\end{itemize}

\begin{equation}\label{eq:proof:BCC-1}
\ti{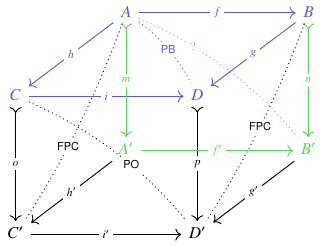} \xrightarrow{\text{take PB}}
\ti{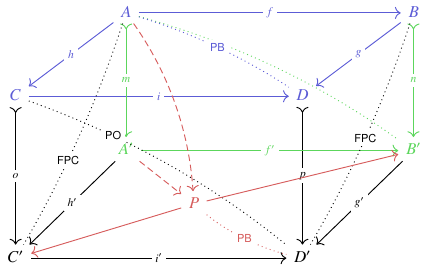}
\end{equation}

\begin{itemize}
\item \kl(cmtSq){(BCC-2)}: the premise of this condition is explicitly depicted in~\eqref{eq:proof:BCC-2}, i.e., the top and right squares are pullbacks, the left square is an FPC, and the front and back squares are pushouts. In order to demonstrate that this entails that the right square is an FPC, we take the final pullback complement of the sequence of morphisms $B-g\rightarrow D\rightarrowtail D'$ (which is admissible since by assumption $\bfC$ \kl(ssm){has FPCs along $\cM$-morphisms}), obtaining a sequence of morphisms $B\rightarrow F\rightarrow D'$.
\begin{itemize}
\item Since the front and left squares are morphisms in $\mathsf{PB}_h(\bfC,\cM)$ and thus pullbacks, by \kl{pullback-pullback composition} so is the vertical diagonal square $ADD'A'$ that arises as the composite of the front and left squares. %
Thus by the \kl{universal property of FPCs}, there exists a unique morphism $A'\rightarrow F$ as indicated with a dashed arrow. Since the back square is a pushout, by the \kl{universal property of pushouts} there exists a unique morphism $B'\rightarrow F$, again indicated with a dashed arrow.
\item Noting that the resulting configuration corresponds precisely to the precondition of \kl(cmtSq){(BCC-1)}, we find that the square $ABFA'$ over $A'\rightarrow F$ is a pushout. Thus by the \kl{universal property of pushouts}, $B'\rightarrow F$ is an isomorphism, and thus the right square is an FPC, which proves \kl(cmtSq){(BCC-2)}.
\end{itemize}
\end{itemize}

\begin{equation}\label{eq:proof:BCC-2}
\ti{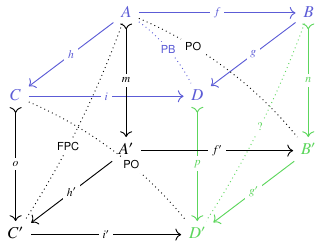} \xrightarrow{\text{take FPC}}
\ti{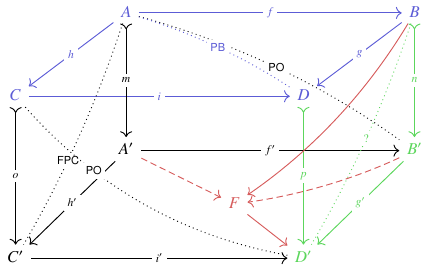}
\end{equation}

\qed \end{proof}

\noindent\textbf{Theorem~\ref{thm:trgtFPC-FP}.}
Let $\bfC$ be a category with a \kl{stable system of monics}  $\cM$ and that \kl(ssm){has FPCs along $\cM$-morphisms}. %
Then the target functor $T: \mathsf{FPC}_v(\bfC,\cM)\rightarrow \bfC\vert_{\cM}$ is a \kl{Grothendieck opfibration}.
\begin{proof}
Let us first provide the claim in more explicit form, i.e., by instantiating the defining properties of a \kl{Grothendieck opfibration} (cf.\ equation~\eqref{eq:def:grothendieckOpfibration}) to the case at hand, where we use the shorthand notation $T$ for the target functor:
\begin{equation}\label{eq:proof:trgtFPCgopfClaim}
\ti{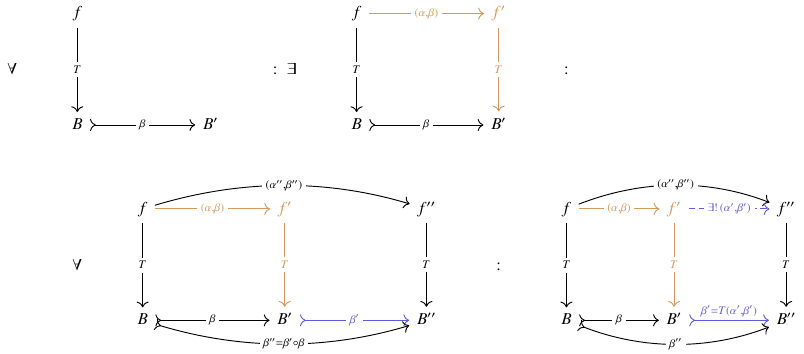}
\end{equation}
Recalling the definition of the category $\mathsf{FPC}_h(\bfC,\cM)$ (Definition~\ref{def:cstCats}), we may further expand the claim into the following equivalent form:

\begin{equation}\label{eq:proof:trgtFPCgopfClaimExplicit}
\ti{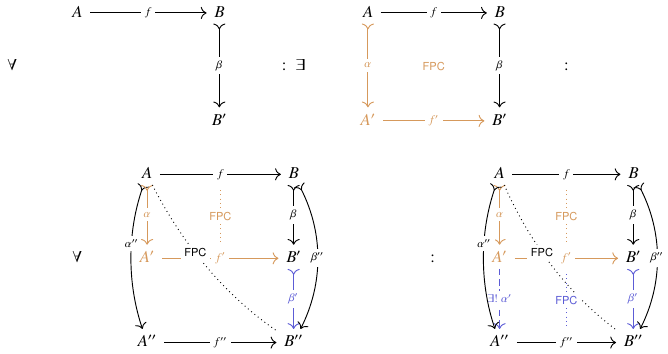}
\end{equation}
The first part of the claim, i.e., the existence of suitable liftings follows since $\bfC$ by assumption \kl(ssm){has FPCs along $\cM$-morphisms}. In order to prove the claim that these liftings have the \kl(gopf){op-Cartesianity} property, we take a pullback to arrive at the diagram below:
\begin{equation}\label{eq:proof:trgtFPCgopfClaimExplicit}
\ti{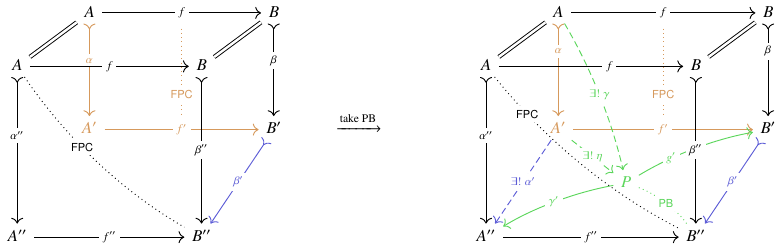}
\end{equation}
\begin{itemize}
\item By \kl{pullback-pullback decomposition}, the square over $g'$ is a pullback.
\item By \kl(ssm){stability of $\cM$-morphisms under pullback}, $P-\gamma'\rightarrow A''$ is an $\cM$-morphism. Since $A-\alpha''\rightarrow A''$ is in $\cM$ and $\alpha''=\gamma'\circ \gamma$, by \kl(ssm){decomposition property of $\cM$-morphisms}, we find that $\gamma\in \cM$.
\item Invoking \kl{vertical FPC-pullback decomposition}, the square under $P$ and the square over $g'$ are FPCs.
\item By the \kl{universal property of FPCs}, there exists thus a morphism $A'-\eta\rightarrow P$, which  by the \kl{universal property of FPCs} is an isomorphism.
\end{itemize}
Up until this point, we have proved that there exists a morphism $\alpha'=\gamma'\circ\eta$, and that the square under $\alpha'$ is an FPC. It remains to prove uniqueness of $\alpha'$. To this end, upon closer inspection of the second diagram in~\eqref{eq:proof:trgtFPCgopfClaimExplicit}, since $\beta'$ is in $\cM$ and thus in particular a monomorphism, the right vertical square is a pullback, and thus by \kl{pullback-pullback composition}, the composite of the right and back vertical squares is a pullback. Therefore, we may identify $\alpha'$ as the morphism that according to the \kl{universal property of FPCs} is guaranteed to exist (mediating before the aforementioned pullback square and the FPC in the front vertical square), and that is moreover unique as per the universal property.
\qed \end{proof}

\noindent\textbf{Lemma~\ref{prop:mIPCexistence}.}
Let $\bfC$ be a category with a \kl{stable system of monics} $\cM$. If \kl(ssm){pushouts along $\cM$-morphisms are stable under $\cM$-pullbacks}, %
and if \kl(lemMIPCexistence){pushouts along $\cM$-morphisms are pullbacks}, %
then $\bfC$ \kl{has multi-initial pushout complements (mIPCs) along $\cM$-morphisms}.
\begin{proof}
By definition, for every composable sequence of morphisms $A-f\rightarrow B\rtail\beta\rightarrow B'$ (i.e., with $\beta\in \cM$), the \kl{multi-initial pushout complement} $\mIPC{f}{\beta}$ consists of all composable sequences of morphisms $A\rtail\alpha\rightarrow A'-f'\rightarrow B'$ such that the resulting commutative square is a pushout:
\begin{equation}
\mIPC{f}{\beta} := \{(A\rtail\alpha\rightarrow A', A'-f'\rightarrow B')\in \mor{\bfC} \times \mor{\bfC} \mid \alpha\in \cM \land (f',\beta) =\pO{\alpha,f}\}\,.
\end{equation}
Since this class may in general be empty, it is non-trivial to prove the \kl(mIPC){universal property} of \kl{mIPCs}. To this end, let us construct the diagrams below:
\begin{equation}
\ti{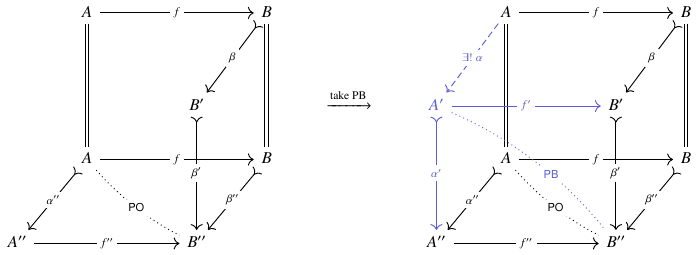}
\end{equation}
Here, the left diagram encodes the premise of the \kl(mIPC){universal property of mIPCs}. The \emph{existence} part of the universal property may be demonstrated as follows:
\begin{itemize}
\item Taking a pullback as indicated to obtain the right diagram (which is admissible since by assumption $\bfC$ \kl(ssm){has pullbacks along $\cM$-morphisms}), we obtain morphisms $f'$, $\alpha$ and $\alpha'$.
\item By \kl(ssm){stability of $\cM$-morphisms under pullback}, $\alpha'$ is in $\cM$. Since $\alpha''=\alpha'\circ \alpha$ is in $\cM$ as well, by the \kl(ssm){decomposition property of $\cM$-morphisms}, we find that $\alpha\in \cM$.
\item Since $\alpha$ and $\beta$ are in $\cM$, the right and left vertical squares are pullbacks. The back vertical square is a pullback, since all squares of this form are such.
\item By assumption, \kl(ssm){pushouts along $\cM$-morphisms are stable under $\cM$-pullbacks}, hence the top square is a pushout. Thus by \kl{pushout-pushout decomposition}, so is the front square.
\end{itemize}
We have thus exhibited an element of $\mIPC{f}{\beta}$. %

It remains to prove the \emph{essential uniqueness} property of \kl{mIPCs}. Suppose we were given another pair of vertically composable pushouts as follows:
\begin{equation*}
\ti{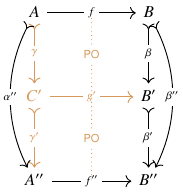}
\end{equation*}
By assumption, \kl(lemMIPCexistence){pushouts along $\cM$-morphisms are pullbacks}, hence the pushout square $C'B'B''A''$ is also a pullback, which by the \kl{universal property of pullbacks} entails the existence of a unique isomorphism $C'\rightarrow A'$.
\qed \end{proof}

\noindent\textbf{Theorem~\ref{thm:FPCfact}.}
Let $\bfC$ be a category with a \kl{stable system of monics} $\cM$, that is \kl{($\cE$,$\cM$)-structured}, that \kl(ssm){has pushouts and FPCs along $\cM$-morphisms}, such that \kl(ssm){$\cM$-morphisms are stable under pushout}, and such that \kl(ssm){pushouts along $\cM$-morphisms are stable under $\cM$-pullbacks}.  %
Then the category $\mathsf{FPC}_v(\bfC, \cM)$ is \kl{(auto-augmented, inert)-structured}. Here, the class of \kl(FPCfact){auto-augmented FPCs}  is defined as
\begin{equation}\label{eq:defAAfpc}
\ti{defAAfpA} \in \mor{\mathsf{FPC}_v(\bfC,\cM)}\vert_{auto-augmented}
\; :\Leftrightarrow\; \exists\;
\ti{defAAfpcB}
\end{equation}
In words: an FPC square along an $\cM$-morphism (seen as a morphism in $\mathsf{FPC}_v(\bfC,\cM)$ is auto-augmented iff when taking a pushout of the span within the FPC, the mediating morphism into the cospan object of the FPC is a morphism in $\cE$.\footnote{Note that since we admit arbitrary morphisms of $\bfC$ for the horizontal morphisms, the mediating morphism would in general be a morphism with a non-trivial \kl(EMS){$\cE$-$\cM$-factorization}, hence for this morphism to be an $\cE$-morphism is indeed a non-trivial requirement.}
Moreover, the class of \kl(FPCfact){inert FPCs} is defined as 
\begin{equation}\label{eq:defIfpc}
\mor{\mathsf{FPC}_v(\bfC,\cM)}\vert_{inert} :=
\left.\left\lbrace\ti{defIfpc}\;\right\vert \alpha \in \cE\cap \cM = \iso{\bfC}\right\rbrace
\end{equation}
\begin{proof}
In order to demonstrate that the two classes of morphisms are both \kl(EMS){closed under composition with isomorphisms}, note first that this is true by definition for the \kl(FPCfact){inert FPCs}. For the \kl(FPCfact){auto-augmented FPCs}, it is useful to observe the following auxiliary fact about isomorphisms in $\mathsf{FPC}_v(\bfC, \cM)$:
\begin{equation}\label{eq:isosFPCv}
\forall
\ti{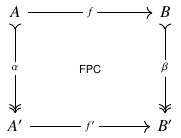} \in \iso{\mathsf{FPC}_v(\bfC,\cM)}
\; :\;
\ti{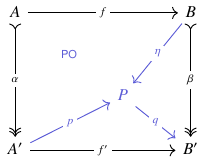}\;\Rightarrow \eta,q\in \cM\cap \cE=\iso{\bfC}
\end{equation}
In words: for every isomorphism in $\mathsf{FPC}_v(\bfC, \cM)$, which is an FPC square where the vertical morphisms are in $\cM\cap\cE=\iso{\bfC}$, if we take a pushout $A' \rightarrow P \leftarrow \eta- B$ of its span $A'\leftarrow \alpha - A\rightarrow B$, which by the \kl{universal property of pushouts} yields a mediating morphism $P-q\rightarrow B'$, then $\eta$  is in $\cM\cap\cE=\iso{\bfC}$ since isomorphisms are stable under pushout, hence $q = \beta\circ \eta^{-1}$ is the composite of two isomorphisms, and thus itself an isomorphism. Consequently, we find that isomorphisms in $\mathsf{FPC}_v(\bfC,\cM)$ are also pushout squares, and they are moreover both \kl(FPCfact){auto-augmented FPCs} and \kl(FPCfact){inert FPCs}. To conclude that the class of \kl(FPCfact){auto-augmented FPCs} is closed under composition with isomorphisms, it suffices then to consider the following diagrams:
\begin{equation}\label{eq:FPCvCompIsoAux}
(i)\; \ti{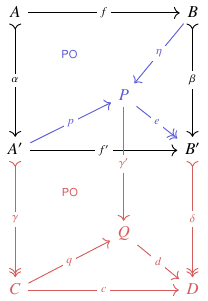}
\qquad\qquad  (ii)\; \ti{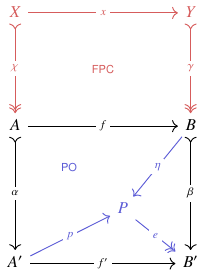}
\end{equation}
Diagram $(i)$ above demonstrates that post-composing an isomorphism with an  \kl(FPCfact){auto-augmented FPC}, and taking the indicated pushouts, the morphism $\gamma'$ in~\eqref{eq:FPCvCompIsoAux}(i) is an isomorphism by stability of isomorphisms under pushout, hence in particular also an $\cE$-morphism; thus the morphism $d$ which satisfies $d\circ \gamma'=\delta\circ e$ is an $\cE$-morphism, which proves that the composite square is indeed an \kl(FPCfact){auto-augmented FPC}. For diagram~\eqref{eq:FPCvCompIsoAux}(ii), which illustrates the pre-composition of an \kl(FPCfact){auto-augmented FPC} with an isomorphism, since isomorphisms in $\mathsf{FPC}_v(\bfC,\cM)$ are as demonstrated above also pushout squares, we find that the vertical composition of the pushout squares in~\eqref{eq:FPCvCompIsoAux}(ii) yields a pushout square with mediating morphism $e$ that is an $\cE$-morphism, which demonstrates that the pre-composition of an \kl(FPCfact){auto-augmented FPC} with an isomorphism yields an \kl(FPCfact){auto-augmented FPC}.

The next part of the proof amounts to showing that $\mathsf{FPC}_v(\bfC,\cM)$ \kl(EMS){has (auto-augmented, inert)-factorizations of morphisms}. To this end, for every morphism in $\mathsf{FPC}_v(\bfC,\cM)$, i.e., for an FPC square as in diagram~\ref{eq:proof:FPCvFact} below, we exhibit a factorization into an \kl(FPCfact){auto-augmented FPC} and an \kl(FPCfact){inert FPC} as follows:
\begin{itemize}
\item Take a pushout, thus obtaining a mediating morphism $P-p\rightarrow B'$.
\item Applying the \kl(EMS){$\cE$-$\cM$-factorization} to $p$ yields an $\cE$-morphism $P-e\rightarrow E$ and an $\cM$-morphism $E-m\rightarrow B'$.
\begin{itemize}
\item Since every square of the form as the square under $A'- e\circ g\rightarrow E$ is a pullback, by \kl{pullback-pullback decomposition} also the square over $e\circ g$ is a pullback. Thus by \kl{vertical FPC-pullback decomposition}, both the squares over and under $e\circ g$ are FPC squares.
\item By the \kl(ssm){decomposition property of $\cM$-morphisms}, since $\beta=m\circ (e\circ \pi)$ and $m$ are in $\cM$, $(e\circ \pi)$ is in $\cM$.
\end{itemize}
We thus confirm that the top subdiagram over $e\circ g$ manifestly has the structure of an  \kl(FPCfact){auto-augmented FPC}, while the bottom subdiagram over $e\circ g$ encodes an \kl(FPCfact){inert FPC}.
\end{itemize}

\begin{equation}\label{eq:proof:FPCvFact}
\ti{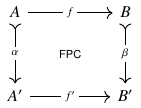}\; \xrightarrow{\text{take PO}}\;
\ti{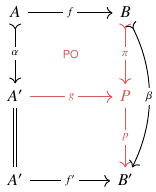}\; \xrightarrow{\text{(epi,$\cM$)-fact.}}\;
\ti{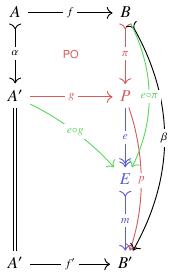}
\end{equation}
It thus remains to prove that $\mathsf{FPC}_v(\bfC,\cM)$ \kl(EMS){has a unique (auto-augmented,inert) diagonalization property}. More explicitly, considering a diagram as in the left of~\eqref{eq:proof:FPCvUDprop} below, where the top square is  an \kl(FPCfact){auto-augmented FPC}, while the bottom square is an \kl(FPCfact){inert FPC}:
\begin{equation}\label{eq:proof:FPCvUDprop}
\ti{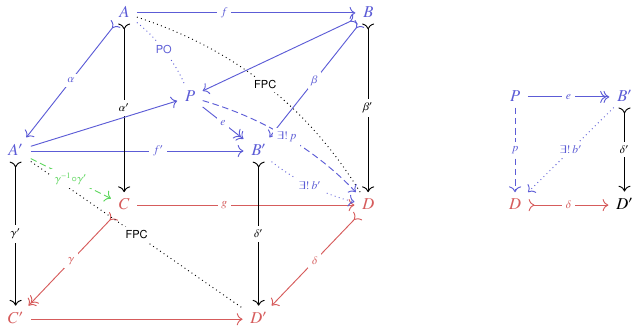}
\end{equation}
\begin{itemize}
\item Since by definition of  \kl(FPCfact){inert FPCs} $\gamma$ is an isomorphism, we obtain a morphism $ \gamma^{-1}\circ \gamma'$, which is moreover in $\cM$ (since also $\gamma'$ is in $\cM$ by definition of morphisms in $\mathsf{FPC}_v(\bfC,\cM)$).
\item The existence of the morphism $ \gamma^{-1}\circ \gamma'$ in turn reveals that there exists a cospan $A'\xrightarrow{g\circ \gamma^{-1}\circ \gamma'} D\xleftarrow{\beta'}B$, which by the \kl{universal property of pushouts} entails the existence of a morphism $P-p\to D$. 
\item Extracting the subdiagram as in the right of~\eqref{eq:proof:FPCvUDprop} above, we find that by the \kl(EMS){unique $(E,M)$-diagonalization property} there exists a unique morphism $B'-b'\to D$ such that $\delta\circ b' =\delta'$; the latter then entails by the  \kl(ssm){decomposition property of $\cM$-morphisms} that $b'$ is in $\cM$.
\item Since the front and bottom squares are FPCs and thus in particular pullbacks, and since the diagram commutes, by \kl{pullback-pullback decomposition} the diagonal square containing the morphisms $\gamma^{-1}\circ \gamma'$ and $b'$ is a pullback.
\item Applying the two variants of \kl{vertical FPC-FPC decomposition}, we may then conclude that the diagonal square containing the morphisms $\gamma^{-1}\circ \gamma'$ and $b'$ is the unique FPC square that simultaneously decomposes both the back and the front FPCs, thus concluding the proof.
\end{itemize}
\qed\end{proof}

\noindent\textbf{Theorem~\ref{thm:SourceFPCvRMOF}.}
Let $\bfC$ be a category with a \kl{stable system of monics} $\cM$, %
that is \kl{($\cE$, $\cM$)-structured}, %
that \kl(ssm){has pullbacks, pushouts and FPCs along $\cM$-morphisms}, %
such that \kl(ssm){$\cM$-morphisms are stable under pushout}, %
and such that \kl(ssm){pushouts along $\cM$-morphisms are stable under $\cM$-pullbacks}. %
Then $S:\mathsf{FPC}_v(\bfC,\cM)\rightarrow \bfC\vert_{\cM}$ is a \kl{residual multi-opfibration}.
\begin{proof}
Let us first utilize the assumptions on the underlying category $\bfC$ in order to provide the following construction on FPC squares that are morphisms in $\mathsf{FPC}_v(\bfC,\cM)$ as described in~\eqref{eq:FPCvRMOPFaux} below:
\begin{itemize}
\item Taking a pushout of the span $A'\leftarrow \alpha- A-f\rightarrow B$, we obtain a cospan $A'-p\rightarrow P\leftarrow \beta-B$ where $\beta\in \cM$ by \kl(ssm){stability of $\cM$-morphisms under pushout}, as well as a unique mediating morphism $P-\beta'\rightarrow B''$.
\item Applying \kl(EMS){$\cE$-$\cM$-factorization} to $\beta'$, we obtain an $\cE$-morphism $P-e\rightarrow E$ and an $\cM$-morphism $E-m\rightarrow B''$ such that $\beta'=m\circ e$. 
\item Taking a pullback of the cospan $A''-f''\rightarrow B''\leftarrow m-E$, we obtain a span $A''\leftarrow \iota' - I -i\rightarrow E$, where by \kl(ssm){stability of $\cM$-morphisms under pullback} we find that $\iota'\in \cM$, and a unique mediating morphism $A'-\iota\rightarrow I$, which by the \kl(ssm){decomposition property of $\cM$-morphisms} is also in $\cM$.
\item Finally, applying \kl{pullback-pullback decomposition} followed by \kl{vertical FPC-FPC decomposition}, we may conclude that the bottom square is an FPC, and so is the vertical composite of the top and middle square. Moreover, the middle square is a \kl{FPC-pushout-augmentation} for the top (pushout) square.
\end{itemize}
\begin{equation}\label{eq:FPCvRMOPFaux}
\ti{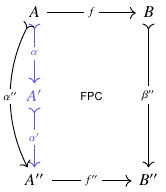}\;\xrightarrow{\text{take PO}}\;
\ti{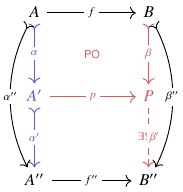}\;\xrightarrow[\text{\& take PB}]{\text{$\cE$-$\cM$-fact.}}\;
\ti{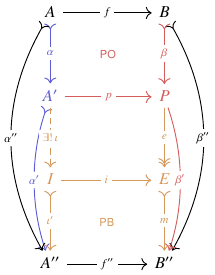}
\end{equation}
It is clear that the vertical decomposition of the original FPC square provided via the above construction is essentially unique, in the sense that both the \kl(EMS){$\cE$-$\cM$-factorization} as well as the pullback taken in the last two steps are unique only up to isomorphisms.\footnote{Note that also the pushout taken in the first step is unique only up to isomorphisms, yet the pushout itself is not retained as part of the data of the vertical decomposition into two FPC squares, hence in this sense does not contribute to the effective ``degrees of freedom'' of the construction.} More explicitly, we have the following chain of arguments demonstrating the existence of unique isomorphisms mediating between any two vertical FPC decompositions obtained via the above procedure:
\begin{itemize}
\item If as in~\eqref{eq:proof:FPCvRMOPFeu} below $A'-q\rightarrow Q\leftarrow \gamma-B$ is another pushout of $A'\leftarrow \alpha-A-f\rightarrow B$, yielding also a unique mediating morphism $Q-\gamma'\rightarrow B''$, by the \kl{universal property of pushouts} there exists a unique isomorphism $P-\pi\rightarrow Q$ such that $\gamma \circ f = \pi\circ \beta\circ f = q\circ \alpha=\pi\circ p\circ \alpha$.%
\item If $Q-\tilde{e}\rightarrow \tilde{E}-\tilde{m}\rightarrow B''$ %
is an \kl(EMS){$\cE$-$\cM$-factorization} of $Q-\gamma'\rightarrow B''$, %
since $Q - e\circ \pi^{-1}\rightarrow E -m\rightarrow B''$ is another \kl(EMS){$\cE$-$\cM$-factorization} of $Q-\gamma'\rightarrow B''$, %
and since \kl(EMS){$\cE$-$\cM$-factorizations are essentially unique}, there exists a unique isomorphism $E-\varepsilon\rightarrow \tilde{E}$ such that $\varepsilon\circ e = \tilde{e}\circ \pi$ and $\tilde{m}\circ \varepsilon = m$.
\item Finally, if $A''\leftarrow \tilde{\iota}'-\tilde{I}-\tilde{i}\rightarrow \tilde{E}$ is a pullback of $A''-f''\rightarrow B''\leftarrow \tilde{m}-\tilde{E}$, with $A'-\tilde{\iota}\rightarrow \tilde{I}$ the unique mediating morphism, by the \kl{universal property of pullbacks}, there exists a unique isomorphism $I-\varphi\rightarrow \tilde{I}$ that makes the diagram commute.
\end{itemize}

\begin{equation}\label{eq:proof:FPCvRMOPFeu}
\ti{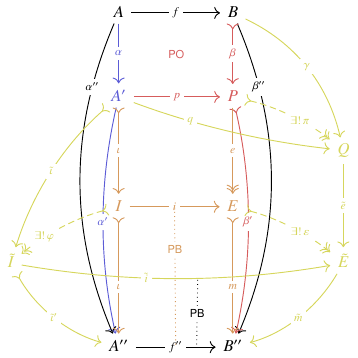}
\end{equation}
In summary, we have proved that for a given span of the form $A'\leftarrow \alpha-A-f\rightarrow B$ with $\alpha$ in $\cM$, the class of \kl{FPAs} for any pushout over $(\alpha,f)$ provides an explicit construction of \kl(rmof){residual multi-op-Cartesian liftings} (with \kl(rmof){residues} realized via \kl{FPAs}; compare~\eqref{eq:def:residualMultiOpfibration}), while the above-mentioned arguments demonstrate that this construction indeed yields the requisite vertical decomposition property of FPC squares up to residues (thus realizing the \kl(rmof){universal property of residual multi-opfibrations}) in an \kl(rmof){essentially unique} form.
\qed \end{proof}

\subsection{Proofs of Section 5}

\noindent\textbf{Lemma~\ref{lem:constrMS} (\cite{ehrig:2006fund}; \cite{GABRIEL_2014}, Fact~A.3.7).}
Let $\bfC$ be a \kl{finitary} \kl{vertical weak adhesive HLR category} with respect to a \kl{stable system of monics} $\cM$, and denote by $\cE$ the class of \kl(ssm){extremal morphisms} with respect to $\cM$.
\begin{enumerate}[label=(\roman*)]
\item \kl(msEM){Existence:} If $\bfC$ has binary coproducts, then every cospan of $\cM$-morphisms $A\rtail a\rightarrow Z\leftarrow b \ltail B$ factors essentially uniquely through a cospan of $\cM$-morphisms $A\rtail y_A\rightarrow Y\leftarrow y_B \ltail B$ and an $\cM$-morphism $Y\rtail m\rightarrow Z$, where $m$ is obtained via the \kl(EMS){$\cE$-$\cM$-factorization} $A+B-e\twoheadrightarrow Y \rtail m\rightarrow Z$ of the induced morphism $A+B-[a,b]\rightarrow Z$, and where $y_A=e\circ in_A$ and $y_B=e\circ in_B$.
\item \kl(msEM){Construction:} if $\bfC$ in addition has an \kl(ssm){$\cM$-initial object} $\mIO$, then $\Msum{A}{B}$ consists of cospans of $\cM$-morphisms obtained as pushouts $A\rtail p_A \rightarrow P\leftarrow p_B\ltail B$ of $\cM$-spans $A\leftarrow x_A \ltail X\rtail x_B\rightarrow B$ (i.e., ``$\cM$-partial overlaps'') extended by $\cE$-morphisms $P-q\twoheadrightarrow Q$ such that $q_A=q\circ p_A$ and $q_B=q\circ p_B$ are in $\cM$.
\item \kl(msEM){Refinements:} if $\bfC$ in addition \kl{has pullbacks}, and if \kl(ssm){pushouts along $\cM$-morphisms} in $\bfC$ are \kl(PO){stable under pullbacks}, then the extension morphisms $P-q\twoheadrightarrow Q$ are morphisms in $\cE\cap \mono{\bfC}$ (so-called ``refinements'').
\end{enumerate}
\begin{proof}
Even though most of this proof is in principle ``folklore'' in the graph rewriting literature~\cite{ehrig:2006fund,GABRIEL_2014}, we provide full details here, since we wish to demonstrate the claims in the generality stated, plus the presented statement regarding \kl(msEM){refinements} is a slight generalization of the corresponding statement in~\cite{BehrHK21}. As depicted in the diagram below left, the \kl(msEM){existence} of \kl{$\cM$-multi-sums} is guaranteed via \kl(EMS){$\cE$-$\cM$-factorization} of the induced morphism $A+B-[a,b]\rightarrow Z$, where $y_A=e\circ in_A$ and $y_B=e\circ in_B$ are in $\cM$ by the \kl(ssm){decomposition property of $\cM$-morphisms}:
\begin{equation}\label{eq:multiSumDef}
\ti{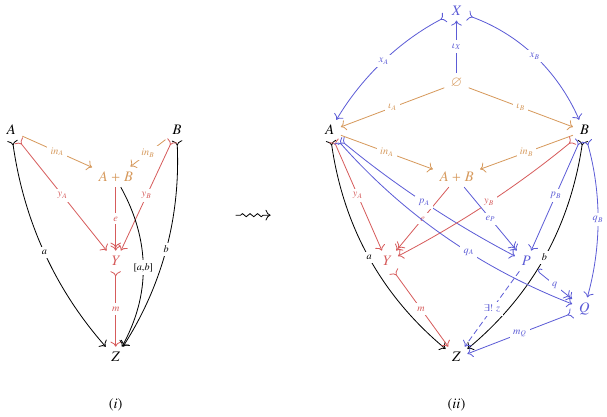}
\end{equation}
In order to prove that the \kl(msEM){construction} provided for the case that $\bfC$ has an \kl(ssm){$\cM$-initial object} $\mIO$ (i.e., in addition to being a \kl{vertical weak adhesive HLR category}) is sound and characterizes the \kl{$\cM$-multi-sums} in $\bfC$ uniquely, consider diagram $(b)$ in~\eqref{eq:multiSumDef}.
\begin{itemize}
\item We first demonstrate that for every pushout $A\rtail p_A \rightarrow P\leftarrow p_B\ltail B$ of an $\cM$-span $A\leftarrow x_A \ltail X\rtail x_B\rightarrow B$, the induced morphism $A+B-[p_A,p_B]\rightarrow P$ is an $\cE$-morphism. To this end, construct the commutative cube below, where the top square is a pushout (cf.\ e.g.\ \cite[Fact~2.6]{GABRIEL_2014}), the bottom square is a pushout:
\begin{equation}\label{eq:msConstrProofA}
\ti{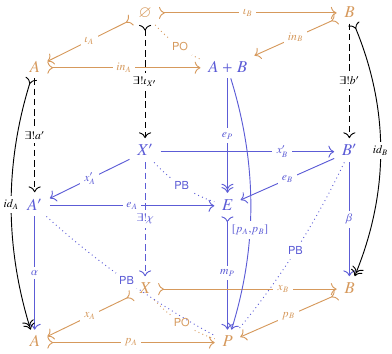}
\end{equation}
\begin{itemize}
\item Construct an \kl(EMS){$\cE$-$\cM$-factorization} $A+B-e_P\twoheadrightarrow E \rtail m_P\rightarrow P$ of the induced morphism $A+B-[p_A,p_B]\rightarrow P$.
\item Since $\bfC$ \kl(ssm){has pullbacks along $\cM$-morphisms}, we can take pullbacks in order to obtain the bottom front and right vertical squares in~\eqref{eq:msConstrProofA}, which by the \kl{universal property of pullbacks} induces unique morphisms $A-a'\rightarrow A'$ and $B-b'\rightarrow B'$, and by \kl(ssm){stability of $\cM$-morphisms under pullback}, all morphisms of the two pullback squares are in $\cM$.
\item Since $\bfC$ has an \kl(EMS){$\cE$-$\cM$-factorization}, isomorphisms such as in particular identity morphisms are both in $\cM$ and in $\cE$; since $id_A=\alpha\circ a'$, $id_B=\beta\circ b'$, by \kl(ssm){extremality} $\alpha$ and $\beta$ being $\cM$-morphisms implies that they are isomorphisms, and hence so are $a'$ and $b'$.
\item Take another pullback to obtain the middle horizontal square in~\eqref{eq:msConstrProofA}, which induces the unique $\cM$-morphism $\mIO\rtail \iota_{X'}\rightarrow X$ (since $\mIO$ is an \kl(ssm){$\cM$-initial object}). Since according to Corollary~\ref{cor:adhPOPB} pushouts along $\cM$-morphisms are pullbacks, we also obtain a unique morphism $X'-x\rightarrow X$ via the \kl{universal property of pullbacks}.
\item By \kl{pullback-pullback decomposition}, the bottom back and bottom left vertical squares are pullbacks. By stability of isomorphisms under pullback, $X'-x\rightarrow X$ is an isomorphism.
\item Since the bottom square is a \kl{vertical weak VK square}, and since all four vertical squares in the bottom half of the diagram are \kl(ssm){pullbacks along $\cM$-morphisms}, the middle horizontal square is a pushout. Thus by the \kl{universal property of pushouts}, $m_P$ is an isomorphism, which proves that $A+B-[p_A,p_B]\rightarrow P$ is an $\cE$-morphism.
\end{itemize}
\item It remains to demonstrate that any factorization of a cospan of $\cM$-morphisms $A\rtail a\rightarrow Z\leftarrow b\ltail B$ obtained via \kl(EMS){$\cE$-$\cM$-factorization} of the induced morphism $A+B-[a,b]\rightarrow Z$ may be equivalently obtained via extending a pushout $A\rtail p_A \rightarrow P\leftarrow p_B\ltail B$ of some $\cM$-span $A\leftarrow x_A \ltail X\rtail x_B\rightarrow B$ with an $\cE$-morphism $P-e\twoheadrightarrow Q$. To this end, consider yet again diagram $(b)$ in~\eqref{eq:multiSumDef}:
\begin{itemize}
\item Take a pullback of $A\rtail a\rightarrow Z\leftarrow b\ltail B$ to obtain a span $A\leftarrow x_A \ltail X\rtail x_B\rightarrow B$, which by \kl(ssm){stability of $\cM$-morphisms under pullback} is a span of $\cM$-morphisms.
\item Take a pushout $A\rtail p_A \rightarrow P\leftarrow p_B\ltail B$ of $A\leftarrow x_A \ltail X\rtail x_B\rightarrow B$, which by the \kl{universal property of pushouts} yields a unique morphism $P-z\rightarrow Z$. Moreover, since $\bfC$ is a \kl{vertical weak adhesive HLR category}, both $p_A$ and $p_B$ are $\cM$-morphisms.
\item Take an \kl(EMS){$\cE$-$\cM$-factorization} $P-q\twoheadrightarrow Q\rtail m_Q\rightarrow Z$.
\begin{itemize}
\item Let $q_A=q\circ p_A$ and $q_B=q\circ p_B$; since both $\cM$ and $\cE$ are \kl(EMS){closed under composition}, $q_A$ and $q_B$ are in $\cM$, while $q\circ e_P$ is in $\cE$.
\item By \kl(EMS){essential uniqueness} of \kl(EMS){$\cE$-$\cM$-factorizations}, there exists a unique isomorphism $Y\rightarrow Q$.
\end{itemize}
\end{itemize}
This concludes the proof of the soundness and completeness of our \kl(msEM){construction} for \kl{$\cM$-multi-sums}. 
\end{itemize}
Finally, let us consider the claim regarding \kl(msEM){refinements}, whereby if $\bfC$ in addition to being a \kl{vertical weak adhesive HLR category} also \kl{has pullbacks}, and moreover satisfies the property that pushouts along $\cM$-morphisms are \kl(PO){stable under pullbacks}, then the morphism $P-q\rightarrow Q$ of the above \kl(msEM){construction} of \kl{$\cM$-multi-sums} is both in $\cE$ and a monomorphism. To this end, first consider diagram $(i)$ in equation~\eqref{eq:proofMEmsMonoE} below:
\begin{itemize}
\item By \kl{pullback-pullback decomposition}, the left and back vertical squares in the bottom of diagram $(i)$ are pullbacks, thus by stability of isomorphisms under pullback, $Q'$ is isomorphic to $X$, hence also the upper left and back vertical squares are pullbacks.
\item By \kl{pullback-pushout decomposition}, the upper front and right vertical squares are pullbacks.
\end{itemize}

To finish the proof, construct diagram $(ii)$ in~\eqref{eq:proofMEmsMonoE} via taking a pullback (i.e., the square under $R$):
\begin{itemize}
\item As the span $\langle id_A, p_A\rangle$ is a pullback of the cospan $\rangle q_A,q\langle $, by \kl{pullback-pullback decomposition} $\langle id_A,r_A\rangle $ is a pullback of $\rangle a,p_R\langle$. Analogously, as $\langle id_B, p_B\rangle$ is a pullback of the cospan $\rangle q_B,q\langle $, by \kl{pullback-pullback decomposition} $\langle id_B,r_B\rangle $ is a pullback of $\rangle b,p_R\langle$.
\item Since the inner bottom horizontal square (i.e., the square marked $\mathsf{PO}$ into $P$) is a pushout of $\cM$-morphisms, and the vertical squares over its boundary are all pullbacks, by the assumed \kl(PO){stability under pullbacks} the inner horizontal middle square is a pushout.
\item By the \kl{universal property of pullbacks}, we find that $R-p_R\rightarrow P$ is an isomorphism (and thus also $P-r_P\rightarrow R$).
\end{itemize}
We have thus proved that the span $\langle id_P, id_P\rangle$ is a pullback of the cospan $\rangle q,q\langle$, which entails that $q$ is a monomorphism.

\begin{equation}\label{eq:proofMEmsMonoE}
\ti{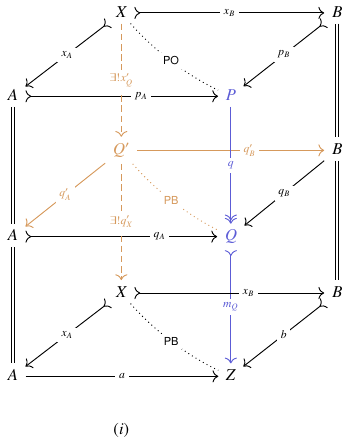}\qquad
\ti{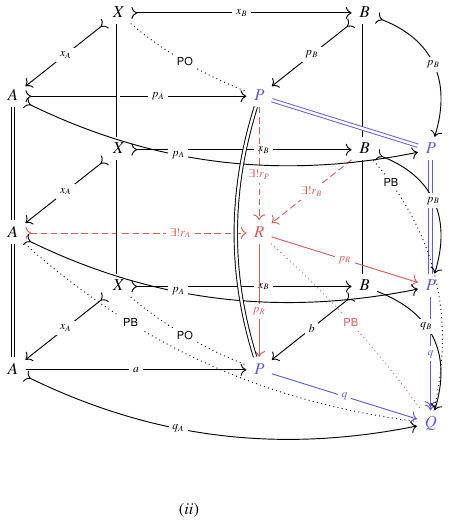}
\end{equation}
\qed\end{proof}

\end{document}